\documentclass[sigplan,screen,pbalance=true]{acmart}
\usepackage[]{stmaryrd}
\usepackage[]{amsthm}
\usepackage[]{polytable}
\usepackage[]{subcaption}
\usepackage[]{tikz}
\usepackage[capitalize]{cleveref}
\newwrite\boxesfile\immediate\openout\boxesfile=suppl.boxes
\newsavebox{\marxupbox}\usepackage{newunicodechar}\usepackage{mdframed}\usepackage[notextcomp,notext,nomath]{stix}\DeclareFontEncoding{LS1}{}{}
\DeclareFontSubstitution{LS1}{stix}{m}{n}
\DeclareSymbolFont{arrows2}{LS1}{stixsf}{m}{it}
\DeclareMathSymbol{\lefttail}{\mathrel}{arrows2}{"B2}


\acmBooktitle{Unpublished}
\newunicodechar{ }{\,}
\newunicodechar{¬}{\ensuremath{\neg}} %
\newunicodechar{²}{^2}
\newunicodechar{³}{^3}
\newunicodechar{·}{\ensuremath{\cdot}}
\newunicodechar{¹}{^1}
\newunicodechar{×}{\ensuremath{\times}} %
\newunicodechar{÷}{\ensuremath{\div}} %
\newunicodechar{ʲ}{^j}
\newunicodechar{ʳ}{^r}
\newunicodechar{ˡ}{^l}
\newunicodechar{̷}{\not} %
\newunicodechar{Γ}{\ensuremath{\Gamma}}   %
\newunicodechar{Δ}{\ensuremath{\Delta}} %
\newunicodechar{Η}{\ensuremath{\textrm{H}}} %
\newunicodechar{Θ}{\ensuremath{\Theta}} %
\newunicodechar{Λ}{\ensuremath{\Lambda}} %
\newunicodechar{Ξ}{\ensuremath{\Xi}} %
\newunicodechar{Π}{\ensuremath{\Pi}}   %
\newunicodechar{Σ}{\ensuremath{\Sigma}} %
\newunicodechar{Φ}{\ensuremath{\Phi}} %
\newunicodechar{Ψ}{\ensuremath{\Psi}} %
\newunicodechar{Ω}{\ensuremath{\Omega}} %
\newunicodechar{α}{\ensuremath{\mathnormal{\alpha}}}
\newunicodechar{β}{\ensuremath{\beta}} %
\newunicodechar{γ}{\ensuremath{\mathnormal{\gamma}}} %
\newunicodechar{δ}{\ensuremath{\mathnormal{\delta}}} %
\newunicodechar{ε}{\ensuremath{\mathnormal{\varepsilon}}} %
\newunicodechar{ζ}{\ensuremath{\mathnormal{\zeta}}} %
\newunicodechar{η}{\ensuremath{\mathnormal{\eta}}} %
\newunicodechar{θ}{\ensuremath{\mathnormal{\theta}}} %
\newunicodechar{ι}{\ensuremath{\mathnormal{\iota}}} %
\newunicodechar{κ}{\ensuremath{\mathnormal{\kappa}}} %
\newunicodechar{λ}{\ensuremath{\mathnormal{\lambda}}} %
\newunicodechar{μ}{\ensuremath{\mathnormal{\mu}}} %
\newunicodechar{ν}{\ensuremath{\mathnormal{\mu}}} %
\newunicodechar{ξ}{\ensuremath{\mathnormal{\xi}}} %
\newunicodechar{π}{\ensuremath{\mathnormal{\pi}}}
\newunicodechar{π}{\ensuremath{\mathnormal{\pi}}} %
\newunicodechar{ρ}{\ensuremath{\mathnormal{\rho}}} %
\newunicodechar{σ}{\ensuremath{\mathnormal{\sigma}}} %
\newunicodechar{τ}{\ensuremath{\mathnormal{\tau}}} %
\newunicodechar{φ}{\ensuremath{\mathnormal{\phi}}} %
\newunicodechar{χ}{\ensuremath{\mathnormal{\chi}}} %
\newunicodechar{ψ}{\ensuremath{\mathnormal{\psi}}} %
\newunicodechar{ω}{\ensuremath{\mathnormal{\omega}}} %
\newunicodechar{ϕ}{\ensuremath{\mathnormal{\phi}}} %
\newunicodechar{ϵ}{\ensuremath{\mathnormal{\epsilon}}} %
\newunicodechar{ᵏ}{^k}
\newunicodechar{ᵢ}{\ensuremath{_i}} %
\newunicodechar{ᶿ}{^\theta}
\newunicodechar{ }{\quad}
\newunicodechar{†}{\dagger}
\newunicodechar{ }{\,}
\newunicodechar{′}{\ensuremath{^\prime}}  %
\newunicodechar{″}{\ensuremath{^\second}} %
\newunicodechar{‴}{\ensuremath{^\third}}  %
\newunicodechar{ⁱ}{^i}
\newunicodechar{⁵}{\ensuremath{^5}}
\newunicodechar{⁺}{\ensuremath{^+}} 
\newunicodechar{⁻}{\ensuremath{^-}} 
\newunicodechar{ⁿ}{^n}
\newunicodechar{₀}{\ensuremath{_0}} %
\newunicodechar{₁}{\ensuremath{_1}} 
\newunicodechar{₂}{\ensuremath{_2}} 
\newunicodechar{₃}{\ensuremath{_3}}
\newunicodechar{₊}{\ensuremath{_+}} 
\newunicodechar{₋}{\ensuremath{_-}} 
\newunicodechar{ₙ}{_n} %
\newunicodechar{ℂ}{\ensuremath{\mathbb{C}}} %
\newunicodechar{ℒ}{\ensuremath{\mathscr{L}}}
\newunicodechar{ℕ}{\mathbb{N}} %
\newunicodechar{ℚ}{\ensuremath{\mathbb{Q}}}
\newunicodechar{ℝ}{\ensuremath{\mathbb{R}}} %
\newunicodechar{ℤ}{\ensuremath{\mathbb{Z}}} %
\newunicodechar{ℳ}{\mathscr{M}}
\newunicodechar{⅋}{\ensuremath{\parr}} %
\newunicodechar{←}{\ensuremath{\leftarrow}} %
\newunicodechar{↑}{\ensuremath{\uparrow}} %
\newunicodechar{→}{\ensuremath{\rightarrow}} %
\newunicodechar{↔}{\ensuremath{\leftrightarrow}} %
\newunicodechar{↖}{\nwarrow} %
\newunicodechar{↗}{\nearrow} %
\newunicodechar{↝}{\ensuremath{\leadsto}}
\newunicodechar{↦}{\ensuremath{\mapsto}}
\newunicodechar{⇆}{\ensuremath{\leftrightarrows}} %
\newunicodechar{⇐}{\ensuremath{\Leftarrow}} %
\newunicodechar{⇒}{\ensuremath{\Rightarrow}} %
\newunicodechar{⇔}{\ensuremath{\Leftrightarrow}} %
\newunicodechar{∀}{\ensuremath{\forall}}   %
\newunicodechar{∂}{\ensuremath{\partial}}
\newunicodechar{∃}{\ensuremath{\exists}} %
\newunicodechar{∅}{\ensuremath{\varnothing}} %
\newunicodechar{∈}{\ensuremath{\in}}
\newunicodechar{∉}{\ensuremath{\not\in}} %
\newunicodechar{∋}{\ensuremath{\ni}}  %
\newunicodechar{∎}{\ensuremath{\qed}}
\newunicodechar{∏}{\prod}
\newunicodechar{∑}{\sum}
\newunicodechar{∗}{\ensuremath{\ast}} %
\newunicodechar{∘}{\ensuremath{\circ}} %
\newunicodechar{∙}{\ensuremath{\bullet}} 
\newunicodechar{∝}{\propto}
\newunicodechar{∞}{\ensuremath{\infty}} %
\newunicodechar{∣}{\ensuremath{\mid}} %
\newunicodechar{∧}{\wedge}%
\newunicodechar{∨}{\vee}%
\newunicodechar{∩}{\ensuremath{\cap}} %
\newunicodechar{∪}{\ensuremath{\cup}} %
\newunicodechar{∫}{\int}
\newunicodechar{∷}{::} %
\newunicodechar{∼}{\ensuremath{\sim}} %
\newunicodechar{≃}{\ensuremath{\simeq}} %
\newunicodechar{≅}{\ensuremath{\cong}} %
\newunicodechar{≈}{\ensuremath{\approx}} %
\newunicodechar{≜}{\ensuremath{\stackrel{\scriptscriptstyle {\triangle}}{=}}} 
\newunicodechar{≝}{\ensuremath{\stackrel{\scriptscriptstyle {\text{def}}}{=}}}
\newunicodechar{≟}{\ensuremath{\stackrel {_\text{\textbf{?}}}{\text{\textbf{=}}\negthickspace\negthickspace\text{\textbf{=}}}}} 
\newunicodechar{≠}{\ensuremath{\neq}}%
\newunicodechar{≡}{\ensuremath{\equiv}}%
\newunicodechar{≤}{\ensuremath{\le}} %
\newunicodechar{≥}{\ensuremath{\ge}} %
\newunicodechar{⊂}{\ensuremath{\subset}} %
\newunicodechar{⊃}{\ensuremath{\supset}} %
\newunicodechar{⊆}{\ensuremath{\subseteq}} %
\newunicodechar{⊇}{\ensuremath{\supseteq}} %
\newunicodechar{⊎}{\ensuremath{\uplus}} %
\newunicodechar{⊑}{\ensuremath{\sqsubseteq}} %
\newunicodechar{⊒}{\ensuremath{\sqsupseteq}} %
\newunicodechar{⊓}{\ensuremath{\sqcap}} %
\newunicodechar{⊔}{\ensuremath{\sqcup}} %
\newunicodechar{⊕}{\ensuremath{\oplus}} %
\newunicodechar{⊗}{\ensuremath{\otimes}} %
\newunicodechar{⊛}{\ensuremath{\circledast}}
\newunicodechar{⊢}{\ensuremath{\vdash}} %
\newunicodechar{⊤}{\ensuremath{\top}}
\newunicodechar{⊥}{\ensuremath{\bot}} 
\newunicodechar{⊧}{\models} %
\newunicodechar{⊨}{\models} %
\newunicodechar{⊩}{\Vdash}
\newunicodechar{⊬}{\not \vdash}
\newunicodechar{⊸}{\ensuremath{\multimap}} %
\newunicodechar{⋁}{\ensuremath{\bigvee}}
\newunicodechar{⋃}{\ensuremath{\bigcup}} %
\newunicodechar{⋄}{\ensuremath{\diamond}} %
\newunicodechar{⋅}{\ensuremath{\cdot}}
\newunicodechar{⋆}{\ensuremath{\star}} %
\newunicodechar{⋮}{\ensuremath{\vdots}} %
\newunicodechar{⋯}{\ensuremath{\cdots}} %
\newunicodechar{─}{---}
\newunicodechar{■}{\ensuremath{\blacksquare}} 
\newunicodechar{□}{\ensuremath{\square}} %
\newunicodechar{▴}{\ensuremath{\blacktriangledown}}
\newunicodechar{▵}{\ensuremath{\triangle}}
\newunicodechar{▹}{\ensuremath{\rhd}} %
\newunicodechar{▾}{\ensuremath{\blacktriangle}}
\newunicodechar{▿}{\ensuremath{\triangledown}}
\newunicodechar{◃}{\triangleleft{}}
\newunicodechar{◅}{\ensuremath{\triangleleft{}}}
\newunicodechar{◇}{\ensuremath{\diamond}} %
\newunicodechar{◽}{\ensuremath{\square}}
\newunicodechar{★}{\ensuremath{\star}}   %
\newunicodechar{♭}{\ensuremath{\flat}} %
\newunicodechar{♯}{\ensuremath{\sharp}} %
\newunicodechar{✓}{\ensuremath{\checkmark}} %
\newunicodechar{⟂}{\ensuremath{^\bot}} 
\newunicodechar{⟦}{\ensuremath{\llbracket}}
\newunicodechar{⟧}{\ensuremath{\rrbracket}}
\newunicodechar{⟨}{\ensuremath{\langle}} %
\newunicodechar{⟩}{\ensuremath{\rangle}} %
\newunicodechar{⟶}{{\longrightarrow}}
\newunicodechar{⟷} {\ensuremath{\leftrightarrow}}
\newunicodechar{⟹}{\ensuremath{\Longrightarrow}} %
\newunicodechar{ⱼ}{_j} %
\newunicodechar{𝒟}{\ensuremath{\mathcal{D}}} %
\newunicodechar{𝒢}{\ensuremath{\mathcal{G}}}
\newunicodechar{𝒦}{\ensuremath{\mathcal{K}}} %
\newunicodechar{𝒫}{\ensuremath{\mathcal{P}}}
\newunicodechar{𝔸}{\ensuremath{\mathds{A}}} %
\newunicodechar{𝔹}{\ensuremath{\mathds{B}}} %
\newunicodechar{𝔼}{\mathds{E}}
\newunicodechar{𝟙}{\ensuremath{\mathds{1}}} %
\title{Evaluating Linear Functions to Symmetric Monoidal Categories – Version With Additional Details }\author{Jean-Philippe Bernardy}\affiliation{\institution{Sweden}\city{Gothenburg}\country{University of Gothenburg}}\email{jean-philippe.bernardy@gu.se}\author{Arnaud Spiwack}\affiliation{\institution{France}\city{Paris}\country{Tweag}}\email{arnaud.spiwack@tweag.io}\begin{abstract}A number of domain specific languages, such as circuits or
data-science workflows, are best expressed as diagrams of boxes
connected by wires. Unfortunately, functional languages have traditionally been ill-equipped to
embed this sort of languages. The \ensuremath{\mathsf{Arrow}} abstraction is an
approximation, but we argue that it does not capture the right
properties.

A faithful abstraction is Symmetric Monoidal Categories ({\sc{}smc}s), but,
so far,
it hasn't been convenient to use. We show how the advent of linear typing in Haskell lets us
bridge this gap. We provide a library which lets us program in {\sc{}smc}s with linear functions instead of {\sc{}smc} combinators. This considerably
lowers the syntactic overhead of the {\sc{}edsl} to be on par
with that of monadic {\sc{}dsl}s. A remarkable feature of our library is that, contrary to previously
known methods for categories, it does not use any metaprogramming.
\end{abstract}
\begin{document}\maketitle{} 
\begin{mdframed}[linewidth=0pt,hidealllines,innerleftmargin=0pt,innerrightmargin=0pt,backgroundcolor=gray!15]\emph{Preamble to the supplementary material.} This version contains additional details. Such material is always marked with a gray background, such as the one which is observable in this note.\end{mdframed} \newpage{} \section{Introduction}\label{0} 
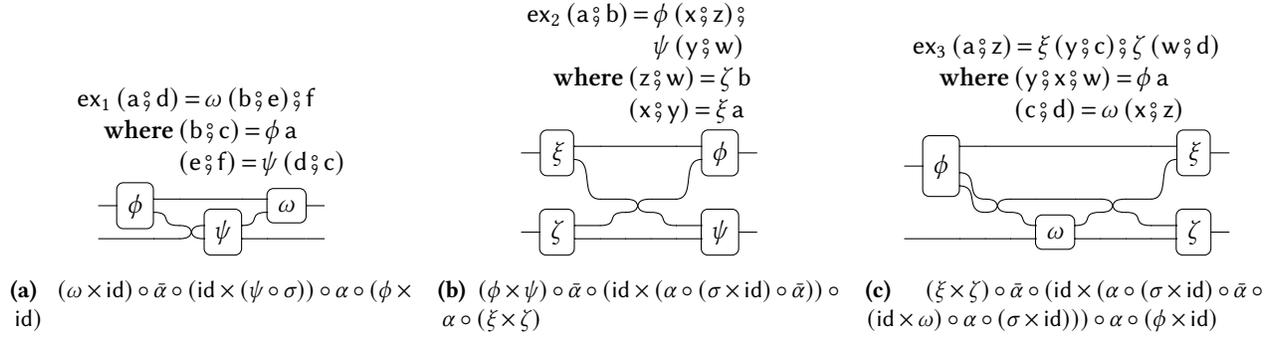
\begin{figure*}\begin{subfigure}[b]{0.3\textwidth}\begin{center} \ensuremath{\begin{parray}\column{B}{@{}>{}l<{}@{}}\column[0em]{1}{@{}>{}l<{}@{}}\column[1em]{2}{@{}>{}l<{}@{}}\column{3}{@{}>{}l<{}@{}}\column{E}{@{}>{}l<{}@{}}%
\>[1]{}{\mathsf{ex}_{1}\mskip 3.0mu\allowbreak{}\mathnormal{(}\mskip 0.0mu\mathsf{a}\mskip 2.0mu\mathnormal{\fatsemi }\mskip 3.0mu\mathsf{d}\mskip 0.0mu\mathnormal{)}\allowbreak{}\mskip 3.0mu\mathnormal{=}\mskip 3.0mu\mathsf{ω}\mskip 3.0mu\allowbreak{}\mathnormal{(}\mskip 0.0mu\mathsf{b}\mskip 2.0mu\mathnormal{\fatsemi }\mskip 3.0mu\mathsf{e}\mskip 0.0mu\mathnormal{)}\allowbreak{}\mskip 2.0mu\mathnormal{\fatsemi }\mskip 3.0mu\mathsf{f}}\<[E]{}\\
\>[2]{}{\mathbf{where}\mskip 3.0mu}\>[3]{}{\allowbreak{}\mathnormal{(}\mskip 0.0mu\mathsf{b}\mskip 2.0mu\mathnormal{\fatsemi }\mskip 3.0mu\mathsf{c}\mskip 0.0mu\mathnormal{)}\allowbreak{}\mskip 3.0mu\mathnormal{=}\mskip 3.0mu\mathsf{φ}\mskip 3.0mu\mathsf{a}}\<[E]{}\\
\>[3]{}{\allowbreak{}\mathnormal{(}\mskip 0.0mu\mathsf{e}\mskip 2.0mu\mathnormal{\fatsemi }\mskip 3.0mu\mathsf{f}\mskip 0.0mu\mathnormal{)}\allowbreak{}\mskip 3.0mu\mathnormal{=}\mskip 3.0mu\mathsf{ψ}\mskip 3.0mu\allowbreak{}\mathnormal{(}\mskip 0.0mu\mathsf{d}\mskip 2.0mu\mathnormal{\fatsemi }\mskip 3.0mu\mathsf{c}\mskip 0.0mu\mathnormal{)}\allowbreak{}}\<[E]{}\end{parray}} 
{\begin{tikzpicture}\path[-,draw=black,line width=0.4000pt,line cap=butt,line join=miter,dash pattern=](-54.2399pt,-2.5000pt)--(-47.2399pt,-2.5000pt);
\path[-,draw=black,line width=0.4000pt,line cap=butt,line join=miter,dash pattern=](-54.2399pt,-15.0000pt)--(-47.2399pt,-15.0000pt);
\path[-,draw=black,line width=0.4000pt,line cap=butt,line join=miter,dash pattern=](24.4050pt,-2.5000pt)--(31.4050pt,-2.5000pt);
\path[-,draw=black,line width=0.4000pt,line cap=butt,line join=miter,dash pattern=](24.4050pt,-15.0000pt)--(31.4050pt,-15.0000pt);
\path[-,draw=black,line width=0.4000pt,line cap=butt,line join=miter,dash pattern=](-33.3650pt,0.0000pt)--(-33.3650pt,0.0000pt);
\path[-,draw=black,line width=0.4000pt,line cap=butt,line join=miter,dash pattern=](-33.3650pt,-5.0000pt)--(-33.3650pt,-5.0000pt);
\path[-,draw=black,line width=0.4000pt,line cap=butt,line join=miter,dash pattern=](-33.3650pt,-15.0000pt)--(-33.3650pt,-15.0000pt);
\path[-,draw=black,line width=0.4000pt,line cap=butt,line join=miter,dash pattern=](-23.7650pt,0.0000pt)--(-23.7650pt,0.0000pt);
\path[-,draw=black,line width=0.4000pt,line cap=butt,line join=miter,dash pattern=](-23.7650pt,-10.0000pt)--(-23.7650pt,-10.0000pt);
\path[-,draw=black,line width=0.4000pt,line cap=butt,line join=miter,dash pattern=](-23.7650pt,-15.0000pt)--(-23.7650pt,-15.0000pt);
\path[-,draw=black,line width=0.4000pt,line cap=butt,line join=miter,dash pattern=](0.0000pt,0.0000pt)--(0.0000pt,0.0000pt);
\path[-,draw=black,line width=0.4000pt,line cap=butt,line join=miter,dash pattern=](0.0000pt,-10.0000pt)--(0.0000pt,-10.0000pt);
\path[-,draw=black,line width=0.4000pt,line cap=butt,line join=miter,dash pattern=](0.0000pt,-15.0000pt)--(0.0000pt,-15.0000pt);
\path[-,draw=black,line width=0.4000pt,line cap=butt,line join=miter,dash pattern=](9.6000pt,0.0000pt)--(9.6000pt,0.0000pt);
\path[-,draw=black,line width=0.4000pt,line cap=butt,line join=miter,dash pattern=](9.6000pt,-5.0000pt)--(9.6000pt,-5.0000pt);
\path[-,draw=black,line width=0.4000pt,line cap=butt,line join=miter,dash pattern=](9.6000pt,-15.0000pt)--(9.6000pt,-15.0000pt);
\path[-,line width=0.4000pt,line cap=butt,line join=miter,dash pattern=](13.6000pt,-0.2550pt)--(20.4050pt,-0.2550pt)--(20.4050pt,-4.7450pt)--(13.6000pt,-4.7450pt)--cycle;
\node[anchor=north west,inner sep=0] at (13.6000pt,-0.2550pt){\savebox{\marxupbox}{{\ensuremath{ω}}}\immediate\write\boxesfile{1}\immediate\write\boxesfile{\number\wd\marxupbox}\immediate\write\boxesfile{\number\ht\marxupbox}\immediate\write\boxesfile{\number\dp\marxupbox}\box\marxupbox};
\path[-,line width=0.4000pt,line cap=butt,line join=miter,dash pattern=](9.6000pt,3.7450pt)--(24.4050pt,3.7450pt)--(24.4050pt,-8.7450pt)--(9.6000pt,-8.7450pt)--cycle;
\path[-,draw=black,line width=0.4000pt,line cap=butt,line join=miter,dash pattern=](9.6000pt,1.5000pt)..controls(9.6000pt,2.8807pt)and(10.7193pt,4.0000pt)..(12.1000pt,4.0000pt)--(21.9050pt,4.0000pt)..controls(23.2857pt,4.0000pt)and(24.4050pt,2.8807pt)..(24.4050pt,1.5000pt)--(24.4050pt,-6.5000pt)..controls(24.4050pt,-7.8807pt)and(23.2857pt,-9.0000pt)..(21.9050pt,-9.0000pt)--(12.1000pt,-9.0000pt)..controls(10.7193pt,-9.0000pt)and(9.6000pt,-7.8807pt)..(9.6000pt,-6.5000pt)--cycle;
\path[-,draw=black,line width=0.4000pt,line cap=butt,line join=miter,dash pattern=](24.4050pt,-2.5000pt)--(24.4050pt,-2.5000pt);
\path[-,draw=black,line width=0.4000pt,line cap=butt,line join=miter,dash pattern=](9.6000pt,0.0000pt)--(9.6000pt,0.0000pt);
\path[-,draw=black,line width=0.4000pt,line cap=butt,line join=miter,dash pattern=](9.6000pt,-5.0000pt)--(9.6000pt,-5.0000pt);
\path[-,draw=black,line width=0.4000pt,line cap=butt,line join=miter,dash pattern=](9.6000pt,-15.0000pt)--(24.4050pt,-15.0000pt);
\path[-,draw=black,line width=0.4000pt,line cap=butt,line join=miter,dash pattern=](0.0000pt,0.0000pt)--(9.6000pt,0.0000pt);
\path[-,draw=black,line width=0.4000pt,line cap=butt,line join=miter,dash pattern=](0.0000pt,-10.0000pt)..controls(8.0000pt,-10.0000pt)and(1.6000pt,-5.0000pt)..(9.6000pt,-5.0000pt);
\path[-,draw=black,line width=0.4000pt,line cap=butt,line join=miter,dash pattern=](0.0000pt,-15.0000pt)--(9.6000pt,-15.0000pt);
\path[-,draw=black,line width=0.4000pt,line cap=butt,line join=miter,dash pattern=](-23.7650pt,0.0000pt)--(0.0000pt,0.0000pt);
\path[-,draw=black,line width=0.4000pt,line cap=butt,line join=miter,dash pattern=](-14.1650pt,-10.0000pt)--(-14.1650pt,-10.0000pt);
\path[-,draw=black,line width=0.4000pt,line cap=butt,line join=miter,dash pattern=](-14.1650pt,-15.0000pt)--(-14.1650pt,-15.0000pt);
\path[-,line width=0.4000pt,line cap=butt,line join=miter,dash pattern=](-10.1650pt,-7.8075pt)--(-4.0000pt,-7.8075pt)--(-4.0000pt,-17.1925pt)--(-10.1650pt,-17.1925pt)--cycle;
\node[anchor=north west,inner sep=0] at (-10.1650pt,-7.8075pt){\savebox{\marxupbox}{{\ensuremath{ψ}}}\immediate\write\boxesfile{2}\immediate\write\boxesfile{\number\wd\marxupbox}\immediate\write\boxesfile{\number\ht\marxupbox}\immediate\write\boxesfile{\number\dp\marxupbox}\box\marxupbox};
\path[-,line width=0.4000pt,line cap=butt,line join=miter,dash pattern=](-14.1650pt,-3.8075pt)--(0.0000pt,-3.8075pt)--(0.0000pt,-21.1925pt)--(-14.1650pt,-21.1925pt)--cycle;
\path[-,draw=black,line width=0.4000pt,line cap=butt,line join=miter,dash pattern=](-14.1650pt,-6.3075pt)..controls(-14.1650pt,-4.9268pt)and(-13.0457pt,-3.8075pt)..(-11.6650pt,-3.8075pt)--(-2.5000pt,-3.8075pt)..controls(-1.1193pt,-3.8075pt)and(0.0000pt,-4.9268pt)..(0.0000pt,-6.3075pt)--(0.0000pt,-18.6925pt)..controls(0.0000pt,-20.0732pt)and(-1.1193pt,-21.1925pt)..(-2.5000pt,-21.1925pt)--(-11.6650pt,-21.1925pt)..controls(-13.0457pt,-21.1925pt)and(-14.1650pt,-20.0732pt)..(-14.1650pt,-18.6925pt)--cycle;
\path[-,draw=black,line width=0.4000pt,line cap=butt,line join=miter,dash pattern=](0.0000pt,-10.0000pt)--(0.0000pt,-10.0000pt);
\path[-,draw=black,line width=0.4000pt,line cap=butt,line join=miter,dash pattern=](0.0000pt,-15.0000pt)--(0.0000pt,-15.0000pt);
\path[-,draw=black,line width=0.4000pt,line cap=butt,line join=miter,dash pattern=](-14.1650pt,-10.0000pt)--(-14.1650pt,-10.0000pt);
\path[-,draw=black,line width=0.4000pt,line cap=butt,line join=miter,dash pattern=](-14.1650pt,-15.0000pt)--(-14.1650pt,-15.0000pt);
\path[-,draw=black,line width=0.4000pt,line cap=butt,line join=miter,dash pattern=](-23.7650pt,-10.0000pt)..controls(-15.7650pt,-10.0000pt)and(-22.1650pt,-15.0000pt)..(-14.1650pt,-15.0000pt);
\path[-,draw=black,line width=0.4000pt,line cap=butt,line join=miter,dash pattern=](-23.7650pt,-15.0000pt)..controls(-15.7650pt,-15.0000pt)and(-22.1650pt,-10.0000pt)..(-14.1650pt,-10.0000pt);
\path[-,draw=black,line width=0.4000pt,line cap=butt,line join=miter,dash pattern=](-33.3650pt,0.0000pt)--(-23.7650pt,0.0000pt);
\path[-,draw=black,line width=0.4000pt,line cap=butt,line join=miter,dash pattern=](-33.3650pt,-5.0000pt)..controls(-25.3650pt,-5.0000pt)and(-31.7650pt,-10.0000pt)..(-23.7650pt,-10.0000pt);
\path[-,draw=black,line width=0.4000pt,line cap=butt,line join=miter,dash pattern=](-33.3650pt,-15.0000pt)--(-23.7650pt,-15.0000pt);
\path[-,line width=0.4000pt,line cap=butt,line join=miter,dash pattern=](-43.2399pt,2.1925pt)--(-37.3650pt,2.1925pt)--(-37.3650pt,-7.1925pt)--(-43.2399pt,-7.1925pt)--cycle;
\node[anchor=north west,inner sep=0] at (-43.2399pt,2.1925pt){\savebox{\marxupbox}{{\ensuremath{φ}}}\immediate\write\boxesfile{3}\immediate\write\boxesfile{\number\wd\marxupbox}\immediate\write\boxesfile{\number\ht\marxupbox}\immediate\write\boxesfile{\number\dp\marxupbox}\box\marxupbox};
\path[-,line width=0.4000pt,line cap=butt,line join=miter,dash pattern=](-47.2399pt,6.1925pt)--(-33.3650pt,6.1925pt)--(-33.3650pt,-11.1925pt)--(-47.2399pt,-11.1925pt)--cycle;
\path[-,draw=black,line width=0.4000pt,line cap=butt,line join=miter,dash pattern=](-47.2399pt,3.6925pt)..controls(-47.2399pt,5.0732pt)and(-46.1207pt,6.1925pt)..(-44.7399pt,6.1925pt)--(-35.8650pt,6.1925pt)..controls(-34.4843pt,6.1925pt)and(-33.3650pt,5.0732pt)..(-33.3650pt,3.6925pt)--(-33.3650pt,-8.6925pt)..controls(-33.3650pt,-10.0732pt)and(-34.4843pt,-11.1925pt)..(-35.8650pt,-11.1925pt)--(-44.7399pt,-11.1925pt)..controls(-46.1207pt,-11.1925pt)and(-47.2399pt,-10.0732pt)..(-47.2399pt,-8.6925pt)--cycle;
\path[-,draw=black,line width=0.4000pt,line cap=butt,line join=miter,dash pattern=](-33.3650pt,0.0000pt)--(-33.3650pt,0.0000pt);
\path[-,draw=black,line width=0.4000pt,line cap=butt,line join=miter,dash pattern=](-33.3650pt,-5.0000pt)--(-33.3650pt,-5.0000pt);
\path[-,draw=black,line width=0.4000pt,line cap=butt,line join=miter,dash pattern=](-47.2399pt,-2.5000pt)--(-47.2399pt,-2.5000pt);
\path[-,draw=black,line width=0.4000pt,line cap=butt,line join=miter,dash pattern=](-47.2399pt,-15.0000pt)--(-33.3650pt,-15.0000pt);
\end{tikzpicture}}\end{center}\caption{\ensuremath{\allowbreak{}\mathnormal{(}\mskip 0.0mu}\ensuremath{\mathsf{ω}\mskip 3.0mu}\ensuremath{\allowbreak{}\mathnormal{×}\allowbreak{}\mskip 3.0mu}\ensuremath{\mathsf{id}\mskip 0.0mu}\ensuremath{\mathnormal{)}\allowbreak{}\mskip 3.0mu}\ensuremath{\allowbreak{}\mathnormal{∘}\allowbreak{}\mskip 3.0mu}\ensuremath{\bar{α}\mskip 3.0mu}\ensuremath{\allowbreak{}\mathnormal{∘}\allowbreak{}\mskip 3.0mu}\ensuremath{\allowbreak{}\mathnormal{(}\mskip 0.0mu}\ensuremath{\mathsf{id}\mskip 3.0mu}\ensuremath{\allowbreak{}\mathnormal{×}\allowbreak{}\mskip 3.0mu}\ensuremath{\allowbreak{}\mathnormal{(}\mskip 0.0mu}\ensuremath{\mathsf{ψ}\mskip 3.0mu}\ensuremath{\allowbreak{}\mathnormal{∘}\allowbreak{}\mskip 3.0mu}\ensuremath{σ\mskip 0.0mu}\ensuremath{\mathnormal{)}\allowbreak{}\mskip 0.0mu}\ensuremath{\mathnormal{)}\allowbreak{}\mskip 3.0mu}\ensuremath{\allowbreak{}\mathnormal{∘}\allowbreak{}\mskip 3.0mu}\ensuremath{α\mskip 3.0mu}\ensuremath{\allowbreak{}\mathnormal{∘}\allowbreak{}\mskip 3.0mu}\ensuremath{\allowbreak{}\mathnormal{(}\mskip 0.0mu}\ensuremath{\mathsf{φ}\mskip 3.0mu}\ensuremath{\allowbreak{}\mathnormal{×}\allowbreak{}\mskip 3.0mu}\ensuremath{\mathsf{id}\mskip 0.0mu}\ensuremath{\mathnormal{)}\allowbreak{}}}\label{4}\end{subfigure}\quad{}\begin{subfigure}[b]{0.3\textwidth}\begin{center} \ensuremath{\begin{parray}\column{B}{@{}>{}l<{}@{}}\column[0em]{1}{@{}>{}l<{}@{}}\column[1em]{2}{@{}>{}l<{}@{}}\column{3}{@{}>{}l<{}@{}}\column{4}{@{}>{}l<{}@{}}\column{E}{@{}>{}l<{}@{}}%
\>[1]{}{\mathsf{ex}_{2}\mskip 3.0mu\allowbreak{}\mathnormal{(}\mskip 0.0mu\mathsf{a}\mskip 2.0mu\mathnormal{\fatsemi }\mskip 3.0mu\mathsf{b}\mskip 0.0mu\mathnormal{)}\allowbreak{}\mskip 3.0mu\mathnormal{=}\mskip 3.0mu}\>[4]{}{\mathsf{φ}\mskip 3.0mu\allowbreak{}\mathnormal{(}\mskip 0.0mu\mathsf{x}\mskip 2.0mu\mathnormal{\fatsemi }\mskip 3.0mu\mathsf{z}\mskip 0.0mu\mathnormal{)}\allowbreak{}\mskip 2.0mu\mathnormal{\fatsemi }}\<[E]{}\\
\>[4]{}{\mathsf{ψ}\mskip 3.0mu\allowbreak{}\mathnormal{(}\mskip 0.0mu\mathsf{y}\mskip 2.0mu\mathnormal{\fatsemi }\mskip 3.0mu\mathsf{w}\mskip 0.0mu\mathnormal{)}\allowbreak{}}\<[E]{}\\
\>[2]{}{\mathbf{where}\mskip 3.0mu}\>[3]{}{\allowbreak{}\mathnormal{(}\mskip 0.0mu\mathsf{z}\mskip 2.0mu\mathnormal{\fatsemi }\mskip 3.0mu\mathsf{w}\mskip 0.0mu\mathnormal{)}\allowbreak{}\mskip 3.0mu\mathnormal{=}\mskip 3.0mu\mathsf{ζ}\mskip 3.0mu\mathsf{b}}\<[E]{}\\
\>[3]{}{\allowbreak{}\mathnormal{(}\mskip 0.0mu\mathsf{x}\mskip 2.0mu\mathnormal{\fatsemi }\mskip 3.0mu\mathsf{y}\mskip 0.0mu\mathnormal{)}\allowbreak{}\mskip 3.0mu\mathnormal{=}\mskip 3.0mu\mathsf{ξ}\mskip 3.0mu\mathsf{a}}\<[E]{}\end{parray}} 
{\begin{tikzpicture}\path[-,draw=black,line width=0.4000pt,line cap=butt,line join=miter,dash pattern=](-82.2950pt,8.6925pt)--(-75.2950pt,8.6925pt);
\path[-,draw=black,line width=0.4000pt,line cap=butt,line join=miter,dash pattern=](-82.2950pt,-21.3075pt)--(-75.2950pt,-21.3075pt);
\path[-,draw=black,line width=0.4000pt,line cap=butt,line join=miter,dash pattern=](0.0000pt,8.6925pt)--(7.0000pt,8.6925pt);
\path[-,draw=black,line width=0.4000pt,line cap=butt,line join=miter,dash pattern=](0.0000pt,-21.3075pt)--(7.0000pt,-21.3075pt);
\path[-,draw=black,line width=0.4000pt,line cap=butt,line join=miter,dash pattern=](-62.1650pt,11.1925pt)--(-62.1650pt,11.1925pt);
\path[-,draw=black,line width=0.4000pt,line cap=butt,line join=miter,dash pattern=](-62.1650pt,6.1925pt)--(-62.1650pt,6.1925pt);
\path[-,draw=black,line width=0.4000pt,line cap=butt,line join=miter,dash pattern=](-62.1650pt,-18.8075pt)--(-62.1650pt,-18.8075pt);
\path[-,draw=black,line width=0.4000pt,line cap=butt,line join=miter,dash pattern=](-62.1650pt,-23.8075pt)--(-62.1650pt,-23.8075pt);
\path[-,draw=black,line width=0.4000pt,line cap=butt,line join=miter,dash pattern=](-52.5650pt,11.1925pt)--(-52.5650pt,11.1925pt);
\path[-,draw=black,line width=0.4000pt,line cap=butt,line join=miter,dash pattern=](-52.5650pt,-8.8075pt)--(-52.5650pt,-8.8075pt);
\path[-,draw=black,line width=0.4000pt,line cap=butt,line join=miter,dash pattern=](-52.5650pt,-18.8075pt)--(-52.5650pt,-18.8075pt);
\path[-,draw=black,line width=0.4000pt,line cap=butt,line join=miter,dash pattern=](-52.5650pt,-23.8075pt)--(-52.5650pt,-23.8075pt);
\path[-,draw=black,line width=0.4000pt,line cap=butt,line join=miter,dash pattern=](-23.7650pt,11.1925pt)--(-23.7650pt,11.1925pt);
\path[-,draw=black,line width=0.4000pt,line cap=butt,line join=miter,dash pattern=](-23.7650pt,-8.8075pt)--(-23.7650pt,-8.8075pt);
\path[-,draw=black,line width=0.4000pt,line cap=butt,line join=miter,dash pattern=](-23.7650pt,-18.8075pt)--(-23.7650pt,-18.8075pt);
\path[-,draw=black,line width=0.4000pt,line cap=butt,line join=miter,dash pattern=](-23.7650pt,-23.8075pt)--(-23.7650pt,-23.8075pt);
\path[-,draw=black,line width=0.4000pt,line cap=butt,line join=miter,dash pattern=](-14.1650pt,11.1925pt)--(-14.1650pt,11.1925pt);
\path[-,draw=black,line width=0.4000pt,line cap=butt,line join=miter,dash pattern=](-14.1650pt,6.1925pt)--(-14.1650pt,6.1925pt);
\path[-,draw=black,line width=0.4000pt,line cap=butt,line join=miter,dash pattern=](-14.1650pt,-18.8075pt)--(-14.1650pt,-18.8075pt);
\path[-,draw=black,line width=0.4000pt,line cap=butt,line join=miter,dash pattern=](-14.1650pt,-23.8075pt)--(-14.1650pt,-23.8075pt);
\path[-,line width=0.4000pt,line cap=butt,line join=miter,dash pattern=](-10.0200pt,13.3850pt)--(-4.1450pt,13.3850pt)--(-4.1450pt,4.0000pt)--(-10.0200pt,4.0000pt)--cycle;
\node[anchor=north west,inner sep=0] at (-10.0200pt,13.3850pt){\savebox{\marxupbox}{{\ensuremath{φ}}}\immediate\write\boxesfile{5}\immediate\write\boxesfile{\number\wd\marxupbox}\immediate\write\boxesfile{\number\ht\marxupbox}\immediate\write\boxesfile{\number\dp\marxupbox}\box\marxupbox};
\path[-,line width=0.4000pt,line cap=butt,line join=miter,dash pattern=](-14.0200pt,17.3850pt)--(-0.1450pt,17.3850pt)--(-0.1450pt,0.0000pt)--(-14.0200pt,0.0000pt)--cycle;
\path[-,draw=black,line width=0.4000pt,line cap=butt,line join=miter,dash pattern=](-14.0200pt,14.8850pt)..controls(-14.0200pt,16.2657pt)and(-12.9007pt,17.3850pt)..(-11.5200pt,17.3850pt)--(-2.6450pt,17.3850pt)..controls(-1.2643pt,17.3850pt)and(-0.1450pt,16.2657pt)..(-0.1450pt,14.8850pt)--(-0.1450pt,2.5000pt)..controls(-0.1450pt,1.1193pt)and(-1.2643pt,0.0000pt)..(-2.6450pt,0.0000pt)--(-11.5200pt,0.0000pt)..controls(-12.9007pt,0.0000pt)and(-14.0200pt,1.1193pt)..(-14.0200pt,2.5000pt)--cycle;
\path[-,draw=black,line width=0.4000pt,line cap=butt,line join=miter,dash pattern=](0.0000pt,8.6925pt)--(-0.1450pt,8.6925pt);
\path[-,draw=black,line width=0.4000pt,line cap=butt,line join=miter,dash pattern=](-14.1650pt,11.1925pt)--(-14.0200pt,11.1925pt);
\path[-,draw=black,line width=0.4000pt,line cap=butt,line join=miter,dash pattern=](-14.1650pt,6.1925pt)--(-14.0200pt,6.1925pt);
\path[-,line width=0.4000pt,line cap=butt,line join=miter,dash pattern=](-10.1650pt,-16.6150pt)--(-4.0000pt,-16.6150pt)--(-4.0000pt,-26.0000pt)--(-10.1650pt,-26.0000pt)--cycle;
\node[anchor=north west,inner sep=0] at (-10.1650pt,-16.6150pt){\savebox{\marxupbox}{{\ensuremath{ψ}}}\immediate\write\boxesfile{6}\immediate\write\boxesfile{\number\wd\marxupbox}\immediate\write\boxesfile{\number\ht\marxupbox}\immediate\write\boxesfile{\number\dp\marxupbox}\box\marxupbox};
\path[-,line width=0.4000pt,line cap=butt,line join=miter,dash pattern=](-14.1650pt,-12.6150pt)--(0.0000pt,-12.6150pt)--(0.0000pt,-30.0000pt)--(-14.1650pt,-30.0000pt)--cycle;
\path[-,draw=black,line width=0.4000pt,line cap=butt,line join=miter,dash pattern=](-14.1650pt,-15.1150pt)..controls(-14.1650pt,-13.7343pt)and(-13.0457pt,-12.6150pt)..(-11.6650pt,-12.6150pt)--(-2.5000pt,-12.6150pt)..controls(-1.1193pt,-12.6150pt)and(0.0000pt,-13.7343pt)..(0.0000pt,-15.1150pt)--(0.0000pt,-27.5000pt)..controls(0.0000pt,-28.8807pt)and(-1.1193pt,-30.0000pt)..(-2.5000pt,-30.0000pt)--(-11.6650pt,-30.0000pt)..controls(-13.0457pt,-30.0000pt)and(-14.1650pt,-28.8807pt)..(-14.1650pt,-27.5000pt)--cycle;
\path[-,draw=black,line width=0.4000pt,line cap=butt,line join=miter,dash pattern=](0.0000pt,-21.3075pt)--(0.0000pt,-21.3075pt);
\path[-,draw=black,line width=0.4000pt,line cap=butt,line join=miter,dash pattern=](-14.1650pt,-18.8075pt)--(-14.1650pt,-18.8075pt);
\path[-,draw=black,line width=0.4000pt,line cap=butt,line join=miter,dash pattern=](-14.1650pt,-23.8075pt)--(-14.1650pt,-23.8075pt);
\path[-,draw=black,line width=0.4000pt,line cap=butt,line join=miter,dash pattern=](-23.7650pt,11.1925pt)--(-14.1650pt,11.1925pt);
\path[-,draw=black,line width=0.4000pt,line cap=butt,line join=miter,dash pattern=](-23.7650pt,-8.8075pt)..controls(-15.7650pt,-8.8075pt)and(-22.1650pt,6.1925pt)..(-14.1650pt,6.1925pt);
\path[-,draw=black,line width=0.4000pt,line cap=butt,line join=miter,dash pattern=](-23.7650pt,-18.8075pt)--(-14.1650pt,-18.8075pt);
\path[-,draw=black,line width=0.4000pt,line cap=butt,line join=miter,dash pattern=](-23.7650pt,-23.8075pt)--(-14.1650pt,-23.8075pt);
\path[-,draw=black,line width=0.4000pt,line cap=butt,line join=miter,dash pattern=](-52.5650pt,11.1925pt)--(-23.7650pt,11.1925pt);
\path[-,draw=black,line width=0.4000pt,line cap=butt,line join=miter,dash pattern=](-42.9650pt,-8.8075pt)--(-42.9650pt,-8.8075pt);
\path[-,draw=black,line width=0.4000pt,line cap=butt,line join=miter,dash pattern=](-42.9650pt,-13.8075pt)--(-42.9650pt,-13.8075pt);
\path[-,draw=black,line width=0.4000pt,line cap=butt,line join=miter,dash pattern=](-42.9650pt,-23.8075pt)--(-42.9650pt,-23.8075pt);
\path[-,draw=black,line width=0.4000pt,line cap=butt,line join=miter,dash pattern=](-33.3650pt,-8.8075pt)--(-33.3650pt,-8.8075pt);
\path[-,draw=black,line width=0.4000pt,line cap=butt,line join=miter,dash pattern=](-33.3650pt,-13.8075pt)--(-33.3650pt,-13.8075pt);
\path[-,draw=black,line width=0.4000pt,line cap=butt,line join=miter,dash pattern=](-33.3650pt,-23.8075pt)--(-33.3650pt,-23.8075pt);
\path[-,draw=black,line width=0.4000pt,line cap=butt,line join=miter,dash pattern=](-33.3650pt,-8.8075pt)--(-23.7650pt,-8.8075pt);
\path[-,draw=black,line width=0.4000pt,line cap=butt,line join=miter,dash pattern=](-33.3650pt,-13.8075pt)..controls(-25.3650pt,-13.8075pt)and(-31.7650pt,-18.8075pt)..(-23.7650pt,-18.8075pt);
\path[-,draw=black,line width=0.4000pt,line cap=butt,line join=miter,dash pattern=](-33.3650pt,-23.8075pt)--(-23.7650pt,-23.8075pt);
\path[-,draw=black,line width=0.4000pt,line cap=butt,line join=miter,dash pattern=](-42.9650pt,-8.8075pt)..controls(-34.9650pt,-8.8075pt)and(-41.3650pt,-13.8075pt)..(-33.3650pt,-13.8075pt);
\path[-,draw=black,line width=0.4000pt,line cap=butt,line join=miter,dash pattern=](-42.9650pt,-13.8075pt)..controls(-34.9650pt,-13.8075pt)and(-41.3650pt,-8.8075pt)..(-33.3650pt,-8.8075pt);
\path[-,draw=black,line width=0.4000pt,line cap=butt,line join=miter,dash pattern=](-42.9650pt,-23.8075pt)--(-33.3650pt,-23.8075pt);
\path[-,draw=black,line width=0.4000pt,line cap=butt,line join=miter,dash pattern=](-52.5650pt,-8.8075pt)--(-42.9650pt,-8.8075pt);
\path[-,draw=black,line width=0.4000pt,line cap=butt,line join=miter,dash pattern=](-52.5650pt,-18.8075pt)..controls(-44.5650pt,-18.8075pt)and(-50.9650pt,-13.8075pt)..(-42.9650pt,-13.8075pt);
\path[-,draw=black,line width=0.4000pt,line cap=butt,line join=miter,dash pattern=](-52.5650pt,-23.8075pt)--(-42.9650pt,-23.8075pt);
\path[-,draw=black,line width=0.4000pt,line cap=butt,line join=miter,dash pattern=](-62.1650pt,11.1925pt)--(-52.5650pt,11.1925pt);
\path[-,draw=black,line width=0.4000pt,line cap=butt,line join=miter,dash pattern=](-62.1650pt,6.1925pt)..controls(-54.1650pt,6.1925pt)and(-60.5650pt,-8.8075pt)..(-52.5650pt,-8.8075pt);
\path[-,draw=black,line width=0.4000pt,line cap=butt,line join=miter,dash pattern=](-62.1650pt,-18.8075pt)--(-52.5650pt,-18.8075pt);
\path[-,draw=black,line width=0.4000pt,line cap=butt,line join=miter,dash pattern=](-62.1650pt,-23.8075pt)--(-52.5650pt,-23.8075pt);
\path[-,line width=0.4000pt,line cap=butt,line join=miter,dash pattern=](-71.0275pt,13.3350pt)--(-66.4325pt,13.3350pt)--(-66.4325pt,4.0500pt)--(-71.0275pt,4.0500pt)--cycle;
\node[anchor=north west,inner sep=0] at (-71.0275pt,13.3350pt){\savebox{\marxupbox}{{\ensuremath{ξ}}}\immediate\write\boxesfile{7}\immediate\write\boxesfile{\number\wd\marxupbox}\immediate\write\boxesfile{\number\ht\marxupbox}\immediate\write\boxesfile{\number\dp\marxupbox}\box\marxupbox};
\path[-,line width=0.4000pt,line cap=butt,line join=miter,dash pattern=](-75.0275pt,17.3350pt)--(-62.4325pt,17.3350pt)--(-62.4325pt,0.0500pt)--(-75.0275pt,0.0500pt)--cycle;
\path[-,draw=black,line width=0.4000pt,line cap=butt,line join=miter,dash pattern=](-75.0275pt,14.8350pt)..controls(-75.0275pt,16.2157pt)and(-73.9082pt,17.3350pt)..(-72.5275pt,17.3350pt)--(-64.9325pt,17.3350pt)..controls(-63.5518pt,17.3350pt)and(-62.4325pt,16.2157pt)..(-62.4325pt,14.8350pt)--(-62.4325pt,2.5500pt)..controls(-62.4325pt,1.1693pt)and(-63.5518pt,0.0500pt)..(-64.9325pt,0.0500pt)--(-72.5275pt,0.0500pt)..controls(-73.9082pt,0.0500pt)and(-75.0275pt,1.1693pt)..(-75.0275pt,2.5500pt)--cycle;
\path[-,draw=black,line width=0.4000pt,line cap=butt,line join=miter,dash pattern=](-62.1650pt,11.1925pt)--(-62.4325pt,11.1925pt);
\path[-,draw=black,line width=0.4000pt,line cap=butt,line join=miter,dash pattern=](-62.1650pt,6.1925pt)--(-62.4325pt,6.1925pt);
\path[-,draw=black,line width=0.4000pt,line cap=butt,line join=miter,dash pattern=](-75.2950pt,8.6925pt)--(-75.0275pt,8.6925pt);
\path[-,line width=0.4000pt,line cap=butt,line join=miter,dash pattern=](-71.2950pt,-16.6150pt)--(-66.1650pt,-16.6150pt)--(-66.1650pt,-26.0000pt)--(-71.2950pt,-26.0000pt)--cycle;
\node[anchor=north west,inner sep=0] at (-71.2950pt,-16.6150pt){\savebox{\marxupbox}{{\ensuremath{ζ}}}\immediate\write\boxesfile{8}\immediate\write\boxesfile{\number\wd\marxupbox}\immediate\write\boxesfile{\number\ht\marxupbox}\immediate\write\boxesfile{\number\dp\marxupbox}\box\marxupbox};
\path[-,line width=0.4000pt,line cap=butt,line join=miter,dash pattern=](-75.2950pt,-12.6150pt)--(-62.1650pt,-12.6150pt)--(-62.1650pt,-30.0000pt)--(-75.2950pt,-30.0000pt)--cycle;
\path[-,draw=black,line width=0.4000pt,line cap=butt,line join=miter,dash pattern=](-75.2950pt,-15.1150pt)..controls(-75.2950pt,-13.7343pt)and(-74.1757pt,-12.6150pt)..(-72.7950pt,-12.6150pt)--(-64.6650pt,-12.6150pt)..controls(-63.2843pt,-12.6150pt)and(-62.1650pt,-13.7343pt)..(-62.1650pt,-15.1150pt)--(-62.1650pt,-27.5000pt)..controls(-62.1650pt,-28.8807pt)and(-63.2843pt,-30.0000pt)..(-64.6650pt,-30.0000pt)--(-72.7950pt,-30.0000pt)..controls(-74.1757pt,-30.0000pt)and(-75.2950pt,-28.8807pt)..(-75.2950pt,-27.5000pt)--cycle;
\path[-,draw=black,line width=0.4000pt,line cap=butt,line join=miter,dash pattern=](-62.1650pt,-18.8075pt)--(-62.1650pt,-18.8075pt);
\path[-,draw=black,line width=0.4000pt,line cap=butt,line join=miter,dash pattern=](-62.1650pt,-23.8075pt)--(-62.1650pt,-23.8075pt);
\path[-,draw=black,line width=0.4000pt,line cap=butt,line join=miter,dash pattern=](-75.2950pt,-21.3075pt)--(-75.2950pt,-21.3075pt);
\end{tikzpicture}}\end{center}\caption{\ensuremath{\allowbreak{}\mathnormal{(}\mskip 0.0mu}\ensuremath{\mathsf{φ}\mskip 3.0mu}\ensuremath{\allowbreak{}\mathnormal{×}\allowbreak{}\mskip 3.0mu}\ensuremath{\mathsf{ψ}\mskip 0.0mu}\ensuremath{\mathnormal{)}\allowbreak{}\mskip 3.0mu}\ensuremath{\allowbreak{}\mathnormal{∘}\allowbreak{}\mskip 3.0mu}\ensuremath{\bar{α}\mskip 3.0mu}\ensuremath{\allowbreak{}\mathnormal{∘}\allowbreak{}\mskip 3.0mu}\ensuremath{\allowbreak{}\mathnormal{(}\mskip 0.0mu}\ensuremath{\mathsf{id}\mskip 3.0mu}\ensuremath{\allowbreak{}\mathnormal{×}\allowbreak{}\mskip 3.0mu}\ensuremath{\allowbreak{}\mathnormal{(}\mskip 0.0mu}\ensuremath{α\mskip 3.0mu}\ensuremath{\allowbreak{}\mathnormal{∘}\allowbreak{}\mskip 3.0mu}\ensuremath{\allowbreak{}\mathnormal{(}\mskip 0.0mu}\ensuremath{σ\mskip 3.0mu}\ensuremath{\allowbreak{}\mathnormal{×}\allowbreak{}\mskip 3.0mu}\ensuremath{\mathsf{id}\mskip 0.0mu}\ensuremath{\mathnormal{)}\allowbreak{}\mskip 3.0mu}\ensuremath{\allowbreak{}\mathnormal{∘}\allowbreak{}\mskip 3.0mu}\ensuremath{\bar{α}\mskip 0.0mu}\ensuremath{\mathnormal{)}\allowbreak{}\mskip 0.0mu}\ensuremath{\mathnormal{)}\allowbreak{}\mskip 3.0mu}\ensuremath{\allowbreak{}\mathnormal{∘}\allowbreak{}\mskip 3.0mu}\ensuremath{α\mskip 3.0mu}\ensuremath{\allowbreak{}\mathnormal{∘}\allowbreak{}\mskip 3.0mu}\ensuremath{\allowbreak{}\mathnormal{(}\mskip 0.0mu}\ensuremath{\mathsf{ξ}\mskip 3.0mu}\ensuremath{\allowbreak{}\mathnormal{×}\allowbreak{}\mskip 3.0mu}\ensuremath{\mathsf{ζ}\mskip 0.0mu}\ensuremath{\mathnormal{)}\allowbreak{}}}\label{9}\end{subfigure}\quad{}\begin{subfigure}[b]{0.3\textwidth}\begin{center} \ensuremath{\begin{parray}\column{B}{@{}>{}l<{}@{}}\column[0em]{1}{@{}>{}l<{}@{}}\column[1em]{2}{@{}>{}l<{}@{}}\column{3}{@{}>{}l<{}@{}}\column{4}{@{}>{}l<{}@{}}\column{E}{@{}>{}l<{}@{}}%
\>[1]{}{\mathsf{ex}_{3}\mskip 3.0mu\allowbreak{}\mathnormal{(}\mskip 0.0mu\mathsf{a}\mskip 2.0mu\mathnormal{\fatsemi }\mskip 3.0mu\mathsf{z}\mskip 0.0mu\mathnormal{)}\allowbreak{}\mskip 3.0mu\mathnormal{=}\mskip 3.0mu\mathsf{ξ}\mskip 3.0mu\allowbreak{}\mathnormal{(}\mskip 0.0mu\mathsf{y}\mskip 2.0mu\mathnormal{\fatsemi }\mskip 3.0mu\mathsf{c}\mskip 0.0mu\mathnormal{)}\allowbreak{}\mskip 2.0mu}\>[4]{}{\mathnormal{\fatsemi }\mskip 3.0mu\mathsf{ζ}\mskip 3.0mu\allowbreak{}\mathnormal{(}\mskip 0.0mu\mathsf{w}\mskip 2.0mu\mathnormal{\fatsemi }\mskip 3.0mu\mathsf{d}\mskip 0.0mu\mathnormal{)}\allowbreak{}}\<[E]{}\\
\>[2]{}{\mathbf{where}\mskip 3.0mu}\>[3]{}{\allowbreak{}\mathnormal{(}\mskip 0.0mu\mathsf{y}\mskip 2.0mu\mathnormal{\fatsemi }\mskip 3.0mu\mathsf{x}\mskip 2.0mu\mathnormal{\fatsemi }\mskip 3.0mu\mathsf{w}\mskip 0.0mu\mathnormal{)}\allowbreak{}\mskip 3.0mu\mathnormal{=}\mskip 3.0mu\mathsf{φ}\mskip 3.0mu\mathsf{a}}\<[E]{}\\
\>[3]{}{\allowbreak{}\mathnormal{(}\mskip 0.0mu\mathsf{c}\mskip 2.0mu\mathnormal{\fatsemi }\mskip 3.0mu\mathsf{d}\mskip 0.0mu\mathnormal{)}\allowbreak{}\mskip 3.0mu\mathnormal{=}\mskip 3.0mu\mathsf{ω}\mskip 3.0mu\allowbreak{}\mathnormal{(}\mskip 0.0mu\mathsf{x}\mskip 2.0mu\mathnormal{\fatsemi }\mskip 3.0mu\mathsf{z}\mskip 0.0mu\mathnormal{)}\allowbreak{}}\<[E]{}\end{parray}} 
{\begin{tikzpicture}\path[-,draw=black,line width=0.4000pt,line cap=butt,line join=miter,dash pattern=](-116.0099pt,-7.5000pt)--(-109.0099pt,-7.5000pt);
\path[-,draw=black,line width=0.4000pt,line cap=butt,line join=miter,dash pattern=](-116.0099pt,-35.0000pt)--(-109.0099pt,-35.0000pt);
\path[-,draw=black,line width=0.4000pt,line cap=butt,line join=miter,dash pattern=](0.0000pt,-2.5000pt)--(7.0000pt,-2.5000pt);
\path[-,draw=black,line width=0.4000pt,line cap=butt,line join=miter,dash pattern=](0.0000pt,-32.5000pt)--(7.0000pt,-32.5000pt);
\path[-,draw=black,line width=0.4000pt,line cap=butt,line join=miter,dash pattern=](-95.1350pt,0.0000pt)--(-95.1350pt,0.0000pt);
\path[-,draw=black,line width=0.4000pt,line cap=butt,line join=miter,dash pattern=](-95.1350pt,-10.0000pt)--(-95.1350pt,-10.0000pt);
\path[-,draw=black,line width=0.4000pt,line cap=butt,line join=miter,dash pattern=](-95.1350pt,-15.0000pt)--(-95.1350pt,-15.0000pt);
\path[-,draw=black,line width=0.4000pt,line cap=butt,line join=miter,dash pattern=](-95.1350pt,-35.0000pt)--(-95.1350pt,-35.0000pt);
\path[-,draw=black,line width=0.4000pt,line cap=butt,line join=miter,dash pattern=](-85.5350pt,0.0000pt)--(-85.5350pt,0.0000pt);
\path[-,draw=black,line width=0.4000pt,line cap=butt,line join=miter,dash pattern=](-85.5350pt,-20.0000pt)--(-85.5350pt,-20.0000pt);
\path[-,draw=black,line width=0.4000pt,line cap=butt,line join=miter,dash pattern=](-85.5350pt,-25.0000pt)--(-85.5350pt,-25.0000pt);
\path[-,draw=black,line width=0.4000pt,line cap=butt,line join=miter,dash pattern=](-85.5350pt,-35.0000pt)--(-85.5350pt,-35.0000pt);
\path[-,draw=black,line width=0.4000pt,line cap=butt,line join=miter,dash pattern=](-22.7300pt,0.0000pt)--(-22.7300pt,0.0000pt);
\path[-,draw=black,line width=0.4000pt,line cap=butt,line join=miter,dash pattern=](-22.7300pt,-20.0000pt)--(-22.7300pt,-20.0000pt);
\path[-,draw=black,line width=0.4000pt,line cap=butt,line join=miter,dash pattern=](-22.7300pt,-30.0000pt)--(-22.7300pt,-30.0000pt);
\path[-,draw=black,line width=0.4000pt,line cap=butt,line join=miter,dash pattern=](-22.7300pt,-35.0000pt)--(-22.7300pt,-35.0000pt);
\path[-,draw=black,line width=0.4000pt,line cap=butt,line join=miter,dash pattern=](-13.1300pt,0.0000pt)--(-13.1300pt,0.0000pt);
\path[-,draw=black,line width=0.4000pt,line cap=butt,line join=miter,dash pattern=](-13.1300pt,-5.0000pt)--(-13.1300pt,-5.0000pt);
\path[-,draw=black,line width=0.4000pt,line cap=butt,line join=miter,dash pattern=](-13.1300pt,-30.0000pt)--(-13.1300pt,-30.0000pt);
\path[-,draw=black,line width=0.4000pt,line cap=butt,line join=miter,dash pattern=](-13.1300pt,-35.0000pt)--(-13.1300pt,-35.0000pt);
\path[-,line width=0.4000pt,line cap=butt,line join=miter,dash pattern=](-8.8625pt,2.1425pt)--(-4.2675pt,2.1425pt)--(-4.2675pt,-7.1425pt)--(-8.8625pt,-7.1425pt)--cycle;
\node[anchor=north west,inner sep=0] at (-8.8625pt,2.1425pt){\savebox{\marxupbox}{{\ensuremath{ξ}}}\immediate\write\boxesfile{10}\immediate\write\boxesfile{\number\wd\marxupbox}\immediate\write\boxesfile{\number\ht\marxupbox}\immediate\write\boxesfile{\number\dp\marxupbox}\box\marxupbox};
\path[-,line width=0.4000pt,line cap=butt,line join=miter,dash pattern=](-12.8625pt,6.1425pt)--(-0.2675pt,6.1425pt)--(-0.2675pt,-11.1425pt)--(-12.8625pt,-11.1425pt)--cycle;
\path[-,draw=black,line width=0.4000pt,line cap=butt,line join=miter,dash pattern=](-12.8625pt,3.6425pt)..controls(-12.8625pt,5.0232pt)and(-11.7432pt,6.1425pt)..(-10.3625pt,6.1425pt)--(-2.7675pt,6.1425pt)..controls(-1.3868pt,6.1425pt)and(-0.2675pt,5.0232pt)..(-0.2675pt,3.6425pt)--(-0.2675pt,-8.6425pt)..controls(-0.2675pt,-10.0232pt)and(-1.3868pt,-11.1425pt)..(-2.7675pt,-11.1425pt)--(-10.3625pt,-11.1425pt)..controls(-11.7432pt,-11.1425pt)and(-12.8625pt,-10.0232pt)..(-12.8625pt,-8.6425pt)--cycle;
\path[-,draw=black,line width=0.4000pt,line cap=butt,line join=miter,dash pattern=](0.0000pt,-2.5000pt)--(-0.2675pt,-2.5000pt);
\path[-,draw=black,line width=0.4000pt,line cap=butt,line join=miter,dash pattern=](-13.1300pt,0.0000pt)--(-12.8625pt,0.0000pt);
\path[-,draw=black,line width=0.4000pt,line cap=butt,line join=miter,dash pattern=](-13.1300pt,-5.0000pt)--(-12.8625pt,-5.0000pt);
\path[-,line width=0.4000pt,line cap=butt,line join=miter,dash pattern=](-9.1300pt,-27.8075pt)--(-4.0000pt,-27.8075pt)--(-4.0000pt,-37.1925pt)--(-9.1300pt,-37.1925pt)--cycle;
\node[anchor=north west,inner sep=0] at (-9.1300pt,-27.8075pt){\savebox{\marxupbox}{{\ensuremath{ζ}}}\immediate\write\boxesfile{11}\immediate\write\boxesfile{\number\wd\marxupbox}\immediate\write\boxesfile{\number\ht\marxupbox}\immediate\write\boxesfile{\number\dp\marxupbox}\box\marxupbox};
\path[-,line width=0.4000pt,line cap=butt,line join=miter,dash pattern=](-13.1300pt,-23.8075pt)--(0.0000pt,-23.8075pt)--(0.0000pt,-41.1925pt)--(-13.1300pt,-41.1925pt)--cycle;
\path[-,draw=black,line width=0.4000pt,line cap=butt,line join=miter,dash pattern=](-13.1300pt,-26.3075pt)..controls(-13.1300pt,-24.9268pt)and(-12.0107pt,-23.8075pt)..(-10.6300pt,-23.8075pt)--(-2.5000pt,-23.8075pt)..controls(-1.1193pt,-23.8075pt)and(0.0000pt,-24.9268pt)..(0.0000pt,-26.3075pt)--(0.0000pt,-38.6925pt)..controls(0.0000pt,-40.0732pt)and(-1.1193pt,-41.1925pt)..(-2.5000pt,-41.1925pt)--(-10.6300pt,-41.1925pt)..controls(-12.0107pt,-41.1925pt)and(-13.1300pt,-40.0732pt)..(-13.1300pt,-38.6925pt)--cycle;
\path[-,draw=black,line width=0.4000pt,line cap=butt,line join=miter,dash pattern=](0.0000pt,-32.5000pt)--(0.0000pt,-32.5000pt);
\path[-,draw=black,line width=0.4000pt,line cap=butt,line join=miter,dash pattern=](-13.1300pt,-30.0000pt)--(-13.1300pt,-30.0000pt);
\path[-,draw=black,line width=0.4000pt,line cap=butt,line join=miter,dash pattern=](-13.1300pt,-35.0000pt)--(-13.1300pt,-35.0000pt);
\path[-,draw=black,line width=0.4000pt,line cap=butt,line join=miter,dash pattern=](-22.7300pt,0.0000pt)--(-13.1300pt,0.0000pt);
\path[-,draw=black,line width=0.4000pt,line cap=butt,line join=miter,dash pattern=](-22.7300pt,-20.0000pt)..controls(-14.7300pt,-20.0000pt)and(-21.1300pt,-5.0000pt)..(-13.1300pt,-5.0000pt);
\path[-,draw=black,line width=0.4000pt,line cap=butt,line join=miter,dash pattern=](-22.7300pt,-30.0000pt)--(-13.1300pt,-30.0000pt);
\path[-,draw=black,line width=0.4000pt,line cap=butt,line join=miter,dash pattern=](-22.7300pt,-35.0000pt)--(-13.1300pt,-35.0000pt);
\path[-,draw=black,line width=0.4000pt,line cap=butt,line join=miter,dash pattern=](-85.5350pt,0.0000pt)--(-22.7300pt,0.0000pt);
\path[-,draw=black,line width=0.4000pt,line cap=butt,line join=miter,dash pattern=](-75.9350pt,-20.0000pt)--(-75.9350pt,-20.0000pt);
\path[-,draw=black,line width=0.4000pt,line cap=butt,line join=miter,dash pattern=](-75.9350pt,-25.0000pt)--(-75.9350pt,-25.0000pt);
\path[-,draw=black,line width=0.4000pt,line cap=butt,line join=miter,dash pattern=](-75.9350pt,-35.0000pt)--(-75.9350pt,-35.0000pt);
\path[-,draw=black,line width=0.4000pt,line cap=butt,line join=miter,dash pattern=](-51.5300pt,-20.0000pt)--(-51.5300pt,-20.0000pt);
\path[-,draw=black,line width=0.4000pt,line cap=butt,line join=miter,dash pattern=](-51.5300pt,-30.0000pt)--(-51.5300pt,-30.0000pt);
\path[-,draw=black,line width=0.4000pt,line cap=butt,line join=miter,dash pattern=](-51.5300pt,-35.0000pt)--(-51.5300pt,-35.0000pt);
\path[-,draw=black,line width=0.4000pt,line cap=butt,line join=miter,dash pattern=](-41.9300pt,-20.0000pt)--(-41.9300pt,-20.0000pt);
\path[-,draw=black,line width=0.4000pt,line cap=butt,line join=miter,dash pattern=](-41.9300pt,-25.0000pt)--(-41.9300pt,-25.0000pt);
\path[-,draw=black,line width=0.4000pt,line cap=butt,line join=miter,dash pattern=](-41.9300pt,-35.0000pt)--(-41.9300pt,-35.0000pt);
\path[-,draw=black,line width=0.4000pt,line cap=butt,line join=miter,dash pattern=](-32.3300pt,-20.0000pt)--(-32.3300pt,-20.0000pt);
\path[-,draw=black,line width=0.4000pt,line cap=butt,line join=miter,dash pattern=](-32.3300pt,-25.0000pt)--(-32.3300pt,-25.0000pt);
\path[-,draw=black,line width=0.4000pt,line cap=butt,line join=miter,dash pattern=](-32.3300pt,-35.0000pt)--(-32.3300pt,-35.0000pt);
\path[-,draw=black,line width=0.4000pt,line cap=butt,line join=miter,dash pattern=](-32.3300pt,-20.0000pt)--(-22.7300pt,-20.0000pt);
\path[-,draw=black,line width=0.4000pt,line cap=butt,line join=miter,dash pattern=](-32.3300pt,-25.0000pt)..controls(-24.3300pt,-25.0000pt)and(-30.7300pt,-30.0000pt)..(-22.7300pt,-30.0000pt);
\path[-,draw=black,line width=0.4000pt,line cap=butt,line join=miter,dash pattern=](-32.3300pt,-35.0000pt)--(-22.7300pt,-35.0000pt);
\path[-,draw=black,line width=0.4000pt,line cap=butt,line join=miter,dash pattern=](-41.9300pt,-20.0000pt)..controls(-33.9300pt,-20.0000pt)and(-40.3300pt,-25.0000pt)..(-32.3300pt,-25.0000pt);
\path[-,draw=black,line width=0.4000pt,line cap=butt,line join=miter,dash pattern=](-41.9300pt,-25.0000pt)..controls(-33.9300pt,-25.0000pt)and(-40.3300pt,-20.0000pt)..(-32.3300pt,-20.0000pt);
\path[-,draw=black,line width=0.4000pt,line cap=butt,line join=miter,dash pattern=](-41.9300pt,-35.0000pt)--(-32.3300pt,-35.0000pt);
\path[-,draw=black,line width=0.4000pt,line cap=butt,line join=miter,dash pattern=](-51.5300pt,-20.0000pt)--(-41.9300pt,-20.0000pt);
\path[-,draw=black,line width=0.4000pt,line cap=butt,line join=miter,dash pattern=](-51.5300pt,-30.0000pt)..controls(-43.5300pt,-30.0000pt)and(-49.9300pt,-25.0000pt)..(-41.9300pt,-25.0000pt);
\path[-,draw=black,line width=0.4000pt,line cap=butt,line join=miter,dash pattern=](-51.5300pt,-35.0000pt)--(-41.9300pt,-35.0000pt);
\path[-,draw=black,line width=0.4000pt,line cap=butt,line join=miter,dash pattern=](-66.3350pt,-20.0000pt)--(-66.3350pt,-20.0000pt);
\path[-,draw=black,line width=0.4000pt,line cap=butt,line join=miter,dash pattern=](-66.3350pt,-30.0000pt)--(-66.3350pt,-30.0000pt);
\path[-,draw=black,line width=0.4000pt,line cap=butt,line join=miter,dash pattern=](-66.3350pt,-35.0000pt)--(-66.3350pt,-35.0000pt);
\path[-,draw=black,line width=0.4000pt,line cap=butt,line join=miter,dash pattern=](-66.3350pt,-20.0000pt)--(-51.5300pt,-20.0000pt);
\path[-,line width=0.4000pt,line cap=butt,line join=miter,dash pattern=](-62.3350pt,-30.2550pt)--(-55.5300pt,-30.2550pt)--(-55.5300pt,-34.7450pt)--(-62.3350pt,-34.7450pt)--cycle;
\node[anchor=north west,inner sep=0] at (-62.3350pt,-30.2550pt){\savebox{\marxupbox}{{\ensuremath{ω}}}\immediate\write\boxesfile{12}\immediate\write\boxesfile{\number\wd\marxupbox}\immediate\write\boxesfile{\number\ht\marxupbox}\immediate\write\boxesfile{\number\dp\marxupbox}\box\marxupbox};
\path[-,line width=0.4000pt,line cap=butt,line join=miter,dash pattern=](-66.3350pt,-26.2550pt)--(-51.5300pt,-26.2550pt)--(-51.5300pt,-38.7450pt)--(-66.3350pt,-38.7450pt)--cycle;
\path[-,draw=black,line width=0.4000pt,line cap=butt,line join=miter,dash pattern=](-66.3350pt,-28.5000pt)..controls(-66.3350pt,-27.1193pt)and(-65.2157pt,-26.0000pt)..(-63.8350pt,-26.0000pt)--(-54.0300pt,-26.0000pt)..controls(-52.6493pt,-26.0000pt)and(-51.5300pt,-27.1193pt)..(-51.5300pt,-28.5000pt)--(-51.5300pt,-36.5000pt)..controls(-51.5300pt,-37.8807pt)and(-52.6493pt,-39.0000pt)..(-54.0300pt,-39.0000pt)--(-63.8350pt,-39.0000pt)..controls(-65.2157pt,-39.0000pt)and(-66.3350pt,-37.8807pt)..(-66.3350pt,-36.5000pt)--cycle;
\path[-,draw=black,line width=0.4000pt,line cap=butt,line join=miter,dash pattern=](-51.5300pt,-30.0000pt)--(-51.5300pt,-30.0000pt);
\path[-,draw=black,line width=0.4000pt,line cap=butt,line join=miter,dash pattern=](-51.5300pt,-35.0000pt)--(-51.5300pt,-35.0000pt);
\path[-,draw=black,line width=0.4000pt,line cap=butt,line join=miter,dash pattern=](-66.3350pt,-30.0000pt)--(-66.3350pt,-30.0000pt);
\path[-,draw=black,line width=0.4000pt,line cap=butt,line join=miter,dash pattern=](-66.3350pt,-35.0000pt)--(-66.3350pt,-35.0000pt);
\path[-,draw=black,line width=0.4000pt,line cap=butt,line join=miter,dash pattern=](-75.9350pt,-20.0000pt)--(-66.3350pt,-20.0000pt);
\path[-,draw=black,line width=0.4000pt,line cap=butt,line join=miter,dash pattern=](-75.9350pt,-25.0000pt)..controls(-67.9350pt,-25.0000pt)and(-74.3350pt,-30.0000pt)..(-66.3350pt,-30.0000pt);
\path[-,draw=black,line width=0.4000pt,line cap=butt,line join=miter,dash pattern=](-75.9350pt,-35.0000pt)--(-66.3350pt,-35.0000pt);
\path[-,draw=black,line width=0.4000pt,line cap=butt,line join=miter,dash pattern=](-85.5350pt,-20.0000pt)..controls(-77.5350pt,-20.0000pt)and(-83.9350pt,-25.0000pt)..(-75.9350pt,-25.0000pt);
\path[-,draw=black,line width=0.4000pt,line cap=butt,line join=miter,dash pattern=](-85.5350pt,-25.0000pt)..controls(-77.5350pt,-25.0000pt)and(-83.9350pt,-20.0000pt)..(-75.9350pt,-20.0000pt);
\path[-,draw=black,line width=0.4000pt,line cap=butt,line join=miter,dash pattern=](-85.5350pt,-35.0000pt)--(-75.9350pt,-35.0000pt);
\path[-,draw=black,line width=0.4000pt,line cap=butt,line join=miter,dash pattern=](-95.1350pt,0.0000pt)--(-85.5350pt,0.0000pt);
\path[-,draw=black,line width=0.4000pt,line cap=butt,line join=miter,dash pattern=](-95.1350pt,-10.0000pt)..controls(-87.1350pt,-10.0000pt)and(-93.5350pt,-20.0000pt)..(-85.5350pt,-20.0000pt);
\path[-,draw=black,line width=0.4000pt,line cap=butt,line join=miter,dash pattern=](-95.1350pt,-15.0000pt)..controls(-87.1350pt,-15.0000pt)and(-93.5350pt,-25.0000pt)..(-85.5350pt,-25.0000pt);
\path[-,draw=black,line width=0.4000pt,line cap=butt,line join=miter,dash pattern=](-95.1350pt,-35.0000pt)--(-85.5350pt,-35.0000pt);
\path[-,line width=0.4000pt,line cap=butt,line join=miter,dash pattern=](-105.0099pt,-2.8075pt)--(-99.1350pt,-2.8075pt)--(-99.1350pt,-12.1925pt)--(-105.0099pt,-12.1925pt)--cycle;
\node[anchor=north west,inner sep=0] at (-105.0099pt,-2.8075pt){\savebox{\marxupbox}{{\ensuremath{φ}}}\immediate\write\boxesfile{13}\immediate\write\boxesfile{\number\wd\marxupbox}\immediate\write\boxesfile{\number\ht\marxupbox}\immediate\write\boxesfile{\number\dp\marxupbox}\box\marxupbox};
\path[-,line width=0.4000pt,line cap=butt,line join=miter,dash pattern=](-109.0099pt,1.1925pt)--(-95.1350pt,1.1925pt)--(-95.1350pt,-16.1925pt)--(-109.0099pt,-16.1925pt)--cycle;
\path[-,draw=black,line width=0.4000pt,line cap=butt,line join=miter,dash pattern=](-109.0099pt,1.5000pt)..controls(-109.0099pt,2.8807pt)and(-107.8906pt,4.0000pt)..(-106.5099pt,4.0000pt)--(-97.6350pt,4.0000pt)..controls(-96.2543pt,4.0000pt)and(-95.1350pt,2.8807pt)..(-95.1350pt,1.5000pt)--(-95.1350pt,-16.5000pt)..controls(-95.1350pt,-17.8807pt)and(-96.2543pt,-19.0000pt)..(-97.6350pt,-19.0000pt)--(-106.5099pt,-19.0000pt)..controls(-107.8906pt,-19.0000pt)and(-109.0099pt,-17.8807pt)..(-109.0099pt,-16.5000pt)--cycle;
\path[-,draw=black,line width=0.4000pt,line cap=butt,line join=miter,dash pattern=](-95.1350pt,0.0000pt)--(-95.1350pt,0.0000pt);
\path[-,draw=black,line width=0.4000pt,line cap=butt,line join=miter,dash pattern=](-95.1350pt,-10.0000pt)--(-95.1350pt,-10.0000pt);
\path[-,draw=black,line width=0.4000pt,line cap=butt,line join=miter,dash pattern=](-95.1350pt,-15.0000pt)--(-95.1350pt,-15.0000pt);
\path[-,draw=black,line width=0.4000pt,line cap=butt,line join=miter,dash pattern=](-109.0099pt,-7.5000pt)--(-109.0099pt,-7.5000pt);
\path[-,draw=black,line width=0.4000pt,line cap=butt,line join=miter,dash pattern=](-109.0099pt,-35.0000pt)--(-95.1350pt,-35.0000pt);
\end{tikzpicture}}\end{center}\caption{\ensuremath{\allowbreak{}\mathnormal{(}\mskip 0.0mu}\ensuremath{\mathsf{ξ}\mskip 3.0mu}\ensuremath{\allowbreak{}\mathnormal{×}\allowbreak{}\mskip 3.0mu}\ensuremath{\mathsf{ζ}\mskip 0.0mu}\ensuremath{\mathnormal{)}\allowbreak{}\mskip 3.0mu}\ensuremath{\allowbreak{}\mathnormal{∘}\allowbreak{}\mskip 3.0mu}\ensuremath{\bar{α}\mskip 3.0mu}\ensuremath{\allowbreak{}\mathnormal{∘}\allowbreak{}\mskip 3.0mu}\ensuremath{\allowbreak{}\mathnormal{(}\mskip 0.0mu}\ensuremath{\mathsf{id}\mskip 3.0mu}\ensuremath{\allowbreak{}\mathnormal{×}\allowbreak{}\mskip 3.0mu}\ensuremath{\allowbreak{}\mathnormal{(}\mskip 0.0mu}\ensuremath{α\mskip 3.0mu}\ensuremath{\allowbreak{}\mathnormal{∘}\allowbreak{}\mskip 3.0mu}\ensuremath{\allowbreak{}\mathnormal{(}\mskip 0.0mu}\ensuremath{σ\mskip 3.0mu}\ensuremath{\allowbreak{}\mathnormal{×}\allowbreak{}\mskip 3.0mu}\ensuremath{\mathsf{id}\mskip 0.0mu}\ensuremath{\mathnormal{)}\allowbreak{}\mskip 3.0mu}\ensuremath{\allowbreak{}\mathnormal{∘}\allowbreak{}\mskip 3.0mu}\ensuremath{\bar{α}\mskip 3.0mu}\ensuremath{\allowbreak{}\mathnormal{∘}\allowbreak{}\mskip 3.0mu}\ensuremath{\allowbreak{}\mathnormal{(}\mskip 0.0mu}\ensuremath{\mathsf{id}\mskip 3.0mu}\ensuremath{\allowbreak{}\mathnormal{×}\allowbreak{}\mskip 3.0mu}\ensuremath{\mathsf{ω}\mskip 0.0mu}\ensuremath{\mathnormal{)}\allowbreak{}\mskip 3.0mu}\ensuremath{\allowbreak{}\mathnormal{∘}\allowbreak{}\mskip 3.0mu}\ensuremath{α\mskip 3.0mu}\ensuremath{\allowbreak{}\mathnormal{∘}\allowbreak{}\mskip 3.0mu}\ensuremath{\allowbreak{}\mathnormal{(}\mskip 0.0mu}\ensuremath{σ\mskip 3.0mu}\ensuremath{\allowbreak{}\mathnormal{×}\allowbreak{}\mskip 3.0mu}\ensuremath{\mathsf{id}\mskip 0.0mu}\ensuremath{\mathnormal{)}\allowbreak{}\mskip 0.0mu}\ensuremath{\mathnormal{)}\allowbreak{}\mskip 0.0mu}\ensuremath{\mathnormal{)}\allowbreak{}\mskip 3.0mu}\ensuremath{\allowbreak{}\mathnormal{∘}\allowbreak{}\mskip 3.0mu}\ensuremath{α\mskip 3.0mu}\ensuremath{\allowbreak{}\mathnormal{∘}\allowbreak{}\mskip 3.0mu}\ensuremath{\allowbreak{}\mathnormal{(}\mskip 0.0mu}\ensuremath{\mathsf{φ}\mskip 3.0mu}\ensuremath{\allowbreak{}\mathnormal{×}\allowbreak{}\mskip 3.0mu}\ensuremath{\mathsf{id}\mskip 0.0mu}\ensuremath{\mathnormal{)}\allowbreak{}}}\label{14}\end{subfigure}\caption{A few {\sc{}smc} morphisms, their encoding as functions, and their string diagram representations.}\label{15}\end{figure*} 
\paragraph{Parable}\label{16} 
  Frankie is designing a
domain-specific language ({\sc{}dsl}), and by working out examples
on paper, realises that the best way to describe objects in that
{\sc{}dsl} is by box-and-wires diagrams, similar to those in
\cref{15}. The story does not say what Frankie
intends to use the {\sc{}dsl} for. Maybe it has to do with linear algebra,
parallel computing, or even quantum computations (see
\cref{36}): this kind of
pattern occurs in many contexts.  Following accepted
functional programming methodologies, Frankie searches for the right
abstraction and finds out that Symmetric Monoidal Categories ({\sc{}smc} for
short) capture said diagrams precisely
\citep[Section 3]{selinger_string_diagram_survey_2011}.
Accordingly, Frankie starts coding examples using the combinators of {\sc{}smc}s (\cref{20}),
but disappointment is great after writing a few examples: everything is expressed in point-free
style, resulting in cryptic expressions such as \ensuremath{\allowbreak{}\mathnormal{(}\mskip 0.0mu}\ensuremath{\mathsf{ξ}\mskip 3.0mu}\ensuremath{\allowbreak{}\mathnormal{×}\allowbreak{}\mskip 3.0mu}\ensuremath{\mathsf{ζ}\mskip 0.0mu}\ensuremath{\mathnormal{)}\allowbreak{}\mskip 3.0mu}\ensuremath{\allowbreak{}\mathnormal{∘}\allowbreak{}\mskip 3.0mu}\ensuremath{\bar{α}\mskip 3.0mu}\ensuremath{\allowbreak{}\mathnormal{∘}\allowbreak{}\mskip 3.0mu}\ensuremath{\allowbreak{}\mathnormal{(}\mskip 0.0mu}\ensuremath{\mathsf{id}\mskip 3.0mu}\ensuremath{\allowbreak{}\mathnormal{×}\allowbreak{}\mskip 3.0mu}\ensuremath{\allowbreak{}\mathnormal{(}\mskip 0.0mu}\ensuremath{α\mskip 3.0mu}\ensuremath{\allowbreak{}\mathnormal{∘}\allowbreak{}\mskip 3.0mu}\ensuremath{\allowbreak{}\mathnormal{(}\mskip 0.0mu}\ensuremath{σ\mskip 3.0mu}\ensuremath{\allowbreak{}\mathnormal{×}\allowbreak{}\mskip 3.0mu}\ensuremath{\mathsf{id}\mskip 0.0mu}\ensuremath{\mathnormal{)}\allowbreak{}\mskip 3.0mu}\ensuremath{\allowbreak{}\mathnormal{∘}\allowbreak{}\mskip 3.0mu}\ensuremath{\bar{α}\mskip 3.0mu}\ensuremath{\allowbreak{}\mathnormal{∘}\allowbreak{}\mskip 3.0mu}\ensuremath{\allowbreak{}\mathnormal{(}\mskip 0.0mu}\ensuremath{\mathsf{id}\mskip 3.0mu}\ensuremath{\allowbreak{}\mathnormal{×}\allowbreak{}\mskip 3.0mu}\ensuremath{\mathsf{ω}\mskip 0.0mu}\ensuremath{\mathnormal{)}\allowbreak{}\mskip 3.0mu}\ensuremath{\allowbreak{}\mathnormal{∘}\allowbreak{}\mskip 3.0mu}\ensuremath{α\mskip 3.0mu}\ensuremath{\allowbreak{}\mathnormal{∘}\allowbreak{}\mskip 3.0mu}\ensuremath{\allowbreak{}\mathnormal{(}\mskip 0.0mu}\ensuremath{σ\mskip 3.0mu}\ensuremath{\allowbreak{}\mathnormal{×}\allowbreak{}\mskip 3.0mu}\ensuremath{\mathsf{id}\mskip 0.0mu}\ensuremath{\mathnormal{)}\allowbreak{}\mskip 0.0mu}\ensuremath{\mathnormal{)}\allowbreak{}\mskip 0.0mu}\ensuremath{\mathnormal{)}\allowbreak{}\mskip 3.0mu}\ensuremath{\allowbreak{}\mathnormal{∘}\allowbreak{}\mskip 3.0mu}\ensuremath{α\mskip 3.0mu}\ensuremath{\allowbreak{}\mathnormal{∘}\allowbreak{}\mskip 3.0mu}\ensuremath{\allowbreak{}\mathnormal{(}\mskip 0.0mu}\ensuremath{\mathsf{φ}\mskip 3.0mu}\ensuremath{\allowbreak{}\mathnormal{×}\allowbreak{}\mskip 3.0mu}\ensuremath{\mathsf{id}\mskip 0.0mu}\ensuremath{\mathnormal{)}\allowbreak{}} for the
boxes-and-wires diagram of \cref{14}. It becomes obvious
to Frankie why programming languages have variables: in
a language with variables, the same example can be expressed much more
directly. Something like:

\begin{list}{}{\setlength\leftmargin{1.0em}}\item\relax
 \ensuremath{\begin{parray}\column{B}{@{}>{}l<{}@{}}\column[0em]{1}{@{}>{}l<{}@{}}\column[1em]{2}{@{}>{}l<{}@{}}\column{3}{@{}>{}l<{}@{}}\column{4}{@{}>{}l<{}@{}}\column{E}{@{}>{}l<{}@{}}%
\>[1]{}{\mathsf{ex}_{3}\mskip 3.0mu\allowbreak{}\mathnormal{(}\mskip 0.0mu\mathsf{a}\mskip 2.0mu\mathnormal{\fatsemi }\mskip 3.0mu\mathsf{z}\mskip 0.0mu\mathnormal{)}\allowbreak{}\mskip 3.0mu\mathnormal{=}\mskip 3.0mu\mathsf{ξ}\mskip 3.0mu\allowbreak{}\mathnormal{(}\mskip 0.0mu\mathsf{y}\mskip 2.0mu\mathnormal{\fatsemi }\mskip 3.0mu\mathsf{c}\mskip 0.0mu\mathnormal{)}\allowbreak{}\mskip 2.0mu}\>[4]{}{\mathnormal{\fatsemi }\mskip 3.0mu\mathsf{ζ}\mskip 3.0mu\allowbreak{}\mathnormal{(}\mskip 0.0mu\mathsf{w}\mskip 2.0mu\mathnormal{\fatsemi }\mskip 3.0mu\mathsf{d}\mskip 0.0mu\mathnormal{)}\allowbreak{}}\<[E]{}\\
\>[2]{}{\mathbf{where}\mskip 3.0mu}\>[3]{}{\allowbreak{}\mathnormal{(}\mskip 0.0mu\mathsf{y}\mskip 2.0mu\mathnormal{\fatsemi }\mskip 3.0mu\mathsf{x}\mskip 2.0mu\mathnormal{\fatsemi }\mskip 3.0mu\mathsf{w}\mskip 0.0mu\mathnormal{)}\allowbreak{}\mskip 3.0mu\mathnormal{=}\mskip 3.0mu\mathsf{φ}\mskip 3.0mu\mathsf{a}}\<[E]{}\\
\>[3]{}{\allowbreak{}\mathnormal{(}\mskip 0.0mu\mathsf{c}\mskip 2.0mu\mathnormal{\fatsemi }\mskip 3.0mu\mathsf{d}\mskip 0.0mu\mathnormal{)}\allowbreak{}\mskip 3.0mu\mathnormal{=}\mskip 3.0mu\mathsf{ω}\mskip 3.0mu\allowbreak{}\mathnormal{(}\mskip 0.0mu\mathsf{x}\mskip 2.0mu\mathnormal{\fatsemi }\mskip 3.0mu\mathsf{z}\mskip 0.0mu\mathnormal{)}\allowbreak{}}\<[E]{}\end{parray}} \end{list} 
Now, Frankie could roll-out a special-purpose language for {\sc{}smc}s with
variables, together with some compiler, and integrate it into company
praxis. But this would be quite costly! For instance, Frankie would
have to figure out how to share objects between the {\sc{}dsl} and the host
programs. Deploying one's own compiler can be a tricky business.

But is it, really, Frankie's only choice? Either drop lambda notation and use
point-free style, or use a special-purpose compiler to translate from
lambda notation to {\sc{}smc}s?  In this paper, we demonstrate that no
compromise is necessary: Frankie can use usual functions to encode
diagrams.  Specifically, we show how to evaluate \emph{linear} functions to {\sc{}smc} expressions. We do so by pure evaluation within Haskell.
We require no external tool, no modification to the
compiler nor metaprogramming of any kind.
This makes our solution particularly lightweight, and applicable to every
functional programming language that supports linear types.
Even though we specifically target Linear Haskell
\citep{bernardy_linear_2018}, our technique works in any other
functional languages with linear types, such as Idris2
\citep{brady_idris_2020} or Granule \citep{orchard_quantitative_2019}.

We make the following contributions:
\begin{itemize}\item{}We give a linearly typed {\sc{}api} to construct {\sc{}smc} morphisms
 (\cref{33}).  This {\sc{}api} is only 5 functions long and
 allows the programmer to use the name-binding features of Haskell to
 name intermediate results.\item{}We demonstrate with concrete applications how our {\sc{}api} lets one use
Haskell's functions and variables to concisely define
{\sc{}smc} morphisms (\cref{36}).\item{}We describe an implementation of our {\sc{}api}, and prove its
correctness (\cref{60}).

This implementation was tested on all the
examples shown in this paper. In particular, whenever we show a
function and a corresponding diagram, as in \cref{15}, our
library was used to automatically
generate an {\sc{}smc} representation, which was
in turn converted to a diagram, and imported to the
\LaTeX{} source code of the paper. In this sense, this paper is
self-testing.

The library is available on the Hackage repository: \url{https://hackage.haskell.org/package/linear-smc}.
\end{itemize} The rest of the paper discusses salient points and related work
(\cref{139}), before concluding in \cref{146}.  Before
any of this, we review the underlying concepts and introduce our
notations for them (\cref{17}).

\section{Notations and conventions}\label{17} 
In this section we recall the notions of category theory necessary to
follow our development and examples. In addition we explain our
notation for morphisms and conventions for diagrams.

\subsection{Categories}\label{18} \begin{figure*}\begin{subfigure}[b]{0.45\textwidth}\ensuremath{\begin{parray}\column{B}{@{}>{}l<{}@{}}\column[0em]{1}{@{}>{}l<{}@{}}\column[1em]{2}{@{}>{}l<{}@{}}\column{3}{@{}>{}l<{}@{}}\column{4}{@{}>{}l<{}@{}}\column[4em]{5}{@{}>{}l<{}@{}}\column{E}{@{}>{}l<{}@{}}%
\>[1]{}{\mathbf{class}\mskip 3.0mu\mathsf{Category}\mskip 3.0mu\mathsf{k}\mskip 3.0mu\mathbf{where}}\<[E]{}\\
\>[2]{}{\mathbf{type}\mskip 3.0mu\mathsf{Obj}\mskip 3.0mu\mathsf{k}\mskip 3.0mu\mathnormal{::}\mskip 3.0mu\mathsf{Type}\mskip 3.0mu\mathnormal{\rightarrow }\mskip 3.0mu\mathsf{Constraint}}\<[E]{}\\
\>[2]{}{\mathsf{id}\mskip 3.0mu}\>[3]{}{\mathnormal{::}\mskip 3.0mu\mathsf{Obj}\mskip 3.0mu\mathsf{k}\mskip 3.0mu\mathsf{a}\mskip 3.0mu\mathnormal{\Rightarrow }\mskip 3.0mu\mathsf{a}\mskip 0.0mu\overset{\mathsf{k}}{\leadsto}\mskip 0.0mu\mathsf{a}}\<[E]{}\\
\>[2]{}{\allowbreak{}\mathnormal{(}\mskip 0.0mu\allowbreak{}\mathnormal{∘}\allowbreak{}\mskip 0.0mu\mathnormal{)}\allowbreak{}\mskip 3.0mu}\>[3]{}{\mathnormal{::}\mskip 3.0mu}\>[4]{}{\allowbreak{}\mathnormal{(}\mskip 0.0mu\mathsf{Obj}\mskip 3.0mu\mathsf{k}\mskip 3.0mu\mathsf{a}\mskip 0.0mu\mathnormal{,}\mskip 3.0mu\mathsf{Obj}\mskip 3.0mu\mathsf{k}\mskip 3.0mu\mathsf{b}\mskip 0.0mu\mathnormal{,}\mskip 3.0mu\mathsf{Obj}\mskip 3.0mu\mathsf{k}\mskip 3.0mu\mathsf{c}\mskip 0.0mu\mathnormal{)}\allowbreak{}\mskip 3.0mu\mathnormal{\Rightarrow }}\<[E]{}\\
\>[5]{}{\allowbreak{}\mathnormal{(}\mskip 0.0mu\mathsf{b}\mskip 0.0mu\overset{\mathsf{k}}{\leadsto}\mskip 0.0mu\mathsf{c}\mskip 0.0mu\mathnormal{)}\allowbreak{}\mskip 3.0mu\mathnormal{\rightarrow }\mskip 3.0mu\allowbreak{}\mathnormal{(}\mskip 0.0mu\mathsf{a}\mskip 0.0mu\overset{\mathsf{k}}{\leadsto}\mskip 0.0mu\mathsf{b}\mskip 0.0mu\mathnormal{)}\allowbreak{}\mskip 3.0mu\mathnormal{\rightarrow }\mskip 3.0mu\mathsf{a}\mskip 0.0mu\overset{\mathsf{k}}{\leadsto}\mskip 0.0mu\mathsf{c}}\<[E]{}\end{parray}} \caption{Category structure}\label{19}\end{subfigure}\quad{}\begin{subfigure}[b]{0.45\textwidth} \ensuremath{\begin{parray}\column{B}{@{}>{}l<{}@{}}\column[0em]{1}{@{}>{}l<{}@{}}\column[1em]{2}{@{}>{}l<{}@{}}\column{3}{@{}>{}l<{}@{}}\column{4}{@{}>{}l<{}@{}}\column{5}{@{}>{}l<{}@{}}\column{6}{@{}>{}l<{}@{}}\column{7}{@{}>{}l<{}@{}}\column{8}{@{}>{}l<{}@{}}\column{E}{@{}>{}l<{}@{}}%
\>[1]{}{\mathbf{class}\mskip 3.0mu\allowbreak{}\mathnormal{(}\mskip 0.0mu}\>[4]{}{\mathsf{Category}\mskip 3.0mu\mathsf{k}\mskip 0.0mu\mathnormal{)}\allowbreak{}\mskip 3.0mu\mathnormal{\Rightarrow }\mskip 3.0mu\mathsf{Monoidal}\mskip 3.0mu\mathsf{k}\mskip 3.0mu\mathbf{where}}\<[E]{}\\
\>[2]{}{\allowbreak{}\mathnormal{(}\mskip 0.0mu\allowbreak{}\mathnormal{×}\allowbreak{}\mskip 0.0mu\mathnormal{)}\allowbreak{}\mskip 3.0mu}\>[3]{}{\mathnormal{::}\mskip 3.0mu}\>[8]{}{\allowbreak{}\mathnormal{(}\mskip 0.0mu\mathsf{a}\mskip 0.0mu\overset{\mathsf{k}}{\leadsto}\mskip 0.0mu\mathsf{b}\mskip 0.0mu\mathnormal{)}\allowbreak{}\mskip 3.0mu\mathnormal{\rightarrow }\mskip 3.0mu\allowbreak{}\mathnormal{(}\mskip 0.0mu\mathsf{c}\mskip 0.0mu\overset{\mathsf{k}}{\leadsto}\mskip 0.0mu\mathsf{d}\mskip 0.0mu\mathnormal{)}\allowbreak{}\mskip 3.0mu\mathnormal{\rightarrow }\mskip 3.0mu\allowbreak{}\mathnormal{(}\mskip 0.0mu\mathsf{a}\mskip 3.0mu\mathnormal{⊗}\mskip 3.0mu\mathsf{c}\mskip 0.0mu\mathnormal{)}\allowbreak{}\mskip 0.0mu\overset{\mathsf{k}}{\leadsto}\mskip 0.0mu\allowbreak{}\mathnormal{(}\mskip 0.0mu\mathsf{b}\mskip 3.0mu\mathnormal{⊗}\mskip 3.0mu\mathsf{d}\mskip 0.0mu\mathnormal{)}\allowbreak{}}\<[E]{}\\
\>[2]{}{σ\mskip 3.0mu}\>[3]{}{\mathnormal{::}\mskip 3.0mu}\>[6]{}{\allowbreak{}\mathnormal{(}\mskip 0.0mu\mathsf{a}\mskip 3.0mu\mathnormal{⊗}\mskip 3.0mu\mathsf{b}\mskip 0.0mu\mathnormal{)}\allowbreak{}\mskip 0.0mu\overset{\mathsf{k}}{\leadsto}\mskip 0.0mu\allowbreak{}\mathnormal{(}\mskip 0.0mu\mathsf{b}\mskip 3.0mu\mathnormal{⊗}\mskip 3.0mu\mathsf{a}\mskip 0.0mu\mathnormal{)}\allowbreak{}}\<[E]{}\\
\>[2]{}{α\mskip 3.0mu}\>[3]{}{\mathnormal{::}\mskip 3.0mu}\>[7]{}{\allowbreak{}\mathnormal{(}\mskip 0.0mu\allowbreak{}\mathnormal{(}\mskip 0.0mu\mathsf{a}\mskip 3.0mu\mathnormal{⊗}\mskip 3.0mu\mathsf{b}\mskip 0.0mu\mathnormal{)}\allowbreak{}\mskip 3.0mu\mathnormal{⊗}\mskip 3.0mu\mathsf{c}\mskip 0.0mu\mathnormal{)}\allowbreak{}\mskip 0.0mu\overset{\mathsf{k}}{\leadsto}\mskip 0.0mu\allowbreak{}\mathnormal{(}\mskip 0.0mu\mathsf{a}\mskip 3.0mu\mathnormal{⊗}\mskip 3.0mu\allowbreak{}\mathnormal{(}\mskip 0.0mu\mathsf{b}\mskip 3.0mu\mathnormal{⊗}\mskip 3.0mu\mathsf{c}\mskip 0.0mu\mathnormal{)}\allowbreak{}\mskip 0.0mu\mathnormal{)}\allowbreak{}}\<[E]{}\\
\>[2]{}{\bar{α}\mskip 3.0mu}\>[3]{}{\mathnormal{::}\mskip 3.0mu}\>[7]{}{\allowbreak{}\mathnormal{(}\mskip 0.0mu\mathsf{a}\mskip 3.0mu\mathnormal{⊗}\mskip 3.0mu\allowbreak{}\mathnormal{(}\mskip 0.0mu\mathsf{b}\mskip 3.0mu\mathnormal{⊗}\mskip 3.0mu\mathsf{c}\mskip 0.0mu\mathnormal{)}\allowbreak{}\mskip 0.0mu\mathnormal{)}\allowbreak{}\mskip 0.0mu\overset{\mathsf{k}}{\leadsto}\mskip 0.0mu\allowbreak{}\mathnormal{(}\mskip 0.0mu\allowbreak{}\mathnormal{(}\mskip 0.0mu\mathsf{a}\mskip 3.0mu\mathnormal{⊗}\mskip 3.0mu\mathsf{b}\mskip 0.0mu\mathnormal{)}\allowbreak{}\mskip 3.0mu\mathnormal{⊗}\mskip 3.0mu\mathsf{c}\mskip 0.0mu\mathnormal{)}\allowbreak{}}\<[E]{}\\
\>[2]{}{ρ\mskip 3.0mu}\>[3]{}{\mathnormal{::}\mskip 3.0mu}\>[5]{}{\mathsf{a}\mskip 0.0mu\overset{\mathsf{k}}{\leadsto}\mskip 0.0mu\allowbreak{}\mathnormal{(}\mskip 0.0mu\mathsf{a}\mskip 3.0mu\mathnormal{⊗}\mskip 3.0mu\allowbreak{}\mathnormal{(}\mskip 0.0mu\mathnormal{)}\allowbreak{}\mskip 0.0mu\mathnormal{)}\allowbreak{}}\<[E]{}\\
\>[2]{}{\bar{ρ}\mskip 3.0mu}\>[3]{}{\mathnormal{::}\mskip 3.0mu}\>[5]{}{\allowbreak{}\mathnormal{(}\mskip 0.0mu\mathsf{a}\mskip 3.0mu\mathnormal{⊗}\mskip 3.0mu\allowbreak{}\mathnormal{(}\mskip 0.0mu\mathnormal{)}\allowbreak{}\mskip 0.0mu\mathnormal{)}\allowbreak{}\mskip 0.0mu\overset{\mathsf{k}}{\leadsto}\mskip 0.0mu\mathsf{a}}\<[E]{}\end{parray}} \caption{Symmetric Monoidal Category structure}\label{20}\end{subfigure}\caption{Categorical structures.}\label{21}\end{figure*} 
The fundamental structure is that of a category (\cref{19}).
In general a category \ensuremath{\mathsf{k}} is composed of objects and morphisms, but here we
take objects to be types satisfying a specific constraint
\ensuremath{\mathsf{Obj}}. This choice is convenient because it lets us make the type
of Haskell functions an instance of the \ensuremath{\mathsf{Category}} class. A morphism from \ensuremath{\mathsf{a}} to \ensuremath{\mathsf{b}} is a value of type \ensuremath{\mathsf{k}\mskip 3.0mu}\ensuremath{\mathsf{a}\mskip 3.0mu}\ensuremath{\mathsf{b}}, which
we suggestively note \ensuremath{\mathsf{a}\mskip 0.0mu}\ensuremath{\overset{\mathsf{k}}{\leadsto}\mskip 0.0mu}\ensuremath{\mathsf{b}}. Categories are additionally equipped with an identity at
every type (\ensuremath{\mathsf{id}}), which is represented in diagrams as a line. Additionally, categories have morphism composition (∘),
represented by connecting morphisms with a line
(\cref{28}). This representation neatly captures the laws of
categories: morphisms are equivalent iff they are represented by
topologically equivalent diagrams. (For instance, composing with the
identity simply makes a line longer, and stretching a line is a
topology-preserving transformation.)  In this paper we follow the
usual convention for the directions, even though it means that the
layout of diagrams is inverse to that of Haskell expressions. That is, one can think
of information as flowing from right-to-left in the expression \ensuremath{\mathsf{f}\mskip 3.0mu}\ensuremath{\allowbreak{}\mathnormal{∘}\allowbreak{}\mskip 3.0mu}\ensuremath{\mathsf{g}}, but left-to-right in the diagram representing it.

\begin{figure*}
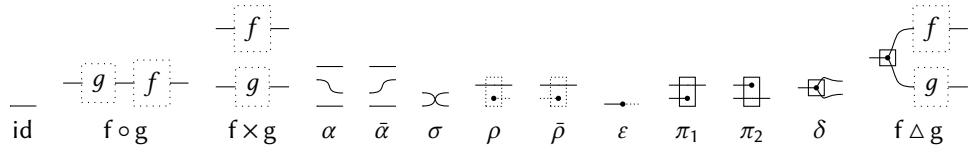

\caption{Diagram-Morphism correspondence.}\label{28}\end{figure*} 
Even though many applications depend crucially on \ensuremath{\mathsf{Obj}} constraints, they are often lengthy, and orthogonal to our main
points.  Thus, to minimise clutter, most of
the time we omit these \ensuremath{\mathsf{Obj}} constraints. To recover them, one
should add an \ensuremath{\mathsf{Obj}} constraint for every relevant type variable,
as well as for the unit type.  Additionally, for {\sc{}smc}s (introduced in
\cref{29} below), one needs closure under monoidal product.

\subsection{Symmetric Monoidal Categories}\label{29} 
Our main objects of study are Symmetric Monoidal Categories
(abbreviated as {\sc{}smc} throughout the paper). They feature a unit object
and the monoidal product (often also called tensor product), written \ensuremath{\mathsf{a}\mskip 3.0mu}\ensuremath{\mathnormal{⊗}\mskip 3.0mu}\ensuremath{\mathsf{b}}. In general the unit can be any type, and
the product can be any type family, but it is sufficient for our applications to let the
unit object be the unit type (written \ensuremath{\allowbreak{}\mathnormal{(}\mskip 0.0mu}\ensuremath{\mathnormal{)}\allowbreak{}}) and the monoidal
product as the product type of Haskell \ensuremath{\allowbreak{}\mathnormal{(}\mskip 0.0mu}\ensuremath{\mathsf{a}\mskip 0.0mu}\ensuremath{\mathnormal{,}\mskip 3.0mu}\ensuremath{\mathsf{b}\mskip 0.0mu}\ensuremath{\mathnormal{)}\allowbreak{}}.  {\sc{}Smc}s provide a number of
ways to manipulate the product of objects.  First, arbitrary morphisms
\ensuremath{\mathsf{f}\mskip 3.0mu}\ensuremath{\mathnormal{:}\mskip 3.0mu}\ensuremath{\mathsf{a}\mskip 0.0mu}\ensuremath{\overset{\mathsf{k}}{\leadsto}\mskip 0.0mu}\ensuremath{\mathsf{b}} and \ensuremath{\mathsf{g}\mskip 3.0mu}\ensuremath{\mathnormal{:}\mskip 3.0mu}\ensuremath{\mathsf{c}\mskip 0.0mu}\ensuremath{\overset{\mathsf{k}}{\leadsto}\mskip 0.0mu}\ensuremath{\mathsf{d}} can be combined using the \ensuremath{\allowbreak{}\mathnormal{(}\mskip 0.0mu}\ensuremath{\allowbreak{}\mathnormal{×}\allowbreak{}\mskip 0.0mu}\ensuremath{\mathnormal{)}\allowbreak{}} combinator: \ensuremath{\mathsf{f}\mskip 3.0mu}\ensuremath{\allowbreak{}\mathnormal{×}\allowbreak{}\mskip 3.0mu}\ensuremath{\mathsf{g}\mskip 3.0mu}\ensuremath{\mathnormal{:}\mskip 3.0mu}\ensuremath{\allowbreak{}\mathnormal{(}\mskip 0.0mu}\ensuremath{\mathsf{a}\mskip 3.0mu}\ensuremath{\mathnormal{⊗}\mskip 3.0mu}\ensuremath{\mathsf{c}\mskip 0.0mu}\ensuremath{\mathnormal{)}\allowbreak{}\mskip 0.0mu}\ensuremath{\overset{\mathsf{k}}{\leadsto}\mskip 0.0mu}\ensuremath{\allowbreak{}\mathnormal{(}\mskip 0.0mu}\ensuremath{\mathsf{b}\mskip 3.0mu}\ensuremath{\mathnormal{⊗}\mskip 3.0mu}\ensuremath{\mathsf{d}\mskip 0.0mu}\ensuremath{\mathnormal{)}\allowbreak{}}. This combinator is
most often also called a product.  In this paper we use different
symbols for the product action on morphisms \ensuremath{\mathsf{f}\mskip 3.0mu}\ensuremath{\allowbreak{}\mathnormal{×}\allowbreak{}\mskip 3.0mu}\ensuremath{\mathsf{g}} and on types \ensuremath{\mathsf{a}\mskip 3.0mu}\ensuremath{\mathnormal{⊗}\mskip 3.0mu}\ensuremath{\mathsf{b}}, hopefully minimising confusion. In
diagrams, the product of morphisms is represented by laying out
the diagram representations of the operands on top of each other. This
means that the product morphism has two lines as output and input.
In general we allow drawing parallel lines in place of a single line
if the corresponding object is a monoidal product.  Consequently, the
rest of the combinators —associators (\ensuremath{α} and \ensuremath{\bar{α}}),
unitors (\ensuremath{ρ} and \ensuremath{\bar{ρ}}) and swap (\ensuremath{σ})— can be drawn
as a (small) descriptive network of lines rather than as abstract
boxes.  For instance, more tightly associated products are represented
by closer parallel lines, and the associators (\ensuremath{α} and
\ensuremath{\bar{α}}) regroup lines accordingly.  The purpose of unitors is
to introduce or eliminate the unit object, whose carrying lines are
drawn dotted. Finally the \ensuremath{σ} morphism exchanges objects in a product.
The reader can refer to \cref{20} for a summary, and the
corresponding diagram representations are shown in
\cref{28}. As in the case of simple categories, a great
advantage of this diagrammatic notation is that diagrams which can be
transformed into one another by continuous deformation (including the
removal of disconnected dotted lines) represent equivalent
morphisms. This property makes the laws of {\sc{}smc}s intuitive, and
because they are extensively documented elsewhere
\citep{barr_category_1999}, we won't repeat them here. We clarify
however that lines can pass each other freely: knots are not taken
into account when checking topological equivalence. (For example, two
consecutive \ensuremath{σ} cancel: \ensuremath{σ\mskip 3.0mu}\ensuremath{\allowbreak{}\mathnormal{∘}\allowbreak{}\mskip 3.0mu}\ensuremath{σ\mskip 3.0mu}\ensuremath{\mathnormal{=}\mskip 3.0mu}\ensuremath{\mathsf{id}}.) This property
corresponds to the ``symmetric'' qualifier in ``\emph{Symmetric} Monoidal Categories'', and it is important to us because it means that
one need not worry about the order of binding or use of
variables when using lambda notation to describe morphisms.

\subsection{Cartesian Categories}\label{30} 
\begin{figure} \ensuremath{\begin{parray}\column{B}{@{}>{}l<{}@{}}\column[0em]{1}{@{}>{}l<{}@{}}\column[1em]{2}{@{}>{}l<{}@{}}\column{3}{@{}>{}l<{}@{}}\column{4}{@{}>{}l<{}@{}}\column{E}{@{}>{}l<{}@{}}%
\>[1]{}{\mathbf{class}\mskip 3.0mu\mathsf{Monoidal}\mskip 3.0mu\mathsf{k}\mskip 3.0mu\mathnormal{\Rightarrow }\mskip 3.0mu\mathsf{Cartesian}\mskip 3.0mu\mathsf{k}\mskip 3.0mu\mathbf{where}}\<[E]{}\\
\>[2]{}{π₁\mskip 3.0mu}\>[3]{}{\mathnormal{::}\mskip 3.0mu}\>[4]{}{\allowbreak{}\mathnormal{(}\mskip 0.0mu\mathsf{a}\mskip 3.0mu\mathnormal{⊗}\mskip 3.0mu\mathsf{b}\mskip 0.0mu\mathnormal{)}\allowbreak{}\mskip 0.0mu\overset{\mathsf{k}}{\leadsto}\mskip 0.0mu\mathsf{a}}\<[E]{}\\
\>[2]{}{π₂\mskip 3.0mu}\>[3]{}{\mathnormal{::}\mskip 3.0mu}\>[4]{}{\allowbreak{}\mathnormal{(}\mskip 0.0mu\mathsf{a}\mskip 3.0mu\mathnormal{⊗}\mskip 3.0mu\mathsf{b}\mskip 0.0mu\mathnormal{)}\allowbreak{}\mskip 0.0mu\overset{\mathsf{k}}{\leadsto}\mskip 0.0mu\mathsf{b}}\<[E]{}\\
\>[2]{}{ε\mskip 3.0mu}\>[3]{}{\mathnormal{::}\mskip 3.0mu}\>[4]{}{\mathsf{a}\mskip 0.0mu\overset{\mathsf{k}}{\leadsto}\mskip 0.0mu\allowbreak{}\mathnormal{(}\mskip 0.0mu\mathnormal{)}\allowbreak{}}\<[E]{}\\
\>[2]{}{δ\mskip 3.0mu}\>[3]{}{\mathnormal{::}\mskip 3.0mu}\>[4]{}{\mathsf{a}\mskip 0.0mu\overset{\mathsf{k}}{\leadsto}\mskip 0.0mu\allowbreak{}\mathnormal{(}\mskip 0.0mu\mathsf{a}\mskip 3.0mu\mathnormal{⊗}\mskip 3.0mu\mathsf{a}\mskip 0.0mu\mathnormal{)}\allowbreak{}}\<[E]{}\\
\>[2]{}{\allowbreak{}\mathnormal{(}\mskip 0.0mu\mathnormal{▵}\mskip 0.0mu\mathnormal{)}\allowbreak{}\mskip 3.0mu}\>[3]{}{\mathnormal{::}\mskip 3.0mu}\>[4]{}{\allowbreak{}\mathnormal{(}\mskip 0.0mu\mathsf{a}\mskip 0.0mu\overset{\mathsf{k}}{\leadsto}\mskip 0.0mu\mathsf{b}\mskip 0.0mu\mathnormal{)}\allowbreak{}\mskip 3.0mu\mathnormal{\rightarrow }\mskip 3.0mu\allowbreak{}\mathnormal{(}\mskip 0.0mu\mathsf{a}\mskip 0.0mu\overset{\mathsf{k}}{\leadsto}\mskip 0.0mu\mathsf{c}\mskip 0.0mu\mathnormal{)}\allowbreak{}\mskip 3.0mu\mathnormal{\rightarrow }\mskip 3.0mu\mathsf{a}\mskip 0.0mu\overset{\mathsf{k}}{\leadsto}\mskip 0.0mu\allowbreak{}\mathnormal{(}\mskip 0.0mu\mathsf{b}\mskip 3.0mu\mathnormal{⊗}\mskip 3.0mu\mathsf{c}\mskip 0.0mu\mathnormal{)}\allowbreak{}}\<[E]{}\end{parray}} \caption{Cartesian structure}\label{31}\end{figure} 

Another key concept is that of cartesian categories
(\cref{31}). Even though they are often presented as standalone
structures, we instead present them as a layer on top of {\sc{}smc}s. More
precisely, we add only new morphisms: no new way to \emph{combine morphisms} is necessary. (In the boxes-and-wires metaphor, we add only new boxes, and
no layout rule is added.)  A
minimal set of such new morphisms is comprised of \ensuremath{ε} and
\ensuremath{δ}, which respectively discard and duplicate an
input. However, it is useful to consider alternative presentations,
which can be more convenient, depending on the purpose.  Instead of
\ensuremath{ε}, one can use projections (\ensuremath{\mathsf{π₁}} and \ensuremath{\mathsf{π₂}}), with \ensuremath{π₁\mskip 3.0mu}\ensuremath{\mathnormal{=}\mskip 3.0mu}\ensuremath{\bar{ρ}\mskip 3.0mu}\ensuremath{\allowbreak{}\mathnormal{∘}\allowbreak{}\mskip 3.0mu}\ensuremath{\allowbreak{}\mathnormal{(}\mskip 0.0mu}\ensuremath{\mathsf{id}\mskip 3.0mu}\ensuremath{\allowbreak{}\mathnormal{×}\allowbreak{}\mskip 3.0mu}\ensuremath{ε\mskip 0.0mu}\ensuremath{\mathnormal{)}\allowbreak{}} and likewise for \ensuremath{\mathsf{π₂}}.
Likewise, but independently, one may use the
combinator (▵) instead of \ensuremath{\mathsf{δ}}, with \ensuremath{\mathsf{f}\mskip 3.0mu}\ensuremath{\mathnormal{▵}\mskip 3.0mu}\ensuremath{\mathsf{g}\mskip 3.0mu}\ensuremath{\mathnormal{=}\mskip 3.0mu}\ensuremath{\allowbreak{}\mathnormal{(}\mskip 0.0mu}\ensuremath{\mathsf{f}\mskip 3.0mu}\ensuremath{\allowbreak{}\mathnormal{×}\allowbreak{}\mskip 3.0mu}\ensuremath{\mathsf{g}\mskip 0.0mu}\ensuremath{\mathnormal{)}\allowbreak{}\mskip 3.0mu}\ensuremath{\allowbreak{}\mathnormal{∘}\allowbreak{}\mskip 3.0mu}\ensuremath{δ}.  Our
diagram notation makes the latter two variants indistinguishable,
while the former two are equivalent under pruning of dotted lines: 
{\begin{tikzpicture}\path[-,draw=black,line width=0.4000pt,line cap=butt,line join=miter,dash pattern=](-23.6000pt,5.0000pt)--(-16.6000pt,5.0000pt);
\path[-,draw=black,line width=0.4000pt,line cap=butt,line join=miter,dash pattern=](-23.6000pt,0.0000pt)--(-16.6000pt,0.0000pt);
\path[-,draw=black,line width=0.4000pt,line cap=butt,line join=miter,dash pattern=](0.0000pt,5.0000pt)--(7.0000pt,5.0000pt);
\path[-,draw=black,line width=0.4000pt,line cap=butt,line join=miter,dash pattern=](-7.0000pt,5.0000pt)--(0.0000pt,5.0000pt);
\path[-,draw=black,line width=0.4000pt,line cap=butt,line join=miter,dash pattern=on 0.4000pt off 1.0000pt](-7.0000pt,0.0000pt)--(0.0000pt,0.0000pt);
\path[-,draw=black,line width=0.4000pt,line cap=butt,line join=miter,dash pattern=](0.0000pt,5.0000pt)--(0.0000pt,5.0000pt);
\path[-,line width=0.4000pt,line cap=butt,line join=miter,dash pattern=on 0.4000pt off 1.0000pt](0.0000pt,5.0000pt)--(0.0000pt,5.0000pt)--(0.0000pt,0.0000pt)--(0.0000pt,0.0000pt)--cycle;
\path[-,draw=black,line width=0.4000pt,line cap=butt,line join=miter,dash pattern=on 0.4000pt off 1.0000pt](-3.0000pt,8.0000pt)--(3.0000pt,8.0000pt)--(3.0000pt,-3.0000pt)--(-3.0000pt,-3.0000pt)--cycle;
\path[-,fill=black,line width=0.4000pt,line cap=butt,line join=miter,dash pattern=](1.0000pt,0.0000pt)..controls(1.0000pt,0.5523pt)and(0.5523pt,1.0000pt)..(0.0000pt,1.0000pt)..controls(-0.5523pt,1.0000pt)and(-1.0000pt,0.5523pt)..(-1.0000pt,0.0000pt)..controls(-1.0000pt,-0.5523pt)and(-0.5523pt,-1.0000pt)..(0.0000pt,-1.0000pt)..controls(0.5523pt,-1.0000pt)and(1.0000pt,-0.5523pt)..(1.0000pt,0.0000pt)--cycle;
\path[-,draw=black,line width=0.4000pt,line cap=butt,line join=miter,dash pattern=](-16.6000pt,5.0000pt)--(-7.0000pt,5.0000pt);
\path[-,fill=black,line width=0.4000pt,line cap=butt,line join=miter,dash pattern=](-6.0000pt,0.0000pt)..controls(-6.0000pt,0.5523pt)and(-6.4477pt,1.0000pt)..(-7.0000pt,1.0000pt)..controls(-7.5523pt,1.0000pt)and(-8.0000pt,0.5523pt)..(-8.0000pt,0.0000pt)..controls(-8.0000pt,-0.5523pt)and(-7.5523pt,-1.0000pt)..(-7.0000pt,-1.0000pt)..controls(-6.4477pt,-1.0000pt)and(-6.0000pt,-0.5523pt)..(-6.0000pt,0.0000pt)--cycle;
\path[-,draw=black,line width=0.4000pt,line cap=butt,line join=miter,dash pattern=](-16.6000pt,0.0000pt)--(-7.0000pt,0.0000pt);
\end{tikzpicture}} = 
{\begin{tikzpicture}\path[-,draw=black,line width=0.4000pt,line cap=butt,line join=miter,dash pattern=](-7.0000pt,5.0000pt)--(0.0000pt,5.0000pt);
\path[-,draw=black,line width=0.4000pt,line cap=butt,line join=miter,dash pattern=](-7.0000pt,0.0000pt)--(0.0000pt,0.0000pt);
\path[-,draw=black,line width=0.4000pt,line cap=butt,line join=miter,dash pattern=](0.0000pt,5.0000pt)--(7.0000pt,5.0000pt);
\path[-,line width=0.4000pt,line cap=butt,line join=miter,dash pattern=](0.0000pt,5.0000pt)--(0.0000pt,5.0000pt)--(0.0000pt,0.0000pt)--(0.0000pt,0.0000pt)--cycle;
\path[-,draw=black,line width=0.4000pt,line cap=butt,line join=miter,dash pattern=](-3.0000pt,8.0000pt)--(3.0000pt,8.0000pt)--(3.0000pt,-3.0000pt)--(-3.0000pt,-3.0000pt)--cycle;
\path[-,fill=black,line width=0.4000pt,line cap=butt,line join=miter,dash pattern=](1.0000pt,0.0000pt)..controls(1.0000pt,0.5523pt)and(0.5523pt,1.0000pt)..(0.0000pt,1.0000pt)..controls(-0.5523pt,1.0000pt)and(-1.0000pt,0.5523pt)..(-1.0000pt,0.0000pt)..controls(-1.0000pt,-0.5523pt)and(-0.5523pt,-1.0000pt)..(0.0000pt,-1.0000pt)..controls(0.5523pt,-1.0000pt)and(1.0000pt,-0.5523pt)..(1.0000pt,0.0000pt)--cycle;
\path[-,draw=black,line width=0.4000pt,line cap=butt,line join=miter,dash pattern=](0.0000pt,5.0000pt)--(0.0000pt,5.0000pt);
\end{tikzpicture}}.

It is enlightening to consider what becomes of the correspondence
between diagram (topological) equivalence and morphism (algebraic)
equivalence in the presence of the above laws. For \ensuremath{\mathsf{ε}}, the metaphor
can be sustained: continuous deformation of lines involving 
{\begin{tikzpicture}\path[-,draw=black,line width=0.4000pt,line cap=butt,line join=miter,dash pattern=](-7.0000pt,0.0000pt)--(0.0000pt,0.0000pt);
\path[-,draw=black,line width=0.4000pt,line cap=butt,line join=miter,dash pattern=on 0.4000pt off 1.0000pt](0.0000pt,0.0000pt)--(7.0000pt,0.0000pt);
\path[-,fill=black,line width=0.4000pt,line cap=butt,line join=miter,dash pattern=](1.0000pt,0.0000pt)..controls(1.0000pt,0.5523pt)and(0.5523pt,1.0000pt)..(0.0000pt,1.0000pt)..controls(-0.5523pt,1.0000pt)and(-1.0000pt,0.5523pt)..(-1.0000pt,0.0000pt)..controls(-1.0000pt,-0.5523pt)and(-0.5523pt,-1.0000pt)..(0.0000pt,-1.0000pt)..controls(0.5523pt,-1.0000pt)and(1.0000pt,-0.5523pt)..(1.0000pt,0.0000pt)--cycle;
\path[-,draw=black,line width=0.4000pt,line cap=butt,line join=miter,dash pattern=](0.0000pt,0.0000pt)--(0.0000pt,0.0000pt);
\end{tikzpicture}} capture its laws. For \ensuremath{\mathsf{δ}}, the topological metaphor
begins to break down. Morphisms can commute with \ensuremath{\mathsf{δ}} in the following
way: \ensuremath{\allowbreak{}\mathnormal{(}\mskip 0.0mu}\ensuremath{\mathsf{f}\mskip 3.0mu}\ensuremath{\allowbreak{}\mathnormal{×}\allowbreak{}\mskip 3.0mu}\ensuremath{\mathsf{f}\mskip 0.0mu}\ensuremath{\mathnormal{)}\allowbreak{}\mskip 3.0mu}\ensuremath{\allowbreak{}\mathnormal{∘}\allowbreak{}\mskip 3.0mu}\ensuremath{δ\mskip 3.0mu}\ensuremath{\mathnormal{=}\mskip 3.0mu}\ensuremath{δ\mskip 3.0mu}\ensuremath{\allowbreak{}\mathnormal{∘}\allowbreak{}\mskip 3.0mu}\ensuremath{\mathsf{f}}.  This breakdown has consequences for
computational applications, as we discuss in
\cref{141}.  
\subsection{Linear types}\label{32} 
We rely on linear types in Haskell in an essential way. Indeed,
every linear function can be interpreted in terms of an {\sc{}smc}. This is a
well known fact, proven for example by
\citet[Ch. 3]{szabo_algebra_2016} or
\citet{benton_lnl_1995}. Unfortunately it does not mean that we have
nothing to do. Indeed, the above result, as it stands, only means that
one can obtain an {\sc{}smc} representation from another \emph{representation} as a (well-typed) lambda term. Such a term is, indeed, constructed by
a compiler, but it is in general not made available to the programs
themselves: some form of metaprogramming would be
required. Unfortunately, outside the Lisp family, such
metaprogramming facilities are often brittle or non-existent. For
instance, the Template Haskell {\sc{}api} is a direct reflection of the
internal representation of source code in use by the Glasgow Haskell
Compiler, and consequently the user-facing {\sc{}api} changes whenever this
internal representation changes.

In this paper we use Linear Haskell as host language, and borrow its
semantics and notations. We refer to \citet{bernardy_linear_2018} if
any doubt should remain, but what the reader should know is that
linear functions are denoted with a lollipop (⊸), and the pointy-headed arrow
(\ensuremath{\mathnormal{\rightarrow }}) corresponds to usual functions, which can use their argument
any number of times. A notable feature of Linear Haskell is that
unrestricted inputs can be embedded in data types (which can themselves
be handled linearly). We make use of this feature in
our implementation (\cref{60}). In sum, any language with the above
feature set is sufficient to host our interface and implementation. In
particular, we do not make use of the ability of Linear Haskell to
quantify over the multiplicity (linear or unrestricted) of function types.

\section{Interface}\label{33} 
With all the basic components in place, we can now reveal the interface
that we provide to construct the morphisms of a symmetric monoidal category \ensuremath{\mathsf{k}} using lambda
notation.  We introduce a single abstract type: \ensuremath{\mathsf{P}\mskip 3.0mu}\ensuremath{\mathsf{k}\mskip 3.0mu}\ensuremath{\mathsf{r}\mskip 3.0mu}\ensuremath{\mathsf{a}}, where
\ensuremath{\mathsf{r}} is a type variable (unique for the morphism under
construction) and \ensuremath{\mathsf{a}} is an object of the category
\ensuremath{\mathsf{k}}. Values of
the type \ensuremath{\mathsf{P}\mskip 3.0mu}\ensuremath{\mathsf{k}\mskip 3.0mu}\ensuremath{\mathsf{r}\mskip 3.0mu}\ensuremath{\mathsf{a}} are called \emph{ports carrying \ensuremath{\mathsf{a}}}. In the
boxes-and-wires metaphor, ports are the output wires of boxes. Indeed, the
type of morphisms \ensuremath{\mathsf{a}\mskip 0.0mu}\ensuremath{\overset{\mathsf{k}}{\leadsto}\mskip 0.0mu}\ensuremath{\mathsf{b}} is encoded as functions of type
\ensuremath{\mathsf{P}\mskip 3.0mu}\ensuremath{\mathsf{k}\mskip 3.0mu}\ensuremath{\mathsf{r}\mskip 3.0mu}\ensuremath{\mathsf{a}\mskip 3.0mu}\ensuremath{\mathnormal{⊸}\mskip 3.0mu}\ensuremath{\mathsf{P}\mskip 3.0mu}\ensuremath{\mathsf{k}\mskip 3.0mu}\ensuremath{\mathsf{r}\mskip 3.0mu}\ensuremath{\mathsf{b}}.  However, the type \ensuremath{\mathsf{P}\mskip 3.0mu}\ensuremath{\mathsf{k}\mskip 3.0mu}\ensuremath{\mathsf{r}\mskip 3.0mu}\ensuremath{\mathsf{a}} is
abstract: it is manipulated solely \textit{via} the combinators of \cref{34}.
(This is enforced according to standard Haskell praxis: the definitions are hidden behind a module boundary, which exports only the prescribed {\sc{}api}.)

\begin{figure}\ensuremath{\begin{parray}\column{B}{@{}>{}l<{}@{}}\column[0em]{1}{@{}>{}l<{}@{}}\column{2}{@{}>{}l<{}@{}}\column{3}{@{}>{}l<{}@{}}\column{4}{@{}>{}l<{}@{}}\column{5}{@{}>{}l<{}@{}}\column{6}{@{}>{}l<{}@{}}\column{7}{@{}>{}l<{}@{}}\column{E}{@{}>{}l<{}@{}}%
\>[1]{}{\mathbf{type}\mskip 3.0mu\mathsf{P}\mskip 3.0mu}\>[2]{}{\mathnormal{::}\mskip 3.0mu\allowbreak{}\mathnormal{(}\mskip 0.0mu\mathsf{Type}\mskip 3.0mu\mathnormal{\rightarrow }\mskip 3.0mu\mathsf{Type}\mskip 3.0mu\mathnormal{\rightarrow }\mskip 3.0mu\mathsf{Type}\mskip 0.0mu\mathnormal{)}\allowbreak{}\mskip 3.0mu\mathnormal{\rightarrow }\mskip 3.0mu\mathsf{Type}\mskip 3.0mu\mathnormal{\rightarrow }\mskip 3.0mu\mathsf{Type}\mskip 3.0mu\mathnormal{\rightarrow }\mskip 3.0mu\mathsf{Type}}\<[E]{}\\
\>[1]{}{\mathsf{unit}\mskip 3.0mu}\>[2]{}{\mathnormal{::}\mskip 3.0mu}\>[5]{}{\mathsf{P}\mskip 3.0mu\mathsf{k}\mskip 3.0mu\mathsf{r}\mskip 3.0mu\allowbreak{}\mathnormal{(}\mskip 0.0mu\mathnormal{)}\allowbreak{}}\<[E]{}\\
\>[1]{}{\mathsf{split}\mskip 3.0mu}\>[2]{}{\mathnormal{::}\mskip 3.0mu}\>[4]{}{\mathsf{P}\mskip 3.0mu\mathsf{k}\mskip 3.0mu\mathsf{r}\mskip 3.0mu\allowbreak{}\mathnormal{(}\mskip 0.0mu\mathsf{a}\mskip 3.0mu\mathnormal{⊗}\mskip 3.0mu\mathsf{b}\mskip 0.0mu\mathnormal{)}\allowbreak{}\mskip 3.0mu\mathnormal{⊸}\mskip 3.0mu\allowbreak{}\mathnormal{(}\mskip 0.0mu\mathsf{P}\mskip 3.0mu\mathsf{k}\mskip 3.0mu\mathsf{r}\mskip 3.0mu\mathsf{a}\mskip 0.0mu\mathnormal{,}\mskip 3.0mu\mathsf{P}\mskip 3.0mu\mathsf{k}\mskip 3.0mu\mathsf{r}\mskip 3.0mu\mathsf{b}\mskip 0.0mu\mathnormal{)}\allowbreak{}}\<[E]{}\\
\>[1]{}{\mathsf{merge}\mskip 3.0mu}\>[2]{}{\mathnormal{::}\mskip 3.0mu}\>[4]{}{\allowbreak{}\mathnormal{(}\mskip 0.0mu\mathsf{P}\mskip 3.0mu\mathsf{k}\mskip 3.0mu\mathsf{r}\mskip 3.0mu\mathsf{a}\mskip 0.0mu\mathnormal{,}\mskip 3.0mu\mathsf{P}\mskip 3.0mu\mathsf{k}\mskip 3.0mu\mathsf{r}\mskip 3.0mu\mathsf{b}\mskip 0.0mu\mathnormal{)}\allowbreak{}\mskip 3.0mu\mathnormal{⊸}\mskip 3.0mu\mathsf{P}\mskip 3.0mu\mathsf{k}\mskip 3.0mu\mathsf{r}\mskip 3.0mu\allowbreak{}\mathnormal{(}\mskip 0.0mu\mathsf{a}\mskip 3.0mu\mathnormal{⊗}\mskip 3.0mu\mathsf{b}\mskip 0.0mu\mathnormal{)}\allowbreak{}}\<[E]{}\\
\>[1]{}{\mathsf{encode}\mskip 3.0mu}\>[2]{}{\mathnormal{::}\mskip 3.0mu}\>[3]{}{\allowbreak{}\mathnormal{(}\mskip 0.0mu\mathsf{a}\mskip 0.0mu\overset{\mathsf{k}}{\leadsto}\mskip 0.0mu\mathsf{b}\mskip 0.0mu\mathnormal{)}\allowbreak{}\mskip 3.0mu\mathnormal{\rightarrow }\mskip 3.0mu\allowbreak{}\mathnormal{(}\mskip 0.0mu\mathsf{P}\mskip 3.0mu\mathsf{k}\mskip 3.0mu\mathsf{r}\mskip 3.0mu\mathsf{a}\mskip 3.0mu\mathnormal{⊸}\mskip 3.0mu\mathsf{P}\mskip 3.0mu\mathsf{k}\mskip 3.0mu\mathsf{r}\mskip 3.0mu\mathsf{b}\mskip 0.0mu\mathnormal{)}\allowbreak{}}\<[E]{}\\
\>[1]{}{\mathsf{decode}\mskip 3.0mu}\>[2]{}{\mathnormal{::}\mskip 3.0mu}\>[6]{}{\allowbreak{}\mathnormal{(}\mskip 0.0mu∀\mskip 3.0mu\mathsf{r}\mskip 1.0mu.\mskip 3.0mu}\>[7]{}{\mathsf{P}\mskip 3.0mu\mathsf{k}\mskip 3.0mu\mathsf{r}\mskip 3.0mu\mathsf{a}\mskip 3.0mu\mathnormal{⊸}\mskip 3.0mu\mathsf{P}\mskip 3.0mu\mathsf{k}\mskip 3.0mu\mathsf{r}\mskip 3.0mu\mathsf{b}\mskip 0.0mu\mathnormal{)}\allowbreak{}\mskip 3.0mu\mathnormal{\rightarrow }\mskip 3.0mu\allowbreak{}\mathnormal{(}\mskip 0.0mu\mathsf{a}\mskip 0.0mu\overset{\mathsf{k}}{\leadsto}\mskip 0.0mu\mathsf{b}\mskip 0.0mu\mathnormal{)}\allowbreak{}}\<[E]{}\end{parray}}\caption{The port {\sc{}api}}\label{34}\end{figure} 
Our bread and butter are the \ensuremath{\mathsf{split}} and \ensuremath{\mathsf{merge}} combinators, which provide the ability to treat ports of type \ensuremath{\mathsf{P}\mskip 3.0mu}\ensuremath{\mathsf{k}\mskip 3.0mu}\ensuremath{\mathsf{r}\mskip 3.0mu}\ensuremath{\allowbreak{}\mathnormal{(}\mskip 0.0mu}\ensuremath{\mathsf{a}\mskip 3.0mu}\ensuremath{\mathnormal{⊗}\mskip 3.0mu}\ensuremath{\mathsf{b}\mskip 0.0mu}\ensuremath{\mathnormal{)}\allowbreak{}} as a pair of ports. In fact,
\ensuremath{\mathsf{split}} and \ensuremath{\mathsf{merge}} are ubiquitous enough to deserve a
shorthand notation, suggestive of the pair-like character of \ensuremath{\mathsf{P}\mskip 3.0mu}\ensuremath{\mathsf{k}\mskip 3.0mu}\ensuremath{\mathsf{r}\mskip 3.0mu}\ensuremath{\allowbreak{}\mathnormal{(}\mskip 0.0mu}\ensuremath{\mathsf{a}\mskip 3.0mu}\ensuremath{\mathnormal{⊗}\mskip 3.0mu}\ensuremath{\mathsf{b}\mskip 0.0mu}\ensuremath{\mathnormal{)}\allowbreak{}}:

\begin{itemize}\item{}We write \ensuremath{\allowbreak{}\mathnormal{(}\mskip 0.0mu}\ensuremath{\mathsf{a}\mskip 2.0mu}\ensuremath{\mathnormal{\fatsemi }\mskip 3.0mu}\ensuremath{\mathsf{b}\mskip 0.0mu}\ensuremath{\mathnormal{)}\allowbreak{}} for \ensuremath{\mathsf{merge}\mskip 3.0mu}\ensuremath{\allowbreak{}\mathnormal{(}\mskip 0.0mu}\ensuremath{\mathsf{a}\mskip 0.0mu}\ensuremath{\mathnormal{,}\mskip 3.0mu}\ensuremath{\mathsf{b}\mskip 0.0mu}\ensuremath{\mathnormal{)}\allowbreak{}}\item{}We also use \ensuremath{\allowbreak{}\mathnormal{(}\mskip 0.0mu}\ensuremath{\mathsf{a}\mskip 2.0mu}\ensuremath{\mathnormal{\fatsemi }\mskip 3.0mu}\ensuremath{\mathsf{b}\mskip 0.0mu}\ensuremath{\mathnormal{)}\allowbreak{}} as a pattern, and interpret it as a call to \ensuremath{\mathsf{split}}.
  For instance, \ensuremath{\mathbf{let}\mskip 3.0mu}\ensuremath{\allowbreak{}\mathnormal{(}\mskip 0.0mu}\ensuremath{\mathsf{a}\mskip 2.0mu}\ensuremath{\mathnormal{\fatsemi }\mskip 3.0mu}\ensuremath{\mathsf{b}\mskip 0.0mu}\ensuremath{\mathnormal{)}\allowbreak{}\mskip 3.0mu}\ensuremath{\mathnormal{=}\mskip 3.0mu}\ensuremath{\mathsf{f}\mskip 3.0mu}\ensuremath{\mathbf{in}\mskip 3.0mu}\ensuremath{\mathsf{u}} means \ensuremath{\mathbf{let}\mskip 3.0mu}\ensuremath{\allowbreak{}\mathnormal{(}\mskip 0.0mu}\ensuremath{\mathsf{a}\mskip 0.0mu}\ensuremath{\mathnormal{,}\mskip 3.0mu}\ensuremath{\mathsf{b}\mskip 0.0mu}\ensuremath{\mathnormal{)}\allowbreak{}\mskip 3.0mu}\ensuremath{\mathnormal{=}\mskip 3.0mu}\ensuremath{\mathsf{split}\mskip 3.0mu}\ensuremath{\mathsf{f}\mskip 3.0mu}\ensuremath{\mathbf{in}\mskip 3.0mu}\ensuremath{\mathsf{u}}\end{itemize} 
Likewise, the presence of \ensuremath{\mathsf{unit}} means that ports of type \ensuremath{\mathsf{P}\mskip 3.0mu}\ensuremath{\mathsf{k}\mskip 3.0mu}\ensuremath{\mathsf{r}\mskip 3.0mu}\ensuremath{\allowbreak{}\mathnormal{(}\mskip 0.0mu}\ensuremath{\mathnormal{)}\allowbreak{}} can be created from thin air, which is useful to embed
constants.  
Finally and crucially, \ensuremath{\mathsf{encode}} and \ensuremath{\mathsf{decode}} provide means
to convert back and forth between morphisms of an {\sc{}smc} (\ensuremath{\mathsf{a}\mskip 0.0mu}\ensuremath{\overset{\mathsf{k}}{\leadsto}\mskip 0.0mu}\ensuremath{\mathsf{b}})
and (\ensuremath{\mathsf{P}\mskip 3.0mu}\ensuremath{\mathsf{k}\mskip 3.0mu}\ensuremath{\mathsf{r}\mskip 3.0mu}\ensuremath{\mathsf{a}\mskip 3.0mu}\ensuremath{\mathnormal{⊸}\mskip 3.0mu}\ensuremath{\mathsf{P}\mskip 3.0mu}\ensuremath{\mathsf{k}\mskip 3.0mu}\ensuremath{\mathsf{r}\mskip 3.0mu}\ensuremath{\mathsf{b}}), the corresponding linear functions.
We see in the type of \ensuremath{\mathsf{decode}} how the type variable \ensuremath{\mathsf{r}} is
introduced, ensuring that ports coming from different functions are
not mixed.  This interface is guaranteed to satisfy the following
properties:

\begin{definition}{Laws of the interface} \begin{itemize}\item{}\ensuremath{\mathsf{split}} and \ensuremath{\mathsf{merge}} are inverses: \ensuremath{\mathsf{split}\mskip 3.0mu}\ensuremath{\allowbreak{}\mathnormal{(}\mskip 0.0mu}\ensuremath{\mathsf{merge}\mskip 3.0mu}\ensuremath{\mathsf{p}\mskip 0.0mu}\ensuremath{\mathnormal{)}\allowbreak{}\mskip 3.0mu}\ensuremath{\mathnormal{=}\mskip 3.0mu}\ensuremath{\mathsf{p}} and \ensuremath{\mathsf{merge}\mskip 3.0mu}\ensuremath{\allowbreak{}\mathnormal{(}\mskip 0.0mu}\ensuremath{\mathsf{split}\mskip 3.0mu}\ensuremath{\mathsf{p}\mskip 0.0mu}\ensuremath{\mathnormal{)}\allowbreak{}\mskip 3.0mu}\ensuremath{\mathnormal{=}\mskip 3.0mu}\ensuremath{\mathsf{p}}\item{}\ensuremath{\mathsf{encode}} and \ensuremath{\mathsf{decode}} are inverses: \ensuremath{\mathsf{encode}\mskip 3.0mu}\ensuremath{\allowbreak{}\mathnormal{(}\mskip 0.0mu}\ensuremath{\mathsf{decode}\mskip 3.0mu}\ensuremath{\mathsf{f}\mskip 0.0mu}\ensuremath{\mathnormal{)}\allowbreak{}\mskip 3.0mu}\ensuremath{\mathnormal{=}\mskip 3.0mu}\ensuremath{\mathsf{f}} and \ensuremath{\mathsf{decode}\mskip 3.0mu}\ensuremath{\allowbreak{}\mathnormal{(}\mskip 0.0mu}\ensuremath{\mathsf{encode}\mskip 3.0mu}\ensuremath{\mathsf{p}\mskip 0.0mu}\ensuremath{\mathnormal{)}\allowbreak{}\mskip 3.0mu}\ensuremath{\mathnormal{=}\mskip 3.0mu}\ensuremath{\mathsf{p}}\item{}\ensuremath{\mathsf{encode}} is a functor: \ensuremath{\mathsf{encode}\mskip 3.0mu}\ensuremath{\mathsf{id}\mskip 3.0mu}\ensuremath{\mathnormal{=}\mskip 3.0mu}\ensuremath{\mathsf{id}} and \ensuremath{\mathsf{encode}\mskip 3.0mu}\ensuremath{\allowbreak{}\mathnormal{(}\mskip 0.0mu}\ensuremath{\mathsf{φ}\mskip 3.0mu}\ensuremath{\allowbreak{}\mathnormal{∘}\allowbreak{}\mskip 3.0mu}\ensuremath{\mathsf{ψ}\mskip 0.0mu}\ensuremath{\mathnormal{)}\allowbreak{}\mskip 3.0mu}\ensuremath{\mathnormal{=}\mskip 3.0mu}\ensuremath{\mathsf{encode}\mskip 3.0mu}\ensuremath{\mathsf{φ}\mskip 3.0mu}\ensuremath{\allowbreak{}\mathnormal{∘}\allowbreak{}\mskip 3.0mu}\ensuremath{\mathsf{encode}\mskip 3.0mu}\ensuremath{\mathsf{ψ}}\item{}\ensuremath{\mathsf{encode}} is compatible with products: \ensuremath{\mathsf{encode}\mskip 3.0mu}\ensuremath{\allowbreak{}\mathnormal{(}\mskip 0.0mu}\ensuremath{\mathsf{φ}\mskip 3.0mu}\ensuremath{\allowbreak{}\mathnormal{×}\allowbreak{}\mskip 3.0mu}\ensuremath{\mathsf{ψ}\mskip 0.0mu}\ensuremath{\mathnormal{)}\allowbreak{}\mskip 3.0mu}\ensuremath{\allowbreak{}\mathnormal{(}\mskip 0.0mu}\ensuremath{\mathsf{a}\mskip 2.0mu}\ensuremath{\mathnormal{\fatsemi }\mskip 3.0mu}\ensuremath{\mathsf{b}\mskip 0.0mu}\ensuremath{\mathnormal{)}\allowbreak{}\mskip 3.0mu}\ensuremath{\mathnormal{=}\mskip 3.0mu}\ensuremath{\allowbreak{}\mathnormal{(}\mskip 0.0mu}\ensuremath{\mathsf{encode}\mskip 3.0mu}\ensuremath{\mathsf{φ}\mskip 3.0mu}\ensuremath{\mathsf{a}\mskip 2.0mu}\ensuremath{\mathnormal{\fatsemi }\mskip 3.0mu}\ensuremath{\mathsf{encode}\mskip 3.0mu}\ensuremath{\mathsf{ψ}\mskip 3.0mu}\ensuremath{\mathsf{b}\mskip 0.0mu}\ensuremath{\mathnormal{)}\allowbreak{}}\item{}\ensuremath{\mathsf{unit}} corresponds to unitors: \ensuremath{\mathsf{encode}\mskip 3.0mu}\ensuremath{ρ\mskip 3.0mu}\ensuremath{\mathsf{a}\mskip 3.0mu}\ensuremath{\mathnormal{=}\mskip 3.0mu}\ensuremath{\allowbreak{}\mathnormal{(}\mskip 0.0mu}\ensuremath{\mathsf{a}\mskip 2.0mu}\ensuremath{\mathnormal{\fatsemi }\mskip 3.0mu}\ensuremath{\mathsf{unit}\mskip 0.0mu}\ensuremath{\mathnormal{)}\allowbreak{}} and \ensuremath{\mathsf{encode}\mskip 3.0mu}\ensuremath{\bar{ρ}\mskip 3.0mu}\ensuremath{\allowbreak{}\mathnormal{(}\mskip 0.0mu}\ensuremath{\mathsf{a}\mskip 2.0mu}\ensuremath{\mathnormal{\fatsemi }\mskip 3.0mu}\ensuremath{\mathsf{unit}\mskip 0.0mu}\ensuremath{\mathnormal{)}\allowbreak{}\mskip 3.0mu}\ensuremath{\mathnormal{=}\mskip 3.0mu}\ensuremath{\mathsf{a}}\item{}\ensuremath{σ}, \ensuremath{α} and \ensuremath{\bar{α}} are consistent between Haskell and the embedded category:
   \begin{itemize}\item{}\ensuremath{\mathsf{encode}\mskip 3.0mu}\ensuremath{σ\mskip 3.0mu}\ensuremath{\allowbreak{}\mathnormal{(}\mskip 0.0mu}\ensuremath{\mathsf{a}\mskip 2.0mu}\ensuremath{\mathnormal{\fatsemi }\mskip 3.0mu}\ensuremath{\mathsf{b}\mskip 0.0mu}\ensuremath{\mathnormal{)}\allowbreak{}\mskip 3.0mu}\ensuremath{\mathnormal{=}\mskip 3.0mu}\ensuremath{\allowbreak{}\mathnormal{(}\mskip 0.0mu}\ensuremath{\mathsf{b}\mskip 2.0mu}\ensuremath{\mathnormal{\fatsemi }\mskip 3.0mu}\ensuremath{\mathsf{a}\mskip 0.0mu}\ensuremath{\mathnormal{)}\allowbreak{}}\item{}\ensuremath{\mathsf{encode}\mskip 3.0mu}\ensuremath{α\mskip 3.0mu}\ensuremath{\allowbreak{}\mathnormal{(}\mskip 0.0mu}\ensuremath{\allowbreak{}\mathnormal{(}\mskip 0.0mu}\ensuremath{\mathsf{a}\mskip 2.0mu}\ensuremath{\mathnormal{\fatsemi }\mskip 3.0mu}\ensuremath{\mathsf{b}\mskip 0.0mu}\ensuremath{\mathnormal{)}\allowbreak{}\mskip 2.0mu}\ensuremath{\mathnormal{\fatsemi }\mskip 3.0mu}\ensuremath{\mathsf{c}\mskip 0.0mu}\ensuremath{\mathnormal{)}\allowbreak{}\mskip 3.0mu}\ensuremath{\mathnormal{=}\mskip 3.0mu}\ensuremath{\allowbreak{}\mathnormal{(}\mskip 0.0mu}\ensuremath{\mathsf{a}\mskip 2.0mu}\ensuremath{\mathnormal{\fatsemi }\mskip 3.0mu}\ensuremath{\allowbreak{}\mathnormal{(}\mskip 0.0mu}\ensuremath{\mathsf{b}\mskip 2.0mu}\ensuremath{\mathnormal{\fatsemi }\mskip 3.0mu}\ensuremath{\mathsf{c}\mskip 0.0mu}\ensuremath{\mathnormal{)}\allowbreak{}\mskip 0.0mu}\ensuremath{\mathnormal{)}\allowbreak{}}\item{}\ensuremath{\mathsf{encode}\mskip 3.0mu}\ensuremath{\bar{α}\mskip 3.0mu}\ensuremath{\allowbreak{}\mathnormal{(}\mskip 0.0mu}\ensuremath{\mathsf{a}\mskip 2.0mu}\ensuremath{\mathnormal{\fatsemi }\mskip 3.0mu}\ensuremath{\allowbreak{}\mathnormal{(}\mskip 0.0mu}\ensuremath{\mathsf{b}\mskip 2.0mu}\ensuremath{\mathnormal{\fatsemi }\mskip 3.0mu}\ensuremath{\mathsf{c}\mskip 0.0mu}\ensuremath{\mathnormal{)}\allowbreak{}\mskip 0.0mu}\ensuremath{\mathnormal{)}\allowbreak{}\mskip 3.0mu}\ensuremath{\mathnormal{=}\mskip 3.0mu}\ensuremath{\allowbreak{}\mathnormal{(}\mskip 0.0mu}\ensuremath{\allowbreak{}\mathnormal{(}\mskip 0.0mu}\ensuremath{\mathsf{a}\mskip 2.0mu}\ensuremath{\mathnormal{\fatsemi }\mskip 3.0mu}\ensuremath{\mathsf{b}\mskip 0.0mu}\ensuremath{\mathnormal{)}\allowbreak{}\mskip 2.0mu}\ensuremath{\mathnormal{\fatsemi }\mskip 3.0mu}\ensuremath{\mathsf{c}\mskip 0.0mu}\ensuremath{\mathnormal{)}\allowbreak{}}\end{itemize}\end{itemize}\label{35}\end{definition} Stating the laws which involve products does require a bit of care.
For instance, it would not have been type-correct to write \ensuremath{\mathsf{encode}\mskip 3.0mu}\ensuremath{\allowbreak{}\mathnormal{(}\mskip 0.0mu}\ensuremath{\mathsf{f}\mskip 3.0mu}\ensuremath{\allowbreak{}\mathnormal{×}\allowbreak{}\mskip 3.0mu}\ensuremath{\mathsf{g}\mskip 0.0mu}\ensuremath{\mathnormal{)}\allowbreak{}\mskip 3.0mu}\ensuremath{\mathnormal{=}\mskip 3.0mu}\ensuremath{\mathsf{encode}\mskip 3.0mu}\ensuremath{\mathsf{f}\mskip 3.0mu}\ensuremath{\allowbreak{}\mathnormal{×}\allowbreak{}\mskip 3.0mu}\ensuremath{\mathsf{encode}\mskip 3.0mu}\ensuremath{\mathsf{g}} nor \ensuremath{\mathsf{encode}\mskip 3.0mu}\ensuremath{σ\mskip 3.0mu}\ensuremath{\mathnormal{=}\mskip 3.0mu}\ensuremath{σ}: going
through \ensuremath{\mathsf{split}} and \ensuremath{\mathsf{merge}} is necessary.

Another aspect to consider is that many of these laws refer to an
equality on ports. Because the type of ports is abstract, we cannot
define it yet: its concrete definition will be provided together with
the concrete definition of ports.
However, we can already give an intuition for it in terms of diagrams: two ports are equal if
they are one and the same in the diagram.
Even it is abstract, we can already reason with this equality via the following property: two
extensionally equal functions on ports will decode to the same morphism. Formally:
\ensuremath{\allowbreak{}\mathnormal{(}\mskip 0.0mu}\ensuremath{\mathnormal{∀}\mskip 3.0mu}\ensuremath{\mathsf{x}\mskip 1.0mu}\ensuremath{.\mskip 3.0mu}\ensuremath{\mathsf{f}\mskip 3.0mu}\ensuremath{\mathsf{x}\mskip 3.0mu}\ensuremath{\mathnormal{=}\mskip 3.0mu}\ensuremath{\mathsf{g}\mskip 3.0mu}\ensuremath{\mathsf{x}\mskip 0.0mu}\ensuremath{\mathnormal{)}\allowbreak{}\mskip 3.0mu}\ensuremath{\mathnormal{\rightarrow }\mskip 3.0mu}\ensuremath{\mathsf{decode}\mskip 3.0mu}\ensuremath{\mathsf{f}\mskip 3.0mu}\ensuremath{\mathnormal{=}\mskip 3.0mu}\ensuremath{\mathsf{decode}\mskip 3.0mu}\ensuremath{\mathsf{g}}.

Without introducing any additional concept, we can already observe some
benefits of the above interface. First, one can use all the facilities
of a higher-order language to construct elements of \ensuremath{\mathsf{a}\mskip 0.0mu}\ensuremath{\overset{\mathsf{k}}{\leadsto}\mskip 0.0mu}\ensuremath{\mathsf{b}},
even though \ensuremath{\mathsf{k}} does not have an internal notion of functions (it
need not be a closed category).  We owe this benefit to the host
language evaluation, which takes care of evaluating all intermediate
redexes. It can be illustrated by the existence of currying
combinators:

\begin{list}{}{\setlength\leftmargin{1.0em}}\item\relax
 \ensuremath{\begin{parray}\column{B}{@{}>{}l<{}@{}}\column[0em]{1}{@{}>{}l<{}@{}}\column{2}{@{}>{}l<{}@{}}\column{3}{@{}>{}l<{}@{}}\column{4}{@{}>{}l<{}@{}}\column{5}{@{}>{}l<{}@{}}\column{6}{@{}>{}l<{}@{}}\column{E}{@{}>{}l<{}@{}}%
\>[1]{}{\mathsf{curry}\mskip 3.0mu}\>[2]{}{\mathnormal{::}\mskip 3.0mu\allowbreak{}\mathnormal{(}\mskip 0.0mu\mathsf{Monoidal}\mskip 3.0mu\mathsf{k}\mskip 0.0mu}\>[4]{}{\mathnormal{)}\allowbreak{}\mskip 3.0mu}\>[5]{}{\mathnormal{\Rightarrow }\mskip 3.0mu}\>[6]{}{\allowbreak{}\mathnormal{(}\mskip 0.0mu\mathsf{P}\mskip 3.0mu\mathsf{k}\mskip 3.0mu\mathsf{r}\mskip 3.0mu\allowbreak{}\mathnormal{(}\mskip 0.0mu\mathsf{a}\mskip 3.0mu\mathnormal{⊗}\mskip 3.0mu\mathsf{b}\mskip 0.0mu\mathnormal{)}\allowbreak{}\mskip 3.0mu\mathnormal{⊸}\mskip 3.0mu\mathsf{P}\mskip 3.0mu\mathsf{k}\mskip 3.0mu\mathsf{r}\mskip 3.0mu\mathsf{c}\mskip 0.0mu\mathnormal{)}\allowbreak{}}\<[E]{}\\
\>[2]{}{\mathnormal{⊸}\mskip 3.0mu\allowbreak{}\mathnormal{(}\mskip 0.0mu\mathsf{P}\mskip 3.0mu\mathsf{k}\mskip 3.0mu\mathsf{r}\mskip 3.0mu\mathsf{a}\mskip 3.0mu\mathnormal{⊸}\mskip 3.0mu\mathsf{P}\mskip 3.0mu\mathsf{k}\mskip 3.0mu\mathsf{r}\mskip 3.0mu\mathsf{b}\mskip 3.0mu\mathnormal{⊸}\mskip 3.0mu\mathsf{P}\mskip 3.0mu\mathsf{k}\mskip 3.0mu\mathsf{r}\mskip 3.0mu\mathsf{c}\mskip 0.0mu\mathnormal{)}\allowbreak{}}\<[E]{}\\
\>[1]{}{\mathsf{curry}\mskip 3.0mu\mathsf{f}\mskip 3.0mu\mathsf{a}\mskip 3.0mu\mathsf{b}\mskip 3.0mu\mathnormal{=}\mskip 3.0mu\mathsf{f}\mskip 3.0mu}\>[3]{}{\allowbreak{}\mathnormal{(}\mskip 0.0mu\mathsf{a}\mskip 2.0mu\mathnormal{\fatsemi }\mskip 3.0mu\mathsf{b}\mskip 0.0mu\mathnormal{)}\allowbreak{}}\<[E]{}\\
\>[1]{}{\mathsf{uncurry}\mskip 3.0mu}\>[2]{}{\mathnormal{::}\mskip 3.0mu\allowbreak{}\mathnormal{(}\mskip 0.0mu\mathsf{Monoidal}\mskip 3.0mu\mathsf{k}\mskip 0.0mu}\>[4]{}{\mathnormal{)}\allowbreak{}\mskip 3.0mu}\>[5]{}{\mathnormal{\Rightarrow }\mskip 3.0mu}\>[6]{}{\allowbreak{}\mathnormal{(}\mskip 0.0mu\mathsf{P}\mskip 3.0mu\mathsf{k}\mskip 3.0mu\mathsf{r}\mskip 3.0mu\mathsf{a}\mskip 3.0mu\mathnormal{⊸}\mskip 3.0mu\mathsf{P}\mskip 3.0mu\mathsf{k}\mskip 3.0mu\mathsf{r}\mskip 3.0mu\mathsf{b}\mskip 3.0mu\mathnormal{⊸}\mskip 3.0mu\mathsf{P}\mskip 3.0mu\mathsf{k}\mskip 3.0mu\mathsf{r}\mskip 3.0mu\mathsf{c}\mskip 0.0mu\mathnormal{)}\allowbreak{}}\<[E]{}\\
\>[2]{}{\mathnormal{⊸}\mskip 3.0mu\allowbreak{}\mathnormal{(}\mskip 0.0mu\mathsf{P}\mskip 3.0mu\mathsf{k}\mskip 3.0mu\mathsf{r}\mskip 3.0mu\allowbreak{}\mathnormal{(}\mskip 0.0mu\mathsf{a}\mskip 3.0mu\mathnormal{⊗}\mskip 3.0mu\mathsf{b}\mskip 0.0mu\mathnormal{)}\allowbreak{}\mskip 3.0mu\mathnormal{⊸}\mskip 3.0mu\mathsf{P}\mskip 3.0mu\mathsf{k}\mskip 3.0mu\mathsf{r}\mskip 3.0mu\mathsf{c}\mskip 0.0mu\mathnormal{)}\allowbreak{}}\<[E]{}\end{parray}} \\ \ensuremath{\begin{parray}\column{B}{@{}>{}l<{}@{}}\column[0em]{1}{@{}>{}l<{}@{}}\column{E}{@{}>{}l<{}@{}}%
\>[1]{}{\mathsf{uncurry}\mskip 3.0mu\mathsf{f}\mskip 3.0mu\mathsf{p}\mskip 3.0mu\mathnormal{=}\mskip 3.0mu\mathbf{case}\mskip 3.0mu\mathsf{split}\mskip 3.0mu\mathsf{p}\mskip 3.0mu\mathbf{of}\mskip 3.0mu\allowbreak{}\mathnormal{(}\mskip 0.0mu\mathsf{a}\mskip 0.0mu\mathnormal{,}\mskip 3.0mu\mathsf{b}\mskip 0.0mu\mathnormal{)}\allowbreak{}\mskip 3.0mu\mathnormal{\rightarrow }\mskip 3.0mu\mathsf{f}\mskip 3.0mu\mathsf{a}\mskip 3.0mu\mathsf{b}}\<[E]{}\end{parray}} \end{list} 
Second, if the category \ensuremath{\mathsf{k}} happens to be cartesian, then we can freely copy
and discard ports. This is done by encoding \ensuremath{\mathsf{ε}} and \ensuremath{\mathsf{δ}}, as follows:

\begin{list}{}{\setlength\leftmargin{1.0em}}\item\relax
\ensuremath{\begin{parray}\column{B}{@{}>{}l<{}@{}}\column[0em]{1}{@{}>{}l<{}@{}}\column{2}{@{}>{}l<{}@{}}\column{3}{@{}>{}l<{}@{}}\column{4}{@{}>{}l<{}@{}}\column{5}{@{}>{}l<{}@{}}\column{E}{@{}>{}l<{}@{}}%
\>[1]{}{\mathsf{copy}\mskip 3.0mu}\>[2]{}{\mathnormal{::}\mskip 3.0mu\allowbreak{}\mathnormal{(}\mskip 0.0mu\mathsf{Cartesian}\mskip 3.0mu\mathsf{k}\mskip 0.0mu}\>[4]{}{\mathnormal{)}\allowbreak{}\mskip 3.0mu\mathnormal{\Rightarrow }\mskip 3.0mu\mathsf{P}\mskip 3.0mu\mathsf{k}\mskip 3.0mu\mathsf{r}\mskip 3.0mu\mathsf{a}\mskip 3.0mu\mathnormal{⊸}\mskip 3.0mu\mathsf{P}\mskip 3.0mu\mathsf{k}\mskip 3.0mu\mathsf{r}\mskip 3.0mu\allowbreak{}\mathnormal{(}\mskip 0.0mu\mathsf{a}\mskip 3.0mu\mathnormal{⊗}\mskip 3.0mu\mathsf{a}\mskip 0.0mu\mathnormal{)}\allowbreak{}}\<[E]{}\\
\>[1]{}{\mathsf{copy}\mskip 3.0mu}\>[2]{}{\mathnormal{=}\mskip 3.0mu\mathsf{encode}\mskip 3.0muδ}\<[E]{}\\
\>[1]{}{\mathsf{discard}\mskip 3.0mu}\>[3]{}{\mathnormal{::}\mskip 3.0mu\allowbreak{}\mathnormal{(}\mskip 0.0mu\mathsf{Cartesian}\mskip 3.0mu\mathsf{k}\mskip 0.0mu}\>[5]{}{\mathnormal{)}\allowbreak{}\mskip 3.0mu\mathnormal{\Rightarrow }\mskip 3.0mu\mathsf{P}\mskip 3.0mu\mathsf{k}\mskip 3.0mu\mathsf{r}\mskip 3.0mu\mathsf{a}\mskip 3.0mu\mathnormal{⊸}\mskip 3.0mu\mathsf{P}\mskip 3.0mu\mathsf{k}\mskip 3.0mu\mathsf{r}\mskip 3.0mu\allowbreak{}\mathnormal{(}\mskip 0.0mu\mathnormal{)}\allowbreak{}}\<[E]{}\\
\>[1]{}{\mathsf{discard}\mskip 3.0mu}\>[3]{}{\mathnormal{=}\mskip 3.0mu\mathsf{encode}\mskip 3.0muε}\<[E]{}\end{parray}}\end{list}  It is worth stressing that \ensuremath{\mathsf{copy}} and \ensuremath{\mathsf{discard}} are not part
of the abstract interface. Indeed, in the above the morphisms \ensuremath{\mathsf{δ}} and
\ensuremath{\mathsf{ε}} are treated as black boxes by our implementation, just like any
other morphism of \ensuremath{\mathsf{k}} would be. Consequently the implementation
does not assume that any law holds for them, and in particular it
cannot commute any morphism with (this instance of) \ensuremath{\mathsf{δ}} using the law
\ensuremath{\mathsf{f}\mskip 3.0mu}\ensuremath{\allowbreak{}\mathnormal{×}\allowbreak{}\mskip 3.0mu}\ensuremath{\mathsf{f}\mskip 3.0mu}\ensuremath{\allowbreak{}\mathnormal{∘}\allowbreak{}\mskip 3.0mu}\ensuremath{\mathsf{δ}\mskip 3.0mu}\ensuremath{\mathnormal{=}\mskip 3.0mu}\ensuremath{\mathsf{δ}\mskip 3.0mu}\ensuremath{\allowbreak{}\mathnormal{∘}\allowbreak{}\mskip 3.0mu}\ensuremath{\mathsf{f}}. We come back to this aspect in
\cref{141}.
More generally, thanks to the
\ensuremath{\mathsf{encode}} combinator, every morphism of \ensuremath{\mathsf{k}} can be
turned into a Haskell function on ports.

\section{Applications}\label{36} 
In this section, we put the port {\sc{}api} of \cref{34} to
use. Through two examples of diagrammatic languages, we illustrate how
convenient it is to describe box-and-wire diagrams as functions on
ports.

\subsection{Quantum circuits}\label{37} 
\begin{figure*}
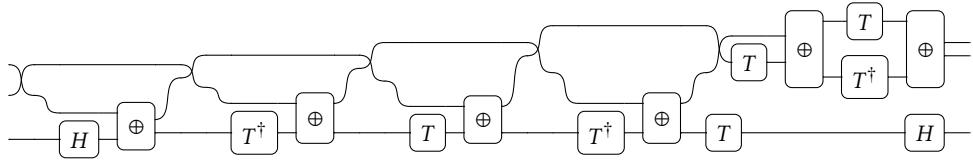
\begin{center}{\small{}
{
}}\end{center}\caption{Toffoli gate in terms of \ensuremath{\ensuremath{H}}, \ensuremath{\ensuremath{T}} and \ensuremath{\mathnormal{⊕}}.}\label{53}\end{figure*} 
In quantum computing one of the common ways to represent programs is
as \emph{quantum circuits}.  Take, for instance, the circuit
of \cref{53}, which is an implementation of the Toffoli gate in terms of simpler
quantum gates.

For our purposes, it suffices to treat the atomic gates in
\cref{53} as abstract. Regardless, if a reader may be
interested in looking up their definitions, the gate \ensuremath{\ensuremath{H}} stands
for the Hadamard gate, \ensuremath{\ensuremath{T}} for the T gate, and \ensuremath{\mathnormal{⊕}} for
the controlled-not gate.\footnote{Refer for example to
\url{https://en.wikipedia.org/wiki/Quantum\_logic\_gate} and
\url{https://en.wikipedia.org/wiki/Toffoli\_gate}.} 
Quantum circuits closely resemble traditional Boolean circuits except
that a circuit represents not a Boolean function, but a unitary matrix
on some finite dimensional \ensuremath{ℂ}-vector space.  For our purposes,
unitary matrices have two important properties. First, they form an
{\sc{}smc}, which we call \ensuremath{\mathsf{U}}.
(This is why quantum circuits can be written as
boxes-and-wires diagrams.)
\begin{mdframed}[linewidth=0pt,hidealllines,innerleftmargin=0pt,innerrightmargin=0pt,backgroundcolor=gray!15] 
A possible implementation of the \ensuremath{\mathsf{U}} category is
to let \ensuremath{\mathsf{a}\mskip 0.0mu}\ensuremath{\overset{\mathsf{U}}{\leadsto}\mskip 0.0mu}\ensuremath{\mathsf{b}} be a matrix whose indices range \ensuremath{\mathsf{a}} and \ensuremath{\mathsf{b}}.
\begin{list}{}{\setlength\leftmargin{1.0em}}\item\relax
\ensuremath{\begin{parray}\column{B}{@{}>{}l<{}@{}}\column[0em]{1}{@{}>{}l<{}@{}}\column{E}{@{}>{}l<{}@{}}%
\>[1]{}{\mathbf{data}\mskip 3.0mu\mathsf{U}\mskip 3.0mu\mathsf{a}\mskip 3.0mu\mathsf{b}\mskip 3.0mu\mathnormal{=}\mskip 3.0mu\mathsf{U}\mskip 3.0mu\allowbreak{}\mathnormal{\{}\mskip 0.0mu\mathsf{fromM}\mskip 3.0mu\mathnormal{::}\mskip 3.0mu\mathsf{Array}\mskip 3.0mu\allowbreak{}\mathnormal{(}\mskip 0.0mu\mathsf{a}\mskip 0.0mu\mathnormal{,}\mskip 3.0mu\mathsf{b}\mskip 0.0mu\mathnormal{)}\allowbreak{}\mskip 3.0muℂ\mskip 0.0mu\mathnormal{\}}\allowbreak{}}\<[E]{}\end{parray}}\end{list} 
Thus this means in particular that all objects in this category must
be finite types: \ensuremath{\mathsf{Finite}\mskip 3.0mu}\ensuremath{\mathsf{a}\mskip 3.0mu}\ensuremath{\mathnormal{=}\mskip 3.0mu}\ensuremath{\allowbreak{}\mathnormal{(}\mskip 0.0mu}\ensuremath{\mathsf{Bounded}\mskip 3.0mu}\ensuremath{\mathsf{a}\mskip 0.0mu}\ensuremath{\mathnormal{,}\mskip 3.0mu}\ensuremath{\mathsf{Ix}\mskip 3.0mu}\ensuremath{\mathsf{a}\mskip 0.0mu}\ensuremath{\mathnormal{,}\mskip 3.0mu}\ensuremath{\mathsf{Eq}\mskip 3.0mu}\ensuremath{\mathsf{a}\mskip 0.0mu}\ensuremath{\mathnormal{)}\allowbreak{}}. This way
we can construct matrices using the following function:
\begin{list}{}{\setlength\leftmargin{1.0em}}\item\relax
\ensuremath{\begin{parray}\column{B}{@{}>{}l<{}@{}}\column[0em]{1}{@{}>{}l<{}@{}}\column{2}{@{}>{}l<{}@{}}\column[2em]{3}{@{}>{}l<{}@{}}\column{4}{@{}>{}l<{}@{}}\column{E}{@{}>{}l<{}@{}}%
\>[1]{}{\mathsf{tabulate}\mskip 3.0mu\mathnormal{::}\mskip 3.0mu\allowbreak{}\mathnormal{(}\mskip 0.0mu\mathsf{Finite}\mskip 3.0mu\mathsf{a}\mskip 0.0mu\mathnormal{,}\mskip 3.0mu\mathsf{Finite}\mskip 3.0mu\mathsf{b}\mskip 0.0mu\mathnormal{)}\allowbreak{}\mskip 3.0mu\mathnormal{\Rightarrow }\mskip 3.0mu\allowbreak{}\mathnormal{(}\mskip 0.0mu\mathsf{a}\mskip 3.0mu\mathnormal{\rightarrow }\mskip 3.0mu\mathsf{b}\mskip 3.0mu\mathnormal{\rightarrow }\mskip 3.0muℂ\mskip 0.0mu\mathnormal{)}\allowbreak{}\mskip 3.0mu\mathnormal{\rightarrow }\mskip 3.0mu\mathsf{a}\mskip 0.0mu\overset{\mathsf{U}}{\leadsto}\mskip 0.0mu\mathsf{b}}\<[E]{}\\
\>[1]{}{\mathsf{tabulate}\mskip 3.0mu\mathsf{f}\mskip 3.0mu\mathnormal{=}\mskip 3.0mu\mathsf{U}\mskip 3.0mu\allowbreak{}\mathnormal{(}\mskip 0.0mu}\>[2]{}{\mathsf{array}\mskip 3.0mu\allowbreak{}\mathnormal{(}\mskip 0.0mu}\>[4]{}{\allowbreak{}\mathnormal{(}\mskip 0.0mu\mathsf{minBound}\mskip 0.0mu\mathnormal{,}\mskip 3.0mu\mathsf{minBound}\mskip 0.0mu\mathnormal{)}\allowbreak{}\mskip 0.0mu\mathnormal{,}}\<[E]{}\\
\>[4]{}{\allowbreak{}\mathnormal{(}\mskip 0.0mu\mathsf{maxBound}\mskip 0.0mu\mathnormal{,}\mskip 3.0mu\mathsf{maxBound}\mskip 0.0mu\mathnormal{)}\allowbreak{}\mskip 0.0mu\mathnormal{)}\allowbreak{}}\<[E]{}\\
\>[3]{}{\allowbreak{}\mathnormal{[}\mskip 0.0mu\allowbreak{}\mathnormal{(}\mskip 0.0mu\allowbreak{}\mathnormal{(}\mskip 0.0mu\mathsf{i}\mskip 0.0mu\mathnormal{,}\mskip 3.0mu\mathsf{j}\mskip 0.0mu\mathnormal{)}\allowbreak{}\mskip 0.0mu\mathnormal{,}\mskip 3.0mu\mathsf{f}\mskip 3.0mu\mathsf{i}\mskip 3.0mu\mathsf{j}\mskip 0.0mu\mathnormal{)}\allowbreak{}\mskip 3.0mu\mathnormal{|}\mskip 3.0mu\mathsf{i}\mskip 3.0mu\mathnormal{\leftarrow }\mskip 3.0mu\mathsf{inhabitants}\mskip 0.0mu\mathnormal{,}\mskip 3.0mu\mathsf{j}\mskip 3.0mu\mathnormal{\leftarrow }\mskip 3.0mu\mathsf{inhabitants}\mskip 0.0mu\mathnormal{]}\allowbreak{}\mskip 0.0mu\mathnormal{)}\allowbreak{}}\<[E]{}\end{parray}}\end{list} 
Besides, the main tool for implementation is the Kronecker delta:
\begin{list}{}{\setlength\leftmargin{1.0em}}\item\relax
\ensuremath{\begin{parray}\column{B}{@{}>{}l<{}@{}}\column[0em]{1}{@{}>{}l<{}@{}}\column{E}{@{}>{}l<{}@{}}%
\>[1]{}{\mathsf{delta}\mskip 3.0mu\mathnormal{::}\mskip 3.0mu\allowbreak{}\mathnormal{(}\mskip 0.0mu\mathsf{Eq}\mskip 3.0mu\mathsf{a}\mskip 0.0mu\mathnormal{)}\allowbreak{}\mskip 3.0mu\mathnormal{\Rightarrow }\mskip 3.0mu\mathsf{a}\mskip 3.0mu\mathnormal{\rightarrow }\mskip 3.0mu\mathsf{a}\mskip 3.0mu\mathnormal{\rightarrow }\mskip 3.0muℂ}\<[E]{}\\
\>[1]{}{\mathsf{delta}\mskip 3.0mu\mathsf{x}\mskip 3.0mu\mathsf{y}\mskip 3.0mu\mathnormal{=}\mskip 3.0mu\mathbf{if}\mskip 3.0mu\mathsf{x}\mskip 3.0mu\mathnormal{\equiv}\mskip 3.0mu\mathsf{y}\mskip 3.0mu\mathbf{then}\mskip 3.0mu\mathrm{1}\mskip 3.0mu\mathbf{else}\mskip 3.0mu\mathrm{0}}\<[E]{}\end{parray}}\end{list} 
We can then construct the \ensuremath{\mathsf{Monoidal}\mskip 3.0mu}\ensuremath{\mathsf{U}} instance:
\begin{list}{}{\setlength\leftmargin{1.0em}}\item\relax
\ensuremath{\begin{parray}\column{B}{@{}>{}l<{}@{}}\column[0em]{1}{@{}>{}l<{}@{}}\column[1em]{2}{@{}>{}l<{}@{}}\column{3}{@{}>{}l<{}@{}}\column{4}{@{}>{}l<{}@{}}\column{5}{@{}>{}l<{}@{}}\column{6}{@{}>{}l<{}@{}}\column{7}{@{}>{}l<{}@{}}\column{8}{@{}>{}l<{}@{}}\column{9}{@{}>{}l<{}@{}}\column{E}{@{}>{}l<{}@{}}%
\>[1]{}{\mathbf{instance}\mskip 3.0mu\mathsf{Category}\mskip 3.0mu\mathsf{U}\mskip 3.0mu\mathbf{where}}\<[E]{}\\
\>[2]{}{\mathbf{type}\mskip 3.0mu\mathsf{Obj}\mskip 3.0mu\mathsf{U}\mskip 3.0mu\mathnormal{=}\mskip 3.0mu\mathsf{Finite}}\<[E]{}\\
\>[2]{}{\mathsf{id}\mskip 3.0mu\mathnormal{=}\mskip 3.0mu\mathsf{tabulate}\mskip 3.0mu\mathsf{delta}}\<[E]{}\\
\>[2]{}{\mathsf{U}\mskip 3.0mu\mathsf{g}\mskip 3.0mu\allowbreak{}\mathnormal{∘}\allowbreak{}\mskip 3.0mu\mathsf{U}\mskip 3.0mu\mathsf{f}\mskip 3.0mu\mathnormal{=}\mskip 3.0mu\mathsf{tabulate}\mskip 3.0mu\allowbreak{}\mathnormal{(}\mskip 0.0muλ\mskip 3.0mu\mathsf{i}\mskip 3.0mu\mathsf{j}\mskip 3.0mu\mathnormal{\rightarrow }\mskip 3.0mu}\>[7]{}{\mathsf{summation}}\<[E]{}\\
\>[7]{}{\allowbreak{}\mathnormal{(}\mskip 0.0muλ\mskip 3.0mu\mathsf{k}\mskip 3.0mu\mathnormal{\rightarrow }\mskip 3.0mu\mathsf{f}\mskip 3.0mu\mathnormal{!}\mskip 3.0mu\allowbreak{}\mathnormal{(}\mskip 0.0mu\mathsf{i}\mskip 0.0mu\mathnormal{,}\mskip 3.0mu\mathsf{k}\mskip 0.0mu\mathnormal{)}\allowbreak{}\mskip 3.0mu}\>[8]{}{\mathnormal{*}\mskip 3.0mu}\>[9]{}{\mathsf{g}\mskip 3.0mu\mathnormal{!}\mskip 3.0mu\allowbreak{}\mathnormal{(}\mskip 0.0mu\mathsf{k}\mskip 0.0mu\mathnormal{,}\mskip 3.0mu\mathsf{j}\mskip 0.0mu\mathnormal{)}\allowbreak{}\mskip 0.0mu\mathnormal{)}\allowbreak{}\mskip 0.0mu\mathnormal{)}\allowbreak{}}\<[E]{}\\
\>[1]{}{\mathbf{instance}\mskip 3.0mu\mathsf{Monoidal}\mskip 3.0mu\mathsf{U}\mskip 3.0mu\mathbf{where}}\<[E]{}\\
\>[2]{}{\mathsf{U}\mskip 3.0mu\mathsf{f}\mskip 3.0mu\allowbreak{}\mathnormal{×}\allowbreak{}\mskip 3.0mu\mathsf{U}\mskip 3.0mu\mathsf{g}\mskip 3.0mu\mathnormal{=}\mskip 3.0mu\mathsf{tabulate}\mskip 3.0mu\allowbreak{}\mathnormal{(}\mskip 0.0muλ\mskip 3.0mu\allowbreak{}\mathnormal{(}\mskip 0.0mu\mathsf{a}\mskip 0.0mu\mathnormal{,}\mskip 3.0mu\mathsf{c}\mskip 0.0mu\mathnormal{)}\allowbreak{}\mskip 3.0mu\allowbreak{}\mathnormal{(}\mskip 0.0mu\mathsf{b}\mskip 0.0mu\mathnormal{,}\mskip 3.0mu\mathsf{d}\mskip 0.0mu\mathnormal{)}\allowbreak{}\mskip 3.0mu\mathnormal{\rightarrow }\mskip 3.0mu\mathsf{f}\mskip 3.0mu\mathnormal{!}\mskip 3.0mu\allowbreak{}\mathnormal{(}\mskip 0.0mu\mathsf{a}\mskip 0.0mu\mathnormal{,}\mskip 3.0mu\mathsf{b}\mskip 0.0mu\mathnormal{)}\allowbreak{}\mskip 3.0mu\mathnormal{*}\mskip 3.0mu\mathsf{g}\mskip 3.0mu\mathnormal{!}\mskip 3.0mu\allowbreak{}\mathnormal{(}\mskip 0.0mu\mathsf{c}\mskip 0.0mu\mathnormal{,}\mskip 3.0mu\mathsf{d}\mskip 0.0mu\mathnormal{)}\allowbreak{}\mskip 0.0mu\mathnormal{)}\allowbreak{}}\<[E]{}\\
\>[2]{}{ρ\mskip 3.0mu\mathnormal{=}\mskip 3.0mu\mathsf{tabulate}\mskip 3.0mu\allowbreak{}\mathnormal{(}\mskip 0.0muλ\mskip 3.0mu\mathsf{x}\mskip 3.0mu\allowbreak{}\mathnormal{(}\mskip 0.0mu\mathsf{y}\mskip 0.0mu\mathnormal{,}\mskip 3.0mu\allowbreak{}\mathnormal{(}\mskip 0.0mu\mathnormal{)}\allowbreak{}\mskip 0.0mu\mathnormal{)}\allowbreak{}\mskip 3.0mu\mathnormal{\rightarrow }\mskip 3.0mu\mathsf{delta}\mskip 3.0mu\mathsf{x}\mskip 3.0mu\mathsf{y}\mskip 0.0mu\mathnormal{)}\allowbreak{}}\<[E]{}\\
\>[2]{}{\bar{ρ}\mskip 3.0mu\mathnormal{=}\mskip 3.0mu\mathsf{tabulate}\mskip 3.0mu\allowbreak{}\mathnormal{(}\mskip 0.0muλ\mskip 3.0mu\allowbreak{}\mathnormal{(}\mskip 0.0mu\mathsf{y}\mskip 0.0mu\mathnormal{,}\mskip 3.0mu\allowbreak{}\mathnormal{(}\mskip 0.0mu\mathnormal{)}\allowbreak{}\mskip 0.0mu\mathnormal{)}\allowbreak{}\mskip 3.0mu\mathsf{x}\mskip 3.0mu\mathnormal{\rightarrow }\mskip 3.0mu\mathsf{delta}\mskip 3.0mu\mathsf{x}\mskip 3.0mu\mathsf{y}\mskip 0.0mu\mathnormal{)}\allowbreak{}}\<[E]{}\\
\>[2]{}{α\mskip 3.0mu\mathnormal{=}\mskip 3.0mu\mathsf{tabulate}\mskip 3.0mu}\>[3]{}{\allowbreak{}\mathnormal{(}\mskip 0.0muλ\mskip 3.0mu}\>[5]{}{\allowbreak{}\mathnormal{(}\mskip 0.0mu\allowbreak{}\mathnormal{(}\mskip 0.0mu\mathsf{x}\mskip 0.0mu\mathnormal{,}\mskip 3.0mu\mathsf{y}\mskip 0.0mu\mathnormal{)}\allowbreak{}\mskip 0.0mu\mathnormal{,}\mskip 3.0mu\mathsf{z}\mskip 0.0mu\mathnormal{)}\allowbreak{}\mskip 3.0mu\allowbreak{}\mathnormal{(}\mskip 0.0mu\mathsf{x'}\mskip 0.0mu\mathnormal{,}\mskip 3.0mu\allowbreak{}\mathnormal{(}\mskip 0.0mu\mathsf{y'}\mskip 0.0mu\mathnormal{,}\mskip 3.0mu\mathsf{z'}\mskip 0.0mu\mathnormal{)}\allowbreak{}\mskip 0.0mu\mathnormal{)}\allowbreak{}\mskip 3.0mu\mathnormal{\rightarrow }}\<[E]{}\\
\>[5]{}{\mathsf{delta}\mskip 3.0mu\allowbreak{}\mathnormal{(}\mskip 0.0mu\allowbreak{}\mathnormal{(}\mskip 0.0mu\mathsf{x}\mskip 0.0mu\mathnormal{,}\mskip 3.0mu\mathsf{y}\mskip 0.0mu\mathnormal{)}\allowbreak{}\mskip 0.0mu\mathnormal{,}\mskip 3.0mu\mathsf{z}\mskip 0.0mu\mathnormal{)}\allowbreak{}\mskip 3.0mu\allowbreak{}\mathnormal{(}\mskip 0.0mu\allowbreak{}\mathnormal{(}\mskip 0.0mu\mathsf{x'}\mskip 0.0mu\mathnormal{,}\mskip 3.0mu\mathsf{y'}\mskip 0.0mu\mathnormal{)}\allowbreak{}\mskip 0.0mu\mathnormal{,}\mskip 3.0mu\mathsf{z'}\mskip 0.0mu\mathnormal{)}\allowbreak{}\mskip 0.0mu\mathnormal{)}\allowbreak{}}\<[E]{}\\
\>[2]{}{\bar{α}\mskip 3.0mu\mathnormal{=}\mskip 3.0mu\mathsf{tabulate}\mskip 3.0mu}\>[4]{}{\allowbreak{}\mathnormal{(}\mskip 0.0muλ\mskip 3.0mu}\>[6]{}{\allowbreak{}\mathnormal{(}\mskip 0.0mu\mathsf{x'}\mskip 0.0mu\mathnormal{,}\mskip 3.0mu\allowbreak{}\mathnormal{(}\mskip 0.0mu\mathsf{y'}\mskip 0.0mu\mathnormal{,}\mskip 3.0mu\mathsf{z'}\mskip 0.0mu\mathnormal{)}\allowbreak{}\mskip 0.0mu\mathnormal{)}\allowbreak{}\mskip 3.0mu\allowbreak{}\mathnormal{(}\mskip 0.0mu\allowbreak{}\mathnormal{(}\mskip 0.0mu\mathsf{x}\mskip 0.0mu\mathnormal{,}\mskip 3.0mu\mathsf{y}\mskip 0.0mu\mathnormal{)}\allowbreak{}\mskip 0.0mu\mathnormal{,}\mskip 3.0mu\mathsf{z}\mskip 0.0mu\mathnormal{)}\allowbreak{}\mskip 3.0mu\mathnormal{\rightarrow }}\<[E]{}\\
\>[6]{}{\mathsf{delta}\mskip 3.0mu\allowbreak{}\mathnormal{(}\mskip 0.0mu\allowbreak{}\mathnormal{(}\mskip 0.0mu\mathsf{x}\mskip 0.0mu\mathnormal{,}\mskip 3.0mu\mathsf{y}\mskip 0.0mu\mathnormal{)}\allowbreak{}\mskip 0.0mu\mathnormal{,}\mskip 3.0mu\mathsf{z}\mskip 0.0mu\mathnormal{)}\allowbreak{}\mskip 3.0mu\allowbreak{}\mathnormal{(}\mskip 0.0mu\allowbreak{}\mathnormal{(}\mskip 0.0mu\mathsf{x'}\mskip 0.0mu\mathnormal{,}\mskip 3.0mu\mathsf{y'}\mskip 0.0mu\mathnormal{)}\allowbreak{}\mskip 0.0mu\mathnormal{,}\mskip 3.0mu\mathsf{z'}\mskip 0.0mu\mathnormal{)}\allowbreak{}\mskip 0.0mu\mathnormal{)}\allowbreak{}}\<[E]{}\\
\>[2]{}{σ\mskip 3.0mu\mathnormal{=}\mskip 3.0mu\mathsf{tabulate}\mskip 3.0mu\mathnormal{\$}\mskip 3.0muλ\mskip 3.0mu\allowbreak{}\mathnormal{(}\mskip 0.0mu\mathsf{x}\mskip 0.0mu\mathnormal{,}\mskip 3.0mu\mathsf{y}\mskip 0.0mu\mathnormal{)}\allowbreak{}\mskip 3.0mu\allowbreak{}\mathnormal{(}\mskip 0.0mu\mathsf{y'}\mskip 0.0mu\mathnormal{,}\mskip 3.0mu\mathsf{x'}\mskip 0.0mu\mathnormal{)}\allowbreak{}\mskip 3.0mu\mathnormal{\rightarrow }\mskip 3.0mu\mathsf{delta}\mskip 3.0mu\allowbreak{}\mathnormal{(}\mskip 0.0mu\mathsf{x}\mskip 0.0mu\mathnormal{,}\mskip 3.0mu\mathsf{y}\mskip 0.0mu\mathnormal{)}\allowbreak{}\mskip 3.0mu\allowbreak{}\mathnormal{(}\mskip 0.0mu\mathsf{x'}\mskip 0.0mu\mathnormal{,}\mskip 3.0mu\mathsf{y'}\mskip 0.0mu\mathnormal{)}\allowbreak{}}\<[E]{}\end{parray}}\end{list} 
Morphism composition is matrix product, and the product \ensuremath{\allowbreak{}\mathnormal{(}\mskip 0.0mu}\ensuremath{\allowbreak{}\mathnormal{×}\allowbreak{}\mskip 0.0mu}\ensuremath{\mathnormal{)}\allowbreak{}} is
implemented as the Kronecker product.

To be complete, we also would need to show that each method
implemented above preserves the unitary character of matrices.  These
proofs can be easily looked up, but for the reader who might prefer to
reconstruct them, the key property is that a matrix is unitary iff its
determinant is 1: \ensuremath{\mathsf{norm}\mskip 3.0mu}\ensuremath{\allowbreak{}\mathnormal{(}\mskip 0.0mu}\ensuremath{\mathsf{det}\mskip 3.0mu}\ensuremath{\mathsf{u}\mskip 0.0mu}\ensuremath{\mathnormal{)}\allowbreak{}\mskip 3.0mu}\ensuremath{\mathnormal{=}\mskip 3.0mu}\ensuremath{\mathrm{1}}. Then one needs to check
that this property is preserved by each operation. The properties to
invoke are \ensuremath{\mathsf{det}\mskip 3.0mu}\ensuremath{\allowbreak{}\mathnormal{(}\mskip 0.0mu}\ensuremath{\mathsf{u}\mskip 3.0mu}\ensuremath{\allowbreak{}\mathnormal{∘}\allowbreak{}\mskip 3.0mu}\ensuremath{\mathsf{v}\mskip 0.0mu}\ensuremath{\mathnormal{)}\allowbreak{}\mskip 3.0mu}\ensuremath{\mathnormal{=}\mskip 3.0mu}\ensuremath{\mathsf{det}\mskip 3.0mu}\ensuremath{\mathsf{u}\mskip 3.0mu}\ensuremath{\mathnormal{·}\mskip 3.0mu}\ensuremath{\mathsf{det}\mskip 3.0mu}\ensuremath{\mathsf{v}} and \ensuremath{\mathsf{det}\mskip 3.0mu}\ensuremath{\allowbreak{}\mathnormal{(}\mskip 0.0mu}\ensuremath{\mathsf{u}\mskip 3.0mu}\ensuremath{\allowbreak{}\mathnormal{×}\allowbreak{}\mskip 3.0mu}\ensuremath{\mathsf{v}\mskip 0.0mu}\ensuremath{\mathnormal{)}\allowbreak{}\mskip 3.0mu}\ensuremath{\mathnormal{=}\mskip 3.0mu}\ensuremath{\allowbreak{}\mathnormal{(}\mskip 0.0mu}\ensuremath{\mathsf{det}\mskip 3.0mu}\ensuremath{\mathsf{u}\mskip 0.0mu}\ensuremath{\mathnormal{)}\allowbreak{}\mskip 3.0mu}\ensuremath{\string^\mskip 3.0mu}\ensuremath{\mathsf{n}\mskip 3.0mu}\ensuremath{\mathnormal{·}\mskip 3.0mu}\ensuremath{\allowbreak{}\mathnormal{(}\mskip 0.0mu}\ensuremath{\mathsf{det}\mskip 3.0mu}\ensuremath{\mathsf{v}\mskip 0.0mu}\ensuremath{\mathnormal{)}\allowbreak{}\mskip 3.0mu}\ensuremath{\string^\mskip 3.0mu}\ensuremath{\mathsf{m}}, where \ensuremath{\mathsf{n}} and \ensuremath{\mathsf{m}} are the
respective dimensions of \ensuremath{\mathsf{u}} and \ensuremath{\mathsf{v}}.

\end{mdframed} Second, unitary matrices can be inverted
by taking their conjugate transpose.
Notice for example the use of the gate \ensuremath{\ensuremath{T^{\dagger}}} in \cref{53}. It is the conjugate transpose
 of \ensuremath{\ensuremath{T}}. That is, \ensuremath{\ensuremath{T^{\dagger}}} is not a primitive gate, but one
defined in terms of \ensuremath{\ensuremath{T}} using the function

\begin{list}{}{\setlength\leftmargin{1.0em}}\item\relax
\ensuremath{\begin{parray}\column{B}{@{}>{}l<{}@{}}\column[0em]{1}{@{}>{}l<{}@{}}\column{2}{@{}>{}l<{}@{}}\column{E}{@{}>{}l<{}@{}}%
\>[1]{}{\mathsf{conjugateTranspose}\mskip 3.0mu\mathnormal{::}\mskip 3.0mu}\>[2]{}{\mathsf{U}\mskip 3.0mu\mathsf{b}\mskip 3.0mu\mathsf{a}\mskip 3.0mu\mathnormal{\rightarrow }\mskip 3.0mu\mathsf{U}\mskip 3.0mu\mathsf{a}\mskip 3.0mu\mathsf{b}}\<[E]{}\end{parray}}\end{list} 
It would be inconvenient to have to return to the low-level {\sc{}smc} interface every time we want to invert a matrix: what we really want
is to lift the \ensuremath{\mathsf{U}}-level interface to ports (\ensuremath{\mathsf{P}\mskip 3.0mu}\ensuremath{\mathsf{U}}) once and
for all, then work entirely with ports. Fortunately,
we can do just that. The only difference with lifting simple morphisms (\ensuremath{\mathsf{a}\mskip 0.0mu}\ensuremath{\overset{\mathsf{U}}{\leadsto}\mskip 0.0mu}\ensuremath{\mathsf{b}}) is that lifting \ensuremath{\mathsf{conjugateTranspose}} yields a
higher-order function:

\begin{list}{}{\setlength\leftmargin{1.0em}}\item\relax
\ensuremath{\begin{parray}\column{B}{@{}>{}l<{}@{}}\column[0em]{1}{@{}>{}l<{}@{}}\column{2}{@{}>{}l<{}@{}}\column{3}{@{}>{}l<{}@{}}\column{4}{@{}>{}l<{}@{}}\column{E}{@{}>{}l<{}@{}}%
\>[1]{}{\mathsf{invert}\mskip 3.0mu\mathnormal{::}\mskip 3.0mu}\>[2]{}{\allowbreak{}\mathnormal{(}\mskip 0.0mu∀\mskip 3.0mu\mathsf{s}\mskip 1.0mu.\mskip 3.0mu}\>[3]{}{\mathsf{P}\mskip 3.0mu\mathsf{U}\mskip 3.0mu\mathsf{s}\mskip 3.0mu\mathsf{a}\mskip 3.0mu\mathnormal{⊸}\mskip 3.0mu\mathsf{P}\mskip 3.0mu\mathsf{U}\mskip 3.0mu\mathsf{s}\mskip 3.0mu\mathsf{b}\mskip 0.0mu\mathnormal{)}\allowbreak{}\mskip 3.0mu\mathnormal{\rightarrow }\mskip 3.0mu\allowbreak{}\mathnormal{(}\mskip 0.0mu∀\mskip 3.0mu\mathsf{r}\mskip 1.0mu.\mskip 3.0mu}\>[4]{}{\mathsf{P}\mskip 3.0mu\mathsf{U}\mskip 3.0mu\mathsf{r}\mskip 3.0mu\mathsf{b}\mskip 3.0mu\mathnormal{⊸}\mskip 3.0mu\mathsf{P}\mskip 3.0mu\mathsf{U}\mskip 3.0mu\mathsf{r}\mskip 3.0mu\mathsf{a}\mskip 0.0mu\mathnormal{)}\allowbreak{}}\<[E]{}\\
\>[1]{}{\mathsf{invert}\mskip 3.0mu\mathsf{f}\mskip 3.0mu\mathnormal{=}\mskip 3.0mu\mathsf{encode}\mskip 3.0mu\allowbreak{}\mathnormal{(}\mskip 0.0mu\mathsf{conjugateTranspose}\mskip 3.0mu\allowbreak{}\mathnormal{(}\mskip 0.0mu\mathsf{decode}\mskip 3.0mu\mathsf{f}\mskip 0.0mu\mathnormal{)}\allowbreak{}\mskip 0.0mu\mathnormal{)}\allowbreak{}}\<[E]{}\end{parray}}\end{list} 
Consequently we do not have to encode the diagram of
\cref{53} using the methods of the \ensuremath{\mathsf{Monoidal}} class, but we can use the
more familiar lambda notation, manipulating ports. We do so assuming
the gates \ensuremath{\ensuremath{H}}, \ensuremath{\ensuremath{T}}, and \ensuremath{\mathnormal{⊕}}, which we can leave
abstract with the following types:

\begin{list}{}{\setlength\leftmargin{1.0em}}\item\relax
\ensuremath{\begin{parray}\column{B}{@{}>{}l<{}@{}}\column[0em]{1}{@{}>{}l<{}@{}}\column{E}{@{}>{}l<{}@{}}%
\>[1]{}{\ensuremath{H}\mskip 3.0mu\mathnormal{::}\mskip 3.0mu\mathsf{P}\mskip 3.0mu\mathsf{U}\mskip 3.0mu\mathsf{r}\mskip 3.0mu\mathsf{Bool}\mskip 3.0mu\mathnormal{⊸}\mskip 3.0mu\mathsf{P}\mskip 3.0mu\mathsf{U}\mskip 3.0mu\mathsf{r}\mskip 3.0mu\mathsf{Bool}}\<[E]{}\\
\>[1]{}{\ensuremath{T}\mskip 3.0mu\mathnormal{::}\mskip 3.0mu\mathsf{P}\mskip 3.0mu\mathsf{U}\mskip 3.0mu\mathsf{r}\mskip 3.0mu\mathsf{Bool}\mskip 3.0mu\mathnormal{⊸}\mskip 3.0mu\mathsf{P}\mskip 3.0mu\mathsf{U}\mskip 3.0mu\mathsf{r}\mskip 3.0mu\mathsf{Bool}}\<[E]{}\\
\>[1]{}{\allowbreak{}\mathnormal{(}\mskip 0.0mu\mathnormal{⊕}\mskip 0.0mu\mathnormal{)}\allowbreak{}\mskip 3.0mu\mathnormal{::}\mskip 3.0mu\mathsf{P}\mskip 3.0mu\mathsf{U}\mskip 3.0mu\mathsf{r}\mskip 3.0mu\mathsf{Bool}\mskip 3.0mu\mathnormal{⊸}\mskip 3.0mu\mathsf{P}\mskip 3.0mu\mathsf{U}\mskip 3.0mu\mathsf{r}\mskip 3.0mu\mathsf{Bool}\mskip 3.0mu\mathnormal{⊸}\mskip 3.0mu\allowbreak{}\mathnormal{(}\mskip 0.0mu\mathsf{P}\mskip 3.0mu\mathsf{U}\mskip 3.0mu\mathsf{r}\mskip 3.0mu\mathsf{Bool}\mskip 0.0mu\mathnormal{,}\mskip 3.0mu\mathsf{P}\mskip 3.0mu\mathsf{U}\mskip 3.0mu\mathsf{r}\mskip 3.0mu\mathsf{Bool}\mskip 0.0mu\mathnormal{)}\allowbreak{}}\<[E]{}\end{parray}}\end{list} 
Now, we can define the Toffoli gate circuit as follows

\begin{list}{}{\setlength\leftmargin{1.0em}}\item\relax
\ensuremath{\begin{parray}\column{B}{@{}>{}l<{}@{}}\column[0em]{1}{@{}>{}l<{}@{}}\column{2}{@{}>{}l<{}@{}}\column{3}{@{}>{}l<{}@{}}\column{4}{@{}>{}l<{}@{}}\column{5}{@{}>{}l<{}@{}}\column{6}{@{}>{}l<{}@{}}\column{7}{@{}>{}l<{}@{}}\column{8}{@{}>{}l<{}@{}}\column{9}{@{}>{}l<{}@{}}\column{E}{@{}>{}l<{}@{}}%
\>[1]{}{\mathsf{toffoli}\mskip 3.0mu}\>[2]{}{\mathnormal{::}\mskip 3.0mu}\>[3]{}{\mathsf{P}\mskip 3.0mu\mathsf{U}\mskip 3.0mu\mathsf{r}\mskip 3.0mu\allowbreak{}\mathnormal{(}\mskip 0.0mu\allowbreak{}\mathnormal{(}\mskip 0.0mu\mathsf{Bool}\mskip 3.0mu\mathnormal{⊗}\mskip 3.0mu\mathsf{Bool}\mskip 0.0mu\mathnormal{)}\allowbreak{}\mskip 3.0mu\mathnormal{⊗}\mskip 3.0mu\mathsf{Bool}\mskip 0.0mu\mathnormal{)}\allowbreak{}}\<[E]{}\\
\>[2]{}{\mathnormal{⊸}\mskip 3.0mu}\>[3]{}{\mathsf{P}\mskip 3.0mu\mathsf{U}\mskip 3.0mu\mathsf{r}\mskip 3.0mu\allowbreak{}\mathnormal{(}\mskip 0.0mu\allowbreak{}\mathnormal{(}\mskip 0.0mu\mathsf{Bool}\mskip 3.0mu\mathnormal{⊗}\mskip 3.0mu\mathsf{Bool}\mskip 0.0mu\mathnormal{)}\allowbreak{}\mskip 3.0mu\mathnormal{⊗}\mskip 3.0mu\mathsf{Bool}\mskip 0.0mu\mathnormal{)}\allowbreak{}}\<[E]{}\\
\>[1]{}{\mathsf{toffoli}\mskip 3.0mu\mathsf{c}_{1}\mskip 3.0mu\mathsf{c}_{2}\mskip 3.0mu\mathsf{x}\mskip 3.0mu\mathnormal{=}\mskip 3.0mu}\>[4]{}{\mathsf{c}_{1}\mskip 3.0mu}\>[5]{}{\mathnormal{⊕}\mskip 3.0mu}\>[6]{}{\ensuremath{H}\mskip 3.0mu}\>[7]{}{\mathsf{x}\mskip 3.0mu}\>[8]{}{\mathnormal{\&}\mskip 3.0mu}\>[9]{}{λ\mskip 3.0mu\allowbreak{}\mathnormal{(}\mskip 0.0mu\mathsf{c}_{1}\mskip 0.0mu\mathnormal{,}\mskip 3.0mu\mathsf{x}\mskip 0.0mu\mathnormal{)}\allowbreak{}\mskip 3.0mu\mathnormal{\rightarrow }}\<[E]{}\\
\>[4]{}{\mathsf{c}_{2}\mskip 3.0mu}\>[5]{}{\mathnormal{⊕}\mskip 3.0mu}\>[6]{}{\ensuremath{T^{\dagger}}\mskip 3.0mu}\>[7]{}{\mathsf{x}\mskip 3.0mu}\>[8]{}{\mathnormal{\&}\mskip 3.0mu}\>[9]{}{λ\mskip 3.0mu\allowbreak{}\mathnormal{(}\mskip 0.0mu\mathsf{c}_{2}\mskip 0.0mu\mathnormal{,}\mskip 3.0mu\mathsf{x}\mskip 0.0mu\mathnormal{)}\allowbreak{}\mskip 3.0mu\mathnormal{\rightarrow }}\<[E]{}\\
\>[4]{}{\mathsf{c}_{1}\mskip 3.0mu}\>[5]{}{\mathnormal{⊕}\mskip 3.0mu}\>[6]{}{\ensuremath{T}\mskip 3.0mu}\>[7]{}{\mathsf{x}\mskip 3.0mu}\>[8]{}{\mathnormal{\&}\mskip 3.0mu}\>[9]{}{λ\mskip 3.0mu\allowbreak{}\mathnormal{(}\mskip 0.0mu\mathsf{c}_{1}\mskip 0.0mu\mathnormal{,}\mskip 3.0mu\mathsf{x}\mskip 0.0mu\mathnormal{)}\allowbreak{}\mskip 3.0mu\mathnormal{\rightarrow }}\<[E]{}\\
\>[4]{}{\mathsf{c}_{2}\mskip 3.0mu}\>[5]{}{\mathnormal{⊕}\mskip 3.0mu}\>[6]{}{\ensuremath{T^{\dagger}}\mskip 3.0mu}\>[7]{}{\mathsf{x}\mskip 3.0mu}\>[8]{}{\mathnormal{\&}\mskip 3.0mu}\>[9]{}{λ\mskip 3.0mu\allowbreak{}\mathnormal{(}\mskip 0.0mu\mathsf{c}_{2}\mskip 0.0mu\mathnormal{,}\mskip 3.0mu\mathsf{x}\mskip 0.0mu\mathnormal{)}\allowbreak{}\mskip 3.0mu\mathnormal{\rightarrow }}\<[E]{}\\
\>[4]{}{\mathsf{c}_{2}\mskip 3.0mu}\>[5]{}{\mathnormal{⊕}\mskip 3.0mu}\>[6]{}{\ensuremath{T}\mskip 3.0mu}\>[7]{}{\mathsf{c}_{1}\mskip 3.0mu}\>[8]{}{\mathnormal{\&}\mskip 3.0mu}\>[9]{}{λ\mskip 3.0mu\allowbreak{}\mathnormal{(}\mskip 0.0mu\mathsf{c}_{2}\mskip 0.0mu\mathnormal{,}\mskip 3.0mu\mathsf{y}\mskip 0.0mu\mathnormal{)}\allowbreak{}\mskip 3.0mu\mathnormal{\rightarrow }}\<[E]{}\\
\>[4]{}{\allowbreak{}\mathnormal{(}\mskip 0.0mu\ensuremath{T}\mskip 3.0mu\mathsf{c}_{2}\mskip 3.0mu\mathnormal{⊕}\mskip 3.0mu\ensuremath{T^{\dagger}}\mskip 3.0mu\mathsf{y}\mskip 0.0mu\mathnormal{)}\allowbreak{}\mskip 2.0mu\mathnormal{\fatsemi }\mskip 3.0mu\allowbreak{}\mathnormal{(}\mskip 0.0mu\ensuremath{H}\mskip 3.0mu\allowbreak{}\mathnormal{(}\mskip 0.0mu\ensuremath{T}\mskip 3.0mu\mathsf{x}\mskip 0.0mu\mathnormal{)}\allowbreak{}\mskip 0.0mu\mathnormal{)}\allowbreak{}}\<[E]{}\\
\>[4]{}{\mathbf{where}\mskip 3.0mu\ensuremath{T^{\dagger}}\mskip 3.0mu\mathnormal{=}\mskip 3.0mu\mathsf{invert}\mskip 3.0mu\ensuremath{T}}\<[E]{}\end{parray}}\end{list} 
We use explicit β-redexes instead of let-bindings here because we want
to reuse some variable names: since using a linear variable makes
it unavailable in the remainder of the function, we may freely reuse
its name. Unfortunately, Haskell only has recursive lets, so if we
were to write \ensuremath{\mathbf{let}\mskip 3.0mu}\ensuremath{\allowbreak{}\mathnormal{(}\mskip 0.0mu}\ensuremath{\mathsf{c}_{1}\mskip 2.0mu}\ensuremath{\mathnormal{\fatsemi }\mskip 3.0mu}\ensuremath{\mathsf{x}\mskip 0.0mu}\ensuremath{\mathnormal{)}\allowbreak{}\mskip 3.0mu}\ensuremath{\mathnormal{=}\mskip 3.0mu}\ensuremath{\mathsf{c}_{1}\mskip 3.0mu}\ensuremath{\mathnormal{⊕}\mskip 3.0mu}\ensuremath{\ensuremath{H}\mskip 3.0mu}\ensuremath{\mathsf{x}\mskip 3.0mu}\ensuremath{\mathbf{in}\mskip 3.0mu}\ensuremath{\mathnormal{…}}, Haskell
would try to define both \ensuremath{\mathsf{c}_{1}} and \ensuremath{\mathsf{x}} recursively, which is
not the intended behaviour. To this effect, we use the reverse-order
linear application operator \ensuremath{\allowbreak{}\mathnormal{(}\mskip 0.0mu}\ensuremath{\mathnormal{\&}\mskip 0.0mu}\ensuremath{\mathnormal{)}\allowbreak{}} which is defined as

\begin{list}{}{\setlength\leftmargin{1.0em}}\item\relax
\ensuremath{\begin{parray}\column{B}{@{}>{}l<{}@{}}\column[0em]{1}{@{}>{}l<{}@{}}\column{2}{@{}>{}l<{}@{}}\column{E}{@{}>{}l<{}@{}}%
\>[1]{}{\allowbreak{}\mathnormal{(}\mskip 0.0mu\mathnormal{\&}\mskip 0.0mu\mathnormal{)}\allowbreak{}\mskip 3.0mu\mathnormal{::}\mskip 3.0mu}\>[2]{}{\mathsf{a}\mskip 3.0mu\mathnormal{⊸}\mskip 3.0mu\allowbreak{}\mathnormal{(}\mskip 0.0mu\mathsf{a}\mskip 3.0mu\mathnormal{⊸}\mskip 3.0mu\mathsf{b}\mskip 0.0mu\mathnormal{)}\allowbreak{}\mskip 3.0mu\mathnormal{⊸}\mskip 3.0mu\mathsf{b}}\<[E]{}\\
\>[1]{}{\mathsf{x}\mskip 3.0mu\mathnormal{\&}\mskip 3.0mu\mathsf{f}\mskip 3.0mu\mathnormal{=}\mskip 3.0mu\mathsf{f}\mskip 3.0mu\mathsf{x}}\<[E]{}\end{parray}}\end{list} 
This is a specificity of Haskell. In a language with non-recursive lets
the definition of \ensuremath{\mathsf{toffoli}} would look even more natural.

\subsection{Workflow orchestration}\label{54} 
\begin{figure}
{\begin{tikzpicture}\path[-,draw=black,line width=0.4000pt,line cap=butt,line join=miter,dash pattern=](-43.8950pt,23.3850pt)--(-36.8950pt,23.3850pt);
\path[-,draw=black,line width=0.4000pt,line cap=butt,line join=miter,dash pattern=](-43.8950pt,0.0000pt)--(-36.8950pt,0.0000pt);
\path[-,draw=black,line width=0.4000pt,line cap=butt,line join=miter,dash pattern=](0.0000pt,25.8850pt)--(7.0000pt,25.8850pt);
\path[-,draw=black,line width=0.4000pt,line cap=butt,line join=miter,dash pattern=](0.0000pt,2.5000pt)--(7.0000pt,2.5000pt);
\path[-,draw=black,line width=0.4000pt,line cap=butt,line join=miter,dash pattern=](-22.7300pt,25.8850pt)--(-22.7300pt,25.8850pt);
\path[-,draw=black,line width=0.4000pt,line cap=butt,line join=miter,dash pattern=](-22.7300pt,20.8850pt)--(-22.7300pt,20.8850pt);
\path[-,draw=black,line width=0.4000pt,line cap=butt,line join=miter,dash pattern=](-22.7300pt,0.0000pt)--(-22.7300pt,0.0000pt);
\path[-,draw=black,line width=0.4000pt,line cap=butt,line join=miter,dash pattern=](-13.1300pt,25.8850pt)--(-13.1300pt,25.8850pt);
\path[-,draw=black,line width=0.4000pt,line cap=butt,line join=miter,dash pattern=](-13.1300pt,5.0000pt)--(-13.1300pt,5.0000pt);
\path[-,draw=black,line width=0.4000pt,line cap=butt,line join=miter,dash pattern=](-13.1300pt,0.0000pt)--(-13.1300pt,0.0000pt);
\path[-,line width=0.4000pt,line cap=butt,line join=miter,dash pattern=](-8.8625pt,30.5275pt)--(-4.2675pt,30.5275pt)--(-4.2675pt,21.2425pt)--(-8.8625pt,21.2425pt)--cycle;
\node[anchor=north west,inner sep=0] at (-8.8625pt,30.5275pt){\savebox{\marxupbox}{{\ensuremath{ξ}}}\immediate\write\boxesfile{55}\immediate\write\boxesfile{\number\wd\marxupbox}\immediate\write\boxesfile{\number\ht\marxupbox}\immediate\write\boxesfile{\number\dp\marxupbox}\box\marxupbox};
\path[-,line width=0.4000pt,line cap=butt,line join=miter,dash pattern=](-12.8625pt,34.5275pt)--(-0.2675pt,34.5275pt)--(-0.2675pt,17.2425pt)--(-12.8625pt,17.2425pt)--cycle;
\path[-,draw=black,line width=0.4000pt,line cap=butt,line join=miter,dash pattern=](-12.8625pt,32.0275pt)..controls(-12.8625pt,33.4082pt)and(-11.7432pt,34.5275pt)..(-10.3625pt,34.5275pt)--(-2.7675pt,34.5275pt)..controls(-1.3868pt,34.5275pt)and(-0.2675pt,33.4082pt)..(-0.2675pt,32.0275pt)--(-0.2675pt,19.7425pt)..controls(-0.2675pt,18.3618pt)and(-1.3868pt,17.2425pt)..(-2.7675pt,17.2425pt)--(-10.3625pt,17.2425pt)..controls(-11.7432pt,17.2425pt)and(-12.8625pt,18.3618pt)..(-12.8625pt,19.7425pt)--cycle;
\path[-,draw=black,line width=0.4000pt,line cap=butt,line join=miter,dash pattern=](0.0000pt,25.8850pt)--(-0.2675pt,25.8850pt);
\path[-,draw=black,line width=0.4000pt,line cap=butt,line join=miter,dash pattern=](-13.1300pt,25.8850pt)--(-12.8625pt,25.8850pt);
\path[-,line width=0.4000pt,line cap=butt,line join=miter,dash pattern=](-9.1300pt,7.1925pt)--(-4.0000pt,7.1925pt)--(-4.0000pt,-2.1925pt)--(-9.1300pt,-2.1925pt)--cycle;
\node[anchor=north west,inner sep=0] at (-9.1300pt,7.1925pt){\savebox{\marxupbox}{{\ensuremath{ζ}}}\immediate\write\boxesfile{56}\immediate\write\boxesfile{\number\wd\marxupbox}\immediate\write\boxesfile{\number\ht\marxupbox}\immediate\write\boxesfile{\number\dp\marxupbox}\box\marxupbox};
\path[-,line width=0.4000pt,line cap=butt,line join=miter,dash pattern=](-13.1300pt,11.1925pt)--(0.0000pt,11.1925pt)--(0.0000pt,-6.1925pt)--(-13.1300pt,-6.1925pt)--cycle;
\path[-,draw=black,line width=0.4000pt,line cap=butt,line join=miter,dash pattern=](-13.1300pt,8.6925pt)..controls(-13.1300pt,10.0732pt)and(-12.0107pt,11.1925pt)..(-10.6300pt,11.1925pt)--(-2.5000pt,11.1925pt)..controls(-1.1193pt,11.1925pt)and(0.0000pt,10.0732pt)..(0.0000pt,8.6925pt)--(0.0000pt,-3.6925pt)..controls(0.0000pt,-5.0732pt)and(-1.1193pt,-6.1925pt)..(-2.5000pt,-6.1925pt)--(-10.6300pt,-6.1925pt)..controls(-12.0107pt,-6.1925pt)and(-13.1300pt,-5.0732pt)..(-13.1300pt,-3.6925pt)--cycle;
\path[-,draw=black,line width=0.4000pt,line cap=butt,line join=miter,dash pattern=](0.0000pt,2.5000pt)--(0.0000pt,2.5000pt);
\path[-,draw=black,line width=0.4000pt,line cap=butt,line join=miter,dash pattern=](-13.1300pt,5.0000pt)--(-13.1300pt,5.0000pt);
\path[-,draw=black,line width=0.4000pt,line cap=butt,line join=miter,dash pattern=](-13.1300pt,0.0000pt)--(-13.1300pt,0.0000pt);
\path[-,draw=black,line width=0.4000pt,line cap=butt,line join=miter,dash pattern=](-22.7300pt,25.8850pt)--(-13.1300pt,25.8850pt);
\path[-,draw=black,line width=0.4000pt,line cap=butt,line join=miter,dash pattern=](-22.7300pt,20.8850pt)..controls(-14.7300pt,20.8850pt)and(-21.1300pt,5.0000pt)..(-13.1300pt,5.0000pt);
\path[-,draw=black,line width=0.4000pt,line cap=butt,line join=miter,dash pattern=](-22.7300pt,0.0000pt)--(-13.1300pt,0.0000pt);
\path[-,line width=0.4000pt,line cap=butt,line join=miter,dash pattern=](-32.7500pt,28.0775pt)--(-26.8750pt,28.0775pt)--(-26.8750pt,18.6925pt)--(-32.7500pt,18.6925pt)--cycle;
\node[anchor=north west,inner sep=0] at (-32.7500pt,28.0775pt){\savebox{\marxupbox}{{\ensuremath{φ}}}\immediate\write\boxesfile{57}\immediate\write\boxesfile{\number\wd\marxupbox}\immediate\write\boxesfile{\number\ht\marxupbox}\immediate\write\boxesfile{\number\dp\marxupbox}\box\marxupbox};
\path[-,line width=0.4000pt,line cap=butt,line join=miter,dash pattern=](-36.7500pt,32.0775pt)--(-22.8750pt,32.0775pt)--(-22.8750pt,14.6925pt)--(-36.7500pt,14.6925pt)--cycle;
\path[-,draw=black,line width=0.4000pt,line cap=butt,line join=miter,dash pattern=](-36.7500pt,29.5775pt)..controls(-36.7500pt,30.9582pt)and(-35.6307pt,32.0775pt)..(-34.2500pt,32.0775pt)--(-25.3750pt,32.0775pt)..controls(-23.9943pt,32.0775pt)and(-22.8750pt,30.9582pt)..(-22.8750pt,29.5775pt)--(-22.8750pt,17.1925pt)..controls(-22.8750pt,15.8118pt)and(-23.9943pt,14.6925pt)..(-25.3750pt,14.6925pt)--(-34.2500pt,14.6925pt)..controls(-35.6307pt,14.6925pt)and(-36.7500pt,15.8118pt)..(-36.7500pt,17.1925pt)--cycle;
\path[-,draw=black,line width=0.4000pt,line cap=butt,line join=miter,dash pattern=](-22.7300pt,25.8850pt)--(-22.8750pt,25.8850pt);
\path[-,draw=black,line width=0.4000pt,line cap=butt,line join=miter,dash pattern=](-22.7300pt,20.8850pt)--(-22.8750pt,20.8850pt);
\path[-,draw=black,line width=0.4000pt,line cap=butt,line join=miter,dash pattern=](-36.8950pt,23.3850pt)--(-36.7500pt,23.3850pt);
\path[-,line width=0.4000pt,line cap=butt,line join=miter,dash pattern=](-32.8950pt,4.6925pt)--(-26.7300pt,4.6925pt)--(-26.7300pt,-4.6925pt)--(-32.8950pt,-4.6925pt)--cycle;
\node[anchor=north west,inner sep=0] at (-32.8950pt,4.6925pt){\savebox{\marxupbox}{{\ensuremath{ψ}}}\immediate\write\boxesfile{58}\immediate\write\boxesfile{\number\wd\marxupbox}\immediate\write\boxesfile{\number\ht\marxupbox}\immediate\write\boxesfile{\number\dp\marxupbox}\box\marxupbox};
\path[-,line width=0.4000pt,line cap=butt,line join=miter,dash pattern=](-36.8950pt,8.6925pt)--(-22.7300pt,8.6925pt)--(-22.7300pt,-8.6925pt)--(-36.8950pt,-8.6925pt)--cycle;
\path[-,draw=black,line width=0.4000pt,line cap=butt,line join=miter,dash pattern=](-36.8950pt,6.1925pt)..controls(-36.8950pt,7.5732pt)and(-35.7757pt,8.6925pt)..(-34.3950pt,8.6925pt)--(-25.2300pt,8.6925pt)..controls(-23.8493pt,8.6925pt)and(-22.7300pt,7.5732pt)..(-22.7300pt,6.1925pt)--(-22.7300pt,-6.1925pt)..controls(-22.7300pt,-7.5732pt)and(-23.8493pt,-8.6925pt)..(-25.2300pt,-8.6925pt)--(-34.3950pt,-8.6925pt)..controls(-35.7757pt,-8.6925pt)and(-36.8950pt,-7.5732pt)..(-36.8950pt,-6.1925pt)--cycle;
\path[-,draw=black,line width=0.4000pt,line cap=butt,line join=miter,dash pattern=](-22.7300pt,0.0000pt)--(-22.7300pt,0.0000pt);
\path[-,draw=black,line width=0.4000pt,line cap=butt,line join=miter,dash pattern=](-36.8950pt,0.0000pt)--(-36.8950pt,0.0000pt);
\end{tikzpicture}}\caption{A workflow corresponding to the morphism \ensuremath{\allowbreak{}\mathnormal{(}\mskip 0.0mu}\ensuremath{\mathsf{ξ}\mskip 3.0mu}\ensuremath{\allowbreak{}\mathnormal{×}\allowbreak{}\mskip 3.0mu}\ensuremath{\mathsf{ζ}\mskip 0.0mu}\ensuremath{\mathnormal{)}\allowbreak{}\mskip 3.0mu}\ensuremath{\allowbreak{}\mathnormal{∘}\allowbreak{}\mskip 3.0mu}\ensuremath{α\mskip 3.0mu}\ensuremath{\allowbreak{}\mathnormal{∘}\allowbreak{}\mskip 3.0mu}\ensuremath{\allowbreak{}\mathnormal{(}\mskip 0.0mu}\ensuremath{\mathsf{φ}\mskip 3.0mu}\ensuremath{\allowbreak{}\mathnormal{×}\allowbreak{}\mskip 3.0mu}\ensuremath{\mathsf{ψ}\mskip 0.0mu}\ensuremath{\mathnormal{)}\allowbreak{}}.}\label{59}\end{figure} 
Consider a type \ensuremath{\mathsf{Step}\mskip 3.0mu}\ensuremath{\mathsf{a}\mskip 3.0mu}\ensuremath{\mathsf{b}} representing
computations from type \ensuremath{\mathsf{a}} to type \ensuremath{\mathsf{b}}: a value of type
\ensuremath{\mathsf{Step}\mskip 3.0mu}\ensuremath{\mathsf{a}\mskip 3.0mu}\ensuremath{\mathsf{b}} may be some Haskell function, or it can run an
external command.  Whatever it is, we make the assumption that the side effects
embedded in a \ensuremath{\mathsf{Step}} are commutative.  That is, it never
matters if step \ensuremath{\mathsf{φ}} is run before step \ensuremath{\mathsf{ψ}} or the other way
around. And, in fact, if there is no data dependencies between
\ensuremath{\mathsf{φ}} ant \ensuremath{\mathsf{ψ}}, we want to run them in parallel.

What we want to do, in this scenario, is to compose individual steps
together to form bigger computations, typically called a
\emph{workflow}.  In \cref{59} we show a simple,
albeit typical, workflow.

What would a {\sc{}dsl} to that effect look like? A first attempt may be to organise the {\sc{}dsl} around a monad \ensuremath{\mathsf{M}}, and define
the workflow of \cref{59} as follows:

\begin{list}{}{\setlength\leftmargin{1.0em}}\item\relax
\ensuremath{\begin{parray}\column{B}{@{}>{}l<{}@{}}\column[0em]{1}{@{}>{}l<{}@{}}\column[1em]{2}{@{}>{}l<{}@{}}\column{E}{@{}>{}l<{}@{}}%
\>[1]{}{\mathbf{type}\mskip 3.0mu\mathsf{Workflow}\mskip 3.0mu\mathsf{a}\mskip 3.0mu\mathsf{b}\mskip 3.0mu\mathnormal{=}\mskip 3.0mu\mathsf{a}\mskip 3.0mu\mathnormal{\rightarrow }\mskip 3.0mu\mathsf{M}\mskip 3.0mu\mathsf{b}}\<[E]{}\\
\>[1]{}{\mathsf{workflowM}\mskip 3.0mu\mathnormal{::}\mskip 3.0mu\mathsf{Workflow}\mskip 3.0mu\allowbreak{}\mathnormal{(}\mskip 0.0mu\mathsf{A}\mskip 0.0mu\mathnormal{,}\mskip 3.0mu\mathsf{B}\mskip 0.0mu\mathnormal{)}\allowbreak{}\mskip 3.0mu\allowbreak{}\mathnormal{(}\mskip 0.0mu\mathsf{C}\mskip 0.0mu\mathnormal{,}\mskip 3.0mu\mathsf{D}\mskip 0.0mu\mathnormal{)}\allowbreak{}}\<[E]{}\\
\>[1]{}{\mathsf{workflowM}\mskip 3.0mu\allowbreak{}\mathnormal{(}\mskip 0.0mu\mathsf{a}\mskip 0.0mu\mathnormal{,}\mskip 3.0mu\mathsf{b}\mskip 0.0mu\mathnormal{)}\allowbreak{}\mskip 3.0mu\mathnormal{=}\mskip 3.0mu\mathbf{do}}\<[E]{}\\
\>[2]{}{\allowbreak{}\mathnormal{(}\mskip 0.0mu\mathsf{x}\mskip 0.0mu\mathnormal{,}\mskip 3.0mu\mathsf{y}\mskip 0.0mu\mathnormal{)}\allowbreak{}\mskip 3.0mu\mathnormal{\leftarrow }\mskip 3.0mu\mathsf{φ}\mskip 3.0mu\mathsf{a}}\<[E]{}\\
\>[2]{}{\mathsf{z}\mskip 3.0mu\mathnormal{\leftarrow }\mskip 3.0mu\mathsf{ψ}\mskip 3.0mu\mathsf{b}}\<[E]{}\\
\>[2]{}{\mathsf{c}\mskip 3.0mu\mathnormal{\leftarrow }\mskip 3.0mu\mathsf{ξ}\mskip 3.0mu\mathsf{x}}\<[E]{}\\
\>[2]{}{\mathsf{d}\mskip 3.0mu\mathnormal{\leftarrow }\mskip 3.0mu\mathsf{ζ}\mskip 3.0mu\allowbreak{}\mathnormal{(}\mskip 0.0mu\mathsf{y}\mskip 0.0mu\mathnormal{,}\mskip 3.0mu\mathsf{z}\mskip 0.0mu\mathnormal{)}\allowbreak{}}\<[E]{}\\
\>[2]{}{\mathsf{return}\mskip 3.0mu\allowbreak{}\mathnormal{(}\mskip 0.0mu\mathsf{c}\mskip 0.0mu\mathnormal{,}\mskip 3.0mu\mathsf{d}\mskip 0.0mu\mathnormal{)}\allowbreak{}}\<[E]{}\end{parray}}\end{list} The problem with this monadic {\sc{}dsl}, however, is that it forces us
to fully sequentialise our workflow: φ runs before ψ, which runs
before ξ, which runs before ζ. This is wasteful: a glance at
\cref{59} makes it obvious that φ and ψ can be run in
parallel, as well as ξ and ζ, etc. Running independent steps in parallel may
be crucial to performance. But the monad abstraction makes
the inherent parallelism fundamentally unrecoverable.

To improve upon this state of affairs, one could attempt to leverage an
applicative functor structure that \ensuremath{\mathsf{M}} may exhibit. Accordingly
one can recover parallelism as follows:
\begin{list}{}{\setlength\leftmargin{1.0em}}\item\relax
\ensuremath{\begin{parray}\column{B}{@{}>{}l<{}@{}}\column[0em]{1}{@{}>{}l<{}@{}}\column[1em]{2}{@{}>{}l<{}@{}}\column{E}{@{}>{}l<{}@{}}%
\>[1]{}{\mathsf{workflowA}\mskip 3.0mu\mathnormal{::}\mskip 3.0mu\mathsf{Workflow}\mskip 3.0mu\allowbreak{}\mathnormal{(}\mskip 0.0mu\mathsf{A}\mskip 0.0mu\mathnormal{,}\mskip 3.0mu\mathsf{B}\mskip 0.0mu\mathnormal{)}\allowbreak{}\mskip 3.0mu\allowbreak{}\mathnormal{(}\mskip 0.0mu\mathsf{C}\mskip 0.0mu\mathnormal{,}\mskip 3.0mu\mathsf{D}\mskip 0.0mu\mathnormal{)}\allowbreak{}}\<[E]{}\\
\>[1]{}{\mathsf{workflowA}\mskip 3.0mu\allowbreak{}\mathnormal{(}\mskip 0.0mu\mathsf{a}\mskip 0.0mu\mathnormal{,}\mskip 3.0mu\mathsf{b}\mskip 0.0mu\mathnormal{)}\allowbreak{}\mskip 3.0mu\mathnormal{=}\mskip 3.0mu\mathbf{do}}\<[E]{}\\
\>[2]{}{\allowbreak{}\mathnormal{(}\mskip 0.0mu\allowbreak{}\mathnormal{(}\mskip 0.0mu\mathsf{x}\mskip 0.0mu\mathnormal{,}\mskip 3.0mu\mathsf{y}\mskip 0.0mu\mathnormal{)}\allowbreak{}\mskip 0.0mu\mathnormal{,}\mskip 3.0mu\mathsf{z}\mskip 0.0mu\mathnormal{)}\allowbreak{}\mskip 3.0mu\mathnormal{\leftarrow }\mskip 3.0mu\allowbreak{}\mathnormal{(}\mskip 0.0mu\mathnormal{,}\mskip 0.0mu\mathnormal{)}\allowbreak{}\mskip 3.0mu\mathnormal{<{\mkern-6mu}\${\mkern-6mu}>}\mskip 3.0mu\mathsf{φ}\mskip 3.0mu\mathsf{a}\mskip 3.0mu\mathnormal{<{\mkern-12mu}*{\mkern-12mu}>}\mskip 3.0mu\mathsf{ψ}\mskip 3.0mu\mathsf{b}}\<[E]{}\\
\>[2]{}{\allowbreak{}\mathnormal{(}\mskip 0.0mu\mathsf{c}\mskip 0.0mu\mathnormal{,}\mskip 3.0mu\mathsf{d}\mskip 0.0mu\mathnormal{)}\allowbreak{}\mskip 3.0mu\mathnormal{\leftarrow }\mskip 3.0mu\allowbreak{}\mathnormal{(}\mskip 0.0mu\mathnormal{,}\mskip 0.0mu\mathnormal{)}\allowbreak{}\mskip 3.0mu\mathnormal{<{\mkern-6mu}\${\mkern-6mu}>}\mskip 3.0mu\mathsf{ξ}\mskip 3.0mu\mathsf{x}\mskip 3.0mu\mathnormal{<{\mkern-12mu}*{\mkern-12mu}>}\mskip 3.0mu\mathsf{ζ}\mskip 3.0mu\allowbreak{}\mathnormal{(}\mskip 0.0mu\mathsf{y}\mskip 0.0mu\mathnormal{,}\mskip 3.0mu\mathsf{z}\mskip 0.0mu\mathnormal{)}\allowbreak{}}\<[E]{}\\
\>[2]{}{\mathsf{return}\mskip 3.0mu\allowbreak{}\mathnormal{(}\mskip 0.0mu\mathsf{c}\mskip 0.0mu\mathnormal{,}\mskip 3.0mu\mathsf{d}\mskip 0.0mu\mathnormal{)}\allowbreak{}}\<[E]{}\end{parray}}\end{list}  This is
the style advocated, in the context of database query batching, by the Haxl library \citep{haxl_2014}. {\sc{}Ghc} even features an extension
(\ensuremath{\mathsf{ApplicativeDo}} \citep{marlow_desugaring_2016}) that automatically translates code written
using the \ensuremath{\mathbf{do}}-notation (as in \ensuremath{\mathsf{workflowM}}) to use applicative combinators for parallel commands as \ensuremath{\mathsf{workflowA}} does.
Unfortunately, even \ensuremath{\mathsf{workflowA}} doesn't fully expose all
the parallelism opportunities: \ensuremath{\mathsf{workflowA}} will run both φ and ψ in parallel, but it
will wait until both are completed before starting either ξ or ζ. But
only the result of φ is necessary to run ξ. If ψ takes more time to run than φ, then
this is wasteful.

One could try to rewrite the workflow as follows:
\begin{list}{}{\setlength\leftmargin{1.0em}}\item\relax
\ensuremath{\begin{parray}\column{B}{@{}>{}l<{}@{}}\column[0em]{1}{@{}>{}l<{}@{}}\column[1em]{2}{@{}>{}l<{}@{}}\column{3}{@{}>{}l<{}@{}}\column{E}{@{}>{}l<{}@{}}%
\>[1]{}{\mathsf{workflowA'}\mskip 3.0mu\mathnormal{::}\mskip 3.0mu\mathsf{Workflow}\mskip 3.0mu\allowbreak{}\mathnormal{(}\mskip 0.0mu\mathsf{A}\mskip 0.0mu\mathnormal{,}\mskip 3.0mu\mathsf{B}\mskip 0.0mu\mathnormal{)}\allowbreak{}\mskip 3.0mu\allowbreak{}\mathnormal{(}\mskip 0.0mu\mathsf{C}\mskip 0.0mu\mathnormal{,}\mskip 3.0mu\mathsf{D}\mskip 0.0mu\mathnormal{)}\allowbreak{}}\<[E]{}\\
\>[1]{}{\mathsf{workflowA'}\mskip 3.0mu\allowbreak{}\mathnormal{(}\mskip 0.0mu\mathsf{a}\mskip 0.0mu\mathnormal{,}\mskip 3.0mu\mathsf{b}\mskip 0.0mu\mathnormal{)}\allowbreak{}\mskip 3.0mu\mathnormal{=}\mskip 3.0mu\mathbf{do}}\<[E]{}\\
\>[2]{}{\allowbreak{}\mathnormal{(}\mskip 0.0mu\allowbreak{}\mathnormal{(}\mskip 0.0mu\mathsf{y}\mskip 0.0mu\mathnormal{,}\mskip 3.0mu\mathsf{c}\mskip 0.0mu\mathnormal{)}\allowbreak{}\mskip 0.0mu\mathnormal{,}\mskip 3.0mu\mathsf{z}\mskip 0.0mu\mathnormal{)}\allowbreak{}\mskip 3.0mu\mathnormal{\leftarrow }\mskip 3.0mu\allowbreak{}\mathnormal{(}\mskip 0.0mu\mathnormal{,}\mskip 0.0mu\mathnormal{)}\allowbreak{}\mskip 3.0mu\mathnormal{<{\mkern-6mu}\${\mkern-6mu}>}\mskip 3.0mu\mathsf{part}_{1}\mskip 3.0mu\mathnormal{<{\mkern-12mu}*{\mkern-12mu}>}\mskip 3.0mu\mathsf{ψ}\mskip 3.0mu\mathsf{b}}\<[E]{}\\
\>[2]{}{\mathsf{d}\mskip 3.0mu\mathnormal{\leftarrow }\mskip 3.0mu\mathsf{ζ}\mskip 3.0mu\allowbreak{}\mathnormal{(}\mskip 0.0mu\mathsf{y}\mskip 0.0mu\mathnormal{,}\mskip 3.0mu\mathsf{z}\mskip 0.0mu\mathnormal{)}\allowbreak{}}\<[E]{}\\
\>[2]{}{\mathsf{return}\mskip 3.0mu\allowbreak{}\mathnormal{(}\mskip 0.0mu\mathsf{c}\mskip 0.0mu\mathnormal{,}\mskip 3.0mu\mathsf{d}\mskip 0.0mu\mathnormal{)}\allowbreak{}}\<[E]{}\\
\>[2]{}{\mathbf{where}\mskip 3.0mu\mathsf{part}_{1}\mskip 3.0mu\mathnormal{=}\mskip 3.0mu\mathbf{do}\mskip 3.0mu}\>[3]{}{\allowbreak{}\mathnormal{(}\mskip 0.0mu\mathsf{x}\mskip 0.0mu\mathnormal{,}\mskip 3.0mu\mathsf{y}\mskip 0.0mu\mathnormal{)}\allowbreak{}\mskip 3.0mu\mathnormal{\leftarrow }\mskip 3.0mu\mathsf{φ}\mskip 3.0mu\mathsf{a}}\<[E]{}\\
\>[3]{}{\mathsf{c}\mskip 3.0mu\mathnormal{\leftarrow }\mskip 3.0mu\mathsf{ξ}\mskip 3.0mu\mathsf{x}}\<[E]{}\\
\>[3]{}{\mathsf{return}\mskip 3.0mu\allowbreak{}\mathnormal{(}\mskip 0.0mu\mathsf{y}\mskip 0.0mu\mathnormal{,}\mskip 3.0mu\mathsf{c}\mskip 0.0mu\mathnormal{)}\allowbreak{}}\<[E]{}\end{parray}}\end{list} Now ξ can start as soon as φ completes, and run in parallel with
ψ. But ζ has to wait for ξ to complete before it can start.
In sum, the combined Applicative-Monadic interface prevents \emph{any} implementation to fully expose the parallelism opportunities
inherent in the workflow.

Haskell offers another abstraction, called \emph{arrows} \citep{hughes_generalising_2000}, to model parallelism. This is how
\citet{pars_algebraic_2020} model workflows. In this style, our example
would look like:
\begin{list}{}{\setlength\leftmargin{1.0em}}\item\relax
\ensuremath{\begin{parray}\column{B}{@{}>{}l<{}@{}}\column[0em]{1}{@{}>{}l<{}@{}}\column{2}{@{}>{}l<{}@{}}\column{E}{@{}>{}l<{}@{}}%
\>[1]{}{\mathbf{data}\mskip 3.0mu\mathsf{Workflow}\mskip 3.0mu\mathsf{a}\mskip 3.0mu\mathsf{b}}\<[E]{}\\
\>[1]{}{\mathbf{instance}\mskip 3.0mu\mathsf{Arrow}\mskip 3.0mu\mathsf{Workflow}}\<[E]{}\\
\>[1]{}{\mathsf{workflowArr}\mskip 3.0mu\mathnormal{::}\mskip 3.0mu\mathsf{Workflow}\mskip 3.0mu\allowbreak{}\mathnormal{(}\mskip 0.0mu\mathsf{A}\mskip 0.0mu\mathnormal{,}\mskip 3.0mu\mathsf{B}\mskip 0.0mu\mathnormal{)}\allowbreak{}\mskip 3.0mu\allowbreak{}\mathnormal{(}\mskip 0.0mu\mathsf{C}\mskip 0.0mu\mathnormal{,}\mskip 3.0mu\mathsf{D}\mskip 0.0mu\mathnormal{)}\allowbreak{}}\<[E]{}\\
\>[1]{}{\mathsf{workflowArr}\mskip 3.0mu\mathnormal{=}\mskip 3.0mu}\>[2]{}{\allowbreak{}\mathnormal{(}\mskip 0.0mu\mathsf{φ}\mskip 3.0mu\allowbreak{}\mathnormal{*\mkern-4mu *\mkern-4mu *}\allowbreak{}\mskip 3.0mu\mathsf{ψ}\mskip 0.0mu\mathnormal{)}\allowbreak{}\mskip 3.0mu\allowbreak{}\mathnormal{>\mkern-3mu >\mkern-3mu >}\allowbreak{}}\<[E]{}\\
\>[2]{}{\mathsf{arr}\mskip 3.0mu\allowbreak{}\mathnormal{(}\mskip 0.0muλ\mskip 3.0mu\allowbreak{}\mathnormal{(}\mskip 0.0mu\allowbreak{}\mathnormal{(}\mskip 0.0mu\mathsf{x}\mskip 0.0mu\mathnormal{,}\mskip 3.0mu\mathsf{y}\mskip 0.0mu\mathnormal{)}\allowbreak{}\mskip 0.0mu\mathnormal{,}\mskip 3.0mu\mathsf{z}\mskip 0.0mu\mathnormal{)}\allowbreak{}\mskip 3.0mu\mathnormal{\rightarrow }\mskip 3.0mu\allowbreak{}\mathnormal{(}\mskip 0.0mu\mathsf{x}\mskip 0.0mu\mathnormal{,}\mskip 3.0mu\allowbreak{}\mathnormal{(}\mskip 0.0mu\mathsf{y}\mskip 0.0mu\mathnormal{,}\mskip 3.0mu\mathsf{z}\mskip 0.0mu\mathnormal{)}\allowbreak{}\mskip 0.0mu\mathnormal{)}\allowbreak{}\mskip 0.0mu\mathnormal{)}\allowbreak{}\mskip 3.0mu\allowbreak{}\mathnormal{>\mkern-3mu >\mkern-3mu >}\allowbreak{}}\<[E]{}\\
\>[2]{}{\allowbreak{}\mathnormal{(}\mskip 0.0mu\mathsf{ξ}\mskip 3.0mu\allowbreak{}\mathnormal{*\mkern-4mu *\mkern-4mu *}\allowbreak{}\mskip 3.0mu\mathsf{ζ}\mskip 0.0mu\mathnormal{)}\allowbreak{}}\<[E]{}\end{parray}}\end{list} 
Or, using the built-in notation for arrows \citep{paterson_new_2001} \begin{list}{}{\setlength\leftmargin{1.0em}}\item\relax
\ensuremath{\begin{parray}\column{B}{@{}>{}l<{}@{}}\column[0em]{1}{@{}>{}l<{}@{}}\column[1em]{2}{@{}>{}l<{}@{}}\column{E}{@{}>{}l<{}@{}}%
\>[1]{}{\mathsf{workflowArr'}\mskip 3.0mu\mathnormal{::}\mskip 3.0mu\mathsf{Workflow}\mskip 3.0mu\allowbreak{}\mathnormal{(}\mskip 0.0mu\mathsf{A}\mskip 0.0mu\mathnormal{,}\mskip 3.0mu\mathsf{B}\mskip 0.0mu\mathnormal{)}\allowbreak{}\mskip 3.0mu\allowbreak{}\mathnormal{(}\mskip 0.0mu\mathsf{C}\mskip 0.0mu\mathnormal{,}\mskip 3.0mu\mathsf{D}\mskip 0.0mu\mathnormal{)}\allowbreak{}}\<[E]{}\\
\>[1]{}{\mathsf{workflowArr'}\mskip 3.0mu\mathnormal{=}\mskip 3.0mu\mathsf{proc}\mskip 3.0mu\allowbreak{}\mathnormal{(}\mskip 0.0mu\mathsf{a}\mskip 0.0mu\mathnormal{,}\mskip 3.0mu\mathsf{b}\mskip 0.0mu\mathnormal{)}\allowbreak{}\mskip 3.0mu\mathnormal{\rightarrow }\mskip 3.0mu\mathbf{do}}\<[E]{}\\
\>[2]{}{\allowbreak{}\mathnormal{(}\mskip 0.0mu\mathsf{x}\mskip 0.0mu\mathnormal{,}\mskip 3.0mu\mathsf{y}\mskip 0.0mu\mathnormal{)}\allowbreak{}\mskip 3.0mu\mathnormal{\leftarrow }\mskip 3.0mu\mathsf{φ}\mskip 3.0mu\mathnormal{\lefttail }\mskip 3.0mu\mathsf{a}}\<[E]{}\\
\>[2]{}{\mathsf{z}\mskip 3.0mu\mathnormal{\leftarrow }\mskip 3.0mu\mathsf{ψ}\mskip 3.0mu\mathnormal{\lefttail }\mskip 3.0mu\mathsf{b}}\<[E]{}\\
\>[2]{}{\mathsf{c}\mskip 3.0mu\mathnormal{\leftarrow }\mskip 3.0mu\mathsf{ξ}\mskip 3.0mu\mathnormal{\lefttail }\mskip 3.0mu\mathsf{x}}\<[E]{}\\
\>[2]{}{\mathsf{d}\mskip 3.0mu\mathnormal{\leftarrow }\mskip 3.0mu\mathsf{ζ}\mskip 3.0mu\mathnormal{\lefttail }\mskip 3.0mu\allowbreak{}\mathnormal{(}\mskip 0.0mu\mathsf{y}\mskip 0.0mu\mathnormal{,}\mskip 3.0mu\mathsf{z}\mskip 0.0mu\mathnormal{)}\allowbreak{}}\<[E]{}\\
\>[2]{}{\mathsf{returnA}\mskip 3.0mu\mathnormal{\lefttail }\mskip 3.0mu\allowbreak{}\mathnormal{(}\mskip 0.0mu\mathsf{c}\mskip 0.0mu\mathnormal{,}\mskip 3.0mu\mathsf{d}\mskip 0.0mu\mathnormal{)}\allowbreak{}}\<[E]{}\end{parray}}\end{list} In
\ensuremath{\mathsf{workflowArr}}, like \ensuremath{\mathsf{workflowA}}, \ensuremath{\mathsf{ξ}} must run after
\ensuremath{\mathsf{ψ}}. It is also possible to write a version of the workflow
which, like \ensuremath{\mathsf{workflowA'}}, has \ensuremath{\mathsf{ξ}} running in parallel with
\ensuremath{\mathsf{ψ}}, but \ensuremath{\mathsf{ζ}} must run after \ensuremath{\mathsf{ξ}}.
In sum, This arrow-based {\sc{}dsl} suffers from the same problem as
the applicative {\sc{}dsl}: some over-sequentialisation is unavoidable.\footnote{ This problem with \ensuremath{\mathsf{Arrow}} can be attributed to the \ensuremath{\mathsf{arr}} combinator. Because \ensuremath{\mathsf{arr}} embeds a Haskell function, it is opaque and
thus prevents any efficient scheduling strategy between the morphisms connected to it.
Of course, for some specific \ensuremath{\mathsf{Arrow}} instances, one can provide the combinators
of an {\sc{}smc} and recover a better behaviour when using them instead of \ensuremath{\mathsf{arr}}. However this does not apply when using the arrow notation, because it always
desugars to calls to \ensuremath{\mathsf{arr}}.} Indeed, a situation just as this one, in
an industrial workflow, was one of the motivations behind this paper. It was
impossible to optimise resources usage in that workflow due to the
limitation of the arrow abstraction, wasting resources.

In contrast, if workflows are given an {\sc{}smc} instance, all the
parallelism of \cref{59} is exposed and can be exploited by
the workflow scheduler.
\begin{list}{}{\setlength\leftmargin{1.0em}}\item\relax
\ensuremath{\begin{parray}\column{B}{@{}>{}l<{}@{}}\column[0em]{1}{@{}>{}l<{}@{}}\column{2}{@{}>{}l<{}@{}}\column{3}{@{}>{}l<{}@{}}\column{4}{@{}>{}l<{}@{}}\column{E}{@{}>{}l<{}@{}}%
\>[1]{}{\mathbf{data}\mskip 3.0mu\mathsf{Workflow}\mskip 3.0mu\mathsf{a}\mskip 3.0mu\mathsf{b}}\<[E]{}\\
\>[1]{}{\mathbf{instance}\mskip 3.0mu\mathsf{Monoidal}\mskip 3.0mu\mathsf{Workflow}}\<[E]{}\\
\>[1]{}{\mathsf{workflowSMC}\mskip 3.0mu}\>[2]{}{\mathnormal{::}\mskip 3.0mu\mathsf{P}\mskip 3.0mu\mathsf{Workflow}\mskip 3.0mu\mathsf{r}\mskip 3.0mu\allowbreak{}\mathnormal{(}\mskip 0.0mu\mathsf{K}\mskip 3.0mu\mathsf{A}\mskip 0.0mu\mathnormal{)}\allowbreak{}\mskip 3.0mu\mathnormal{⊸}\mskip 3.0mu\mathsf{P}\mskip 3.0mu\mathsf{Workflow}\mskip 3.0mu\mathsf{r}\mskip 3.0mu\allowbreak{}\mathnormal{(}\mskip 0.0mu\mathsf{K}\mskip 3.0mu\mathsf{B}\mskip 0.0mu\mathnormal{)}\allowbreak{}}\<[E]{}\\
\>[2]{}{\mathnormal{⊸}\mskip 3.0mu\mathsf{P}\mskip 3.0mu\mathsf{Workflow}\mskip 3.0mu\allowbreak{}\mathnormal{(}\mskip 0.0mu\mathsf{K}\mskip 3.0mu\mathsf{C}\mskip 3.0mu\mathnormal{⊗}\mskip 3.0mu\mathsf{K}\mskip 3.0mu\mathsf{D}\mskip 0.0mu\mathnormal{)}\allowbreak{}}\<[E]{}\\
\>[1]{}{\mathsf{workflowSMC}\mskip 3.0mu\mathsf{a}\mskip 3.0mu\mathsf{b}\mskip 3.0mu\mathnormal{=}\mskip 3.0mu}\>[3]{}{\mathsf{φ}\mskip 3.0mu\mathsf{a}\mskip 3.0mu}\>[4]{}{\mathnormal{\&}\mskip 3.0muλ\mskip 3.0mu\allowbreak{}\mathnormal{(}\mskip 0.0mu\mathsf{x}\mskip 2.0mu\mathnormal{\fatsemi }\mskip 3.0mu\mathsf{y}\mskip 0.0mu\mathnormal{)}\allowbreak{}\mskip 3.0mu\mathnormal{\rightarrow }}\<[E]{}\\
\>[3]{}{\mathsf{ψ}\mskip 3.0mu\mathsf{b}\mskip 3.0mu}\>[4]{}{\mathnormal{\&}\mskip 3.0muλ\mskip 3.0mu\mathsf{z}\mskip 3.0mu\mathnormal{\rightarrow }}\<[E]{}\\
\>[3]{}{\mathsf{ξ}\mskip 3.0mu\mathsf{x}\mskip 3.0mu}\>[4]{}{\mathnormal{\&}\mskip 3.0muλ\mskip 3.0mu\mathsf{c}\mskip 3.0mu\mathnormal{\rightarrow }}\<[E]{}\\
\>[3]{}{\mathsf{ζ}\mskip 3.0mu\mathsf{y}\mskip 3.0mu\mathsf{z}\mskip 3.0mu}\>[4]{}{\mathnormal{\&}\mskip 3.0muλ\mskip 3.0mu\mathsf{d}\mskip 3.0mu\mathnormal{\rightarrow }}\<[E]{}\\
\>[3]{}{\allowbreak{}\mathnormal{(}\mskip 0.0mu\mathsf{c}\mskip 2.0mu\mathnormal{\fatsemi }\mskip 3.0mu\mathsf{d}\mskip 0.0mu\mathnormal{)}\allowbreak{}}\<[E]{}\end{parray}}\end{list} This version is syntactically close to the monadic
\ensuremath{\mathsf{workflowM}} implementation: the chief difference is the use of reverse application (\ensuremath{\mathnormal{\&}}) instead of the monadic bind. Yet, all the parallelism is retained!

A noteworthy element of this workflow {\sc{}dsl} is the presence of \ensuremath{\mathsf{K}} wrappers around the types \ensuremath{\mathsf{A}}, \ensuremath{\mathsf{B}}, \ensuremath{\mathsf{C}}, and
\ensuremath{\mathsf{D}}. The rationale is that synchronisation points will be at the
level of atomic types, and \ensuremath{\mathsf{K}} indicates such atomic types. That
is, if two sub-workflows are connected by the type \ensuremath{\mathsf{K}\mskip 3.0mu}\ensuremath{\allowbreak{}\mathnormal{(}\mskip 0.0mu}\ensuremath{\mathsf{a}\mskip 3.0mu}\ensuremath{\mathnormal{⊗}\mskip 3.0mu}\ensuremath{\mathsf{b}\mskip 0.0mu}\ensuremath{\mathnormal{)}\allowbreak{}},
then there can be no parallelisation between them. However, if they
are connected by \ensuremath{\mathsf{K}\mskip 3.0mu}\ensuremath{\mathsf{a}\mskip 3.0mu}\ensuremath{\mathnormal{⊗}\mskip 3.0mu}\ensuremath{\mathsf{K}\mskip 3.0mu}\ensuremath{\mathsf{b}}, then parallelisation can be
discovered by the scheduler.  (Another option to identify atomic types
would have to let \ensuremath{\mathnormal{⊗}} be different from the native Haskell
product.)

\begin{mdframed}[linewidth=0pt,hidealllines,innerleftmargin=0pt,innerrightmargin=0pt,backgroundcolor=gray!15] Composable workflows are implemented as \ensuremath{\mathsf{IO}} actions connecting
two synchronisation \ensuremath{\mathsf{Point}}s. This means in particular that they
embed synchronisation primitives (which reside in \ensuremath{\mathsf{IO}\mskip 3.0mu}\ensuremath{\allowbreak{}\mathnormal{(}\mskip 0.0mu}\ensuremath{\mathnormal{)}\allowbreak{}} in
Concurrent Haskell):

\begin{list}{}{\setlength\leftmargin{1.0em}}\item\relax
\ensuremath{\begin{parray}\column{B}{@{}>{}l<{}@{}}\column[0em]{1}{@{}>{}l<{}@{}}\column[1em]{2}{@{}>{}l<{}@{}}\column{E}{@{}>{}l<{}@{}}%
\>[1]{}{\mathbf{data}\mskip 3.0mu\mathsf{Workflow}\mskip 3.0mu\mathsf{a}\mskip 3.0mu\mathsf{b}}\<[E]{}\\
\>[2]{}{\mathnormal{=}\mskip 3.0mu\mathsf{W}\mskip 3.0mu\allowbreak{}\mathnormal{\{}\mskip 0.0mu\mathsf{taskRun}\mskip 3.0mu\mathnormal{::}\mskip 3.0mu\mathsf{Point}\mskip 3.0mu\mathsf{a}\mskip 3.0mu\mathnormal{\rightarrow }\mskip 3.0mu\mathsf{Point}\mskip 3.0mu\mathsf{b}\mskip 3.0mu\mathnormal{\rightarrow }\mskip 3.0mu\mathsf{IO}\mskip 3.0mu\allowbreak{}\mathnormal{(}\mskip 0.0mu\mathnormal{)}\allowbreak{}\mskip 0.0mu\mathnormal{\}}\allowbreak{}}\<[E]{}\end{parray}}\end{list} 
These \ensuremath{\mathsf{Point}}s
must be at the level of base types (not products thereof) so that
synchronisation is as fine-grained as necessary.  Hence, the
\ensuremath{\mathsf{Point}} type must be defined by structural induction over types,
such that the synchronisation point of a product is the product of
synchronisation points.  For atomic types, synchronisation can be
implemented by any suitable mechanism provided by Haskell.
Here we have chosen the \ensuremath{\mathsf{MVar}}s of concurrent Haskell
\citep{peyton_jones_concurrent_1996}.
Such an induction can be implemented in Haskell by exploiting the
\ensuremath{\mathsf{Obj}} constraint over types. We let it be a type-class
\ensuremath{\mathsf{HasPoint}}, with separate instances for products and for base
types. The form of base type is required to be \ensuremath{\mathsf{K}\mskip 3.0mu}\ensuremath{\mathsf{a}} with
\ensuremath{\mathbf{data}\mskip 3.0mu}\ensuremath{\mathsf{K}\mskip 3.0mu}\ensuremath{\mathsf{a}\mskip 3.0mu}\ensuremath{\mathnormal{=}\mskip 3.0mu}\ensuremath{\mathsf{K}\mskip 3.0mu}\ensuremath{\mathsf{a}}. If the type \ensuremath{\mathsf{Point}\mskip 3.0mu}\ensuremath{\mathsf{a}} is an associated
data type of the class \ensuremath{\mathsf{HasPoint}}, we get an inductive
definition as desired:

\begin{list}{}{\setlength\leftmargin{1.0em}}\item\relax
\ensuremath{\begin{parray}\column{B}{@{}>{}l<{}@{}}\column[0em]{1}{@{}>{}l<{}@{}}\column[1em]{2}{@{}>{}l<{}@{}}\column[2em]{3}{@{}>{}l<{}@{}}\column[3em]{4}{@{}>{}l<{}@{}}\column{E}{@{}>{}l<{}@{}}%
\>[1]{}{\mathbf{class}\mskip 3.0mu\mathsf{HasPoint}\mskip 3.0mu\mathsf{a}\mskip 3.0mu\mathbf{where}}\<[E]{}\\
\>[2]{}{\mathbf{data}\mskip 3.0mu\mathsf{Point}\mskip 3.0mu\mathsf{a}\mskip 3.0mu\mathnormal{::}\mskip 3.0mu\mathsf{Type}}\<[E]{}\\
\>[2]{}{\mathsf{mkPoint}\mskip 3.0mu\mathnormal{::}\mskip 3.0mu\mathsf{IO}\mskip 3.0mu\allowbreak{}\mathnormal{(}\mskip 0.0mu\mathsf{Point}\mskip 3.0mu\mathsf{a}\mskip 0.0mu\mathnormal{)}\allowbreak{}}\<[E]{}\\
\>[2]{}{\mathsf{connect}\mskip 3.0mu\mathnormal{::}\mskip 3.0mu\mathsf{Point}\mskip 3.0mu\mathsf{a}\mskip 3.0mu\mathnormal{\rightarrow }\mskip 3.0mu\mathsf{Point}\mskip 3.0mu\mathsf{a}\mskip 3.0mu\mathnormal{\rightarrow }\mskip 3.0mu\mathsf{IO}\mskip 3.0mu\allowbreak{}\mathnormal{(}\mskip 0.0mu\mathnormal{)}\allowbreak{}}\<[E]{}\\
\>[1]{}{\mathbf{instance}\mskip 3.0mu\allowbreak{}\mathnormal{(}\mskip 0.0mu\mathsf{HasPoint}\mskip 3.0mu\mathsf{a}\mskip 0.0mu\mathnormal{,}\mskip 3.0mu\mathsf{HasPoint}\mskip 3.0mu\mathsf{b}\mskip 0.0mu\mathnormal{)}\allowbreak{}\mskip 3.0mu\mathnormal{\Rightarrow }\mskip 3.0mu\mathsf{HasPoint}\mskip 3.0mu\allowbreak{}\mathnormal{(}\mskip 0.0mu\mathsf{a}\mskip 3.0mu\mathnormal{⊗}\mskip 3.0mu\mathsf{b}\mskip 0.0mu\mathnormal{)}\allowbreak{}\mskip 3.0mu\mathbf{where}}\<[E]{}\\
\>[2]{}{\mathbf{data}\mskip 3.0mu\mathsf{Point}\mskip 3.0mu\allowbreak{}\mathnormal{(}\mskip 0.0mu\mathsf{a}\mskip 3.0mu\mathnormal{⊗}\mskip 3.0mu\mathsf{b}\mskip 0.0mu\mathnormal{)}\allowbreak{}\mskip 3.0mu\mathnormal{=}\mskip 3.0mu\allowbreak{}\mathnormal{(}\mskip 0.0mu\mathsf{Point}\mskip 3.0mu\mathsf{a}\mskip 0.0mu\mathnormal{)}\allowbreak{}\mskip 3.0mu\allowbreak{}\mathnormal{:\kern -3pt *}\allowbreak{}\mskip 3.0mu\allowbreak{}\mathnormal{(}\mskip 0.0mu\mathsf{Point}\mskip 3.0mu\mathsf{b}\mskip 0.0mu\mathnormal{)}\allowbreak{}}\<[E]{}\\
\>[2]{}{\mathsf{mkPoint}\mskip 3.0mu\mathnormal{=}\mskip 3.0mu\mathbf{do}}\<[E]{}\\
\>[3]{}{\mathsf{a}\mskip 3.0mu\mathnormal{\leftarrow }\mskip 3.0mu\mathsf{mkPoint}}\<[E]{}\\
\>[3]{}{\mathsf{b}\mskip 3.0mu\mathnormal{\leftarrow }\mskip 3.0mu\mathsf{mkPoint}}\<[E]{}\\
\>[3]{}{\mathsf{return}\mskip 3.0mu\allowbreak{}\mathnormal{(}\mskip 0.0mu\mathsf{a}\mskip 3.0mu\allowbreak{}\mathnormal{:\kern -3pt *}\allowbreak{}\mskip 3.0mu\mathsf{b}\mskip 0.0mu\mathnormal{)}\allowbreak{}}\<[E]{}\\
\>[2]{}{\mathsf{connect}\mskip 3.0mu\allowbreak{}\mathnormal{(}\mskip 0.0mu\mathsf{a}\mskip 3.0mu\allowbreak{}\mathnormal{:\kern -3pt *}\allowbreak{}\mskip 3.0mu\mathsf{b}\mskip 0.0mu\mathnormal{)}\allowbreak{}\mskip 3.0mu\allowbreak{}\mathnormal{(}\mskip 0.0mu\mathsf{a'}\mskip 3.0mu\allowbreak{}\mathnormal{:\kern -3pt *}\allowbreak{}\mskip 3.0mu\mathsf{b'}\mskip 0.0mu\mathnormal{)}\allowbreak{}\mskip 3.0mu\mathnormal{=}\mskip 3.0mu\mathbf{do}}\<[E]{}\\
\>[3]{}{\mathsf{connect}\mskip 3.0mu\mathsf{a}\mskip 3.0mu\mathsf{a'}}\<[E]{}\\
\>[3]{}{\mathsf{connect}\mskip 3.0mu\mathsf{b}\mskip 3.0mu\mathsf{b'}}\<[E]{}\\
\>[1]{}{\mathbf{instance}\mskip 3.0mu\mathsf{HasPoint}\mskip 3.0mu\allowbreak{}\mathnormal{(}\mskip 0.0mu\mathsf{K}\mskip 3.0mu\mathsf{a}\mskip 0.0mu\mathnormal{)}\allowbreak{}\mskip 3.0mu\mathbf{where}}\<[E]{}\\
\>[2]{}{\mathbf{data}\mskip 3.0mu\mathsf{Point}\mskip 3.0mu\allowbreak{}\mathnormal{(}\mskip 0.0mu\mathsf{K}\mskip 3.0mu\mathsf{a}\mskip 0.0mu\mathnormal{)}\allowbreak{}\mskip 3.0mu\mathnormal{=}\mskip 3.0mu\mathsf{Atom}\mskip 3.0mu\allowbreak{}\mathnormal{(}\mskip 0.0mu\mathsf{MVar}\mskip 3.0mu\mathsf{a}\mskip 0.0mu\mathnormal{)}\allowbreak{}}\<[E]{}\\
\>[2]{}{\mathsf{mkPoint}\mskip 3.0mu\mathnormal{=}\mskip 3.0mu\mathsf{Atom}\mskip 3.0mu\mathnormal{<{\mkern-6mu}\${\mkern-6mu}>}\mskip 3.0mu\mathsf{newEmptyMVar}}\<[E]{}\\
\>[2]{}{\mathsf{connect}\mskip 3.0mu\allowbreak{}\mathnormal{(}\mskip 0.0mu\mathsf{Atom}\mskip 3.0mu\mathsf{a}\mskip 0.0mu\mathnormal{)}\allowbreak{}\mskip 3.0mu\allowbreak{}\mathnormal{(}\mskip 0.0mu\mathsf{Atom}\mskip 3.0mu\mathsf{b}\mskip 0.0mu\mathnormal{)}\allowbreak{}}\<[E]{}\\
\>[4]{}{\mathnormal{=}\mskip 3.0mu\mathsf{forkIO}\mskip 3.0mu\allowbreak{}\mathnormal{(}\mskip 0.0mu\mathsf{takeMVar}\mskip 3.0mu\mathsf{a}\mskip 3.0mu\allowbreak{}\mathnormal{>\mkern-3mu >\mkern-2mu =}\allowbreak{}\mskip 3.0mu\mathsf{putMVar}\mskip 3.0mu\mathsf{b}\mskip 0.0mu\mathnormal{)}\allowbreak{}\mskip 3.0mu\allowbreak{}\mathnormal{>\mkern-3mu >}\allowbreak{}\mskip 3.0mu\mathsf{return}\mskip 3.0mu\allowbreak{}\mathnormal{(}\mskip 0.0mu\mathnormal{)}\allowbreak{}}\<[E]{}\end{parray}}\end{list} 
The product \ensuremath{\allowbreak{}\mathnormal{(}\mskip 0.0mu}\ensuremath{\mathsf{f}\mskip 3.0mu}\ensuremath{\allowbreak{}\mathnormal{×}\allowbreak{}\mskip 3.0mu}\ensuremath{\mathsf{g}\mskip 0.0mu}\ensuremath{\mathnormal{)}\allowbreak{}} is implemented by
running \ensuremath{\mathsf{f}} and \ensuremath{\mathsf{g}} in separate threads (forking one extra
thread).  As described above, the composition \ensuremath{\allowbreak{}\mathnormal{(}\mskip 0.0mu}\ensuremath{\mathsf{f}\mskip 3.0mu}\ensuremath{\allowbreak{}\mathnormal{∘}\allowbreak{}\mskip 3.0mu}\ensuremath{\mathsf{g}\mskip 0.0mu}\ensuremath{\mathnormal{)}\allowbreak{}} runs \ensuremath{\mathsf{f}} and \ensuremath{\mathsf{g}} in parallel, with a new synchronisation point in-between.
This means that if \ensuremath{\mathsf{f}} and \ensuremath{\mathsf{g}} are run as subtasks
with fine-grained dependencies: no unnecessary synchronisation happens.  The \ensuremath{σ} morphism is
implemented by forwarding data as appropriate. The other ones (\ensuremath{α} , \ensuremath{ρ}) follow the
same pattern and are omitted for concision.

\begin{list}{}{\setlength\leftmargin{1.0em}}\item\relax
\ensuremath{\begin{parray}\column{B}{@{}>{}l<{}@{}}\column[0em]{1}{@{}>{}l<{}@{}}\column[1em]{2}{@{}>{}l<{}@{}}\column[2em]{3}{@{}>{}l<{}@{}}\column{4}{@{}>{}l<{}@{}}\column{5}{@{}>{}l<{}@{}}\column{E}{@{}>{}l<{}@{}}%
\>[1]{}{\mathbf{instance}\mskip 3.0mu\mathsf{Category}\mskip 3.0mu\mathsf{Workflow}\mskip 3.0mu\mathbf{where}}\<[E]{}\\
\>[2]{}{\mathbf{type}\mskip 3.0mu\mathsf{Obj}\mskip 3.0mu\mathsf{Workflow}\mskip 3.0mu\mathnormal{=}\mskip 3.0mu\mathsf{HasPoint}}\<[E]{}\\
\>[2]{}{\mathsf{id}\mskip 3.0mu\mathnormal{=}\mskip 3.0mu\mathsf{W}\mskip 3.0mu\mathsf{connect}}\<[E]{}\\
\>[2]{}{\mathsf{W}\mskip 3.0mu\mathsf{f}\mskip 3.0mu\allowbreak{}\mathnormal{∘}\allowbreak{}\mskip 3.0mu\mathsf{W}\mskip 3.0mu\mathsf{g}\mskip 3.0mu\mathnormal{=}\mskip 3.0mu\mathsf{W}\mskip 3.0mu\mathnormal{\$}\mskip 3.0muλ\mskip 3.0mu\mathsf{a}\mskip 3.0mu\mathsf{c}\mskip 3.0mu\mathnormal{\rightarrow }\mskip 3.0mu\mathbf{do}}\<[E]{}\\
\>[3]{}{\mathsf{b}\mskip 3.0mu\mathnormal{\leftarrow }\mskip 3.0mu\mathsf{mkPoint}}\<[E]{}\\
\>[3]{}{\mathsf{forkIO}\mskip 3.0mu\allowbreak{}\mathnormal{(}\mskip 0.0mu\mathsf{g}\mskip 3.0mu\mathsf{a}\mskip 3.0mu\mathsf{b}\mskip 0.0mu\mathnormal{)}\allowbreak{}\mskip 3.0mu\allowbreak{}\mathnormal{>\mkern-3mu >}\allowbreak{}\mskip 3.0mu\mathsf{f}\mskip 3.0mu\mathsf{b}\mskip 3.0mu\mathsf{c}}\<[E]{}\\
\>[1]{}{\mathbf{instance}\mskip 3.0mu\mathsf{Monoidal}\mskip 3.0mu\mathsf{Workflow}\mskip 3.0mu\mathbf{where}}\<[E]{}\\
\>[2]{}{\mathsf{W}\mskip 3.0mu\mathsf{f}\mskip 3.0mu\allowbreak{}\mathnormal{×}\allowbreak{}\mskip 3.0mu\mathsf{W}\mskip 3.0mu\mathsf{g}\mskip 3.0mu\mathnormal{=}\mskip 3.0mu\mathsf{W}\mskip 3.0mu\mathnormal{\$}\mskip 3.0muλ\mskip 3.0mu\allowbreak{}\mathnormal{(}\mskip 0.0mu\mathsf{a}\mskip 3.0mu\allowbreak{}\mathnormal{:\kern -3pt *}\allowbreak{}\mskip 3.0mu\mathsf{b}\mskip 0.0mu\mathnormal{)}\allowbreak{}\mskip 3.0mu\allowbreak{}\mathnormal{(}\mskip 0.0mu\mathsf{c}\mskip 3.0mu\allowbreak{}\mathnormal{:\kern -3pt *}\allowbreak{}\mskip 3.0mu\mathsf{d}\mskip 0.0mu\mathnormal{)}\allowbreak{}\mskip 3.0mu\mathnormal{\rightarrow }\mskip 3.0mu\mathbf{do}}\<[E]{}\\
\>[3]{}{\mathnormal{\_}\mskip 3.0mu\mathnormal{\leftarrow }\mskip 3.0mu\mathsf{forkIO}\mskip 3.0mu\allowbreak{}\mathnormal{(}\mskip 0.0mu\mathsf{f}\mskip 3.0mu\mathsf{a}\mskip 3.0mu\mathsf{c}\mskip 0.0mu\mathnormal{)}\allowbreak{}}\<[E]{}\\
\>[3]{}{\mathsf{g}\mskip 3.0mu\mathsf{b}\mskip 3.0mu\mathsf{d}}\<[E]{}\\
\>[2]{}{σ\mskip 3.0mu\mathnormal{=}\mskip 3.0mu\mathsf{W}\mskip 3.0mu\mathnormal{\$}\mskip 3.0mu}\>[5]{}{λ\mskip 3.0mu\allowbreak{}\mathnormal{(}\mskip 0.0mu\allowbreak{}\mathnormal{(}\mskip 0.0mu\mathsf{a}\mskip 3.0mu\allowbreak{}\mathnormal{:\kern -3pt *}\allowbreak{}\mskip 3.0mu\mathsf{b}\mskip 0.0mu\mathnormal{)}\allowbreak{}\mskip 0.0mu\mathnormal{)}\allowbreak{}\mskip 3.0mu\allowbreak{}\mathnormal{(}\mskip 0.0mu\allowbreak{}\mathnormal{(}\mskip 0.0mu\mathsf{c}\mskip 3.0mu\allowbreak{}\mathnormal{:\kern -3pt *}\allowbreak{}\mskip 3.0mu\mathsf{d}\mskip 0.0mu\mathnormal{)}\allowbreak{}\mskip 0.0mu\mathnormal{)}\allowbreak{}\mskip 3.0mu\mathnormal{\rightarrow }\mskip 3.0mu\mathbf{do}}\<[E]{}\\
\>[3]{}{\mathsf{connect}\mskip 3.0mu\mathsf{a}\mskip 3.0mu\mathsf{d}}\<[E]{}\\
\>[3]{}{\mathsf{connect}\mskip 3.0mu\mathsf{b}\mskip 3.0mu\mathsf{c}}\<[E]{}\\
\>[2]{}{\bar{α}\mskip 3.0mu\mathnormal{=}\mskip 3.0mu\mathsf{W}\mskip 3.0mu\mathnormal{\$}\mskip 3.0muλ\mskip 3.0mu\allowbreak{}\mathnormal{(}\mskip 0.0mu\allowbreak{}\mathnormal{(}\mskip 0.0mu\mathsf{a}\mskip 3.0mu\allowbreak{}\mathnormal{:\kern -3pt *}\allowbreak{}\mskip 3.0mu\allowbreak{}\mathnormal{(}\mskip 0.0mu\allowbreak{}\mathnormal{(}\mskip 0.0mu\mathsf{b}\mskip 3.0mu\allowbreak{}\mathnormal{:\kern -3pt *}\allowbreak{}\mskip 3.0mu\mathsf{c}\mskip 0.0mu\mathnormal{)}\allowbreak{}\mskip 0.0mu\mathnormal{)}\allowbreak{}\mskip 0.0mu\mathnormal{)}\allowbreak{}\mskip 0.0mu\mathnormal{)}\allowbreak{}\mskip 3.0mu\allowbreak{}\mathnormal{(}\mskip 0.0mu\allowbreak{}\mathnormal{(}\mskip 0.0mu\allowbreak{}\mathnormal{(}\mskip 0.0mu\allowbreak{}\mathnormal{(}\mskip 0.0mu\mathsf{d}\mskip 3.0mu\allowbreak{}\mathnormal{:\kern -3pt *}\allowbreak{}\mskip 3.0mu\mathsf{e}\mskip 0.0mu\mathnormal{)}\allowbreak{}\mskip 0.0mu\mathnormal{)}\allowbreak{}\mskip 3.0mu\allowbreak{}\mathnormal{:\kern -3pt *}\allowbreak{}\mskip 3.0mu\mathsf{f}\mskip 0.0mu\mathnormal{)}\allowbreak{}\mskip 0.0mu\mathnormal{)}\allowbreak{}\mskip 3.0mu\mathnormal{\rightarrow }\mskip 3.0mu\mathbf{do}}\<[E]{}\\
\>[3]{}{\mathsf{connect}\mskip 3.0mu\mathsf{a}\mskip 3.0mu\mathsf{d}}\<[E]{}\\
\>[3]{}{\mathsf{connect}\mskip 3.0mu\mathsf{b}\mskip 3.0mu\mathsf{e}}\<[E]{}\\
\>[3]{}{\mathsf{connect}\mskip 3.0mu\mathsf{c}\mskip 3.0mu\mathsf{f}}\<[E]{}\\
\>[2]{}{α\mskip 3.0mu\mathnormal{=}\mskip 3.0mu\mathsf{W}\mskip 3.0mu\mathnormal{\$}\mskip 3.0muλ\mskip 3.0mu\allowbreak{}\mathnormal{(}\mskip 0.0mu\allowbreak{}\mathnormal{(}\mskip 0.0mu\allowbreak{}\mathnormal{(}\mskip 0.0mu\mathsf{a}\mskip 3.0mu\allowbreak{}\mathnormal{:\kern -3pt *}\allowbreak{}\mskip 3.0mu\mathsf{b}\mskip 0.0mu\mathnormal{)}\allowbreak{}\mskip 3.0mu\allowbreak{}\mathnormal{:\kern -3pt *}\allowbreak{}\mskip 3.0mu\mathsf{c}\mskip 0.0mu\mathnormal{)}\allowbreak{}\mskip 0.0mu\mathnormal{)}\allowbreak{}\mskip 3.0mu\allowbreak{}\mathnormal{(}\mskip 0.0mu\allowbreak{}\mathnormal{(}\mskip 0.0mu\mathsf{d}\mskip 3.0mu\allowbreak{}\mathnormal{:\kern -3pt *}\allowbreak{}\mskip 3.0mu\allowbreak{}\mathnormal{(}\mskip 0.0mu\mathsf{e}\mskip 3.0mu\allowbreak{}\mathnormal{:\kern -3pt *}\allowbreak{}\mskip 3.0mu\mathsf{f}\mskip 0.0mu\mathnormal{)}\allowbreak{}\mskip 0.0mu\mathnormal{)}\allowbreak{}\mskip 0.0mu\mathnormal{)}\allowbreak{}\mskip 3.0mu\mathnormal{\rightarrow }\mskip 3.0mu\mathbf{do}}\<[E]{}\\
\>[3]{}{\mathsf{connect}\mskip 3.0mu\mathsf{a}\mskip 3.0mu\mathsf{d}}\<[E]{}\\
\>[3]{}{\mathsf{connect}\mskip 3.0mu\mathsf{b}\mskip 3.0mu\mathsf{e}}\<[E]{}\\
\>[3]{}{\mathsf{connect}\mskip 3.0mu\mathsf{c}\mskip 3.0mu\mathsf{f}}\<[E]{}\\
\>[2]{}{ρ\mskip 3.0mu}\>[4]{}{\mathnormal{=}\mskip 3.0mu\mathsf{W}\mskip 3.0mu\mathnormal{\$}\mskip 3.0muλ\mskip 3.0mu\mathsf{a}\mskip 3.0mu\allowbreak{}\mathnormal{(}\mskip 0.0mu\allowbreak{}\mathnormal{(}\mskip 0.0mu\mathsf{a'}\mskip 3.0mu\allowbreak{}\mathnormal{:\kern -3pt *}\allowbreak{}\mskip 3.0mu\mathnormal{\_}\mskip 0.0mu\mathnormal{)}\allowbreak{}\mskip 0.0mu\mathnormal{)}\allowbreak{}\mskip 3.0mu\mathnormal{\rightarrow }\mskip 3.0mu\mathsf{connect}\mskip 3.0mu\mathsf{a}\mskip 3.0mu\mathsf{a'}}\<[E]{}\\
\>[2]{}{\bar{ρ}\mskip 3.0mu\mathnormal{=}\mskip 3.0mu\mathsf{W}\mskip 3.0mu\mathnormal{\$}\mskip 3.0muλ\mskip 3.0mu\allowbreak{}\mathnormal{(}\mskip 0.0mu\allowbreak{}\mathnormal{(}\mskip 0.0mu\mathsf{a}\mskip 3.0mu\allowbreak{}\mathnormal{:\kern -3pt *}\allowbreak{}\mskip 3.0mu\mathnormal{\_}\mskip 0.0mu\mathnormal{)}\allowbreak{}\mskip 0.0mu\mathnormal{)}\allowbreak{}\mskip 3.0mu\mathsf{a'}\mskip 3.0mu\mathnormal{\rightarrow }\mskip 3.0mu\mathsf{connect}\mskip 3.0mu\mathsf{a}\mskip 3.0mu\mathsf{a'}}\<[E]{}\end{parray}}\end{list} 
The above implementation is only a prototype for illustrative
purposes. For instance, it causes a synchronisation point to happen
even between every two connected atomic tasks. This excessive
synchronisation can induce significant overheads in some
situations. If this is a concern, one can
perform an analysis of the computation graph (say, by first
reifying it as a data type) and eliminate unneeded synchronisation
points. Additionally, applications may need some mechanism to deal with
errors or dead tasks.

A more fundamental limitation of the prototype resides in
synchronisation being of the simplest kind: connection between ports
is realised by simply forwarding data— always in the same
direction. Thus, another extension to the above prototype would be to
support more complex protocols (corresponding to other base types than
\ensuremath{\mathsf{K}\mskip 3.0mu}\ensuremath{\mathsf{a}}). For example, sequential data can be streamed, one element
at a time. Query-reply protocols are also a possibility. In this
light, we can now examine the question of whether tasks form a
cartesian category (in addition to symmetric monoidal). Because the \ensuremath{δ} morphism corresponds to multiplexing, \ensuremath{\mathsf{Workflow}}s can be enthused
with a cartesian structure only if all base protocols are
multiplexable in their input. (The condition that the unit type is the
unit for multiplexing would normally be satisfied as well.)
\end{mdframed} 
\section{Implementation}\label{60} 
In this section we reveal the implementation of our abstract type for
the {\sc{}api} from \cref{33}.  Unfortunately it is not just a
matter of writing down the specification and calculating an
implementation: some amount of creativity is required. \emph{The key idea is
to represent ports as morphisms from the source (\ensuremath{\mathsf{r}}) to the
object of interest.} In terms of diagrams, they represent the portion of the diagram
which connect the source (on the left) to the port.

\begin{figure} \begin{list}{}{\setlength\leftmargin{1.0em}}\item\relax
\ensuremath{\begin{parray}\column{B}{@{}>{}l<{}@{}}\column[0em]{1}{@{}>{}l<{}@{}}\column[1em]{2}{@{}>{}l<{}@{}}\column{3}{@{}>{}l<{}@{}}\column{4}{@{}>{}l<{}@{}}\column{5}{@{}>{}l<{}@{}}\column{6}{@{}>{}l<{}@{}}\column{7}{@{}>{}l<{}@{}}\column{8}{@{}>{}l<{}@{}}\column{9}{@{}>{}l<{}@{}}\column{10}{@{}>{}l<{}@{}}\column{11}{@{}>{}l<{}@{}}\column{12}{@{}>{}l<{}@{}}\column{13}{@{}>{}l<{}@{}}\column{E}{@{}>{}l<{}@{}}%
\>[1]{}{\mathbf{data}\mskip 3.0mu\mathsf{FreeCartesian}\mskip 3.0mu\mathsf{k}\mskip 3.0mu}\>[8]{}{\mathsf{a}\mskip 3.0mu\mathsf{b}\mskip 3.0mu\mathbf{where}}\<[E]{}\\
\>[2]{}{\mathsf{I}\mskip 3.0mu}\>[3]{}{\mathnormal{::}\mskip 3.0mu\mathsf{FreeCartesian}\mskip 3.0mu\mathsf{k}\mskip 3.0mu}\>[6]{}{\mathsf{a}\mskip 3.0mu\mathsf{a}}\<[E]{}\\
\>[2]{}{\allowbreak{}\mathnormal{(}\mskip 0.0mu\allowbreak{}\mathnormal{:\kern -3pt ∘ \kern -3pt:}\allowbreak{}\mskip 0.0mu\mathnormal{)}\allowbreak{}\mskip 3.0mu}\>[3]{}{\mathnormal{::}\mskip 3.0mu}\>[4]{}{\mathsf{FreeCartesian}\mskip 3.0mu\mathsf{k}\mskip 3.0mu}\>[9]{}{\mathsf{b}\mskip 3.0mu\mathsf{c}\mskip 3.0mu\mathnormal{\rightarrow }\mskip 3.0mu\mathsf{FreeCartesian}\mskip 3.0mu\mathsf{k}\mskip 3.0mu}\>[12]{}{\mathsf{a}\mskip 3.0mu\mathsf{b}}\<[E]{}\\
\>[3]{}{\mathnormal{\rightarrow }\mskip 3.0mu\mathsf{FreeCartesian}\mskip 3.0mu\mathsf{k}\mskip 3.0mu}\>[6]{}{\mathsf{a}\mskip 3.0mu\mathsf{c}}\<[E]{}\\
\>[2]{}{\mathsf{Embed}\mskip 3.0mu}\>[3]{}{\mathnormal{::}\mskip 3.0mu}\>[5]{}{\mathsf{k}\mskip 3.0mu\mathsf{a}\mskip 3.0mu\mathsf{b}\mskip 3.0mu\mathnormal{\rightarrow }\mskip 3.0mu\mathsf{FreeCartesian}\mskip 3.0mu\mathsf{k}\mskip 3.0mu}\>[11]{}{\mathsf{a}\mskip 3.0mu\mathsf{b}}\<[E]{}\\
\>[2]{}{\allowbreak{}\mathnormal{(}\mskip 0.0mu\mathnormal{:\kern -3pt ▵ \kern -3pt:}\mskip 0.0mu\mathnormal{)}\allowbreak{}\mskip 3.0mu}\>[3]{}{\mathnormal{::}\mskip 3.0mu}\>[7]{}{\mathsf{FreeCartesian}\mskip 3.0mu\mathsf{k}\mskip 3.0mu}\>[10]{}{\mathsf{a}\mskip 3.0mu\mathsf{b}\mskip 3.0mu\mathnormal{\rightarrow }\mskip 3.0mu\mathsf{FreeCartesian}\mskip 3.0mu\mathsf{k}\mskip 3.0mu}\>[13]{}{\mathsf{a}\mskip 3.0mu\mathsf{c}}\<[E]{}\\
\>[3]{}{\mathnormal{\rightarrow }\mskip 3.0mu\mathsf{FreeCartesian}\mskip 3.0mu\mathsf{k}\mskip 3.0mu}\>[6]{}{\mathsf{a}\mskip 3.0mu\allowbreak{}\mathnormal{(}\mskip 0.0mu\mathsf{b}\mskip 3.0mu\mathnormal{⊗}\mskip 3.0mu\mathsf{c}\mskip 0.0mu\mathnormal{)}\allowbreak{}}\<[E]{}\\
\>[2]{}{\mathsf{P}_{1}\mskip 3.0mu}\>[3]{}{\mathnormal{::}\mskip 3.0mu}\>[4]{}{\mathsf{FreeCartesian}\mskip 3.0mu\mathsf{k}\mskip 3.0mu}\>[9]{}{\allowbreak{}\mathnormal{(}\mskip 0.0mu\mathsf{a}\mskip 3.0mu\mathnormal{⊗}\mskip 3.0mu\mathsf{b}\mskip 0.0mu\mathnormal{)}\allowbreak{}\mskip 3.0mu\mathsf{a}}\<[E]{}\\
\>[2]{}{\mathsf{P}_{2}\mskip 3.0mu}\>[3]{}{\mathnormal{::}\mskip 3.0mu}\>[4]{}{\mathsf{FreeCartesian}\mskip 3.0mu\mathsf{k}\mskip 3.0mu}\>[9]{}{\allowbreak{}\mathnormal{(}\mskip 0.0mu\mathsf{a}\mskip 3.0mu\mathnormal{⊗}\mskip 3.0mu\mathsf{b}\mskip 0.0mu\mathnormal{)}\allowbreak{}\mskip 3.0mu\mathsf{b}}\<[E]{}\end{parray}}\end{list} \begin{list}{}{\setlength\leftmargin{1.0em}}\item\relax
\ensuremath{\begin{parray}\column{B}{@{}>{}l<{}@{}}\column[0em]{1}{@{}>{}l<{}@{}}\column{2}{@{}>{}l<{}@{}}\column{3}{@{}>{}l<{}@{}}\column{4}{@{}>{}l<{}@{}}\column{E}{@{}>{}l<{}@{}}%
\>[1]{}{\mathbf{instance}\mskip 3.0mu\allowbreak{}\mathnormal{(}\mskip 0.0mu}\>[2]{}{\mathsf{Monoidal}\mskip 3.0mu\mathsf{k}\mskip 0.0mu\mathnormal{)}\allowbreak{}\mskip 3.0mu\mathnormal{\Rightarrow }\mskip 3.0mu}\>[3]{}{\mathsf{Monoidal}\mskip 3.0mu\allowbreak{}\mathnormal{(}\mskip 0.0mu\mathsf{FreeCartesian}\mskip 3.0mu\mathsf{k}\mskip 0.0mu}\>[4]{}{\mathnormal{)}\allowbreak{}}\<[E]{}\\
\>[1]{}{\mathbf{instance}\mskip 3.0mu\allowbreak{}\mathnormal{(}\mskip 0.0mu}\>[2]{}{\mathsf{Monoidal}\mskip 3.0mu\mathsf{k}\mskip 0.0mu\mathnormal{)}\allowbreak{}\mskip 3.0mu\mathnormal{\Rightarrow }\mskip 3.0mu\mathsf{Cartesian}\mskip 3.0mu\allowbreak{}\mathnormal{(}\mskip 0.0mu\mathsf{FreeCartesian}\mskip 3.0mu\mathsf{k}\mskip 0.0mu}\>[4]{}{\mathnormal{)}\allowbreak{}}\<[E]{}\end{parray}}\end{list} \caption{Definition of the free cartesian category
over an underlying category \ensuremath{\mathsf{k}}, whose morphisms it \ensuremath{\mathsf{Embed}}s.
\ensuremath{\mathsf{P}_{1}} and \ensuremath{\mathsf{P}_{2}} implement respectively \ensuremath{π₁} and \ensuremath{π₂}, while (:▵:) implements (▵). }\label{61}\end{figure} 
Such morphisms may therefore discard part of the input. This means that
they are not morphisms of the {\sc{}smc} k, but rather morphisms of the free cartesian
category over \ensuremath{\mathsf{k}} (\ensuremath{\mathsf{FreeCartesian}\mskip 3.0mu}\ensuremath{\mathsf{k}\mskip 3.0mu}\ensuremath{\mathsf{r}\mskip 3.0mu}\ensuremath{\mathsf{a}}, see
\cref{61}):

\begin{list}{}{\setlength\leftmargin{1.0em}}\item\relax
\ensuremath{\begin{parray}\column{B}{@{}>{}l<{}@{}}\column[0em]{1}{@{}>{}l<{}@{}}\column[1em]{2}{@{}>{}l<{}@{}}\column{3}{@{}>{}l<{}@{}}\column{4}{@{}>{}l<{}@{}}\column{5}{@{}>{}l<{}@{}}\column{E}{@{}>{}l<{}@{}}%
\>[1]{}{\mathbf{data}\mskip 3.0mu\mathsf{P}\mskip 3.0mu\mathsf{k}\mskip 3.0mu\mathsf{r}\mskip 3.0mu\mathsf{a}\mskip 3.0mu}\>[3]{}{\mathbf{where}}\<[E]{}\\
\>[2]{}{\mathsf{P}\mskip 3.0mu\mathnormal{::}\mskip 3.0mu\mathsf{FreeCartesian}\mskip 3.0mu\mathsf{k}\mskip 3.0mu}\>[4]{}{\mathsf{r}\mskip 3.0mu\mathsf{a}\mskip 3.0mu\mathnormal{\rightarrow }\mskip 3.0mu\mathsf{P}\mskip 3.0mu\mathsf{k}\mskip 3.0mu\mathsf{r}\mskip 3.0mu\mathsf{a}}\<[E]{}\\
\>[1]{}{\mathsf{fromP}\mskip 3.0mu\mathnormal{::}\mskip 3.0mu\mathsf{P}\mskip 3.0mu\mathsf{k}\mskip 3.0mu\mathsf{r}\mskip 3.0mu\mathsf{a}\mskip 3.0mu\mathnormal{\rightarrow }\mskip 3.0mu\mathsf{FreeCartesian}\mskip 3.0mu\mathsf{k}\mskip 3.0mu}\>[5]{}{\mathsf{r}\mskip 3.0mu\mathsf{a}}\<[E]{}\\
\>[1]{}{\mathsf{fromP}\mskip 3.0mu\allowbreak{}\mathnormal{(}\mskip 0.0mu\mathsf{P}\mskip 3.0mu\mathsf{f}\mskip 0.0mu\mathnormal{)}\allowbreak{}\mskip 3.0mu\mathnormal{=}\mskip 3.0mu\mathsf{f}}\<[E]{}\end{parray}}\end{list} 
This free category is implemented as a data type with a constructor for each method in the
\ensuremath{\mathsf{Cartesian}} class, plus a constructor to \ensuremath{\mathsf{Embed}} generators from \ensuremath{\mathsf{k}}.  A subtlety is
that, even though \ensuremath{\mathsf{P}\mskip 3.0mu}\ensuremath{\mathsf{k}\mskip 3.0mu}\ensuremath{\mathsf{r}\mskip 3.0mu}\ensuremath{\mathsf{a}} is used linearly everywhere in the
interface, the free cartesian representation that it embeds can be
duplicated at will.  In Linear Haskell this is subtly noted by using
using the \ensuremath{\mathnormal{\rightarrow }} arrow instead of \ensuremath{\mathnormal{⊸}} in the declaration of
\ensuremath{\mathsf{P}} constructor.  Consequently, when doing \emph{encode φ}, the morphism φ must be
available \emph{unrestricted}, not just linearly. This is not a problem in practice: even if
data cannot be duplicated, closed functions which manipulate such data
can be.

With these technical bits out of the way, let us return to the main
representational idea: a port for the object \emph{a} is a free
cartesian morphism from \ensuremath{\mathsf{r}} to \ensuremath{\mathsf{a}}. Accordingly, the
equality on ports is the usual equality of free cartesian
categories, but quotiented by equations arising from \ensuremath{\mathsf{Embed}} being an
{\sc{}smc} homomorphism:

\begin{list}{}{\setlength\leftmargin{1.0em}}\item\relax
\ensuremath{\begin{parray}\column{B}{@{}>{}l<{}@{}}\column[0em]{1}{@{}>{}l<{}@{}}\column{2}{@{}>{}l<{}@{}}\column{3}{@{}>{}l<{}@{}}\column{E}{@{}>{}l<{}@{}}%
\>[1]{}{\mathsf{Embed}\mskip 3.0mu\mathsf{id}\mskip 3.0mu}\>[2]{}{\mathnormal{=}\mskip 3.0mu}\>[3]{}{\mathsf{id}}\<[E]{}\\
\>[1]{}{\mathsf{Embed}\mskip 3.0mu\allowbreak{}\mathnormal{(}\mskip 0.0mu\mathsf{φ}\mskip 3.0mu\allowbreak{}\mathnormal{∘}\allowbreak{}\mskip 3.0mu\mathsf{ψ}\mskip 0.0mu\mathnormal{)}\allowbreak{}\mskip 3.0mu}\>[2]{}{\mathnormal{=}\mskip 3.0mu}\>[3]{}{\mathsf{Embed}\mskip 3.0mu\mathsf{φ}\mskip 3.0mu\allowbreak{}\mathnormal{∘}\allowbreak{}\mskip 3.0mu\mathsf{Embed}\mskip 3.0mu\mathsf{ψ}}\<[E]{}\\
\>[1]{}{\mathsf{Embed}\mskip 3.0mu\allowbreak{}\mathnormal{(}\mskip 0.0mu\mathsf{φ}\mskip 3.0mu\allowbreak{}\mathnormal{×}\allowbreak{}\mskip 3.0mu\mathsf{ψ}\mskip 0.0mu\mathnormal{)}\allowbreak{}\mskip 3.0mu}\>[2]{}{\mathnormal{=}\mskip 3.0mu}\>[3]{}{\mathsf{Embed}\mskip 3.0mu\mathsf{φ}\mskip 3.0mu\allowbreak{}\mathnormal{×}\allowbreak{}\mskip 3.0mu\mathsf{Embed}\mskip 3.0mu\mathsf{ψ}}\<[E]{}\end{parray}}\end{list} etc.

Because \ensuremath{\mathsf{P}\mskip 3.0mu}\ensuremath{\mathsf{k}\mskip 3.0mu}\ensuremath{\mathsf{r}\mskip 3.0mu}\ensuremath{\mathsf{a}} is a morphism from \ensuremath{\mathsf{r}} to \ensuremath{\mathsf{a}}, the
encoding from \ensuremath{\mathsf{a}\mskip 0.0mu}\ensuremath{\overset{\mathsf{k}}{\leadsto}\mskip 0.0mu}\ensuremath{\mathsf{b}} to \ensuremath{\mathsf{P}\mskip 3.0mu}\ensuremath{\mathsf{k}\mskip 3.0mu}\ensuremath{\mathsf{r}\mskip 3.0mu}\ensuremath{\mathsf{a}\mskip 3.0mu}\ensuremath{\mathnormal{⊸}\mskip 3.0mu}\ensuremath{\mathsf{P}\mskip 3.0mu}\ensuremath{\mathsf{k}\mskip 3.0mu}\ensuremath{\mathsf{r}\mskip 3.0mu}\ensuremath{\mathsf{b}} can be
thought of as a transformation to continuation-passing-style ({\sc{}cps}), albeit
reversed--- perhaps a ``prefix-passing-style''
transformation. For non-linear functions, the encoding
would be given by the Yoneda lemma \citep{boisseau_what_2018} composed with embedding in the free
cartesian category.  The implementation of
the combinators of the interface can then follow the usual (cartesian)
categorical semantics of product and unit types: 
\begin{list}{}{\setlength\leftmargin{1.0em}}\item\relax
\ensuremath{\begin{parray}\column{B}{@{}>{}l<{}@{}}\column[0em]{1}{@{}>{}l<{}@{}}\column{2}{@{}>{}l<{}@{}}\column{E}{@{}>{}l<{}@{}}%
\>[1]{}{\mathsf{encode}\mskip 3.0mu\mathsf{φ}\mskip 3.0mu\allowbreak{}\mathnormal{(}\mskip 0.0mu\mathsf{P}\mskip 3.0mu\mathsf{f}\mskip 0.0mu\mathnormal{)}\allowbreak{}\mskip 3.0mu}\>[2]{}{\mathnormal{=}\mskip 3.0mu\mathsf{P}\mskip 3.0mu\allowbreak{}\mathnormal{(}\mskip 0.0mu\mathsf{Embed}\mskip 3.0mu\mathsf{φ}\mskip 3.0mu\allowbreak{}\mathnormal{∘}\allowbreak{}\mskip 3.0mu\mathsf{f}\mskip 0.0mu\mathnormal{)}\allowbreak{}}\<[E]{}\\
\>[1]{}{\mathsf{unit}\mskip 3.0mu}\>[2]{}{\mathnormal{=}\mskip 3.0mu\mathsf{P}\mskip 3.0muε}\<[E]{}\\
\>[1]{}{\mathsf{split}\mskip 3.0mu\allowbreak{}\mathnormal{(}\mskip 0.0mu\mathsf{P}\mskip 3.0mu\mathsf{f}\mskip 0.0mu\mathnormal{)}\allowbreak{}\mskip 3.0mu}\>[2]{}{\mathnormal{=}\mskip 3.0mu\allowbreak{}\mathnormal{(}\mskip 0.0mu\mathsf{P}\mskip 3.0mu\allowbreak{}\mathnormal{(}\mskip 0.0muπ₁\mskip 3.0mu\allowbreak{}\mathnormal{∘}\allowbreak{}\mskip 3.0mu\mathsf{f}\mskip 0.0mu\mathnormal{)}\allowbreak{}\mskip 0.0mu\mathnormal{,}\mskip 3.0mu\mathsf{P}\mskip 3.0mu\allowbreak{}\mathnormal{(}\mskip 0.0muπ₂\mskip 3.0mu\allowbreak{}\mathnormal{∘}\allowbreak{}\mskip 3.0mu\mathsf{f}\mskip 0.0mu\mathnormal{)}\allowbreak{}\mskip 0.0mu\mathnormal{)}\allowbreak{}}\<[E]{}\\
\>[1]{}{\mathsf{merge}\mskip 3.0mu\allowbreak{}\mathnormal{(}\mskip 0.0mu\mathsf{P}\mskip 3.0mu\mathsf{f}\mskip 0.0mu\mathnormal{,}\mskip 3.0mu\mathsf{P}\mskip 3.0mu\mathsf{g}\mskip 0.0mu\mathnormal{)}\allowbreak{}\mskip 3.0mu}\>[2]{}{\mathnormal{=}\mskip 3.0mu\mathsf{P}\mskip 3.0mu\allowbreak{}\mathnormal{(}\mskip 0.0mu\mathsf{f}\mskip 3.0mu\mathnormal{▵}\mskip 3.0mu\mathsf{g}\mskip 0.0mu\mathnormal{)}\allowbreak{}}\<[E]{}\end{parray}}\end{list} 
The most challenging part of the implementation is \ensuremath{\mathsf{decode}},
which converts \emph{linear} functions between ports to morphisms in
\ensuremath{\mathsf{k}}.

\begin{list}{}{\setlength\leftmargin{1.0em}}\item\relax
\ensuremath{\begin{parray}\column{B}{@{}>{}l<{}@{}}\column[0em]{1}{@{}>{}l<{}@{}}\column{2}{@{}>{}l<{}@{}}\column{3}{@{}>{}l<{}@{}}\column{4}{@{}>{}l<{}@{}}\column{5}{@{}>{}l<{}@{}}\column{E}{@{}>{}l<{}@{}}%
\>[1]{}{\mathsf{decode}\mskip 3.0mu\mathsf{f}\mskip 3.0mu}\>[2]{}{\mathnormal{=}\mskip 3.0mu\mathsf{evalM}\mskip 3.0mu\allowbreak{}\mathnormal{(}\mskip 0.0mu\mathsf{reduce}\mskip 3.0mu\allowbreak{}\mathnormal{(}\mskip 0.0mu\mathsf{extract}\mskip 3.0mu\mathsf{f}\mskip 0.0mu\mathnormal{)}\allowbreak{}\mskip 0.0mu\mathnormal{)}\allowbreak{}}\<[E]{}\\
\>[1]{}{\mathsf{extract}\mskip 3.0mu}\>[2]{}{\mathnormal{::}\mskip 3.0mu}\>[3]{}{\allowbreak{}\mathnormal{(}\mskip 0.0mu∀\mskip 3.0mu\mathsf{r}\mskip 1.0mu.\mskip 3.0mu}\>[4]{}{\mathsf{P}\mskip 3.0mu\mathsf{k}\mskip 3.0mu\mathsf{r}\mskip 3.0mu\mathsf{a}\mskip 3.0mu\mathnormal{⊸}\mskip 3.0mu\mathsf{P}\mskip 3.0mu\mathsf{k}\mskip 3.0mu\mathsf{r}\mskip 3.0mu\mathsf{b}\mskip 0.0mu\mathnormal{)}\allowbreak{}\mskip 3.0mu\mathnormal{\rightarrow }\mskip 3.0mu\mathsf{FreeCartesian}\mskip 3.0mu\mathsf{k}\mskip 3.0mu}\>[5]{}{\mathsf{a}\mskip 3.0mu\mathsf{b}}\<[E]{}\\
\>[1]{}{\mathsf{extract}\mskip 3.0mu\mathsf{f}\mskip 3.0mu}\>[2]{}{\mathnormal{=}\mskip 3.0mu\mathsf{fromP}\mskip 3.0mu\allowbreak{}\mathnormal{(}\mskip 0.0mu\mathsf{f}\mskip 3.0mu\allowbreak{}\mathnormal{(}\mskip 0.0mu\mathsf{P}\mskip 3.0mu\mathsf{id}\mskip 0.0mu\mathnormal{)}\allowbreak{}\mskip 0.0mu\mathnormal{)}\allowbreak{}}\<[E]{}\end{parray}}\end{list} 
As usual in {\sc{}cps}, the first step is to complete the
computation by passing the identity morphism (\ensuremath{\mathsf{extract}}).  Then
the obtained \ensuremath{\mathsf{FreeCartesian}\mskip 3.0mu}\ensuremath{\mathsf{k}} morphism is projected to
the {\sc{}smc} k, which it carries. The next step is \ensuremath{\mathsf{reduce}}, which
projects the free cartesian representation to a free
{\sc{}smc} representation, referred hereafter as \ensuremath{\mathsf{FreeSMC}}. This is the most difficult
operation, and we return to it shortly.
The \ensuremath{\mathsf{evalM}} part maps a morphism of \ensuremath{\mathsf{FreeSMC}\mskip 3.0mu}\ensuremath{\mathsf{k}} back to a
morphism in \ensuremath{\mathsf{k}} ─it is the natural inductive definition on the
structure of free-{\sc{}smc} morphisms.

\begin{mdframed}[linewidth=0pt,hidealllines,innerleftmargin=0pt,innerrightmargin=0pt,backgroundcolor=gray!15] 
\begin{list}{}{\setlength\leftmargin{1.0em}}\item\relax
\ensuremath{\begin{parray}\column{B}{@{}>{}l<{}@{}}\column[0em]{1}{@{}>{}l<{}@{}}\column[1em]{2}{@{}>{}l<{}@{}}\column{3}{@{}>{}l<{}@{}}\column{4}{@{}>{}l<{}@{}}\column{E}{@{}>{}l<{}@{}}%
\>[1]{}{\mathbf{data}\mskip 3.0mu\mathsf{FreeSMC}\mskip 3.0mu\mathsf{k}\mskip 3.0mu\mathsf{a}\mskip 3.0mu\mathsf{b}\mskip 3.0mu\mathbf{where}}\<[E]{}\\
\>[2]{}{\mathsf{I}\mskip 3.0mu}\>[3]{}{\mathnormal{::}\mskip 3.0mu\mathsf{FreeSMC}\mskip 3.0mu\mathsf{k}\mskip 3.0mu}\>[4]{}{\mathsf{a}\mskip 3.0mu\mathsf{a}}\<[E]{}\\
\>[2]{}{\mathsf{Embed}\mskip 3.0mu}\>[3]{}{\mathnormal{::}\mskip 3.0mu\mathsf{k}\mskip 3.0mu\mathsf{a}\mskip 3.0mu\mathsf{b}\mskip 3.0mu\mathnormal{\rightarrow }\mskip 3.0mu\mathsf{FreeSMC}\mskip 3.0mu\mathsf{k}\mskip 3.0mu\mathsf{a}\mskip 3.0mu\mathsf{b}}\<[E]{}\\
\>[2]{}{\mathsf{A}\mskip 3.0mu}\>[3]{}{\mathnormal{::}\mskip 3.0mu\mathsf{FreeSMC}\mskip 3.0mu\mathsf{k}\mskip 3.0mu\allowbreak{}\mathnormal{(}\mskip 0.0mu\allowbreak{}\mathnormal{(}\mskip 0.0mu\mathsf{a}\mskip 3.0mu\mathnormal{⊗}\mskip 3.0mu\mathsf{b}\mskip 0.0mu\mathnormal{)}\allowbreak{}\mskip 3.0mu\mathnormal{⊗}\mskip 3.0mu\mathsf{c}\mskip 0.0mu\mathnormal{)}\allowbreak{}\mskip 3.0mu\allowbreak{}\mathnormal{(}\mskip 0.0mu\mathsf{a}\mskip 3.0mu\mathnormal{⊗}\mskip 3.0mu\allowbreak{}\mathnormal{(}\mskip 0.0mu\mathsf{b}\mskip 3.0mu\mathnormal{⊗}\mskip 3.0mu\mathsf{c}\mskip 0.0mu\mathnormal{)}\allowbreak{}\mskip 0.0mu\mathnormal{)}\allowbreak{}}\<[E]{}\\
\>[2]{}{\mathsf{A'}\mskip 3.0mu}\>[3]{}{\mathnormal{::}\mskip 3.0mu\mathsf{FreeSMC}\mskip 3.0mu\mathsf{k}\mskip 3.0mu\allowbreak{}\mathnormal{(}\mskip 0.0mu\mathsf{a}\mskip 3.0mu\mathnormal{⊗}\mskip 3.0mu\allowbreak{}\mathnormal{(}\mskip 0.0mu\mathsf{b}\mskip 3.0mu\mathnormal{⊗}\mskip 3.0mu\mathsf{c}\mskip 0.0mu\mathnormal{)}\allowbreak{}\mskip 0.0mu\mathnormal{)}\allowbreak{}\mskip 3.0mu\allowbreak{}\mathnormal{(}\mskip 0.0mu\allowbreak{}\mathnormal{(}\mskip 0.0mu\mathsf{a}\mskip 3.0mu\mathnormal{⊗}\mskip 3.0mu\mathsf{b}\mskip 0.0mu\mathnormal{)}\allowbreak{}\mskip 3.0mu\mathnormal{⊗}\mskip 3.0mu\mathsf{c}\mskip 0.0mu\mathnormal{)}\allowbreak{}}\<[E]{}\\
\>[2]{}{\mathsf{S}\mskip 3.0mu}\>[3]{}{\mathnormal{::}\mskip 3.0mu\mathsf{FreeSMC}\mskip 3.0mu\mathsf{k}\mskip 3.0mu\allowbreak{}\mathnormal{(}\mskip 0.0mu\mathsf{a}\mskip 3.0mu\mathnormal{⊗}\mskip 3.0mu\mathsf{b}\mskip 0.0mu\mathnormal{)}\allowbreak{}\mskip 3.0mu\allowbreak{}\mathnormal{(}\mskip 0.0mu\mathsf{b}\mskip 3.0mu\mathnormal{⊗}\mskip 3.0mu\mathsf{a}\mskip 0.0mu\mathnormal{)}\allowbreak{}}\<[E]{}\\
\>[2]{}{\mathsf{U}\mskip 3.0mu}\>[3]{}{\mathnormal{::}\mskip 3.0mu\mathsf{FreeSMC}\mskip 3.0mu\mathsf{k}\mskip 3.0mu\mathsf{a}\mskip 3.0mu\allowbreak{}\mathnormal{(}\mskip 0.0mu\mathsf{a}\mskip 3.0mu\mathnormal{⊗}\mskip 3.0mu\allowbreak{}\mathnormal{(}\mskip 0.0mu\mathnormal{)}\allowbreak{}\mskip 0.0mu\mathnormal{)}\allowbreak{}}\<[E]{}\\
\>[2]{}{\mathsf{U'}\mskip 3.0mu}\>[3]{}{\mathnormal{::}\mskip 3.0mu\mathsf{FreeSMC}\mskip 3.0mu\mathsf{k}\mskip 3.0mu\allowbreak{}\mathnormal{(}\mskip 0.0mu\mathsf{a}\mskip 3.0mu\mathnormal{⊗}\mskip 3.0mu\allowbreak{}\mathnormal{(}\mskip 0.0mu\mathnormal{)}\allowbreak{}\mskip 0.0mu\mathnormal{)}\allowbreak{}\mskip 3.0mu\mathsf{a}}\<[E]{}\\
\>[2]{}{\allowbreak{}\mathnormal{(}\mskip 0.0mu\allowbreak{}\mathnormal{:\kern -3pt ∘ \kern -3pt:}\allowbreak{}\mskip 0.0mu\mathnormal{)}\allowbreak{}\mskip 3.0mu}\>[3]{}{\mathnormal{::}\mskip 3.0mu\mathsf{FreeSMC}\mskip 3.0mu\mathsf{k}\mskip 3.0mu\mathsf{b}\mskip 3.0mu\mathsf{c}\mskip 3.0mu\mathnormal{\rightarrow }\mskip 3.0mu\mathsf{FreeSMC}\mskip 3.0mu\mathsf{k}\mskip 3.0mu\mathsf{a}\mskip 3.0mu\mathsf{b}}\<[E]{}\\
\>[3]{}{\mathnormal{\rightarrow }\mskip 3.0mu\mathsf{FreeSMC}\mskip 3.0mu\mathsf{k}\mskip 3.0mu\mathsf{a}\mskip 3.0mu\mathsf{c}}\<[E]{}\\
\>[2]{}{\allowbreak{}\mathnormal{(}\mskip 0.0mu\mathnormal{:\kern -3pt × \kern -3pt:}\mskip 0.0mu\mathnormal{)}\allowbreak{}\mskip 3.0mu}\>[3]{}{\mathnormal{::}\mskip 3.0mu\mathsf{FreeSMC}\mskip 3.0mu\mathsf{k}\mskip 3.0mu\mathsf{a}\mskip 3.0mu\mathsf{b}\mskip 3.0mu\mathnormal{\rightarrow }\mskip 3.0mu\mathsf{FreeSMC}\mskip 3.0mu\mathsf{k}\mskip 3.0mu\mathsf{c}\mskip 3.0mu\mathsf{d}}\<[E]{}\\
\>[3]{}{\mathnormal{\rightarrow }\mskip 3.0mu\mathsf{FreeSMC}\mskip 3.0mu\mathsf{k}\mskip 3.0mu}\>[4]{}{\allowbreak{}\mathnormal{(}\mskip 0.0mu\mathsf{a}\mskip 3.0mu\mathnormal{⊗}\mskip 3.0mu\mathsf{c}\mskip 0.0mu\mathnormal{)}\allowbreak{}\mskip 3.0mu\allowbreak{}\mathnormal{(}\mskip 0.0mu\mathsf{b}\mskip 3.0mu\mathnormal{⊗}\mskip 3.0mu\mathsf{d}\mskip 0.0mu\mathnormal{)}\allowbreak{}}\<[E]{}\end{parray}}\end{list} 
The equality for \ensuremath{\mathsf{FreeSMC}} is quotiented by the same laws regarding \ensuremath{\mathsf{Embed}} as the \ensuremath{\mathsf{FreeCartesian}} representation.

\begin{list}{}{\setlength\leftmargin{1.0em}}\item\relax
\ensuremath{\begin{parray}\column{B}{@{}>{}l<{}@{}}\column[0em]{1}{@{}>{}l<{}@{}}\column{2}{@{}>{}l<{}@{}}\column{3}{@{}>{}l<{}@{}}\column{E}{@{}>{}l<{}@{}}%
\>[1]{}{\mathsf{evalM}\mskip 3.0mu\mathnormal{::}\mskip 3.0mu\allowbreak{}\mathnormal{(}\mskip 0.0mu}\>[3]{}{\mathsf{Monoidal}\mskip 3.0mu\mathsf{k}\mskip 0.0mu\mathnormal{)}\allowbreak{}\mskip 3.0mu\mathnormal{\Rightarrow }\mskip 3.0mu\mathsf{FreeSMC}\mskip 3.0mu\mathsf{k}\mskip 3.0mu\mathsf{a}\mskip 3.0mu\mathsf{b}\mskip 3.0mu\mathnormal{\rightarrow }\mskip 3.0mu\mathsf{a}\mskip 0.0mu\overset{\mathsf{k}}{\leadsto}\mskip 0.0mu\mathsf{b}}\<[E]{}\\
\>[1]{}{\mathsf{evalM}\mskip 3.0mu\mathsf{I}\mskip 3.0mu}\>[2]{}{\mathnormal{=}\mskip 3.0mu\mathsf{id}}\<[E]{}\\
\>[1]{}{\mathsf{evalM}\mskip 3.0mu\allowbreak{}\mathnormal{(}\mskip 0.0mu\mathsf{f}\mskip 3.0mu\mathnormal{:\kern -3pt × \kern -3pt:}\mskip 3.0mu\mathsf{g}\mskip 0.0mu\mathnormal{)}\allowbreak{}\mskip 3.0mu}\>[2]{}{\mathnormal{=}\mskip 3.0mu\mathsf{evalM}\mskip 3.0mu\mathsf{f}\mskip 3.0mu\allowbreak{}\mathnormal{×}\allowbreak{}\mskip 3.0mu\mathsf{evalM}\mskip 3.0mu\mathsf{g}}\<[E]{}\\
\>[1]{}{\mathsf{evalM}\mskip 3.0mu\allowbreak{}\mathnormal{(}\mskip 0.0mu\mathsf{f}\mskip 3.0mu\allowbreak{}\mathnormal{:\kern -3pt ∘ \kern -3pt:}\allowbreak{}\mskip 3.0mu\mathsf{g}\mskip 0.0mu\mathnormal{)}\allowbreak{}\mskip 3.0mu}\>[2]{}{\mathnormal{=}\mskip 3.0mu\mathsf{evalM}\mskip 3.0mu\mathsf{f}\mskip 3.0mu\allowbreak{}\mathnormal{∘}\allowbreak{}\mskip 3.0mu\mathsf{evalM}\mskip 3.0mu\mathsf{g}}\<[E]{}\\
\>[1]{}{\mathsf{evalM}\mskip 3.0mu\mathsf{A}\mskip 3.0mu}\>[2]{}{\mathnormal{=}\mskip 3.0muα}\<[E]{}\\
\>[1]{}{\mathsf{evalM}\mskip 3.0mu\mathsf{A'}\mskip 3.0mu}\>[2]{}{\mathnormal{=}\mskip 3.0mu\bar{α}}\<[E]{}\\
\>[1]{}{\mathsf{evalM}\mskip 3.0mu\mathsf{S}\mskip 3.0mu}\>[2]{}{\mathnormal{=}\mskip 3.0muσ}\<[E]{}\\
\>[1]{}{\mathsf{evalM}\mskip 3.0mu\mathsf{U}\mskip 3.0mu}\>[2]{}{\mathnormal{=}\mskip 3.0muρ}\<[E]{}\\
\>[1]{}{\mathsf{evalM}\mskip 3.0mu\mathsf{U'}\mskip 3.0mu}\>[2]{}{\mathnormal{=}\mskip 3.0mu\bar{ρ}}\<[E]{}\\
\>[1]{}{\mathsf{evalM}\mskip 3.0mu\allowbreak{}\mathnormal{(}\mskip 0.0mu\mathsf{Embed}\mskip 3.0mu\mathsf{ϕ}\mskip 0.0mu\mathnormal{)}\allowbreak{}\mskip 3.0mu}\>[2]{}{\mathnormal{=}\mskip 3.0mu\mathsf{ϕ}}\<[E]{}\end{parray}}\end{list} \end{mdframed} 
\subsection{Proving the implementation correct}\label{62} 
Even though we have not fully described the implementation yet, we know
enough to prove it correct. (Indeed, the only remaining uncertainty is
in the implementation of \ensuremath{\mathsf{reduce}}, but we already have specified that it must not
change the meaning of morphisms, only project them from free cartesian
to free {\sc{}smc} representations.)

To begin, we show that \ensuremath{\mathsf{decode}} respects the equality on ports.
Indeed, due to this equality being quotiented by \ensuremath{\mathsf{Embed}} being an
{\sc{}smc}-homomorphism, a bit of reasoning is necessary to prove that
functions over ports which are extensionally equal (with the above
equality for outputs) are decoded to equal morphisms:

\begin{lemma}{}\ensuremath{\allowbreak{}\mathnormal{(}\mskip 0.0mu}\ensuremath{\mathnormal{∀}\mskip 3.0mu}\ensuremath{\mathsf{x}\mskip 1.0mu}\ensuremath{.\mskip 3.0mu}\ensuremath{\mathsf{f}\mskip 3.0mu}\ensuremath{\mathsf{x}\mskip 3.0mu}\ensuremath{\mathnormal{=}\mskip 3.0mu}\ensuremath{\mathsf{g}\mskip 3.0mu}\ensuremath{\mathsf{x}\mskip 0.0mu}\ensuremath{\mathnormal{)}\allowbreak{}\mskip 3.0mu}\ensuremath{\mathnormal{\rightarrow }\mskip 3.0mu}\ensuremath{\mathsf{decode}\mskip 3.0mu}\ensuremath{\mathsf{f}\mskip 3.0mu}\ensuremath{\mathnormal{=}\mskip 3.0mu}\ensuremath{\mathsf{decode}\mskip 3.0mu}\ensuremath{\mathsf{g}}\label{63}\end{lemma}\begin{proof} 
The idea is that \ensuremath{\mathsf{decode}} subjects all \ensuremath{\mathsf{FreeCartesian}} morphisms to \ensuremath{\mathsf{evalM}}. Because \ensuremath{\mathsf{evalM}} maps representations that are equal
under the \ensuremath{\mathsf{Embed}} homomorphism equations to equal morphisms
in \ensuremath{\mathsf{k}}, we have our result.

\begin{mdframed}[linewidth=0pt,hidealllines,innerleftmargin=0pt,innerrightmargin=0pt,backgroundcolor=gray!15]Formally, the implication is proven by a transitive application of number of congruences:
\begin{align*}&\ensuremath{\mathnormal{∀}\mskip 3.0mu}\ensuremath{\mathsf{x}\mskip 1.0mu}\ensuremath{.\mskip 3.0mu}\ensuremath{\mathsf{f}\mskip 3.0mu}\ensuremath{\mathsf{x}\mskip 3.0mu}\ensuremath{\allowbreak{}=\allowbreak{}\mskip 3.0mu}\ensuremath{\mathsf{g}\mskip 3.0mu}\ensuremath{\mathsf{x}}&\\\ensuremath{\mathnormal{\rightarrow }}&\emph{}&\\&\ensuremath{\mathsf{f}\mskip 3.0mu}\ensuremath{\allowbreak{}\mathnormal{(}\mskip 0.0mu}\ensuremath{\mathsf{P}\mskip 3.0mu}\ensuremath{\mathsf{id}\mskip 0.0mu}\ensuremath{\mathnormal{)}\allowbreak{}\mskip 3.0mu}\ensuremath{\allowbreak{}=\allowbreak{}\mskip 3.0mu}\ensuremath{\mathsf{g}\mskip 3.0mu}\ensuremath{\allowbreak{}\mathnormal{(}\mskip 0.0mu}\ensuremath{\mathsf{P}\mskip 3.0mu}\ensuremath{\mathsf{id}\mskip 0.0mu}\ensuremath{\mathnormal{)}\allowbreak{}}&\\\ensuremath{\mathnormal{\rightarrow }}&\emph{ by congruence}&\\&\ensuremath{\mathsf{reduce}\mskip 3.0mu}\ensuremath{\allowbreak{}\mathnormal{(}\mskip 0.0mu}\ensuremath{\mathsf{f}\mskip 3.0mu}\ensuremath{\allowbreak{}\mathnormal{(}\mskip 0.0mu}\ensuremath{\mathsf{P}\mskip 3.0mu}\ensuremath{\mathsf{id}\mskip 0.0mu}\ensuremath{\mathnormal{)}\allowbreak{}\mskip 0.0mu}\ensuremath{\mathnormal{)}\allowbreak{}\mskip 3.0mu}\ensuremath{\allowbreak{}=\allowbreak{}\mskip 3.0mu}\ensuremath{\mathsf{reduce}\mskip 3.0mu}\ensuremath{\allowbreak{}\mathnormal{(}\mskip 0.0mu}\ensuremath{\mathsf{g}\mskip 3.0mu}\ensuremath{\allowbreak{}\mathnormal{(}\mskip 0.0mu}\ensuremath{\mathsf{P}\mskip 3.0mu}\ensuremath{\mathsf{id}\mskip 0.0mu}\ensuremath{\mathnormal{)}\allowbreak{}\mskip 0.0mu}\ensuremath{\mathnormal{)}\allowbreak{}}&\\\ensuremath{\mathnormal{\rightarrow }}& by \cref{64}\\&\ensuremath{\mathsf{evalM}\mskip 3.0mu}\ensuremath{\allowbreak{}\mathnormal{(}\mskip 0.0mu}\ensuremath{\mathsf{reduce}\mskip 3.0mu}\ensuremath{\allowbreak{}\mathnormal{(}\mskip 0.0mu}\ensuremath{\mathsf{f}\mskip 3.0mu}\ensuremath{\allowbreak{}\mathnormal{(}\mskip 0.0mu}\ensuremath{\mathsf{P}\mskip 3.0mu}\ensuremath{\mathsf{id}\mskip 0.0mu}\ensuremath{\mathnormal{)}\allowbreak{}\mskip 0.0mu}\ensuremath{\mathnormal{)}\allowbreak{}\mskip 0.0mu}\ensuremath{\mathnormal{)}\allowbreak{}\mskip 3.0mu}\ensuremath{\allowbreak{}=\allowbreak{}\mskip 3.0mu}\ensuremath{\mathsf{evalM}\mskip 3.0mu}\ensuremath{\allowbreak{}\mathnormal{(}\mskip 0.0mu}\ensuremath{\mathsf{reduce}\mskip 3.0mu}\ensuremath{\allowbreak{}\mathnormal{(}\mskip 0.0mu}\ensuremath{\mathsf{g}\mskip 3.0mu}\ensuremath{\allowbreak{}\mathnormal{(}\mskip 0.0mu}\ensuremath{\mathsf{P}\mskip 3.0mu}\ensuremath{\mathsf{id}\mskip 0.0mu}\ensuremath{\mathnormal{)}\allowbreak{}\mskip 0.0mu}\ensuremath{\mathnormal{)}\allowbreak{}\mskip 0.0mu}\ensuremath{\mathnormal{)}\allowbreak{}}&\\\ensuremath{\mathnormal{\rightarrow }}&\emph{ by def.}&\\&\ensuremath{\mathsf{evalM}\mskip 3.0mu}\ensuremath{\allowbreak{}\mathnormal{(}\mskip 0.0mu}\ensuremath{\mathsf{reduce}\mskip 3.0mu}\ensuremath{\allowbreak{}\mathnormal{(}\mskip 0.0mu}\ensuremath{\mathsf{extract}\mskip 3.0mu}\ensuremath{\mathsf{f}\mskip 0.0mu}\ensuremath{\mathnormal{)}\allowbreak{}\mskip 0.0mu}\ensuremath{\mathnormal{)}\allowbreak{}\mskip 3.0mu}\ensuremath{\allowbreak{}=\allowbreak{}\mskip 3.0mu}\ensuremath{\mathsf{evalM}\mskip 3.0mu}\ensuremath{\allowbreak{}\mathnormal{(}\mskip 0.0mu}\ensuremath{\mathsf{reduce}\mskip 3.0mu}\ensuremath{\allowbreak{}\mathnormal{(}\mskip 0.0mu}\ensuremath{\mathsf{extract}\mskip 3.0mu}\ensuremath{\mathsf{g}\mskip 0.0mu}\ensuremath{\mathnormal{)}\allowbreak{}\mskip 0.0mu}\ensuremath{\mathnormal{)}\allowbreak{}}&\\\ensuremath{\mathnormal{\rightarrow }}&\emph{ by def.}&\\&\ensuremath{\mathsf{decode}\mskip 3.0mu}\ensuremath{\mathsf{f}\mskip 3.0mu}\ensuremath{\allowbreak{}=\allowbreak{}\mskip 3.0mu}\ensuremath{\mathsf{decode}\mskip 3.0mu}\ensuremath{\mathsf{g}}&\end{align*}\end{mdframed} 
\end{proof}\begin{mdframed}[linewidth=0pt,hidealllines,innerleftmargin=0pt,innerrightmargin=0pt,backgroundcolor=gray!15]The critical step, which is taken care of by the following lemma, is necessary because we go from an equality on a type where equality is quotiented, to a type where equality is not quotiented.\begin{lemma}{}if \ensuremath{\mathsf{x}\mskip 3.0mu}\ensuremath{\mathnormal{=}\mskip 3.0mu}\ensuremath{\mathsf{y}} then \ensuremath{\mathsf{evalM}\mskip 3.0mu}\ensuremath{\mathsf{x}\mskip 3.0mu}\ensuremath{\mathnormal{=}\mskip 3.0mu}\ensuremath{\mathsf{evalM}\mskip 3.0mu}\ensuremath{\mathsf{y}}\label{64}\end{lemma}\begin{proof} We need to
show that the terms which we deem equal by quotienting the equality of
\ensuremath{\mathsf{FreeCartesian}} are mapped to equal terms by \ensuremath{\mathsf{evalM}}. This
is done case by case, and a simple matter of expanding definitions. We
show two cases here: the others follow the same patterns.

\begin{itemize}\item{}\ensuremath{\mathsf{evalM}\mskip 3.0mu}\ensuremath{\mathsf{id}\mskip 3.0mu}\ensuremath{\allowbreak{}=\allowbreak{}\mskip 3.0mu}\ensuremath{\mathsf{id}\mskip 3.0mu}\ensuremath{\allowbreak{}=\allowbreak{}\mskip 3.0mu}\ensuremath{\mathsf{evalM}\mskip 3.0mu}\ensuremath{\allowbreak{}\mathnormal{(}\mskip 0.0mu}\ensuremath{\mathsf{Embed}\mskip 3.0mu}\ensuremath{\mathsf{Id}\mskip 0.0mu}\ensuremath{\mathnormal{)}\allowbreak{}}\item{}\ensuremath{\mathsf{evalM}\mskip 3.0mu}\ensuremath{\allowbreak{}\mathnormal{(}\mskip 0.0mu}\ensuremath{\mathsf{Embed}\mskip 3.0mu}\ensuremath{\mathsf{φ}\mskip 3.0mu}\ensuremath{\allowbreak{}\mathnormal{×}\allowbreak{}\mskip 3.0mu}\ensuremath{\mathsf{Embed}\mskip 3.0mu}\ensuremath{\mathsf{ψ}\mskip 0.0mu}\ensuremath{\mathnormal{)}\allowbreak{}\mskip 3.0mu}\ensuremath{\allowbreak{}=\allowbreak{}\mskip 3.0mu}\ensuremath{\mathsf{evalM}\mskip 3.0mu}\ensuremath{\allowbreak{}\mathnormal{(}\mskip 0.0mu}\ensuremath{\mathsf{Embed}\mskip 3.0mu}\ensuremath{\mathsf{φ}\mskip 0.0mu}\ensuremath{\mathnormal{)}\allowbreak{}\mskip 3.0mu}\ensuremath{\allowbreak{}\mathnormal{×}\allowbreak{}\mskip 3.0mu}\ensuremath{\mathsf{evalM}\mskip 3.0mu}\ensuremath{\allowbreak{}\mathnormal{(}\mskip 0.0mu}\ensuremath{\mathsf{Embed}\mskip 3.0mu}\ensuremath{\mathsf{ψ}\mskip 0.0mu}\ensuremath{\mathnormal{)}\allowbreak{}\mskip 3.0mu}\ensuremath{\allowbreak{}=\allowbreak{}\mskip 3.0mu}\ensuremath{\mathsf{φ}\mskip 3.0mu}\ensuremath{\allowbreak{}\mathnormal{×}\allowbreak{}\mskip 3.0mu}\ensuremath{\mathsf{ψ}\mskip 3.0mu}\ensuremath{\allowbreak{}=\allowbreak{}\mskip 3.0mu}\ensuremath{\mathsf{evalM}\mskip 3.0mu}\ensuremath{\allowbreak{}\mathnormal{(}\mskip 0.0mu}\ensuremath{\mathsf{Embed}\mskip 3.0mu}\ensuremath{\allowbreak{}\mathnormal{(}\mskip 0.0mu}\ensuremath{\mathsf{φ}\mskip 3.0mu}\ensuremath{\allowbreak{}\mathnormal{×}\allowbreak{}\mskip 3.0mu}\ensuremath{\mathsf{ψ}\mskip 0.0mu}\ensuremath{\mathnormal{)}\allowbreak{}\mskip 0.0mu}\ensuremath{\mathnormal{)}\allowbreak{}}\end{itemize} \end{proof}\end{mdframed} 
We can then prove all the laws given in \cref{35}.
\begin{mdframed}[linewidth=0pt,hidealllines,innerleftmargin=0pt,innerrightmargin=0pt,backgroundcolor=gray!15]\begin{theorem}{}The implementation respects the laws stated in \cref{35}.\label{65}\end{theorem}\begin{proof}Each case can be proven by equational reasoning. (In these reduction steps we assume that the \ensuremath{\mathsf{P}\mskip 3.0mu}\ensuremath{\mathsf{r}} type forms a cartesian category, obtained by lifting the same structure from \ensuremath{\mathsf{FreeCartesian}}.
    This simplification means that we can skip many uninformative conversions between the two types using \ensuremath{\mathsf{P}} and \ensuremath{\mathsf{fromP}}.)\begin{itemize}\item{}split/merge\begin{list}{}{\setlength\leftmargin{1.0em}}\item\relax
\ensuremath{\begin{parray}\column{B}{@{}>{}l<{}@{}}\column[0em]{1}{@{}>{}l<{}@{}}\column[1em]{2}{@{}>{}l<{}@{}}\column{3}{@{}>{}l<{}@{}}\column{E}{@{}>{}l<{}@{}}%
\>[2]{}{\mathsf{split}\mskip 3.0mu\allowbreak{}\mathnormal{(}\mskip 0.0mu\mathsf{merge}\mskip 3.0mu\allowbreak{}\mathnormal{(}\mskip 0.0mu\mathsf{x}\mskip 0.0mu\mathnormal{,}\mskip 3.0mu\mathsf{y}\mskip 0.0mu\mathnormal{)}\allowbreak{}\mskip 0.0mu\mathnormal{)}\allowbreak{}}\<[E]{}\\
\>[1]{}{\allowbreak{}=\allowbreak{}\mskip 0.0mu}\>[3]{}{\quad{}\text{\textit{by def}}}\<[E]{}\\
\>[2]{}{\mathsf{split}\mskip 3.0mu\allowbreak{}\mathnormal{(}\mskip 0.0mu\mathsf{x}\mskip 3.0mu\mathnormal{▵}\mskip 3.0mu\mathsf{y}\mskip 0.0mu\mathnormal{)}\allowbreak{}}\<[E]{}\\
\>[1]{}{\allowbreak{}=\allowbreak{}\mskip 0.0mu}\>[3]{}{\quad{}\text{\textit{by def}}}\<[E]{}\\
\>[2]{}{\allowbreak{}\mathnormal{(}\mskip 0.0muπ₁\mskip 3.0mu\allowbreak{}\mathnormal{∘}\allowbreak{}\mskip 3.0mu\allowbreak{}\mathnormal{(}\mskip 0.0mu\mathsf{x}\mskip 3.0mu\mathnormal{▵}\mskip 3.0mu\mathsf{y}\mskip 0.0mu\mathnormal{)}\allowbreak{}\mskip 0.0mu\mathnormal{,}\mskip 3.0muπ₂\mskip 3.0mu\allowbreak{}\mathnormal{∘}\allowbreak{}\mskip 3.0mu\allowbreak{}\mathnormal{(}\mskip 0.0mu\mathsf{x}\mskip 3.0mu\mathnormal{▵}\mskip 3.0mu\mathsf{y}\mskip 0.0mu\mathnormal{)}\allowbreak{}\mskip 0.0mu\mathnormal{)}\allowbreak{}}\<[E]{}\\
\>[1]{}{\allowbreak{}=\allowbreak{}\mskip 0.0mu}\>[3]{}{\quad{}\text{\textit{by cartesian category properties}}}\<[E]{}\\
\>[2]{}{\allowbreak{}\mathnormal{(}\mskip 0.0mu\mathsf{x}\mskip 0.0mu\mathnormal{,}\mskip 3.0mu\mathsf{y}\mskip 0.0mu\mathnormal{)}\allowbreak{}}\<[E]{}\end{parray}}\end{list}\item{}merge/split\begin{list}{}{\setlength\leftmargin{1.0em}}\item\relax
\ensuremath{\begin{parray}\column{B}{@{}>{}l<{}@{}}\column[0em]{1}{@{}>{}l<{}@{}}\column[1em]{2}{@{}>{}l<{}@{}}\column{3}{@{}>{}l<{}@{}}\column{E}{@{}>{}l<{}@{}}%
\>[2]{}{\mathsf{merge}\mskip 3.0mu\allowbreak{}\mathnormal{(}\mskip 0.0mu\mathsf{split}\mskip 3.0mu\mathsf{f}\mskip 0.0mu\mathnormal{)}\allowbreak{}}\<[E]{}\\
\>[1]{}{\allowbreak{}=\allowbreak{}\mskip 0.0mu}\>[3]{}{\quad{}\text{\textit{by def}}}\<[E]{}\\
\>[2]{}{\allowbreak{}\mathnormal{(}\mskip 0.0mu\mathbf{let}\mskip 3.0mu\allowbreak{}\mathnormal{(}\mskip 0.0mu\mathsf{x}\mskip 0.0mu\mathnormal{,}\mskip 3.0mu\mathsf{y}\mskip 0.0mu\mathnormal{)}\allowbreak{}\mskip 3.0mu\mathnormal{=}\mskip 3.0mu\allowbreak{}\mathnormal{(}\mskip 0.0mu\allowbreak{}\mathnormal{(}\mskip 0.0muπ₁\mskip 3.0mu\allowbreak{}\mathnormal{∘}\allowbreak{}\mskip 3.0mu\mathsf{f}\mskip 0.0mu\mathnormal{)}\allowbreak{}\mskip 0.0mu\mathnormal{,}\mskip 3.0mu\allowbreak{}\mathnormal{(}\mskip 0.0muπ₂\mskip 3.0mu\allowbreak{}\mathnormal{∘}\allowbreak{}\mskip 3.0mu\mathsf{f}\mskip 0.0mu\mathnormal{)}\allowbreak{}\mskip 0.0mu\mathnormal{)}\allowbreak{}\mskip 3.0mu\mathbf{in}\mskip 3.0mu\allowbreak{}\mathnormal{(}\mskip 0.0mu\mathsf{x}\mskip 3.0mu\mathnormal{▵}\mskip 3.0mu\mathsf{y}\mskip 0.0mu\mathnormal{)}\allowbreak{}\mskip 0.0mu\mathnormal{)}\allowbreak{}}\<[E]{}\\
\>[1]{}{\allowbreak{}=\allowbreak{}\mskip 0.0mu}\>[3]{}{\quad{}\text{\textit{by evaluation}}}\<[E]{}\\
\>[2]{}{\allowbreak{}\mathnormal{(}\mskip 0.0mu\allowbreak{}\mathnormal{(}\mskip 0.0muπ₁\mskip 3.0mu\allowbreak{}\mathnormal{∘}\allowbreak{}\mskip 3.0mu\mathsf{f}\mskip 0.0mu\mathnormal{)}\allowbreak{}\mskip 3.0mu\mathnormal{▵}\mskip 3.0mu\allowbreak{}\mathnormal{(}\mskip 0.0muπ₂\mskip 3.0mu\allowbreak{}\mathnormal{∘}\allowbreak{}\mskip 3.0mu\mathsf{f}\mskip 0.0mu\mathnormal{)}\allowbreak{}\mskip 0.0mu\mathnormal{)}\allowbreak{}}\<[E]{}\\
\>[1]{}{\allowbreak{}=\allowbreak{}\mskip 0.0mu}\>[3]{}{\quad{}\text{\textit{by cartestian laws}}}\<[E]{}\\
\>[2]{}{\mathsf{f}}\<[E]{}\end{parray}}\end{list}\item{}decode/encode\begin{list}{}{\setlength\leftmargin{1.0em}}\item\relax
\ensuremath{\begin{parray}\column{B}{@{}>{}l<{}@{}}\column[0em]{1}{@{}>{}l<{}@{}}\column[1em]{2}{@{}>{}l<{}@{}}\column{3}{@{}>{}l<{}@{}}\column{4}{@{}>{}l<{}@{}}\column[4em]{5}{@{}>{}l<{}@{}}\column{E}{@{}>{}l<{}@{}}%
\>[2]{}{\mathsf{decode}\mskip 3.0mu\allowbreak{}\mathnormal{(}\mskip 0.0mu\mathsf{encode}\mskip 3.0mu\mathsf{f}\mskip 0.0mu\mathnormal{)}\allowbreak{}}\<[E]{}\\
\>[1]{}{\allowbreak{}=\allowbreak{}\mskip 0.0mu}\>[3]{}{\quad{}\text{\textit{by def}}}\<[E]{}\\
\>[2]{}{\mathsf{decode}\mskip 3.0mu\allowbreak{}\mathnormal{(}\mskip 0.0muλ\mskip 3.0mu\allowbreak{}\mathnormal{(}\mskip 0.0mu\mathsf{P}\mskip 3.0mu\mathsf{x}\mskip 0.0mu\mathnormal{)}\allowbreak{}\mskip 3.0mu\mathnormal{\rightarrow }\mskip 3.0mu\mathsf{P}\mskip 3.0mu\allowbreak{}\mathnormal{(}\mskip 0.0mu\mathsf{Embed}\mskip 3.0mu\mathsf{f}\mskip 3.0mu\allowbreak{}\mathnormal{∘}\allowbreak{}\mskip 3.0mu\mathsf{x}\mskip 0.0mu\mathnormal{)}\allowbreak{}\mskip 0.0mu\mathnormal{)}\allowbreak{}}\<[E]{}\\
\>[1]{}{\allowbreak{}=\allowbreak{}\mskip 0.0mu}\>[3]{}{\quad{}\text{\textit{by def}}}\<[E]{}\\
\>[2]{}{\mathsf{evalM}\mskip 3.0mu\allowbreak{}\mathnormal{(}\mskip 0.0mu\mathsf{reduce}\mskip 3.0mu\allowbreak{}\mathnormal{(}\mskip 0.0mu\mathsf{extract}\mskip 3.0mu\allowbreak{}\mathnormal{(}\mskip 0.0muλ\mskip 3.0mu\allowbreak{}\mathnormal{(}\mskip 0.0mu\mathsf{P}\mskip 3.0mu\mathsf{x}\mskip 0.0mu\mathnormal{)}\allowbreak{}\mskip 3.0mu\mathnormal{\rightarrow }}\<[E]{}\\
\>[5]{}{\mathsf{P}\mskip 3.0mu\allowbreak{}\mathnormal{(}\mskip 0.0mu\mathsf{Embed}\mskip 3.0mu\mathsf{f}\mskip 3.0mu\allowbreak{}\mathnormal{∘}\allowbreak{}\mskip 3.0mu\mathsf{x}\mskip 0.0mu\mathnormal{)}\allowbreak{}\mskip 0.0mu\mathnormal{)}\allowbreak{}\mskip 0.0mu\mathnormal{)}\allowbreak{}\mskip 0.0mu\mathnormal{)}\allowbreak{}}\<[E]{}\\
\>[1]{}{\allowbreak{}=\allowbreak{}\mskip 0.0mu}\>[3]{}{\quad{}\text{\textit{by def}}}\<[E]{}\\
\>[2]{}{\mathsf{evalM}\mskip 3.0mu\allowbreak{}\mathnormal{(}\mskip 0.0mu\mathsf{reduce}\mskip 3.0mu\allowbreak{}\mathnormal{(}\mskip 0.0mu\allowbreak{}\mathnormal{(}\mskip 0.0muλ\mskip 3.0mu\mathsf{x}\mskip 3.0mu\mathnormal{\rightarrow }\mskip 3.0mu\mathsf{Embed}\mskip 3.0mu\mathsf{f}\mskip 3.0mu\allowbreak{}\mathnormal{∘}\allowbreak{}\mskip 3.0mu\mathsf{x}\mskip 0.0mu\mathnormal{)}\allowbreak{}\mskip 3.0mu\mathsf{id}\mskip 0.0mu\mathnormal{)}\allowbreak{}\mskip 0.0mu\mathnormal{)}\allowbreak{}}\<[E]{}\\
\>[1]{}{\allowbreak{}=\allowbreak{}\mskip 0.0mu}\>[3]{}{\quad{}\text{\textit{by β-reduction}}}\<[E]{}\\
\>[2]{}{\mathsf{evalM}\mskip 3.0mu\allowbreak{}\mathnormal{(}\mskip 0.0mu\mathsf{reduce}\mskip 3.0mu}\>[4]{}{\allowbreak{}\mathnormal{(}\mskip 0.0mu\mathsf{Embed}\mskip 3.0mu\mathsf{f}\mskip 3.0mu\allowbreak{}\mathnormal{∘}\allowbreak{}\mskip 3.0mu\mathsf{id}\mskip 0.0mu\mathnormal{)}\allowbreak{}\mskip 0.0mu\mathnormal{)}\allowbreak{}}\<[E]{}\\
\>[1]{}{\allowbreak{}=\allowbreak{}\mskip 0.0mu}\>[3]{}{\quad{}\text{\textit{by property of host language composition}}}\<[E]{}\\
\>[2]{}{\mathsf{evalM}\mskip 3.0mu\allowbreak{}\mathnormal{(}\mskip 0.0mu\mathsf{reduce}\mskip 3.0mu}\>[4]{}{\allowbreak{}\mathnormal{(}\mskip 0.0mu\mathsf{Embed}\mskip 3.0mu\mathsf{f}\mskip 0.0mu\mathnormal{)}\allowbreak{}\mskip 0.0mu\mathnormal{)}\allowbreak{}}\<[E]{}\\
\>[1]{}{\allowbreak{}=\allowbreak{}\mskip 0.0mu}\>[3]{}{\quad{}\text{\textit{by evalM ∘ reduce ∘ Embed = id}}}\<[E]{}\\
\>[2]{}{\mathsf{f}}\<[E]{}\end{parray}}\end{list}\item{}encode/decode\begin{list}{}{\setlength\leftmargin{1.0em}}\item\relax
\ensuremath{\begin{parray}\column{B}{@{}>{}l<{}@{}}\column[0em]{1}{@{}>{}l<{}@{}}\column[1em]{2}{@{}>{}l<{}@{}}\column[2em]{3}{@{}>{}l<{}@{}}\column{4}{@{}>{}l<{}@{}}\column{E}{@{}>{}l<{}@{}}%
\>[3]{}{\mathsf{encode}\mskip 3.0mu\allowbreak{}\mathnormal{(}\mskip 0.0mu\mathsf{decode}\mskip 3.0mu\mathsf{f}\mskip 0.0mu\mathnormal{)}\allowbreak{}\mskip 3.0mu\allowbreak{}\mathnormal{(}\mskip 0.0mu\mathsf{P}\mskip 3.0mu\mathsf{a}\mskip 0.0mu\mathnormal{)}\allowbreak{}}\<[E]{}\\
\>[1]{}{\allowbreak{}=\allowbreak{}\mskip 0.0mu}\>[4]{}{\quad{}\text{\textit{by def of encode}}}\<[E]{}\\
\>[3]{}{\mathsf{P}\mskip 3.0mu\allowbreak{}\mathnormal{(}\mskip 0.0mu\mathsf{Embed}\mskip 3.0mu\allowbreak{}\mathnormal{(}\mskip 0.0mu\mathsf{decode}\mskip 3.0mu\mathsf{f}\mskip 0.0mu\mathnormal{)}\allowbreak{}\mskip 3.0mu\allowbreak{}\mathnormal{∘}\allowbreak{}\mskip 3.0mu\mathsf{a}\mskip 0.0mu\mathnormal{)}\allowbreak{}}\<[E]{}\\
\>[1]{}{\allowbreak{}=\allowbreak{}\mskip 0.0mu}\>[4]{}{\quad{}\text{\textit{by def of decode}}}\<[E]{}\\
\>[3]{}{\mathsf{P}\mskip 3.0mu\allowbreak{}\mathnormal{(}\mskip 0.0mu\mathsf{Embed}\mskip 3.0mu\allowbreak{}\mathnormal{(}\mskip 0.0mu\mathsf{evalM}\mskip 3.0mu\allowbreak{}\mathnormal{(}\mskip 0.0mu\mathsf{reduce}\mskip 3.0mu\allowbreak{}\mathnormal{(}\mskip 0.0mu\mathsf{fromP}\mskip 3.0mu\allowbreak{}\mathnormal{(}\mskip 0.0mu\mathsf{f}\mskip 3.0mu\mathsf{id}\mskip 0.0mu\mathnormal{)}\allowbreak{}\mskip 0.0mu\mathnormal{)}\allowbreak{}\mskip 0.0mu\mathnormal{)}\allowbreak{}\mskip 0.0mu\mathnormal{)}\allowbreak{}\mskip 3.0mu\allowbreak{}\mathnormal{∘}\allowbreak{}\mskip 3.0mu\mathsf{a}\mskip 0.0mu\mathnormal{)}\allowbreak{}}\<[E]{}\\
\>[1]{}{\allowbreak{}=\allowbreak{}\mskip 0.0mu}\>[4]{}{\quad{}\text{\textit{by Embed ∘ evalM ∘ reduce = id}}}\<[E]{}\\
\>[3]{}{\mathsf{f}\mskip 3.0mu\mathsf{id}\mskip 3.0mu\allowbreak{}\mathnormal{∘}\allowbreak{}\mskip 3.0mu\mathsf{P}\mskip 3.0mu\mathsf{a}}\<[E]{}\\
\>[1]{}{\allowbreak{}=\allowbreak{}\mskip 0.0mu}\>[4]{}{\quad{}\text{\textit{by Covariant Yoneda Lemma (naturality of f)}}}\<[E]{}\\
\>[2]{}{\mathsf{f}\mskip 3.0mu\allowbreak{}\mathnormal{(}\mskip 0.0mu\mathsf{P}\mskip 3.0mu\mathsf{a}\mskip 0.0mu\mathnormal{)}\allowbreak{}}\<[E]{}\end{parray}} The step one way to see that \ensuremath{\mathsf{Embed}\mskip 3.0mu}\ensuremath{\allowbreak{}\mathnormal{∘}\allowbreak{}\mskip 3.0mu}\ensuremath{\mathsf{evalM}\mskip 3.0mu}\ensuremath{\allowbreak{}\mathnormal{∘}\allowbreak{}\mskip 3.0mu}\ensuremath{\mathsf{reduce}\mskip 3.0mu}\ensuremath{\mathnormal{=}\mskip 3.0mu}\ensuremath{\mathsf{id}} is to notice that
\ensuremath{\mathsf{reduce}} does not change the meaning of morphisms, only their representation, from free cartesian to free {\sc{}smc}.
The composition
\ensuremath{\mathsf{Embed}\mskip 3.0mu}\ensuremath{\allowbreak{}\mathnormal{∘}\allowbreak{}\mskip 3.0mu}\ensuremath{\mathsf{evalM}} does the opposite conversion. We have equality because free {\sc{}smc} terms are quotiented by \ensuremath{\mathsf{Embed}} being an homomorphism.
\end{list}\item{}encode/∘\begin{list}{}{\setlength\leftmargin{1.0em}}\item\relax
\ensuremath{\begin{parray}\column{B}{@{}>{}l<{}@{}}\column[0em]{1}{@{}>{}l<{}@{}}\column[1em]{2}{@{}>{}l<{}@{}}\column{3}{@{}>{}l<{}@{}}\column{E}{@{}>{}l<{}@{}}%
\>[2]{}{\mathsf{encode}\mskip 3.0mu\allowbreak{}\mathnormal{(}\mskip 0.0mu\mathsf{φ}\mskip 3.0mu\allowbreak{}\mathnormal{∘}\allowbreak{}\mskip 3.0mu\mathsf{ψ}\mskip 0.0mu\mathnormal{)}\allowbreak{}\mskip 3.0mu\allowbreak{}\mathnormal{(}\mskip 0.0mu\mathsf{P}\mskip 3.0mu\mathsf{f}\mskip 0.0mu\mathnormal{)}\allowbreak{}}\<[E]{}\\
\>[1]{}{\allowbreak{}=\allowbreak{}\mskip 0.0mu}\>[3]{}{\quad{}\text{\textit{by def}}}\<[E]{}\\
\>[2]{}{\mathsf{P}\mskip 3.0mu\allowbreak{}\mathnormal{(}\mskip 0.0mu\mathsf{Embed}\mskip 3.0mu\allowbreak{}\mathnormal{(}\mskip 0.0mu\mathsf{φ}\mskip 3.0mu\allowbreak{}\mathnormal{∘}\allowbreak{}\mskip 3.0mu\mathsf{ψ}\mskip 0.0mu\mathnormal{)}\allowbreak{}\mskip 3.0mu\allowbreak{}\mathnormal{∘}\allowbreak{}\mskip 3.0mu\mathsf{f}\mskip 0.0mu\mathnormal{)}\allowbreak{}}\<[E]{}\\
\>[1]{}{\allowbreak{}=\allowbreak{}\mskip 0.0mu}\>[3]{}{\quad{}\text{\textit{by Embed property}}}\<[E]{}\\
\>[2]{}{\mathsf{P}\mskip 3.0mu\allowbreak{}\mathnormal{(}\mskip 0.0mu\mathsf{Embed}\mskip 3.0mu\mathsf{φ}\mskip 3.0mu\allowbreak{}\mathnormal{∘}\allowbreak{}\mskip 3.0mu\mathsf{Embed}\mskip 3.0mu\mathsf{ψ}\mskip 3.0mu\allowbreak{}\mathnormal{∘}\allowbreak{}\mskip 3.0mu\mathsf{f}\mskip 0.0mu\mathnormal{)}\allowbreak{}}\<[E]{}\\
\>[1]{}{\allowbreak{}=\allowbreak{}\mskip 0.0mu}\>[3]{}{\quad{}\text{\textit{by def of encode}}}\<[E]{}\\
\>[2]{}{\mathsf{encode}\mskip 3.0mu\mathsf{φ}\mskip 3.0mu\allowbreak{}\mathnormal{(}\mskip 0.0mu\mathsf{P}\mskip 3.0mu\allowbreak{}\mathnormal{(}\mskip 0.0mu\mathsf{Embed}\mskip 3.0mu\mathsf{ψ}\mskip 3.0mu\allowbreak{}\mathnormal{∘}\allowbreak{}\mskip 3.0mu\mathsf{f}\mskip 0.0mu\mathnormal{)}\allowbreak{}\mskip 0.0mu\mathnormal{)}\allowbreak{}}\<[E]{}\\
\>[1]{}{\allowbreak{}=\allowbreak{}\mskip 0.0mu}\>[3]{}{\quad{}\text{\textit{by def of encode}}}\<[E]{}\\
\>[2]{}{\mathsf{encode}\mskip 3.0mu\mathsf{φ}\mskip 3.0mu\allowbreak{}\mathnormal{(}\mskip 0.0mu\mathsf{encode}\mskip 3.0mu\mathsf{ψ}\mskip 3.0mu\allowbreak{}\mathnormal{(}\mskip 0.0mu\mathsf{P}\mskip 3.0mu\mathsf{f}\mskip 0.0mu\mathnormal{)}\allowbreak{}\mskip 0.0mu\mathnormal{)}\allowbreak{}}\<[E]{}\\
\>[1]{}{\allowbreak{}=\allowbreak{}\mskip 0.0mu}\>[3]{}{\quad{}\text{\textit{by def of ∘}}}\<[E]{}\\
\>[2]{}{\allowbreak{}\mathnormal{(}\mskip 0.0mu\mathsf{encode}\mskip 3.0mu\mathsf{φ}\mskip 3.0mu\allowbreak{}\mathnormal{∘}\allowbreak{}\mskip 3.0mu\mathsf{encode}\mskip 3.0mu\mathsf{ψ}\mskip 0.0mu\mathnormal{)}\allowbreak{}\mskip 3.0mu\allowbreak{}\mathnormal{(}\mskip 0.0mu\mathsf{P}\mskip 3.0mu\mathsf{f}\mskip 0.0mu\mathnormal{)}\allowbreak{}}\<[E]{}\end{parray}}\end{list}\item{}encode/id\begin{list}{}{\setlength\leftmargin{1.0em}}\item\relax
\ensuremath{\begin{parray}\column{B}{@{}>{}l<{}@{}}\column[0em]{1}{@{}>{}l<{}@{}}\column[1em]{2}{@{}>{}l<{}@{}}\column{3}{@{}>{}l<{}@{}}\column{E}{@{}>{}l<{}@{}}%
\>[2]{}{\mathsf{encode}\mskip 3.0mu\mathsf{id}\mskip 3.0mu\allowbreak{}\mathnormal{(}\mskip 0.0mu\mathsf{P}\mskip 3.0mu\mathsf{f}\mskip 0.0mu\mathnormal{)}\allowbreak{}}\<[E]{}\\
\>[1]{}{\allowbreak{}=\allowbreak{}\mskip 0.0mu}\>[3]{}{\quad{}\text{\textit{by definition of encode}}}\<[E]{}\\
\>[2]{}{\mathsf{P}\mskip 3.0mu\allowbreak{}\mathnormal{(}\mskip 0.0mu\mathsf{Embed}\mskip 3.0mu\mathsf{id}\mskip 3.0mu\allowbreak{}\mathnormal{∘}\allowbreak{}\mskip 3.0mu\mathsf{f}\mskip 0.0mu\mathnormal{)}\allowbreak{}}\<[E]{}\\
\>[1]{}{\allowbreak{}=\allowbreak{}\mskip 0.0mu}\>[3]{}{\quad{}\text{\textit{by Embed property}}}\<[E]{}\\
\>[2]{}{\mathsf{id}\mskip 3.0mu\allowbreak{}\mathnormal{∘}\allowbreak{}\mskip 3.0mu\mathsf{P}\mskip 3.0mu\mathsf{f}}\<[E]{}\\
\>[1]{}{\allowbreak{}=\allowbreak{}\mskip 0.0mu}\>[3]{}{\quad{}\text{\textit{by def}}}\<[E]{}\\
\>[2]{}{\mathsf{P}\mskip 3.0mu\mathsf{f}}\<[E]{}\\
\>[1]{}{\allowbreak{}=\allowbreak{}\mskip 0.0mu}\>[3]{}{\quad{}\text{\textit{by def}}}\<[E]{}\\
\>[2]{}{\mathsf{id}\mskip 3.0mu\allowbreak{}\mathnormal{(}\mskip 0.0mu\mathsf{P}\mskip 3.0mu\mathsf{f}\mskip 0.0mu\mathnormal{)}\allowbreak{}}\<[E]{}\end{parray}}\end{list}\item{}encode/merge\begin{list}{}{\setlength\leftmargin{1.0em}}\item\relax
\ensuremath{\begin{parray}\column{B}{@{}>{}l<{}@{}}\column[0em]{1}{@{}>{}l<{}@{}}\column[1em]{2}{@{}>{}l<{}@{}}\column{3}{@{}>{}l<{}@{}}\column{E}{@{}>{}l<{}@{}}%
\>[2]{}{\mathsf{encode}\mskip 3.0mu\allowbreak{}\mathnormal{(}\mskip 0.0mu\mathsf{φ}\mskip 3.0mu\allowbreak{}\mathnormal{×}\allowbreak{}\mskip 3.0mu\mathsf{ψ}\mskip 0.0mu\mathnormal{)}\allowbreak{}\mskip 3.0mu\allowbreak{}\mathnormal{(}\mskip 0.0mu\mathsf{P}\mskip 3.0mu\mathsf{a}\mskip 2.0mu\mathnormal{\fatsemi }\mskip 3.0mu\mathsf{P}\mskip 3.0mu\mathsf{b}\mskip 0.0mu\mathnormal{)}\allowbreak{}}\<[E]{}\\
\>[1]{}{\allowbreak{}=\allowbreak{}\mskip 0.0mu}\>[3]{}{\quad{}\text{\textit{by def}}}\<[E]{}\\
\>[2]{}{\mathsf{P}\mskip 3.0mu\allowbreak{}\mathnormal{(}\mskip 0.0mu\mathsf{Embed}\mskip 3.0mu\allowbreak{}\mathnormal{(}\mskip 0.0mu\mathsf{φ}\mskip 3.0mu\allowbreak{}\mathnormal{×}\allowbreak{}\mskip 3.0mu\mathsf{ψ}\mskip 0.0mu\mathnormal{)}\allowbreak{}\mskip 3.0mu\allowbreak{}\mathnormal{∘}\allowbreak{}\mskip 3.0mu\allowbreak{}\mathnormal{(}\mskip 0.0mu\mathsf{a}\mskip 3.0mu\mathnormal{▵}\mskip 3.0mu\mathsf{b}\mskip 0.0mu\mathnormal{)}\allowbreak{}\mskip 0.0mu\mathnormal{)}\allowbreak{}}\<[E]{}\\
\>[1]{}{\allowbreak{}=\allowbreak{}\mskip 0.0mu}\>[3]{}{\quad{}\text{\textit{by assumption on Embed}}}\<[E]{}\\
\>[2]{}{\mathsf{P}\mskip 3.0mu\allowbreak{}\mathnormal{(}\mskip 0.0mu\allowbreak{}\mathnormal{(}\mskip 0.0mu\mathsf{Embed}\mskip 3.0mu\mathsf{φ}\mskip 3.0mu\allowbreak{}\mathnormal{×}\allowbreak{}\mskip 3.0mu\mathsf{Embed}\mskip 3.0mu\mathsf{ψ}\mskip 0.0mu\mathnormal{)}\allowbreak{}\mskip 3.0mu\allowbreak{}\mathnormal{∘}\allowbreak{}\mskip 3.0mu\allowbreak{}\mathnormal{(}\mskip 0.0mu\mathsf{a}\mskip 3.0mu\mathnormal{▵}\mskip 3.0mu\mathsf{b}\mskip 0.0mu\mathnormal{)}\allowbreak{}\mskip 0.0mu\mathnormal{)}\allowbreak{}}\<[E]{}\\
\>[1]{}{\allowbreak{}=\allowbreak{}\mskip 0.0mu}\>[3]{}{\quad{}\text{\textit{by properties of free cartesian categories}}}\<[E]{}\\
\>[2]{}{\mathsf{P}\mskip 3.0mu\allowbreak{}\mathnormal{(}\mskip 0.0mu\allowbreak{}\mathnormal{(}\mskip 0.0mu\mathsf{Embed}\mskip 3.0mu\mathsf{φ}\mskip 3.0mu\allowbreak{}\mathnormal{∘}\allowbreak{}\mskip 3.0mu\mathsf{a}\mskip 0.0mu\mathnormal{)}\allowbreak{}\mskip 3.0mu\mathnormal{▵}\mskip 3.0mu\allowbreak{}\mathnormal{(}\mskip 0.0mu\mathsf{Embed}\mskip 3.0mu\mathsf{ψ}\mskip 3.0mu\allowbreak{}\mathnormal{∘}\allowbreak{}\mskip 3.0mu\mathsf{b}\mskip 0.0mu\mathnormal{)}\allowbreak{}\mskip 0.0mu\mathnormal{)}\allowbreak{}}\<[E]{}\\
\>[1]{}{\allowbreak{}=\allowbreak{}\mskip 0.0mu}\>[3]{}{\quad{}\text{\textit{by def}}}\<[E]{}\\
\>[2]{}{\allowbreak{}\mathnormal{(}\mskip 0.0mu\mathsf{encode}\mskip 3.0mu\mathsf{φ}\mskip 3.0mu\allowbreak{}\mathnormal{(}\mskip 0.0mu\mathsf{P}\mskip 3.0mu\mathsf{a}\mskip 0.0mu\mathnormal{)}\allowbreak{}\mskip 2.0mu\mathnormal{\fatsemi }\mskip 3.0mu\mathsf{encode}\mskip 3.0mu\mathsf{ψ}\mskip 3.0mu\allowbreak{}\mathnormal{(}\mskip 0.0mu\mathsf{P}\mskip 3.0mu\mathsf{b}\mskip 0.0mu\mathnormal{)}\allowbreak{}\mskip 0.0mu\mathnormal{)}\allowbreak{}}\<[E]{}\end{parray}}\end{list}\item{}encode/ρ\begin{list}{}{\setlength\leftmargin{1.0em}}\item\relax
\ensuremath{\begin{parray}\column{B}{@{}>{}l<{}@{}}\column[0em]{1}{@{}>{}l<{}@{}}\column[1em]{2}{@{}>{}l<{}@{}}\column{3}{@{}>{}l<{}@{}}\column{E}{@{}>{}l<{}@{}}%
\>[2]{}{\mathsf{encode}\mskip 3.0muρ\mskip 3.0mu\allowbreak{}\mathnormal{(}\mskip 0.0mu\mathsf{P}\mskip 3.0mu\mathsf{a}\mskip 0.0mu\mathnormal{)}\allowbreak{}}\<[E]{}\\
\>[1]{}{\allowbreak{}=\allowbreak{}\mskip 0.0mu}\>[3]{}{\quad{}\text{\textit{by  definition of encode}}}\<[E]{}\\
\>[2]{}{\mathsf{P}\mskip 3.0mu\allowbreak{}\mathnormal{(}\mskip 0.0mu\mathsf{Embed}\mskip 3.0muρ\mskip 3.0mu\allowbreak{}\mathnormal{∘}\allowbreak{}\mskip 3.0mu\mathsf{a}\mskip 0.0mu\mathnormal{)}\allowbreak{}}\<[E]{}\\
\>[1]{}{\allowbreak{}=\allowbreak{}\mskip 0.0mu}\>[3]{}{\quad{}\text{\textit{by  definition of unitor for cartesian categories}}}\<[E]{}\\
\>[2]{}{\mathsf{P}\mskip 3.0mu\allowbreak{}\mathnormal{(}\mskip 0.0mu\allowbreak{}\mathnormal{(}\mskip 0.0mu\mathsf{id}\mskip 3.0mu\mathnormal{▵}\mskip 3.0muε\mskip 0.0mu\mathnormal{)}\allowbreak{}\mskip 3.0mu\allowbreak{}\mathnormal{∘}\allowbreak{}\mskip 3.0mu\mathsf{a}\mskip 0.0mu\mathnormal{)}\allowbreak{}}\<[E]{}\\
\>[1]{}{\allowbreak{}=\allowbreak{}\mskip 0.0mu}\>[3]{}{\quad{}\text{\textit{by  property of ▵}}}\<[E]{}\\
\>[2]{}{\mathsf{P}\mskip 3.0mu\allowbreak{}\mathnormal{(}\mskip 0.0mu\mathsf{a}\mskip 3.0mu\mathnormal{▵}\mskip 3.0mu\allowbreak{}\mathnormal{(}\mskip 0.0muε\mskip 3.0mu\allowbreak{}\mathnormal{∘}\allowbreak{}\mskip 3.0mu\mathsf{a}\mskip 0.0mu\mathnormal{)}\allowbreak{}\mskip 0.0mu\mathnormal{)}\allowbreak{}}\<[E]{}\\
\>[1]{}{\allowbreak{}=\allowbreak{}\mskip 0.0mu}\>[3]{}{\quad{}\text{\textit{by property of ε}}}\<[E]{}\\
\>[2]{}{\mathsf{P}\mskip 3.0mu\allowbreak{}\mathnormal{(}\mskip 0.0mu\mathsf{a}\mskip 3.0mu\mathnormal{▵}\mskip 3.0muε\mskip 0.0mu\mathnormal{)}\allowbreak{}}\<[E]{}\\
\>[1]{}{\allowbreak{}=\allowbreak{}\mskip 0.0mu}\>[3]{}{\quad{}\text{\textit{by definitoin of merge}}}\<[E]{}\\
\>[2]{}{\allowbreak{}\mathnormal{(}\mskip 0.0mu\mathsf{P}\mskip 3.0mu\mathsf{a}\mskip 2.0mu\mathnormal{\fatsemi }\mskip 3.0mu\mathsf{unit}\mskip 0.0mu\mathnormal{)}\allowbreak{}}\<[E]{}\end{parray}}\end{list}\item{}encode/ρ'\begin{list}{}{\setlength\leftmargin{1.0em}}\item\relax
\ensuremath{\begin{parray}\column{B}{@{}>{}l<{}@{}}\column[0em]{1}{@{}>{}l<{}@{}}\column[1em]{2}{@{}>{}l<{}@{}}\column{3}{@{}>{}l<{}@{}}\column{E}{@{}>{}l<{}@{}}%
\>[2]{}{\mathsf{encode}\mskip 3.0mu\bar{ρ}\mskip 3.0mu\allowbreak{}\mathnormal{(}\mskip 0.0mu\mathsf{P}\mskip 3.0mu\mathsf{a}\mskip 2.0mu\mathnormal{\fatsemi }\mskip 3.0mu\mathsf{unit}\mskip 0.0mu\mathnormal{)}\allowbreak{}}\<[E]{}\\
\>[1]{}{\allowbreak{}=\allowbreak{}\mskip 0.0mu}\>[3]{}{\quad{}\text{\textit{by def}}}\<[E]{}\\
\>[2]{}{\mathsf{P}\mskip 3.0mu\allowbreak{}\mathnormal{(}\mskip 0.0mu\mathsf{Embed}\mskip 3.0mu\bar{ρ}\mskip 3.0mu\allowbreak{}\mathnormal{∘}\allowbreak{}\mskip 3.0mu\allowbreak{}\mathnormal{(}\mskip 0.0mu\mathsf{a}\mskip 3.0mu\mathnormal{▵}\mskip 3.0muε\mskip 0.0mu\mathnormal{)}\allowbreak{}\mskip 0.0mu\mathnormal{)}\allowbreak{}}\<[E]{}\\
\>[1]{}{\allowbreak{}=\allowbreak{}\mskip 0.0mu}\>[3]{}{\quad{}\text{\textit{by definition of unitor for cartesian categories}}}\<[E]{}\\
\>[2]{}{\mathsf{P}\mskip 3.0mu\allowbreak{}\mathnormal{(}\mskip 0.0muπ₁\mskip 3.0mu\allowbreak{}\mathnormal{∘}\allowbreak{}\mskip 3.0mu\allowbreak{}\mathnormal{(}\mskip 0.0mu\mathsf{a}\mskip 3.0mu\mathnormal{▵}\mskip 3.0muε\mskip 0.0mu\mathnormal{)}\allowbreak{}\mskip 0.0mu\mathnormal{)}\allowbreak{}}\<[E]{}\\
\>[1]{}{\allowbreak{}=\allowbreak{}\mskip 0.0mu}\>[3]{}{\quad{}\text{\textit{by properties of cartesian categories}}}\<[E]{}\\
\>[2]{}{\mathsf{P}\mskip 3.0mu\mathsf{a}}\<[E]{}\end{parray}}\end{list}\item{}encode/σ\begin{list}{}{\setlength\leftmargin{1.0em}}\item\relax
\ensuremath{\begin{parray}\column{B}{@{}>{}l<{}@{}}\column[0em]{1}{@{}>{}l<{}@{}}\column[1em]{2}{@{}>{}l<{}@{}}\column{3}{@{}>{}l<{}@{}}\column{E}{@{}>{}l<{}@{}}%
\>[2]{}{\mathsf{encode}\mskip 3.0muσ\mskip 3.0mu\allowbreak{}\mathnormal{(}\mskip 0.0mu\mathsf{P}\mskip 3.0mu\mathsf{a}\mskip 2.0mu\mathnormal{\fatsemi }\mskip 3.0mu\mathsf{P}\mskip 3.0mu\mathsf{b}\mskip 0.0mu\mathnormal{)}\allowbreak{}}\<[E]{}\\
\>[1]{}{\allowbreak{}=\allowbreak{}\mskip 0.0mu}\>[3]{}{\quad{}\text{\textit{by def}}}\<[E]{}\\
\>[2]{}{\mathsf{P}\mskip 3.0mu\allowbreak{}\mathnormal{(}\mskip 0.0mu\mathsf{Embed}\mskip 3.0muσ\mskip 3.0mu\allowbreak{}\mathnormal{∘}\allowbreak{}\mskip 3.0mu\allowbreak{}\mathnormal{(}\mskip 0.0mu\mathsf{a}\mskip 3.0mu\mathnormal{▵}\mskip 3.0mu\mathsf{b}\mskip 0.0mu\mathnormal{)}\allowbreak{}\mskip 0.0mu\mathnormal{)}\allowbreak{}}\<[E]{}\\
\>[1]{}{\allowbreak{}=\allowbreak{}\mskip 0.0mu}\>[3]{}{\quad{}\text{\textit{by assumption on Embed}}}\<[E]{}\\
\>[2]{}{\mathsf{P}\mskip 3.0mu\allowbreak{}\mathnormal{(}\mskip 0.0muσ\mskip 3.0mu\allowbreak{}\mathnormal{∘}\allowbreak{}\mskip 3.0mu\allowbreak{}\mathnormal{(}\mskip 0.0mu\mathsf{a}\mskip 3.0mu\mathnormal{▵}\mskip 3.0mu\mathsf{b}\mskip 0.0mu\mathnormal{)}\allowbreak{}\mskip 0.0mu\mathnormal{)}\allowbreak{}}\<[E]{}\\
\>[1]{}{\allowbreak{}=\allowbreak{}\mskip 0.0mu}\>[3]{}{\quad{}\text{\textit{by properties of free cartesian categories}}}\<[E]{}\\
\>[2]{}{\mathsf{P}\mskip 3.0mu\allowbreak{}\mathnormal{(}\mskip 0.0mu\mathsf{b}\mskip 3.0mu\mathnormal{▵}\mskip 3.0mu\mathsf{a}\mskip 0.0mu\mathnormal{)}\allowbreak{}}\<[E]{}\\
\>[1]{}{\allowbreak{}=\allowbreak{}\mskip 0.0mu}\>[3]{}{\quad{}\text{\textit{by def}}}\<[E]{}\\
\>[2]{}{\allowbreak{}\mathnormal{(}\mskip 0.0mu\mathsf{P}\mskip 3.0mu\mathsf{b}\mskip 2.0mu\mathnormal{\fatsemi }\mskip 3.0mu\mathsf{P}\mskip 3.0mu\mathsf{a}\mskip 0.0mu\mathnormal{)}\allowbreak{}}\<[E]{}\end{parray}}\end{list}\item{}encode/α\begin{list}{}{\setlength\leftmargin{1.0em}}\item\relax
\ensuremath{\begin{parray}\column{B}{@{}>{}l<{}@{}}\column[0em]{1}{@{}>{}l<{}@{}}\column[1em]{2}{@{}>{}l<{}@{}}\column[2em]{3}{@{}>{}l<{}@{}}\column{4}{@{}>{}l<{}@{}}\column{5}{@{}>{}l<{}@{}}\column{E}{@{}>{}l<{}@{}}%
\>[2]{}{\mathsf{encode}\mskip 3.0muα\mskip 3.0mu\allowbreak{}\mathnormal{(}\mskip 0.0mu\allowbreak{}\mathnormal{(}\mskip 0.0mu\mathsf{P}\mskip 3.0mu\mathsf{a}\mskip 2.0mu\mathnormal{\fatsemi }\mskip 3.0mu\mathsf{P}\mskip 3.0mu\mathsf{b}\mskip 0.0mu\mathnormal{)}\allowbreak{}\mskip 2.0mu\mathnormal{\fatsemi }\mskip 3.0mu\mathsf{P}\mskip 3.0mu\mathsf{c}\mskip 0.0mu\mathnormal{)}\allowbreak{}}\<[E]{}\\
\>[1]{}{\allowbreak{}=\allowbreak{}\mskip 0.0mu}\>[5]{}{\quad{}\text{\textit{by def}}}\<[E]{}\\
\>[3]{}{\mathsf{P}\mskip 3.0mu\allowbreak{}\mathnormal{(}\mskip 0.0mu\mathsf{Embed}\mskip 3.0muα\mskip 3.0mu\allowbreak{}\mathnormal{∘}\allowbreak{}\mskip 3.0mu\allowbreak{}\mathnormal{(}\mskip 0.0mu\allowbreak{}\mathnormal{(}\mskip 0.0mu\mathsf{a}\mskip 3.0mu\mathnormal{▵}\mskip 3.0mu\mathsf{b}\mskip 0.0mu\mathnormal{)}\allowbreak{}\mskip 3.0mu\mathnormal{▵}\mskip 3.0mu\mathsf{c}\mskip 0.0mu\mathnormal{)}\allowbreak{}\mskip 0.0mu\mathnormal{)}\allowbreak{}}\<[E]{}\\
\>[1]{}{\allowbreak{}=\allowbreak{}\mskip 0.0mu}\>[4]{}{\quad{}\text{\textit{by assumption on Embed}}}\<[E]{}\\
\>[3]{}{\mathsf{P}\mskip 3.0mu\allowbreak{}\mathnormal{(}\mskip 0.0muα\mskip 3.0mu\allowbreak{}\mathnormal{∘}\allowbreak{}\mskip 3.0mu\allowbreak{}\mathnormal{(}\mskip 0.0mu\allowbreak{}\mathnormal{(}\mskip 0.0mu\mathsf{a}\mskip 3.0mu\mathnormal{▵}\mskip 3.0mu\mathsf{b}\mskip 0.0mu\mathnormal{)}\allowbreak{}\mskip 3.0mu\mathnormal{▵}\mskip 3.0mu\mathsf{c}\mskip 0.0mu\mathnormal{)}\allowbreak{}\mskip 0.0mu\mathnormal{)}\allowbreak{}}\<[E]{}\\
\>[1]{}{\allowbreak{}=\allowbreak{}\mskip 0.0mu}\>[4]{}{\quad{}\text{\textit{by properties of free cartesian categories}}}\<[E]{}\\
\>[3]{}{\mathsf{P}\mskip 3.0mu\allowbreak{}\mathnormal{(}\mskip 0.0mu\mathsf{a}\mskip 3.0mu\mathnormal{▵}\mskip 3.0mu\allowbreak{}\mathnormal{(}\mskip 0.0mu\mathsf{b}\mskip 3.0mu\mathnormal{▵}\mskip 3.0mu\mathsf{c}\mskip 0.0mu\mathnormal{)}\allowbreak{}\mskip 0.0mu\mathnormal{)}\allowbreak{}}\<[E]{}\\
\>[1]{}{\allowbreak{}=\allowbreak{}\mskip 0.0mu}\>[5]{}{\quad{}\text{\textit{by def}}}\<[E]{}\\
\>[2]{}{\allowbreak{}\mathnormal{(}\mskip 0.0mu\mathsf{P}\mskip 3.0mu\mathsf{a}\mskip 2.0mu\mathnormal{\fatsemi }\mskip 3.0mu\allowbreak{}\mathnormal{(}\mskip 0.0mu\mathsf{P}\mskip 3.0mu\mathsf{b}\mskip 2.0mu\mathnormal{\fatsemi }\mskip 3.0mu\mathsf{P}\mskip 3.0mu\mathsf{c}\mskip 0.0mu\mathnormal{)}\allowbreak{}\mskip 0.0mu\mathnormal{)}\allowbreak{}}\<[E]{}\end{parray}}\end{list}\item{}encode/α'\begin{list}{}{\setlength\leftmargin{1.0em}}\item\relax
\ensuremath{\begin{parray}\column{B}{@{}>{}l<{}@{}}\column[0em]{1}{@{}>{}l<{}@{}}\column[1em]{2}{@{}>{}l<{}@{}}\column[2em]{3}{@{}>{}l<{}@{}}\column{4}{@{}>{}l<{}@{}}\column{5}{@{}>{}l<{}@{}}\column{E}{@{}>{}l<{}@{}}%
\>[2]{}{\mathsf{encode}\mskip 3.0mu\bar{α}\mskip 3.0mu\allowbreak{}\mathnormal{(}\mskip 0.0mu\mathsf{P}\mskip 3.0mu\mathsf{a}\mskip 2.0mu\mathnormal{\fatsemi }\mskip 3.0mu\allowbreak{}\mathnormal{(}\mskip 0.0mu\mathsf{P}\mskip 3.0mu\mathsf{b}\mskip 2.0mu\mathnormal{\fatsemi }\mskip 3.0mu\mathsf{P}\mskip 3.0mu\mathsf{c}\mskip 0.0mu\mathnormal{)}\allowbreak{}\mskip 0.0mu\mathnormal{)}\allowbreak{}}\<[E]{}\\
\>[1]{}{\allowbreak{}=\allowbreak{}\mskip 0.0mu}\>[5]{}{\quad{}\text{\textit{by def}}}\<[E]{}\\
\>[3]{}{\mathsf{P}\mskip 3.0mu\allowbreak{}\mathnormal{(}\mskip 0.0mu\mathsf{Embed}\mskip 3.0mu\bar{α}\mskip 3.0mu\allowbreak{}\mathnormal{∘}\allowbreak{}\mskip 3.0mu\allowbreak{}\mathnormal{(}\mskip 0.0mu\mathsf{a}\mskip 3.0mu\mathnormal{▵}\mskip 3.0mu\allowbreak{}\mathnormal{(}\mskip 0.0mu\mathsf{b}\mskip 3.0mu\mathnormal{▵}\mskip 3.0mu\mathsf{c}\mskip 0.0mu\mathnormal{)}\allowbreak{}\mskip 0.0mu\mathnormal{)}\allowbreak{}\mskip 0.0mu\mathnormal{)}\allowbreak{}}\<[E]{}\\
\>[1]{}{\allowbreak{}=\allowbreak{}\mskip 0.0mu}\>[4]{}{\quad{}\text{\textit{by assumption on Embed}}}\<[E]{}\\
\>[3]{}{\mathsf{P}\mskip 3.0mu\allowbreak{}\mathnormal{(}\mskip 0.0mu\bar{α}\mskip 3.0mu\allowbreak{}\mathnormal{∘}\allowbreak{}\mskip 3.0mu\allowbreak{}\mathnormal{(}\mskip 0.0mu\mathsf{a}\mskip 3.0mu\mathnormal{▵}\mskip 3.0mu\allowbreak{}\mathnormal{(}\mskip 0.0mu\mathsf{b}\mskip 3.0mu\mathnormal{▵}\mskip 3.0mu\mathsf{c}\mskip 0.0mu\mathnormal{)}\allowbreak{}\mskip 0.0mu\mathnormal{)}\allowbreak{}\mskip 0.0mu\mathnormal{)}\allowbreak{}}\<[E]{}\\
\>[1]{}{\allowbreak{}=\allowbreak{}\mskip 0.0mu}\>[4]{}{\quad{}\text{\textit{by properties of free cartesian categories}}}\<[E]{}\\
\>[3]{}{\mathsf{P}\mskip 3.0mu\allowbreak{}\mathnormal{(}\mskip 0.0mu\allowbreak{}\mathnormal{(}\mskip 0.0mu\mathsf{a}\mskip 3.0mu\mathnormal{▵}\mskip 3.0mu\mathsf{b}\mskip 0.0mu\mathnormal{)}\allowbreak{}\mskip 3.0mu\mathnormal{▵}\mskip 3.0mu\mathsf{c}\mskip 0.0mu\mathnormal{)}\allowbreak{}}\<[E]{}\\
\>[1]{}{\allowbreak{}=\allowbreak{}\mskip 0.0mu}\>[5]{}{\quad{}\text{\textit{by def}}}\<[E]{}\\
\>[2]{}{\allowbreak{}\mathnormal{(}\mskip 0.0mu\allowbreak{}\mathnormal{(}\mskip 0.0mu\mathsf{P}\mskip 3.0mu\mathsf{a}\mskip 2.0mu\mathnormal{\fatsemi }\mskip 3.0mu\mathsf{P}\mskip 3.0mu\mathsf{b}\mskip 0.0mu\mathnormal{)}\allowbreak{}\mskip 2.0mu\mathnormal{\fatsemi }\mskip 3.0mu\mathsf{P}\mskip 3.0mu\mathsf{c}\mskip 0.0mu\mathnormal{)}\allowbreak{}}\<[E]{}\end{parray}}\end{list}\end{itemize}\end{proof}\end{mdframed} 

\subsection{Characterisation of the domain of \ensuremath{\mathsf{reduce}}.}\label{66} 
As mentioned previously, the bulk of the work is to
define (and prove correct) the \ensuremath{\mathsf{reduce}} function, which converts
a \ensuremath{\mathsf{FreeCartesian}} representation into a \ensuremath{\mathsf{FreeSMC}}.  This
\ensuremath{\mathsf{reduce}} function is partial: if its input is not suitable (say if
an input is discarded) then there is no {\sc{}smc} representation. Fortunately,
we only need to deal with representations which have been constructed
using the port interface, namely linear functions built with \ensuremath{\mathsf{encode}},
\ensuremath{\mathsf{merge}}, \ensuremath{\mathsf{split}} and \ensuremath{\mathsf{unit}}.
Our plan is then to
1. prove that the \ensuremath{\mathsf{extract}}ed morphisms are indeed reducible to
the {\sc{}smc} interface, and 2. show how to carry it out algorithmically.
We start by addressing the first problem, and this will put us firmly
on track to address the second one.

\begin{definition}{}A representation \ensuremath{\mathsf{f}\mskip 3.0mu}\ensuremath{\mathnormal{:}\mskip 3.0mu}\ensuremath{\mathsf{FreeCartesian}\mskip 3.0mu}\ensuremath{\mathsf{k}\mskip 3.0mu}\ensuremath{\mathsf{a}\mskip 3.0mu}\ensuremath{\mathsf{b}} is called linear if it
is it defined using only the {\sc{}smc} subset of the cartesian structure.\label{67}\end{definition} 
\begin{definition}{}A representation \ensuremath{\mathsf{f}\mskip 3.0mu}\ensuremath{\mathnormal{:}\mskip 3.0mu}\ensuremath{\mathsf{FreeCartesian}\mskip 3.0mu}\ensuremath{\mathsf{k}\mskip 3.0mu}\ensuremath{\mathsf{a}\mskip 3.0mu}\ensuremath{\mathsf{b}} is called \emph{protolinear} iff it is equivalent, according to the
laws of a cartesian category, to a linear representation \ensuremath{\mathsf{h}}.\label{68}\end{definition} 
\begin{theorem}{}For every function \ensuremath{\mathsf{h}\mskip 3.0mu}\ensuremath{\mathnormal{:}\mskip 3.0mu}\ensuremath{\mathnormal{∀}\mskip 3.0mu}\ensuremath{\mathsf{r}\mskip 1.0mu}\ensuremath{.\mskip 3.0mu}\ensuremath{\mathsf{P}\mskip 3.0mu}\ensuremath{\mathsf{k}\mskip 3.0mu}\ensuremath{\mathsf{r}\mskip 3.0mu}\ensuremath{\mathsf{a}\mskip 3.0mu}\ensuremath{\mathnormal{⊸}\mskip 3.0mu}\ensuremath{\mathsf{P}\mskip 3.0mu}\ensuremath{\mathsf{k}\mskip 3.0mu}\ensuremath{\mathsf{r}\mskip 3.0mu}\ensuremath{\mathsf{b}}, the morphism \ensuremath{\mathsf{extract}\mskip 3.0mu}\ensuremath{\mathsf{h}} is a
 protolinear representation.\label{69}\end{theorem}\begin{proof} 
The idea of the proof is to do an induction on the structure of
\ensuremath{\mathsf{h}}. But in general a computational prefix \ensuremath{\mathsf{f}} of \ensuremath{\mathsf{h}} has
several outputs. That is, the type of \ensuremath{\mathsf{f}} has the form \ensuremath{\mathsf{P}\mskip 3.0mu}\ensuremath{\mathsf{k}\mskip 3.0mu}\ensuremath{\mathsf{r}\mskip 3.0mu}\ensuremath{\mathsf{a}\mskip 3.0mu}\ensuremath{\mathnormal{⊸}\mskip 0.0mu}\ensuremath{\bigotimes_{i}\mskip 3.0mu}\ensuremath{\allowbreak{}\mathnormal{(}\mskip 0.0mu}\ensuremath{\mathsf{P}\mskip 3.0mu}\ensuremath{\mathsf{k}\mskip 3.0mu}\ensuremath{\mathsf{r}\mskip 3.0mu}\ensuremath{\mathsf{t}_{i}\mskip 0.0mu}\ensuremath{\mathnormal{)}\allowbreak{}}, where where the big circled product
operator is a multary version of the monoidal product with right
associativity. The components of such products represent
ports which are available after the prefix \ensuremath{\mathsf{f}} is run (but \ensuremath{\mathsf{h}} is not complete).
Thus, to obtain a protolinear function from \ensuremath{\mathsf{f}},
its outputs must be merged, by a generalised fork (▵) function,
written \ensuremath{\bigtriangleup}, and defined as follows:
\begin{list}{}{\setlength\leftmargin{1.0em}}\item\relax
\ensuremath{\begin{parray}\column{B}{@{}>{}l<{}@{}}\column[0em]{1}{@{}>{}l<{}@{}}\column{E}{@{}>{}l<{}@{}}%
\>[1]{}{\bigtriangleup\mskip 3.0mu\mathnormal{::}\mskip 0.0mu\bigotimes_{i}\mskip 3.0mu\allowbreak{}\mathnormal{(}\mskip 0.0mu\mathsf{P}\mskip 3.0mu\mathsf{k}\mskip 3.0mu\mathsf{a}\mskip 3.0mu\mathsf{t}_{i}\mskip 0.0mu\mathnormal{)}\allowbreak{}\mskip 3.0mu\mathnormal{\rightarrow }\mskip 3.0mu\mathsf{P}\mskip 3.0mu\mathsf{k}\mskip 3.0mu\mathsf{a}\mskip 3.0mu\allowbreak{}\mathnormal{(}\mskip 0.0mu\bigotimes_{i}\mskip 3.0mu\mathsf{t}_{i}\mskip 0.0mu\mathnormal{)}\allowbreak{}}\<[E]{}\\
\>[1]{}{\bigtriangleup\mskip 3.0mu\allowbreak{}\mathnormal{(}\mskip 0.0mu\mathsf{f₁}\mskip 0.0mu\mathnormal{,}\mskip 3.0mu\allowbreak{}\mathnormal{(}\mskip 0.0mu\mathnormal{…}\mskip 0.0mu\mathnormal{,}\mskip 3.0mu\mathsf{fₙ}\mskip 0.0mu\mathnormal{)}\allowbreak{}\mskip 0.0mu\mathnormal{)}\allowbreak{}\mskip 3.0mu\mathnormal{=}\mskip 3.0mu\allowbreak{}\mathnormal{(}\mskip 0.0mu\mathsf{f₁}\mskip 3.0mu\mathnormal{▵}\mskip 3.0mu\allowbreak{}\mathnormal{(}\mskip 0.0mu\mathnormal{…}\mskip 3.0mu\mathnormal{▵}\mskip 3.0mu\mathsf{fₙ}\mskip 0.0mu\mathnormal{)}\allowbreak{}\mskip 0.0mu\mathnormal{)}\allowbreak{}}\<[E]{}\end{parray}}\end{list} 
When there is a single output port, \ensuremath{\bigtriangleup} is the identity, and
thus this theorem is a corollary of the generalised form,
\cref{70}, for a product with one element.
\end{proof} 

\begin{lemma}{}If \ensuremath{\mathsf{f}\mskip 3.0mu}\ensuremath{\mathnormal{:}\mskip 3.0mu}\ensuremath{\mathnormal{∀}\mskip 3.0mu}\ensuremath{\mathsf{r}\mskip 1.0mu}\ensuremath{.\mskip 3.0mu}\ensuremath{\mathsf{P}\mskip 3.0mu}\ensuremath{\mathsf{k}\mskip 3.0mu}\ensuremath{\mathsf{r}\mskip 3.0mu}\ensuremath{\mathsf{a}\mskip 3.0mu}\ensuremath{\mathnormal{⊸}\mskip 0.0mu}\ensuremath{\bigotimes_{i}\mskip 3.0mu}\ensuremath{\allowbreak{}\mathnormal{(}\mskip 0.0mu}\ensuremath{\mathsf{P}\mskip 3.0mu}\ensuremath{\mathsf{k}\mskip 3.0mu}\ensuremath{\mathsf{r}\mskip 3.0mu}\ensuremath{\mathsf{t}_{i}\mskip 0.0mu}\ensuremath{\mathnormal{)}\allowbreak{}}, then \ensuremath{\bigtriangleup\mskip 3.0mu}\ensuremath{\allowbreak{}\mathnormal{(}\mskip 0.0mu}\ensuremath{\mathsf{f}\mskip 3.0mu}\ensuremath{\mathsf{id}\mskip 0.0mu}\ensuremath{\mathnormal{)}\allowbreak{}} is a protolinear
 representation.\label{70}\end{lemma}\begin{proof} First, we need to choose a convenient representation of the function
\ensuremath{\mathsf{f}} itself. A first idea could be to use the term representation of
Haskell. This would however make for a tedious proof, and to fit our
theme, we use a categorical representation for Haskell functions as well.
For this purpose, we make the
simplifying assumption that functions of the type \ensuremath{∀\mskip 3.0mu}\ensuremath{\mathsf{r}\mskip 1.0mu}\ensuremath{.\mskip 3.0mu}\ensuremath{\mathsf{P}\mskip 3.0mu}\ensuremath{\mathsf{k}\mskip 3.0mu}\ensuremath{\mathsf{r}\mskip 3.0mu}\ensuremath{\mathsf{a}\mskip 3.0mu}\ensuremath{\mathnormal{⊸}\mskip 3.0mu}\ensuremath{\mathsf{P}\mskip 3.0mu}\ensuremath{\mathsf{k}\mskip 3.0mu}\ensuremath{\mathsf{r}\mskip 3.0mu}\ensuremath{\mathsf{b}} can be themselves represented as morphisms in another free {\sc{}smc}, the
category of linear functions of Haskell.\footnote{To be fair, this property would only be true of
an idealised language with linear types (\cref{32}). For
an actual programming language, exceptions, non-termination,
etc. should be taken into account. In practice, if the function of
type \ensuremath{\mathnormal{∀}\mskip 3.0mu}\ensuremath{\mathsf{r}\mskip 1.0mu}\ensuremath{.\mskip 3.0mu}\ensuremath{\mathsf{P}\mskip 3.0mu}\ensuremath{\mathsf{k}\mskip 3.0mu}\ensuremath{\mathsf{r}\mskip 3.0mu}\ensuremath{\mathsf{a}\mskip 3.0mu}\ensuremath{\mathnormal{⊸}\mskip 3.0mu}\ensuremath{\mathsf{P}\mskip 3.0mu}\ensuremath{\mathsf{k}\mskip 3.0mu}\ensuremath{\mathsf{r}\mskip 3.0mu}\ensuremath{\mathsf{b}} diverges, the \ensuremath{\mathsf{reduce}} function
also diverges. This means that we are limited to finite quantum circuits or workflows.} 
Additionally, because the type \ensuremath{\mathsf{P}\mskip 3.0mu}\ensuremath{\mathsf{k}\mskip 3.0mu}\ensuremath{\mathsf{a}\mskip 3.0mu}\ensuremath{\mathsf{b}} is abstract, we know
that the only possible generators for this {\sc{}smc} are the primitives
\ensuremath{\mathsf{unit}}, \ensuremath{\mathsf{split}}, \ensuremath{\mathsf{merge}} and \ensuremath{\mathsf{encode}}:
we can assume that other constructions are reduced away by the
Haskell evaluator.

Furthermore, this representation can be assumed without loss of
generality to take the form of a composition \ensuremath{\mathsf{s}_{1}\mskip 3.0mu}\ensuremath{\allowbreak{}\mathnormal{∘}\allowbreak{}\mskip 3.0mu}\ensuremath{\mathnormal{⋯}\mskip 3.0mu}\ensuremath{\allowbreak{}\mathnormal{∘}\allowbreak{}\mskip 3.0mu}\ensuremath{\mathsf{sₙ}}. (This
corresponds to cutting the corresponding diagram in vertical slices \ensuremath{\mathsf{sᵢ}},
each containing a single generator. By topology-preserving
transformations, it is always possible to move generators so that they
fall in separate slices.) 

In fact, without loss of generality, we assume that each slice
\ensuremath{\mathsf{s}} has either of the following forms: 1. \ensuremath{\mathsf{encode}\mskip 3.0mu}\ensuremath{\mathsf{φ}\mskip 3.0mu}\ensuremath{\allowbreak{}\mathnormal{×}\allowbreak{}\mskip 3.0mu}\ensuremath{\mathsf{id}} 2. \ensuremath{α\mskip 3.0mu}\ensuremath{\allowbreak{}\mathnormal{∘}\allowbreak{}\mskip 3.0mu}\ensuremath{\allowbreak{}\mathnormal{(}\mskip 0.0mu}\ensuremath{\mathsf{split}\mskip 3.0mu}\ensuremath{\allowbreak{}\mathnormal{×}\allowbreak{}\mskip 3.0mu}\ensuremath{\mathsf{id}\mskip 0.0mu}\ensuremath{\mathnormal{)}\allowbreak{}} 3. \ensuremath{\allowbreak{}\mathnormal{(}\mskip 0.0mu}\ensuremath{\mathsf{merge}\mskip 3.0mu}\ensuremath{\allowbreak{}\mathnormal{×}\allowbreak{}\mskip 3.0mu}\ensuremath{\mathsf{id}\mskip 0.0mu}\ensuremath{\mathnormal{)}\allowbreak{}\mskip 3.0mu}\ensuremath{\allowbreak{}\mathnormal{∘}\allowbreak{}\mskip 3.0mu}\ensuremath{\bar{α}} 4. \ensuremath{λ\mskip 3.0mu}\ensuremath{\mathsf{x}\mskip 3.0mu}\ensuremath{\mathnormal{\rightarrow }\mskip 3.0mu}\ensuremath{\allowbreak{}\mathnormal{(}\mskip 0.0mu}\ensuremath{\mathsf{unit}\mskip 0.0mu}\ensuremath{\mathnormal{,}\mskip 3.0mu}\ensuremath{\mathsf{x}\mskip 0.0mu}\ensuremath{\mathnormal{)}\allowbreak{}}. That is, we assume that the generators act on the
first component of the slice. (The split and merge cases are composed
with associators to preserve the property that the multary monoidal
products on the input and output are right-associated.) We can make this
assumption because we treat permutations over the monoidal product as
separate slices (Of a separate form, referred to as 5. below). Such a slice does not contain any generator;
rather its role is to stage the next variable(s) to be acted upon by
the next generator.

We can now proceed with the induction.  The base case reduces to
protolinearity of \ensuremath{\mathsf{id}}, which is obvious.  For the induction case,
we assume that \ensuremath{\bigtriangleup\mskip 3.0mu}\ensuremath{\allowbreak{}\mathnormal{(}\mskip 0.0mu}\ensuremath{\mathsf{f}\mskip 3.0mu}\ensuremath{\mathsf{id}\mskip 0.0mu}\ensuremath{\mathnormal{)}\allowbreak{}} is protolinear, and show that so is
\ensuremath{\bigtriangleup\mskip 3.0mu}\ensuremath{\allowbreak{}\mathnormal{(}\mskip 0.0mu}\ensuremath{\allowbreak{}\mathnormal{(}\mskip 0.0mu}\ensuremath{\mathsf{s}\mskip 3.0mu}\ensuremath{\allowbreak{}\mathnormal{∘}\allowbreak{}\mskip 3.0mu}\ensuremath{\mathsf{f}\mskip 0.0mu}\ensuremath{\mathnormal{)}\allowbreak{}\mskip 3.0mu}\ensuremath{\mathsf{id}\mskip 0.0mu}\ensuremath{\mathnormal{)}\allowbreak{}}, for every function \ensuremath{\mathsf{f}} of type
\ensuremath{∀\mskip 3.0mu}\ensuremath{\mathsf{r}\mskip 1.0mu}\ensuremath{.\mskip 3.0mu}\ensuremath{\mathsf{P}\mskip 3.0mu}\ensuremath{\mathsf{k}\mskip 3.0mu}\ensuremath{\mathsf{r}\mskip 3.0mu}\ensuremath{\mathsf{a}\mskip 3.0mu}\ensuremath{\mathnormal{⊸}\mskip 0.0mu}\ensuremath{\bigotimes_{i}\mskip 3.0mu}\ensuremath{\allowbreak{}\mathnormal{(}\mskip 0.0mu}\ensuremath{\mathsf{P}\mskip 3.0mu}\ensuremath{\mathsf{k}\mskip 3.0mu}\ensuremath{\mathsf{r}\mskip 3.0mu}\ensuremath{\mathsf{t}_{i}\mskip 0.0mu}\ensuremath{\mathnormal{)}\allowbreak{}}, and every possible slice
\ensuremath{\mathsf{s}}.

Let us calculate a reduced form for \ensuremath{\allowbreak{}\mathnormal{(}\mskip 0.0mu}\ensuremath{\mathsf{s}\mskip 3.0mu}\ensuremath{\allowbreak{}\mathnormal{∘}\allowbreak{}\mskip 3.0mu}\ensuremath{\mathsf{f}\mskip 0.0mu}\ensuremath{\mathnormal{)}\allowbreak{}\mskip 3.0mu}\ensuremath{\mathsf{id}} for each case:
\begin{mdframed}[linewidth=0pt,hidealllines,innerleftmargin=0pt,innerrightmargin=0pt,backgroundcolor=gray!15]\begin{itemize}\item{}Let \ensuremath{\mathsf{g}\mskip 3.0mu}\ensuremath{\mathnormal{=}\mskip 3.0mu}\ensuremath{\mathsf{encode}\mskip 3.0mu}\ensuremath{\mathsf{φ}\mskip 3.0mu}\ensuremath{\allowbreak{}\mathnormal{×}\allowbreak{}\mskip 3.0mu}\ensuremath{\mathsf{id}}.
  
 \begin{list}{}{\setlength\leftmargin{1.0em}}\item\relax
\ensuremath{\begin{parray}\column{B}{@{}>{}l<{}@{}}\column[0em]{1}{@{}>{}l<{}@{}}\column[1em]{2}{@{}>{}l<{}@{}}\column{3}{@{}>{}l<{}@{}}\column{E}{@{}>{}l<{}@{}}%
\>[2]{}{\bigtriangleup\mskip 3.0mu\allowbreak{}\mathnormal{(}\mskip 0.0mu\allowbreak{}\mathnormal{(}\mskip 0.0mu\allowbreak{}\mathnormal{(}\mskip 0.0mu\mathsf{encode}\mskip 3.0mu\mathsf{φ}\mskip 3.0mu\allowbreak{}\mathnormal{×}\allowbreak{}\mskip 3.0mu\mathsf{id}\mskip 0.0mu\mathnormal{)}\allowbreak{}\mskip 3.0mu\allowbreak{}\mathnormal{∘}\allowbreak{}\mskip 3.0mu\mathsf{f}\mskip 0.0mu\mathnormal{)}\allowbreak{}\mskip 3.0mu\mathsf{id}\mskip 0.0mu\mathnormal{)}\allowbreak{}}\<[E]{}\\
\>[1]{}{\allowbreak{}=\allowbreak{}\mskip 0.0mu}\>[3]{}{\quad{}\text{\textit{by def of ∘}}}\<[E]{}\\
\>[2]{}{\bigtriangleup\mskip 3.0mu\allowbreak{}\mathnormal{(}\mskip 0.0mu\allowbreak{}\mathnormal{(}\mskip 0.0mu\mathsf{encode}\mskip 3.0mu\mathsf{φ}\mskip 3.0mu\allowbreak{}\mathnormal{×}\allowbreak{}\mskip 3.0mu\mathsf{id}\mskip 0.0mu\mathnormal{)}\allowbreak{}\mskip 3.0mu\allowbreak{}\mathnormal{(}\mskip 0.0mu\mathsf{f}\mskip 3.0mu\mathsf{id}\mskip 0.0mu\mathnormal{)}\allowbreak{}\mskip 0.0mu\mathnormal{)}\allowbreak{}}\<[E]{}\\
\>[1]{}{\allowbreak{}=\allowbreak{}\mskip 0.0mu}\>[3]{}{\quad{}\text{\textit{by expansion of pairs}}}\<[E]{}\\
\>[2]{}{\bigtriangleup\mskip 3.0mu\allowbreak{}\mathnormal{(}\mskip 0.0mu\allowbreak{}\mathnormal{(}\mskip 0.0mu\mathsf{encode}\mskip 3.0mu\mathsf{φ}\mskip 3.0mu\allowbreak{}\mathnormal{×}\allowbreak{}\mskip 3.0mu\mathsf{id}\mskip 0.0mu\mathnormal{)}\allowbreak{}\mskip 3.0mu\allowbreak{}\mathnormal{(}\mskip 0.0muπ₁\mskip 3.0mu\allowbreak{}\mathnormal{(}\mskip 0.0mu\mathsf{f}\mskip 3.0mu\mathsf{id}\mskip 0.0mu\mathnormal{)}\allowbreak{}\mskip 0.0mu\mathnormal{,}\mskip 3.0muπ₂\mskip 3.0mu\allowbreak{}\mathnormal{(}\mskip 0.0mu\mathsf{f}\mskip 3.0mu\mathsf{id}\mskip 0.0mu\mathnormal{)}\allowbreak{}\mskip 0.0mu\mathnormal{)}\allowbreak{}\mskip 0.0mu\mathnormal{)}\allowbreak{}}\<[E]{}\\
\>[1]{}{\allowbreak{}=\allowbreak{}\mskip 0.0mu}\>[3]{}{\quad{}\text{\textit{by definition of ×}}}\<[E]{}\\
\>[2]{}{\bigtriangleup\mskip 3.0mu\allowbreak{}\mathnormal{(}\mskip 0.0mu\mathsf{encode}\mskip 3.0mu\mathsf{φ}\mskip 3.0mu\allowbreak{}\mathnormal{(}\mskip 0.0muπ₁\mskip 3.0mu\allowbreak{}\mathnormal{(}\mskip 0.0mu\mathsf{f}\mskip 3.0mu\mathsf{id}\mskip 0.0mu\mathnormal{)}\allowbreak{}\mskip 0.0mu\mathnormal{)}\allowbreak{}\mskip 0.0mu\mathnormal{,}\mskip 3.0muπ₂\mskip 3.0mu\allowbreak{}\mathnormal{(}\mskip 0.0mu\mathsf{f}\mskip 3.0mu\mathsf{id}\mskip 0.0mu\mathnormal{)}\allowbreak{}\mskip 0.0mu\mathnormal{)}\allowbreak{}}\<[E]{}\\
\>[1]{}{\allowbreak{}=\allowbreak{}\mskip 0.0mu}\>[3]{}{\quad{}\text{\textit{by definition of mergeA}}}\<[E]{}\\
\>[2]{}{\allowbreak{}\mathnormal{(}\mskip 0.0mu\mathsf{encode}\mskip 3.0mu\mathsf{φ}\mskip 3.0mu\allowbreak{}\mathnormal{(}\mskip 0.0muπ₁\mskip 3.0mu\allowbreak{}\mathnormal{(}\mskip 0.0mu\mathsf{f}\mskip 3.0mu\mathsf{id}\mskip 0.0mu\mathnormal{)}\allowbreak{}\mskip 0.0mu\mathnormal{)}\allowbreak{}\mskip 3.0mu\mathnormal{▵}\mskip 3.0mu\allowbreak{}\mathnormal{(}\mskip 0.0mu\bigtriangleup\mskip 3.0mu\allowbreak{}\mathnormal{(}\mskip 0.0muπ₂\mskip 3.0mu\allowbreak{}\mathnormal{(}\mskip 0.0mu\mathsf{f}\mskip 3.0mu\mathsf{id}\mskip 0.0mu\mathnormal{)}\allowbreak{}\mskip 0.0mu\mathnormal{)}\allowbreak{}\mskip 0.0mu\mathnormal{)}\allowbreak{}\mskip 0.0mu\mathnormal{)}\allowbreak{}}\<[E]{}\\
\>[1]{}{\allowbreak{}=\allowbreak{}\mskip 0.0mu}\>[3]{}{\quad{}\text{\textit{by def of encode}}}\<[E]{}\\
\>[2]{}{\mathsf{P}\mskip 3.0mu\allowbreak{}\mathnormal{(}\mskip 0.0mu\mathsf{Embed}\mskip 3.0mu\mathsf{φ}\mskip 3.0mu\allowbreak{}\mathnormal{∘}\allowbreak{}\mskip 3.0mu\mathsf{fromP}\mskip 3.0mu\allowbreak{}\mathnormal{(}\mskip 0.0muπ₁\mskip 3.0mu\allowbreak{}\mathnormal{(}\mskip 0.0mu\mathsf{f}\mskip 3.0mu\mathsf{id}\mskip 0.0mu\mathnormal{)}\allowbreak{}\mskip 0.0mu\mathnormal{)}\allowbreak{}\mskip 0.0mu\mathnormal{)}\allowbreak{}\mskip 3.0mu\mathnormal{▵}\mskip 3.0mu\bigtriangleup\mskip 3.0mu\allowbreak{}\mathnormal{(}\mskip 0.0muπ₂\mskip 3.0mu\allowbreak{}\mathnormal{(}\mskip 0.0mu\mathsf{f}\mskip 3.0mu\mathsf{id}\mskip 0.0mu\mathnormal{)}\allowbreak{}\mskip 0.0mu\mathnormal{)}\allowbreak{}}\<[E]{}\\
\>[1]{}{\allowbreak{}=\allowbreak{}\mskip 0.0mu}\>[3]{}{\quad{}\text{\textit{by property of ×/▵}}}\<[E]{}\\
\>[2]{}{\mathsf{P}\mskip 3.0mu\allowbreak{}\mathnormal{(}\mskip 0.0mu\mathsf{Embed}\mskip 3.0mu\mathsf{φ}\mskip 3.0mu\allowbreak{}\mathnormal{×}\allowbreak{}\mskip 3.0mu\mathsf{id}\mskip 0.0mu\mathnormal{)}\allowbreak{}\mskip 3.0mu\allowbreak{}\mathnormal{∘}\allowbreak{}\mskip 3.0mu\allowbreak{}\mathnormal{(}\mskip 0.0muπ₁\mskip 3.0mu\allowbreak{}\mathnormal{(}\mskip 0.0mu\mathsf{f}\mskip 3.0mu\mathsf{id}\mskip 0.0mu\mathnormal{)}\allowbreak{}\mskip 3.0mu\mathnormal{▵}\mskip 3.0mu\bigtriangleup\mskip 3.0mu\allowbreak{}\mathnormal{(}\mskip 0.0muπ₂\mskip 3.0mu\allowbreak{}\mathnormal{(}\mskip 0.0mu\mathsf{f}\mskip 3.0mu\mathsf{id}\mskip 0.0mu\mathnormal{)}\allowbreak{}\mskip 0.0mu\mathnormal{)}\allowbreak{}\mskip 0.0mu\mathnormal{)}\allowbreak{}}\<[E]{}\\
\>[1]{}{\allowbreak{}=\allowbreak{}\mskip 0.0mu}\>[3]{}{\quad{}\text{\textit{by definition of mergeA}}}\<[E]{}\\
\>[2]{}{\mathsf{P}\mskip 3.0mu\allowbreak{}\mathnormal{(}\mskip 0.0mu\mathsf{Embed}\mskip 3.0mu\mathsf{φ}\mskip 3.0mu\allowbreak{}\mathnormal{×}\allowbreak{}\mskip 3.0mu\mathsf{id}\mskip 0.0mu\mathnormal{)}\allowbreak{}\mskip 3.0mu\allowbreak{}\mathnormal{∘}\allowbreak{}\mskip 3.0mu\bigtriangleup\mskip 3.0mu\allowbreak{}\mathnormal{(}\mskip 0.0muπ₁\mskip 3.0mu\allowbreak{}\mathnormal{(}\mskip 0.0mu\mathsf{f}\mskip 3.0mu\mathsf{id}\mskip 0.0mu\mathnormal{)}\allowbreak{}\mskip 0.0mu\mathnormal{,}\mskip 3.0muπ₂\mskip 3.0mu\allowbreak{}\mathnormal{(}\mskip 0.0mu\mathsf{f}\mskip 3.0mu\mathsf{id}\mskip 0.0mu\mathnormal{)}\allowbreak{}\mskip 0.0mu\mathnormal{)}\allowbreak{}}\<[E]{}\\
\>[1]{}{\allowbreak{}=\allowbreak{}\mskip 0.0mu}\>[3]{}{\quad{}\text{\textit{by contraction of pairs}}}\<[E]{}\\
\>[2]{}{\mathsf{P}\mskip 3.0mu\allowbreak{}\mathnormal{(}\mskip 0.0mu\mathsf{Embed}\mskip 3.0mu\mathsf{φ}\mskip 3.0mu\allowbreak{}\mathnormal{×}\allowbreak{}\mskip 3.0mu\mathsf{id}\mskip 0.0mu\mathnormal{)}\allowbreak{}\mskip 3.0mu\allowbreak{}\mathnormal{∘}\allowbreak{}\mskip 3.0mu\bigtriangleup\mskip 3.0mu\allowbreak{}\mathnormal{(}\mskip 0.0mu\mathsf{f}\mskip 3.0mu\mathsf{id}\mskip 0.0mu\mathnormal{)}\allowbreak{}}\<[E]{}\end{parray}}\end{list}  \item{}Let \ensuremath{\mathsf{g}\mskip 3.0mu}\ensuremath{\mathnormal{=}\mskip 3.0mu}\ensuremath{α\mskip 3.0mu}\ensuremath{\allowbreak{}\mathnormal{∘}\allowbreak{}\mskip 3.0mu}\ensuremath{\allowbreak{}\mathnormal{(}\mskip 0.0mu}\ensuremath{\mathsf{split}\mskip 3.0mu}\ensuremath{\allowbreak{}\mathnormal{×}\allowbreak{}\mskip 3.0mu}\ensuremath{\mathsf{id}\mskip 0.0mu}\ensuremath{\mathnormal{)}\allowbreak{}}  
 \begin{list}{}{\setlength\leftmargin{1.0em}}\item\relax
\ensuremath{\begin{parray}\column{B}{@{}>{}l<{}@{}}\column[0em]{1}{@{}>{}l<{}@{}}\column[1em]{2}{@{}>{}l<{}@{}}\column{3}{@{}>{}l<{}@{}}\column{4}{@{}>{}l<{}@{}}\column{5}{@{}>{}l<{}@{}}\column{6}{@{}>{}l<{}@{}}\column{7}{@{}>{}l<{}@{}}\column{8}{@{}>{}l<{}@{}}\column{E}{@{}>{}l<{}@{}}%
\>[2]{}{\bigtriangleup\mskip 3.0mu\allowbreak{}\mathnormal{(}\mskip 0.0mu\allowbreak{}\mathnormal{(}\mskip 0.0muα\mskip 3.0mu\allowbreak{}\mathnormal{∘}\allowbreak{}\mskip 3.0mu\allowbreak{}\mathnormal{(}\mskip 0.0mu\mathsf{split}\mskip 3.0mu\allowbreak{}\mathnormal{×}\allowbreak{}\mskip 3.0mu\mathsf{id}\mskip 0.0mu\mathnormal{)}\allowbreak{}\mskip 3.0mu\allowbreak{}\mathnormal{∘}\allowbreak{}\mskip 3.0mu\mathsf{f}\mskip 0.0mu\mathnormal{)}\allowbreak{}\mskip 3.0mu\mathsf{id}\mskip 0.0mu\mathnormal{)}\allowbreak{}}\<[E]{}\\
\>[1]{}{\allowbreak{}=\allowbreak{}\mskip 0.0mu}\>[3]{}{\quad{}\text{\textit{by def of ∘}}}\<[E]{}\\
\>[2]{}{\bigtriangleup\mskip 3.0mu\allowbreak{}\mathnormal{(}\mskip 0.0muα\mskip 3.0mu}\>[5]{}{\allowbreak{}\mathnormal{(}\mskip 0.0mu\allowbreak{}\mathnormal{(}\mskip 0.0mu\mathsf{split}\mskip 3.0mu\allowbreak{}\mathnormal{×}\allowbreak{}\mskip 3.0mu\mathsf{id}\mskip 0.0mu\mathnormal{)}\allowbreak{}\mskip 3.0mu}\>[8]{}{\allowbreak{}\mathnormal{(}\mskip 0.0mu\mathsf{f}\mskip 3.0mu\mathsf{id}\mskip 0.0mu\mathnormal{)}\allowbreak{}\mskip 0.0mu\mathnormal{)}\allowbreak{}\mskip 0.0mu\mathnormal{)}\allowbreak{}}\<[E]{}\\
\>[1]{}{\allowbreak{}=\allowbreak{}\mskip 0.0mu}\>[3]{}{\quad{}\text{\textit{by pair expansion}}}\<[E]{}\\
\>[2]{}{\bigtriangleup\mskip 3.0mu\allowbreak{}\mathnormal{(}\mskip 0.0muα\mskip 3.0mu}\>[5]{}{\allowbreak{}\mathnormal{(}\mskip 0.0mu\allowbreak{}\mathnormal{(}\mskip 0.0mu\mathsf{split}\mskip 3.0mu\allowbreak{}\mathnormal{×}\allowbreak{}\mskip 3.0mu\mathsf{id}\mskip 0.0mu\mathnormal{)}\allowbreak{}\mskip 3.0mu}\>[8]{}{\allowbreak{}\mathnormal{(}\mskip 0.0muπ₁\mskip 3.0mu\allowbreak{}\mathnormal{(}\mskip 0.0mu\mathsf{f}\mskip 3.0mu\mathsf{id}\mskip 0.0mu\mathnormal{)}\allowbreak{}\mskip 0.0mu\mathnormal{,}\mskip 3.0muπ₂\mskip 3.0mu\allowbreak{}\mathnormal{(}\mskip 0.0mu\mathsf{f}\mskip 3.0mu\mathsf{id}\mskip 0.0mu\mathnormal{)}\allowbreak{}\mskip 0.0mu\mathnormal{)}\allowbreak{}\mskip 0.0mu\mathnormal{)}\allowbreak{}\mskip 0.0mu\mathnormal{)}\allowbreak{}}\<[E]{}\\
\>[1]{}{\allowbreak{}=\allowbreak{}\mskip 0.0mu}\>[3]{}{\quad{}\text{\textit{by def of ×}}}\<[E]{}\\
\>[2]{}{\bigtriangleup\mskip 3.0mu\allowbreak{}\mathnormal{(}\mskip 0.0muα\mskip 3.0mu}\>[5]{}{\allowbreak{}\mathnormal{(}\mskip 0.0mu\allowbreak{}\mathnormal{(}\mskip 0.0mu\mathsf{split}\mskip 3.0mu\allowbreak{}\mathnormal{(}\mskip 0.0muπ₁\mskip 3.0mu\allowbreak{}\mathnormal{(}\mskip 0.0mu\mathsf{f}\mskip 3.0mu\mathsf{id}\mskip 0.0mu\mathnormal{)}\allowbreak{}\mskip 0.0mu\mathnormal{)}\allowbreak{}\mskip 0.0mu\mathnormal{,}\mskip 3.0muπ₂\mskip 3.0mu\allowbreak{}\mathnormal{(}\mskip 0.0mu\mathsf{f}\mskip 3.0mu\mathsf{id}\mskip 0.0mu\mathnormal{)}\allowbreak{}\mskip 0.0mu\mathnormal{)}\allowbreak{}\mskip 0.0mu\mathnormal{)}\allowbreak{}\mskip 0.0mu\mathnormal{)}\allowbreak{}}\<[E]{}\\
\>[1]{}{\allowbreak{}=\allowbreak{}\mskip 0.0mu}\>[3]{}{\quad{}\text{\textit{by def of split}}}\<[E]{}\\
\>[2]{}{\bigtriangleup\mskip 3.0mu\allowbreak{}\mathnormal{(}\mskip 0.0muα\mskip 3.0mu\allowbreak{}\mathnormal{(}\mskip 0.0mu}\>[6]{}{\allowbreak{}\mathnormal{(}\mskip 0.0muπ₁\mskip 3.0mu\allowbreak{}\mathnormal{∘}\allowbreak{}\mskip 3.0mu\allowbreak{}\mathnormal{(}\mskip 0.0muπ₁\mskip 3.0mu\allowbreak{}\mathnormal{(}\mskip 0.0mu\mathsf{f}\mskip 3.0mu\mathsf{id}\mskip 0.0mu\mathnormal{)}\allowbreak{}\mskip 0.0mu\mathnormal{)}\allowbreak{}\mskip 0.0mu\mathnormal{,}\mskip 3.0muπ₂\mskip 3.0mu\allowbreak{}\mathnormal{∘}\allowbreak{}\mskip 3.0mu\allowbreak{}\mathnormal{(}\mskip 0.0muπ₁\mskip 3.0mu\allowbreak{}\mathnormal{(}\mskip 0.0mu\mathsf{f}\mskip 3.0mu\mathsf{id}\mskip 0.0mu\mathnormal{)}\allowbreak{}\mskip 0.0mu\mathnormal{)}\allowbreak{}\mskip 0.0mu\mathnormal{)}\allowbreak{}\mskip 0.0mu\mathnormal{,}}\<[E]{}\\
\>[6]{}{π₂\mskip 3.0mu\allowbreak{}\mathnormal{(}\mskip 0.0mu\mathsf{f}\mskip 3.0mu\mathsf{id}\mskip 0.0mu\mathnormal{)}\allowbreak{}\mskip 0.0mu\mathnormal{)}\allowbreak{}\mskip 0.0mu\mathnormal{)}\allowbreak{}}\<[E]{}\\
\>[1]{}{\allowbreak{}=\allowbreak{}\mskip 0.0mu}\>[3]{}{\quad{}\text{\textit{by def of assoc}}}\<[E]{}\\
\>[2]{}{\bigtriangleup\mskip 3.0mu\allowbreak{}\mathnormal{(}\mskip 0.0muπ₁\mskip 3.0mu\allowbreak{}\mathnormal{∘}\allowbreak{}\mskip 3.0mu\allowbreak{}\mathnormal{(}\mskip 0.0muπ₁\mskip 3.0mu\allowbreak{}\mathnormal{(}\mskip 0.0mu\mathsf{f}\mskip 3.0mu\mathsf{id}\mskip 0.0mu\mathnormal{)}\allowbreak{}\mskip 0.0mu\mathnormal{)}\allowbreak{}\mskip 0.0mu\mathnormal{,}\mskip 3.0mu\allowbreak{}\mathnormal{(}\mskip 0.0muπ₂\mskip 3.0mu\allowbreak{}\mathnormal{∘}\allowbreak{}\mskip 3.0mu\allowbreak{}\mathnormal{(}\mskip 0.0muπ₁\mskip 3.0mu\allowbreak{}\mathnormal{(}\mskip 0.0mu\mathsf{f}\mskip 3.0mu\mathsf{id}\mskip 0.0mu\mathnormal{)}\allowbreak{}\mskip 0.0mu\mathnormal{)}\allowbreak{}\mskip 0.0mu\mathnormal{,}\mskip 3.0muπ₂\mskip 3.0mu\allowbreak{}\mathnormal{(}\mskip 0.0mu\mathsf{f}\mskip 3.0mu\mathsf{id}\mskip 0.0mu\mathnormal{)}\allowbreak{}\mskip 0.0mu\mathnormal{)}\allowbreak{}\mskip 0.0mu\mathnormal{)}\allowbreak{}}\<[E]{}\\
\>[1]{}{\allowbreak{}=\allowbreak{}\mskip 0.0mu}\>[3]{}{\quad{}\text{\textit{by def of mergeA}}}\<[E]{}\\
\>[2]{}{\allowbreak{}\mathnormal{(}\mskip 0.0muπ₁\mskip 3.0mu\allowbreak{}\mathnormal{∘}\allowbreak{}\mskip 3.0mu\allowbreak{}\mathnormal{(}\mskip 0.0muπ₁\mskip 3.0mu\allowbreak{}\mathnormal{(}\mskip 0.0mu\mathsf{f}\mskip 3.0mu\mathsf{id}\mskip 0.0mu\mathnormal{)}\allowbreak{}\mskip 0.0mu\mathnormal{)}\allowbreak{}\mskip 0.0mu\mathnormal{)}\allowbreak{}\mskip 3.0mu\mathnormal{▵}\mskip 3.0mu\bigtriangleup\mskip 3.0mu\allowbreak{}\mathnormal{(}\mskip 0.0muπ₂\mskip 3.0mu\allowbreak{}\mathnormal{∘}\allowbreak{}\mskip 3.0mu\allowbreak{}\mathnormal{(}\mskip 0.0muπ₁\mskip 3.0mu\allowbreak{}\mathnormal{(}\mskip 0.0mu\mathsf{f}\mskip 3.0mu\mathsf{id}\mskip 0.0mu\mathnormal{)}\allowbreak{}\mskip 0.0mu\mathnormal{)}\allowbreak{}\mskip 0.0mu\mathnormal{,}\mskip 3.0muπ₂\mskip 3.0mu\allowbreak{}\mathnormal{(}\mskip 0.0mu\mathsf{f}\mskip 3.0mu\mathsf{id}\mskip 0.0mu\mathnormal{)}\allowbreak{}\mskip 0.0mu\mathnormal{)}\allowbreak{}}\<[E]{}\\
\>[1]{}{\allowbreak{}=\allowbreak{}\mskip 0.0mu}\>[3]{}{\quad{}\text{\textit{by def of mergeA}}}\<[E]{}\\
\>[2]{}{\allowbreak{}\mathnormal{(}\mskip 0.0muπ₁\mskip 3.0mu\allowbreak{}\mathnormal{∘}\allowbreak{}\mskip 3.0mu\allowbreak{}\mathnormal{(}\mskip 0.0muπ₁\mskip 3.0mu\allowbreak{}\mathnormal{(}\mskip 0.0mu\mathsf{f}\mskip 3.0mu\mathsf{id}\mskip 0.0mu\mathnormal{)}\allowbreak{}\mskip 0.0mu\mathnormal{)}\allowbreak{}\mskip 0.0mu\mathnormal{)}\allowbreak{}\mskip 3.0mu\mathnormal{▵}\mskip 3.0mu\allowbreak{}\mathnormal{(}\mskip 0.0mu}\>[7]{}{\allowbreak{}\mathnormal{(}\mskip 0.0muπ₂\mskip 3.0mu\allowbreak{}\mathnormal{∘}\allowbreak{}\mskip 3.0mu\allowbreak{}\mathnormal{(}\mskip 0.0muπ₁\mskip 3.0mu\allowbreak{}\mathnormal{(}\mskip 0.0mu\mathsf{f}\mskip 3.0mu\mathsf{id}\mskip 0.0mu\mathnormal{)}\allowbreak{}\mskip 0.0mu\mathnormal{)}\allowbreak{}\mskip 0.0mu\mathnormal{)}\allowbreak{}\mskip 3.0mu\mathnormal{▵}}\<[E]{}\\
\>[7]{}{\bigtriangleup\mskip 3.0mu\allowbreak{}\mathnormal{(}\mskip 0.0muπ₂\mskip 3.0mu\allowbreak{}\mathnormal{(}\mskip 0.0mu\mathsf{f}\mskip 3.0mu\mathsf{id}\mskip 0.0mu\mathnormal{)}\allowbreak{}\mskip 0.0mu\mathnormal{)}\allowbreak{}\mskip 0.0mu\mathnormal{)}\allowbreak{}}\<[E]{}\\
\>[1]{}{\allowbreak{}=\allowbreak{}\mskip 0.0mu}\>[3]{}{\quad{}\text{\textit{by def of assoc}}}\<[E]{}\\
\>[2]{}{α\mskip 3.0mu\allowbreak{}\mathnormal{∘}\allowbreak{}\mskip 3.0mu\allowbreak{}\mathnormal{(}\mskip 0.0mu}\>[4]{}{\allowbreak{}\mathnormal{(}\mskip 0.0mu\allowbreak{}\mathnormal{(}\mskip 0.0muπ₁\mskip 3.0mu\allowbreak{}\mathnormal{∘}\allowbreak{}\mskip 3.0mu\allowbreak{}\mathnormal{(}\mskip 0.0muπ₁\mskip 3.0mu\allowbreak{}\mathnormal{(}\mskip 0.0mu\mathsf{f}\mskip 3.0mu\mathsf{id}\mskip 0.0mu\mathnormal{)}\allowbreak{}\mskip 0.0mu\mathnormal{)}\allowbreak{}\mskip 0.0mu\mathnormal{)}\allowbreak{}\mskip 3.0mu\mathnormal{▵}\mskip 3.0mu\allowbreak{}\mathnormal{(}\mskip 0.0muπ₂\mskip 3.0mu\allowbreak{}\mathnormal{∘}\allowbreak{}\mskip 3.0mu\allowbreak{}\mathnormal{(}\mskip 0.0muπ₁\mskip 3.0mu\allowbreak{}\mathnormal{(}\mskip 0.0mu\mathsf{f}\mskip 3.0mu\mathsf{id}\mskip 0.0mu\mathnormal{)}\allowbreak{}\mskip 0.0mu\mathnormal{)}\allowbreak{}\mskip 0.0mu\mathnormal{)}\allowbreak{}\mskip 0.0mu\mathnormal{)}\allowbreak{}\mskip 3.0mu\mathnormal{▵}}\<[E]{}\\
\>[4]{}{\bigtriangleup\mskip 3.0mu\allowbreak{}\mathnormal{(}\mskip 0.0muπ₂\mskip 3.0mu\allowbreak{}\mathnormal{(}\mskip 0.0mu\mathsf{f}\mskip 3.0mu\mathsf{id}\mskip 0.0mu\mathnormal{)}\allowbreak{}\mskip 0.0mu\mathnormal{)}\allowbreak{}\mskip 0.0mu\mathnormal{)}\allowbreak{}}\<[E]{}\\
\>[1]{}{\allowbreak{}=\allowbreak{}\mskip 0.0mu}\>[3]{}{\quad{}\text{\textit{by properties of cartesian categories}}}\<[E]{}\\
\>[2]{}{α\mskip 3.0mu\allowbreak{}\mathnormal{∘}\allowbreak{}\mskip 3.0mu\allowbreak{}\mathnormal{(}\mskip 0.0muπ₁\mskip 3.0mu\allowbreak{}\mathnormal{(}\mskip 0.0mu\mathsf{f}\mskip 3.0mu\mathsf{id}\mskip 0.0mu\mathnormal{)}\allowbreak{}\mskip 3.0mu\mathnormal{▵}\mskip 3.0mu\bigtriangleup\mskip 3.0mu\allowbreak{}\mathnormal{(}\mskip 0.0muπ₂\mskip 3.0mu\allowbreak{}\mathnormal{(}\mskip 0.0mu\mathsf{f}\mskip 3.0mu\mathsf{id}\mskip 0.0mu\mathnormal{)}\allowbreak{}\mskip 0.0mu\mathnormal{)}\allowbreak{}\mskip 0.0mu\mathnormal{)}\allowbreak{}}\<[E]{}\\
\>[1]{}{\allowbreak{}=\allowbreak{}\mskip 0.0mu}\>[3]{}{\quad{}\text{\textit{by def of mergeA}}}\<[E]{}\\
\>[2]{}{α\mskip 3.0mu\allowbreak{}\mathnormal{∘}\allowbreak{}\mskip 3.0mu\bigtriangleup\mskip 3.0mu\allowbreak{}\mathnormal{(}\mskip 0.0muπ₁\mskip 3.0mu\allowbreak{}\mathnormal{(}\mskip 0.0mu\mathsf{f}\mskip 3.0mu\mathsf{id}\mskip 0.0mu\mathnormal{)}\allowbreak{}\mskip 0.0mu\mathnormal{,}\mskip 3.0muπ₂\mskip 3.0mu\allowbreak{}\mathnormal{(}\mskip 0.0mu\mathsf{f}\mskip 3.0mu\mathsf{id}\mskip 0.0mu\mathnormal{)}\allowbreak{}\mskip 0.0mu\mathnormal{)}\allowbreak{}}\<[E]{}\\
\>[1]{}{\allowbreak{}=\allowbreak{}\mskip 0.0mu}\>[3]{}{\quad{}\text{\textit{by contraction of pair}}}\<[E]{}\\
\>[2]{}{α\mskip 3.0mu\allowbreak{}\mathnormal{∘}\allowbreak{}\mskip 3.0mu\bigtriangleup\mskip 3.0mu\allowbreak{}\mathnormal{(}\mskip 0.0mu\mathsf{f}\mskip 3.0mu\mathsf{id}\mskip 0.0mu\mathnormal{)}\allowbreak{}}\<[E]{}\end{parray}}\end{list}\item{}Let \ensuremath{\mathsf{g}\mskip 3.0mu}\ensuremath{\mathnormal{=}\mskip 3.0mu}\ensuremath{\allowbreak{}\mathnormal{(}\mskip 0.0mu}\ensuremath{\mathsf{merge}\mskip 3.0mu}\ensuremath{\allowbreak{}\mathnormal{×}\allowbreak{}\mskip 3.0mu}\ensuremath{\mathsf{id}\mskip 0.0mu}\ensuremath{\mathnormal{)}\allowbreak{}\mskip 3.0mu}\ensuremath{\allowbreak{}\mathnormal{∘}\allowbreak{}\mskip 3.0mu}\ensuremath{\bar{α}}.

 \begin{list}{}{\setlength\leftmargin{1.0em}}\item\relax
\ensuremath{\begin{parray}\column{B}{@{}>{}l<{}@{}}\column[0em]{1}{@{}>{}l<{}@{}}\column[1em]{2}{@{}>{}l<{}@{}}\column{3}{@{}>{}l<{}@{}}\column{4}{@{}>{}l<{}@{}}\column{E}{@{}>{}l<{}@{}}%
\>[2]{}{\bigtriangleup\mskip 3.0mu\allowbreak{}\mathnormal{(}\mskip 0.0mu\allowbreak{}\mathnormal{(}\mskip 0.0mu\allowbreak{}\mathnormal{(}\mskip 0.0mu\mathsf{merge}\mskip 3.0mu\allowbreak{}\mathnormal{×}\allowbreak{}\mskip 3.0mu\mathsf{id}\mskip 0.0mu\mathnormal{)}\allowbreak{}\mskip 3.0mu\allowbreak{}\mathnormal{∘}\allowbreak{}\mskip 3.0mu\bar{α}\mskip 3.0mu\allowbreak{}\mathnormal{∘}\allowbreak{}\mskip 3.0mu\mathsf{f}\mskip 0.0mu\mathnormal{)}\allowbreak{}\mskip 3.0mu\mathsf{id}\mskip 0.0mu\mathnormal{)}\allowbreak{}}\<[E]{}\\
\>[1]{}{\allowbreak{}=\allowbreak{}\mskip 0.0mu}\>[3]{}{\quad{}\text{\textit{by def of ∘}}}\<[E]{}\\
\>[2]{}{\bigtriangleup\mskip 3.0mu\allowbreak{}\mathnormal{(}\mskip 0.0mu\allowbreak{}\mathnormal{(}\mskip 0.0mu\mathsf{merge}\mskip 3.0mu\allowbreak{}\mathnormal{×}\allowbreak{}\mskip 3.0mu\mathsf{id}\mskip 0.0mu\mathnormal{)}\allowbreak{}\mskip 3.0mu\allowbreak{}\mathnormal{(}\mskip 0.0mu\bar{α}\mskip 3.0mu\allowbreak{}\mathnormal{(}\mskip 0.0mu\mathsf{f}\mskip 3.0mu\mathsf{id}\mskip 0.0mu\mathnormal{)}\allowbreak{}\mskip 0.0mu\mathnormal{)}\allowbreak{}\mskip 0.0mu\mathnormal{)}\allowbreak{}}\<[E]{}\\
\>[1]{}{\allowbreak{}=\allowbreak{}\mskip 0.0mu}\>[3]{}{\quad{}\text{\textit{by expansion of pairs, def of assoc.}}}\<[E]{}\\
\>[2]{}{\bigtriangleup\mskip 3.0mu\allowbreak{}\mathnormal{(}\mskip 0.0mu\allowbreak{}\mathnormal{(}\mskip 0.0mu\mathsf{merge}\mskip 3.0mu\allowbreak{}\mathnormal{×}\allowbreak{}\mskip 3.0mu\mathsf{id}\mskip 0.0mu\mathnormal{)}\allowbreak{}\mskip 3.0mu\allowbreak{}\mathnormal{(}\mskip 0.0mu}\>[4]{}{\allowbreak{}\mathnormal{(}\mskip 0.0muπ₁\mskip 3.0mu\allowbreak{}\mathnormal{(}\mskip 0.0mu\mathsf{f}\mskip 3.0mu\mathsf{id}\mskip 0.0mu\mathnormal{)}\allowbreak{}\mskip 0.0mu\mathnormal{,}\mskip 3.0muπ₁\mskip 3.0mu\allowbreak{}\mathnormal{(}\mskip 0.0muπ₂\mskip 3.0mu\allowbreak{}\mathnormal{(}\mskip 0.0mu\mathsf{f}\mskip 3.0mu\mathsf{id}\mskip 0.0mu\mathnormal{)}\allowbreak{}\mskip 0.0mu\mathnormal{)}\allowbreak{}\mskip 0.0mu\mathnormal{)}\allowbreak{}\mskip 0.0mu\mathnormal{,}}\<[E]{}\\
\>[4]{}{π₂\mskip 3.0mu\allowbreak{}\mathnormal{(}\mskip 0.0muπ₂\mskip 3.0mu\allowbreak{}\mathnormal{(}\mskip 0.0mu\mathsf{f}\mskip 3.0mu\mathsf{id}\mskip 0.0mu\mathnormal{)}\allowbreak{}\mskip 0.0mu\mathnormal{)}\allowbreak{}\mskip 0.0mu\mathnormal{)}\allowbreak{}\mskip 0.0mu\mathnormal{)}\allowbreak{}}\<[E]{}\\
\>[1]{}{\allowbreak{}=\allowbreak{}\mskip 0.0mu}\>[3]{}{\quad{}\text{\textit{by def of ×}}}\<[E]{}\\
\>[2]{}{\bigtriangleup\mskip 3.0mu\allowbreak{}\mathnormal{(}\mskip 0.0mu\mathsf{merge}\mskip 3.0mu\allowbreak{}\mathnormal{(}\mskip 0.0muπ₁\mskip 3.0mu\allowbreak{}\mathnormal{(}\mskip 0.0mu\mathsf{f}\mskip 3.0mu\mathsf{id}\mskip 0.0mu\mathnormal{)}\allowbreak{}\mskip 0.0mu\mathnormal{,}\mskip 3.0muπ₁\mskip 3.0mu\allowbreak{}\mathnormal{(}\mskip 0.0muπ₂\mskip 3.0mu\allowbreak{}\mathnormal{(}\mskip 0.0mu\mathsf{f}\mskip 3.0mu\mathsf{id}\mskip 0.0mu\mathnormal{)}\allowbreak{}\mskip 0.0mu\mathnormal{)}\allowbreak{}\mskip 0.0mu\mathnormal{)}\allowbreak{}\mskip 0.0mu\mathnormal{,}\mskip 3.0muπ₂\mskip 3.0mu\allowbreak{}\mathnormal{(}\mskip 0.0muπ₂\mskip 3.0mu\allowbreak{}\mathnormal{(}\mskip 0.0mu\mathsf{f}\mskip 3.0mu\mathsf{id}\mskip 0.0mu\mathnormal{)}\allowbreak{}\mskip 0.0mu\mathnormal{)}\allowbreak{}\mskip 0.0mu\mathnormal{)}\allowbreak{}}\<[E]{}\\
\>[1]{}{\allowbreak{}=\allowbreak{}\mskip 0.0mu}\>[3]{}{\quad{}\text{\textit{by def of merge}}}\<[E]{}\\
\>[2]{}{\bigtriangleup\mskip 3.0mu\allowbreak{}\mathnormal{(}\mskip 0.0mu\allowbreak{}\mathnormal{(}\mskip 0.0muπ₁\mskip 3.0mu\allowbreak{}\mathnormal{(}\mskip 0.0mu\mathsf{f}\mskip 3.0mu\mathsf{id}\mskip 0.0mu\mathnormal{)}\allowbreak{}\mskip 3.0mu\mathnormal{▵}\mskip 3.0muπ₁\mskip 3.0mu\allowbreak{}\mathnormal{(}\mskip 0.0muπ₂\mskip 3.0mu\allowbreak{}\mathnormal{(}\mskip 0.0mu\mathsf{f}\mskip 3.0mu\mathsf{id}\mskip 0.0mu\mathnormal{)}\allowbreak{}\mskip 0.0mu\mathnormal{)}\allowbreak{}\mskip 0.0mu\mathnormal{)}\allowbreak{}\mskip 0.0mu\mathnormal{,}\mskip 3.0muπ₂\mskip 3.0mu\allowbreak{}\mathnormal{(}\mskip 0.0muπ₂\mskip 3.0mu\allowbreak{}\mathnormal{(}\mskip 0.0mu\mathsf{f}\mskip 3.0mu\mathsf{id}\mskip 0.0mu\mathnormal{)}\allowbreak{}\mskip 0.0mu\mathnormal{)}\allowbreak{}\mskip 0.0mu\mathnormal{)}\allowbreak{}}\<[E]{}\\
\>[1]{}{\allowbreak{}=\allowbreak{}\mskip 0.0mu}\>[3]{}{\quad{}\text{\textit{by def of mergeA}}}\<[E]{}\\
\>[2]{}{\allowbreak{}\mathnormal{(}\mskip 0.0muπ₁\mskip 3.0mu\allowbreak{}\mathnormal{(}\mskip 0.0mu\mathsf{f}\mskip 3.0mu\mathsf{id}\mskip 0.0mu\mathnormal{)}\allowbreak{}\mskip 3.0mu\mathnormal{▵}\mskip 3.0muπ₁\mskip 3.0mu\allowbreak{}\mathnormal{(}\mskip 0.0muπ₂\mskip 3.0mu\allowbreak{}\mathnormal{(}\mskip 0.0mu\mathsf{f}\mskip 3.0mu\mathsf{id}\mskip 0.0mu\mathnormal{)}\allowbreak{}\mskip 0.0mu\mathnormal{)}\allowbreak{}\mskip 0.0mu\mathnormal{)}\allowbreak{}\mskip 3.0mu\mathnormal{▵}\mskip 3.0mu\bigtriangleup\mskip 3.0mu\allowbreak{}\mathnormal{(}\mskip 0.0muπ₂\mskip 3.0mu\allowbreak{}\mathnormal{(}\mskip 0.0muπ₂\mskip 3.0mu\allowbreak{}\mathnormal{(}\mskip 0.0mu\mathsf{f}\mskip 3.0mu\mathsf{id}\mskip 0.0mu\mathnormal{)}\allowbreak{}\mskip 0.0mu\mathnormal{)}\allowbreak{}\mskip 0.0mu\mathnormal{)}\allowbreak{}}\<[E]{}\\
\>[1]{}{\allowbreak{}=\allowbreak{}\mskip 0.0mu}\>[3]{}{\quad{}\text{\textit{by property of ▵/assoc}}}\<[E]{}\\
\>[2]{}{\bar{α}\mskip 3.0mu\allowbreak{}\mathnormal{∘}\allowbreak{}\mskip 3.0muπ₁\mskip 3.0mu\allowbreak{}\mathnormal{(}\mskip 0.0mu\mathsf{f}\mskip 3.0mu\mathsf{id}\mskip 0.0mu\mathnormal{)}\allowbreak{}\mskip 3.0mu\mathnormal{▵}\mskip 3.0mu\allowbreak{}\mathnormal{(}\mskip 0.0muπ₁\mskip 3.0mu\allowbreak{}\mathnormal{(}\mskip 0.0muπ₂\mskip 3.0mu\allowbreak{}\mathnormal{(}\mskip 0.0mu\mathsf{f}\mskip 3.0mu\mathsf{id}\mskip 0.0mu\mathnormal{)}\allowbreak{}\mskip 0.0mu\mathnormal{)}\allowbreak{}\mskip 3.0mu\mathnormal{▵}\mskip 3.0mu\bigtriangleup\mskip 3.0mu\allowbreak{}\mathnormal{(}\mskip 0.0muπ₂\mskip 3.0mu\allowbreak{}\mathnormal{(}\mskip 0.0muπ₂\mskip 3.0mu\allowbreak{}\mathnormal{(}\mskip 0.0mu\mathsf{f}\mskip 3.0mu\mathsf{id}\mskip 0.0mu\mathnormal{)}\allowbreak{}\mskip 0.0mu\mathnormal{)}\allowbreak{}\mskip 0.0mu\mathnormal{)}\allowbreak{}\mskip 0.0mu\mathnormal{)}\allowbreak{}}\<[E]{}\\
\>[1]{}{\allowbreak{}=\allowbreak{}\mskip 0.0mu}\>[3]{}{\quad{}\text{\textit{by def of mergeA}}}\<[E]{}\\
\>[2]{}{\bar{α}\mskip 3.0mu\allowbreak{}\mathnormal{∘}\allowbreak{}\mskip 3.0muπ₁\mskip 3.0mu\allowbreak{}\mathnormal{(}\mskip 0.0mu\mathsf{f}\mskip 3.0mu\mathsf{id}\mskip 0.0mu\mathnormal{)}\allowbreak{}\mskip 3.0mu\mathnormal{▵}\mskip 3.0mu\bigtriangleup\mskip 3.0mu\allowbreak{}\mathnormal{(}\mskip 0.0muπ₂\mskip 3.0mu\allowbreak{}\mathnormal{(}\mskip 0.0mu\mathsf{f}\mskip 3.0mu\mathsf{id}\mskip 0.0mu\mathnormal{)}\allowbreak{}\mskip 0.0mu\mathnormal{)}\allowbreak{}}\<[E]{}\\
\>[1]{}{\allowbreak{}=\allowbreak{}\mskip 0.0mu}\>[3]{}{\quad{}\text{\textit{by contraction of pairs}}}\<[E]{}\\
\>[2]{}{\bar{α}\mskip 3.0mu\allowbreak{}\mathnormal{∘}\allowbreak{}\mskip 3.0muπ₁\mskip 3.0mu\allowbreak{}\mathnormal{(}\mskip 0.0mu\mathsf{f}\mskip 3.0mu\mathsf{id}\mskip 0.0mu\mathnormal{)}\allowbreak{}\mskip 3.0mu\mathnormal{▵}\mskip 3.0mu\bigtriangleup\mskip 3.0mu\allowbreak{}\mathnormal{(}\mskip 0.0muπ₁\mskip 3.0mu\allowbreak{}\mathnormal{(}\mskip 0.0muπ₂\mskip 3.0mu\allowbreak{}\mathnormal{(}\mskip 0.0mu\mathsf{f}\mskip 3.0mu\mathsf{id}\mskip 0.0mu\mathnormal{)}\allowbreak{}\mskip 0.0mu\mathnormal{)}\allowbreak{}\mskip 0.0mu\mathnormal{,}\mskip 3.0muπ₂\mskip 3.0mu\allowbreak{}\mathnormal{(}\mskip 0.0muπ₂\mskip 3.0mu\allowbreak{}\mathnormal{(}\mskip 0.0mu\mathsf{f}\mskip 3.0mu\mathsf{id}\mskip 0.0mu\mathnormal{)}\allowbreak{}\mskip 0.0mu\mathnormal{)}\allowbreak{}\mskip 0.0mu\mathnormal{)}\allowbreak{}}\<[E]{}\\
\>[1]{}{\allowbreak{}=\allowbreak{}\mskip 0.0mu}\>[3]{}{\quad{}\text{\textit{by def of mergeA}}}\<[E]{}\\
\>[2]{}{\bar{α}\mskip 3.0mu\allowbreak{}\mathnormal{∘}\allowbreak{}\mskip 3.0mu\bigtriangleup\mskip 3.0mu\allowbreak{}\mathnormal{(}\mskip 0.0muπ₁\mskip 3.0mu\allowbreak{}\mathnormal{(}\mskip 0.0mu\mathsf{f}\mskip 3.0mu\mathsf{id}\mskip 0.0mu\mathnormal{)}\allowbreak{}\mskip 0.0mu\mathnormal{,}\mskip 3.0muπ₂\mskip 3.0mu\allowbreak{}\mathnormal{(}\mskip 0.0mu\mathsf{f}\mskip 3.0mu\mathsf{id}\mskip 0.0mu\mathnormal{)}\allowbreak{}\mskip 0.0mu\mathnormal{)}\allowbreak{}}\<[E]{}\\
\>[1]{}{\allowbreak{}=\allowbreak{}\mskip 0.0mu}\>[3]{}{\quad{}\text{\textit{by contraction of pairs}}}\<[E]{}\\
\>[2]{}{\bar{α}\mskip 3.0mu\allowbreak{}\mathnormal{∘}\allowbreak{}\mskip 3.0mu\bigtriangleup\mskip 3.0mu\allowbreak{}\mathnormal{(}\mskip 0.0mu\mathsf{f}\mskip 3.0mu\mathsf{id}\mskip 0.0mu\mathnormal{)}\allowbreak{}}\<[E]{}\end{parray}}\end{list} \item{}Let \ensuremath{\mathsf{g}\mskip 3.0mu}\ensuremath{\mathnormal{=}\mskip 3.0mu}\ensuremath{λ\mskip 3.0mu}\ensuremath{\mathsf{x}\mskip 3.0mu}\ensuremath{\mathnormal{\rightarrow }\mskip 3.0mu}\ensuremath{\allowbreak{}\mathnormal{(}\mskip 0.0mu}\ensuremath{\mathsf{unit}\mskip 0.0mu}\ensuremath{\mathnormal{,}\mskip 3.0mu}\ensuremath{\mathsf{x}\mskip 0.0mu}\ensuremath{\mathnormal{)}\allowbreak{}}.

\begin{list}{}{\setlength\leftmargin{1.0em}}\item\relax
\ensuremath{\begin{parray}\column{B}{@{}>{}l<{}@{}}\column[0em]{1}{@{}>{}l<{}@{}}\column{2}{@{}>{}l<{}@{}}\column{E}{@{}>{}l<{}@{}}%
\>[2]{}{\bigtriangleup\mskip 3.0mu\allowbreak{}\mathnormal{(}\mskip 0.0mu\allowbreak{}\mathnormal{(}\mskip 0.0mu\allowbreak{}\mathnormal{(}\mskip 0.0muλ\mskip 3.0mu\mathsf{x}\mskip 3.0mu\mathnormal{\rightarrow }\mskip 3.0mu\allowbreak{}\mathnormal{(}\mskip 0.0mu\mathsf{unit}\mskip 0.0mu\mathnormal{,}\mskip 3.0mu\mathsf{x}\mskip 0.0mu\mathnormal{)}\allowbreak{}\mskip 0.0mu\mathnormal{)}\allowbreak{}\mskip 3.0mu\allowbreak{}\mathnormal{∘}\allowbreak{}\mskip 3.0mu\mathsf{f}\mskip 0.0mu\mathnormal{)}\allowbreak{}\mskip 3.0mu\mathsf{id}\mskip 0.0mu\mathnormal{)}\allowbreak{}}\<[E]{}\\
\>[1]{}{\allowbreak{}=\allowbreak{}\mskip 0.0mu}\>[2]{}{\quad{}\text{\textit{by def of ∘}}}\<[E]{}\\
\>[2]{}{\bigtriangleup\mskip 3.0mu\allowbreak{}\mathnormal{(}\mskip 0.0mu\allowbreak{}\mathnormal{(}\mskip 0.0mu\allowbreak{}\mathnormal{(}\mskip 0.0muλ\mskip 3.0mu\mathsf{x}\mskip 3.0mu\mathnormal{\rightarrow }\mskip 3.0mu\allowbreak{}\mathnormal{(}\mskip 0.0mu\mathsf{unit}\mskip 0.0mu\mathnormal{,}\mskip 3.0mu\mathsf{x}\mskip 0.0mu\mathnormal{)}\allowbreak{}\mskip 0.0mu\mathnormal{)}\allowbreak{}\mskip 3.0mu\allowbreak{}\mathnormal{∘}\allowbreak{}\mskip 3.0mu\mathsf{f}\mskip 0.0mu\mathnormal{)}\allowbreak{}\mskip 3.0mu\mathsf{id}\mskip 0.0mu\mathnormal{)}\allowbreak{}}\<[E]{}\\
\>[1]{}{\allowbreak{}=\allowbreak{}\mskip 0.0mu}\>[2]{}{\quad{}\text{\textit{by def of unit}}}\<[E]{}\\
\>[2]{}{\bigtriangleup\mskip 3.0mu\allowbreak{}\mathnormal{(}\mskip 0.0mu\mathsf{P}\mskip 3.0mu\mathsf{ε}\mskip 0.0mu\mathnormal{,}\mskip 3.0mu\mathsf{f}\mskip 3.0mu\mathsf{id}\mskip 0.0mu\mathnormal{)}\allowbreak{}}\<[E]{}\\
\>[1]{}{\allowbreak{}=\allowbreak{}\mskip 0.0mu}\>[2]{}{\quad{}\text{\textit{by def of mergeA}}}\<[E]{}\\
\>[2]{}{ε\mskip 3.0mu\mathnormal{▵}\mskip 3.0mu\bigtriangleup\mskip 3.0mu\allowbreak{}\mathnormal{(}\mskip 0.0mu\mathsf{f}\mskip 3.0mu\mathsf{id}\mskip 0.0mu\mathnormal{)}\allowbreak{}}\<[E]{}\\
\>[1]{}{\allowbreak{}=\allowbreak{}\mskip 0.0mu}\>[2]{}{\quad{}\text{\textit{by property of ε}}}\<[E]{}\\
\>[2]{}{\allowbreak{}\mathnormal{(}\mskip 0.0muε\mskip 3.0mu\mathnormal{▵}\mskip 3.0mu\mathsf{id}\mskip 0.0mu\mathnormal{)}\allowbreak{}\mskip 3.0mu\allowbreak{}\mathnormal{∘}\allowbreak{}\mskip 3.0mu\bigtriangleup\mskip 3.0mu\allowbreak{}\mathnormal{(}\mskip 0.0mu\mathsf{f}\mskip 3.0mu\mathsf{id}\mskip 0.0mu\mathnormal{)}\allowbreak{}}\<[E]{}\\
\>[1]{}{\allowbreak{}=\allowbreak{}\mskip 0.0mu}\>[2]{}{\quad{}\text{\textit{by property of swap}}}\<[E]{}\\
\>[2]{}{σ\mskip 3.0mu\allowbreak{}\mathnormal{∘}\allowbreak{}\mskip 3.0mu\allowbreak{}\mathnormal{(}\mskip 0.0mu\mathsf{id}\mskip 3.0mu\mathnormal{▵}\mskip 3.0muε\mskip 0.0mu\mathnormal{)}\allowbreak{}\mskip 3.0mu\allowbreak{}\mathnormal{∘}\allowbreak{}\mskip 3.0mu\bigtriangleup\mskip 3.0mu\allowbreak{}\mathnormal{(}\mskip 0.0mu\mathsf{f}\mskip 3.0mu\mathsf{id}\mskip 0.0mu\mathnormal{)}\allowbreak{}}\<[E]{}\\
\>[1]{}{\allowbreak{}=\allowbreak{}\mskip 0.0mu}\>[2]{}{\quad{}\text{\textit{by definition of unitor}}}\<[E]{}\\
\>[2]{}{σ\mskip 3.0mu\allowbreak{}\mathnormal{∘}\allowbreak{}\mskip 3.0muρ\mskip 3.0mu\allowbreak{}\mathnormal{∘}\allowbreak{}\mskip 3.0mu\bigtriangleup\mskip 3.0mu\allowbreak{}\mathnormal{(}\mskip 0.0mu\mathsf{f}\mskip 3.0mu\mathsf{id}\mskip 0.0mu\mathnormal{)}\allowbreak{}}\<[E]{}\end{parray}}\end{list} \item{}Let \ensuremath{\mathsf{g}\mskip 3.0mu}\ensuremath{\mathnormal{=}\mskip 3.0mu}\ensuremath{\mathsf{θ}} be a permutation.

\begin{list}{}{\setlength\leftmargin{1.0em}}\item\relax
\ensuremath{\begin{parray}\column{B}{@{}>{}l<{}@{}}\column[0em]{1}{@{}>{}l<{}@{}}\column[1em]{2}{@{}>{}l<{}@{}}\column{3}{@{}>{}l<{}@{}}\column{E}{@{}>{}l<{}@{}}%
\>[2]{}{\bigtriangleup\mskip 3.0mu\allowbreak{}\mathnormal{(}\mskip 0.0mu\allowbreak{}\mathnormal{(}\mskip 0.0mu\mathsf{θ}\mskip 3.0mu\allowbreak{}\mathnormal{∘}\allowbreak{}\mskip 3.0mu\mathsf{f}\mskip 0.0mu\mathnormal{)}\allowbreak{}\mskip 3.0mu\mathsf{id}\mskip 0.0mu\mathnormal{)}\allowbreak{}}\<[E]{}\\
\>[1]{}{\allowbreak{}=\allowbreak{}\mskip 0.0mu}\>[3]{}{\quad{}\text{\textit{by def of ∘ }}}\<[E]{}\\
\>[2]{}{\bigtriangleup\mskip 3.0mu\allowbreak{}\mathnormal{(}\mskip 0.0mu\mathsf{θ}\mskip 3.0mu\allowbreak{}\mathnormal{(}\mskip 0.0mu\mathsf{f}\mskip 3.0mu\mathsf{id}\mskip 0.0mu\mathnormal{)}\allowbreak{}\mskip 0.0mu\mathnormal{)}\allowbreak{}}\<[E]{}\\
\>[1]{}{\allowbreak{}=\allowbreak{}}\<[E]{}\\
\>[2]{}{\mathsf{θ}\mskip 3.0mu\allowbreak{}\mathnormal{∘}\allowbreak{}\mskip 3.0mu\bigtriangleup\mskip 3.0mu\allowbreak{}\mathnormal{(}\mskip 0.0mu\mathsf{f}\mskip 3.0mu\mathsf{id}\mskip 0.0mu\mathnormal{)}\allowbreak{}}\<[E]{}\end{parray}}\end{list} 
The last step is justified because θ is representable in any symmetric monoidal category.
Furthermore, because \ensuremath{\bigtriangleup} respects the structure of products, it does not
matter if θ is applied before or after \ensuremath{\bigtriangleup}.
\end{itemize}\end{mdframed} 
Recall that the induction hypothesis is that \ensuremath{\bigtriangleup\mskip 3.0mu}\ensuremath{\allowbreak{}\mathnormal{(}\mskip 0.0mu}\ensuremath{\mathsf{f}\mskip 3.0mu}\ensuremath{\mathsf{id}\mskip 0.0mu}\ensuremath{\mathnormal{)}\allowbreak{}} is
protolinear. This observation alone concludes the argument for the
\ensuremath{\mathsf{split}}, \ensuremath{\mathsf{merge}} and \ensuremath{\mathsf{unit}} cases.  For the other two cases, it suffices to see
that every permutation θ and every generator φ is linear, and we
have protolinearity for the composition.\end{proof} 
\subsection{An algorithm for \ensuremath{\mathsf{reduce}}}\label{71} 
\begin{figure}
{\begin{tikzpicture}\path[-,draw=black,line width=0.4000pt,line cap=butt,line join=miter,dash pattern=](-16.6000pt,-8.0087pt)--(-9.6000pt,-8.0087pt);
\path[-,draw=black,line width=0.4000pt,line cap=butt,line join=miter,dash pattern=](30.3300pt,24.5350pt)--(37.3300pt,24.5350pt);
\path[-,draw=black,line width=0.4000pt,line cap=butt,line join=miter,dash pattern=](30.3300pt,0.0000pt)--(37.3300pt,0.0000pt);
\path[-,draw=black,line width=0.4000pt,line cap=butt,line join=miter,dash pattern=](30.3300pt,-29.5350pt)--(37.3300pt,-29.5350pt);
\path[-,draw=black,line width=0.4000pt,line cap=butt,line join=miter,dash pattern=](0.0000pt,11.0175pt)--(0.0000pt,11.0175pt);
\path[-,draw=black,line width=0.4000pt,line cap=butt,line join=miter,dash pattern=](0.0000pt,-27.0350pt)--(0.0000pt,-27.0350pt);
\path[-,draw=black,line width=0.4000pt,line cap=butt,line join=miter,dash pattern=](9.6000pt,22.0350pt)--(9.6000pt,22.0350pt);
\path[-,draw=black,line width=0.4000pt,line cap=butt,line join=miter,dash pattern=](9.6000pt,0.0000pt)--(9.6000pt,0.0000pt);
\path[-,draw=black,line width=0.4000pt,line cap=butt,line join=miter,dash pattern=](23.3300pt,24.5350pt)--(30.3300pt,24.5350pt);
\path[-,draw=black,line width=0.4000pt,line cap=butt,line join=miter,dash pattern=](23.3300pt,19.5350pt)--(30.3300pt,19.5350pt);
\path[-,line width=0.4000pt,line cap=butt,line join=miter,dash pattern=](30.3300pt,24.5350pt)--(30.3300pt,24.5350pt)--(30.3300pt,19.5350pt)--(30.3300pt,19.5350pt)--cycle;
\path[-,draw=black,line width=0.4000pt,line cap=butt,line join=miter,dash pattern=](27.3300pt,27.5350pt)--(33.3300pt,27.5350pt)--(33.3300pt,16.5350pt)--(27.3300pt,16.5350pt)--cycle;
\path[-,fill=black,line width=0.4000pt,line cap=butt,line join=miter,dash pattern=](31.3300pt,19.5350pt)..controls(31.3300pt,20.0873pt)and(30.8823pt,20.5350pt)..(30.3300pt,20.5350pt)..controls(29.7777pt,20.5350pt)and(29.3300pt,20.0873pt)..(29.3300pt,19.5350pt)..controls(29.3300pt,18.9827pt)and(29.7777pt,18.5350pt)..(30.3300pt,18.5350pt)..controls(30.8823pt,18.5350pt)and(31.3300pt,18.9827pt)..(31.3300pt,19.5350pt)--cycle;
\path[-,draw=black,line width=0.4000pt,line cap=butt,line join=miter,dash pattern=](30.3300pt,24.5350pt)--(30.3300pt,24.5350pt);
\path[-,line width=0.4000pt,line cap=butt,line join=miter,dash pattern=on 0.4000pt off 2.0000pt](13.6000pt,26.7275pt)--(19.3300pt,26.7275pt)--(19.3300pt,17.3425pt)--(13.6000pt,17.3425pt)--cycle;
\node[anchor=north west,inner sep=0] at (13.6000pt,26.7275pt){\savebox{\marxupbox}{{\ensuremath{f}}}\immediate\write\boxesfile{72}\immediate\write\boxesfile{\number\wd\marxupbox}\immediate\write\boxesfile{\number\ht\marxupbox}\immediate\write\boxesfile{\number\dp\marxupbox}\box\marxupbox};
\path[-,line width=0.4000pt,line cap=butt,line join=miter,dash pattern=on 0.4000pt off 2.0000pt](9.6000pt,30.7275pt)--(23.3300pt,30.7275pt)--(23.3300pt,13.3425pt)--(9.6000pt,13.3425pt)--cycle;
\path[-,draw=black,line width=0.4000pt,line cap=butt,line join=miter,dash pattern=on 0.4000pt off 2.0000pt](9.6000pt,30.7275pt)--(23.3300pt,30.7275pt)--(23.3300pt,13.3425pt)--(9.6000pt,13.3425pt)--cycle;
\path[-,draw=black,line width=0.4000pt,line cap=butt,line join=miter,dash pattern=](23.3300pt,24.5350pt)--(23.3300pt,24.5350pt);
\path[-,draw=black,line width=0.4000pt,line cap=butt,line join=miter,dash pattern=](23.3300pt,19.5350pt)--(23.3300pt,19.5350pt);
\path[-,draw=black,line width=0.4000pt,line cap=butt,line join=miter,dash pattern=](9.6000pt,22.0350pt)--(9.6000pt,22.0350pt);
\path[-,line width=0.4000pt,line cap=butt,line join=miter,dash pattern=on 0.4000pt off 2.0000pt](17.6200pt,3.3425pt)--(22.3100pt,3.3425pt)--(22.3100pt,-3.3425pt)--(17.6200pt,-3.3425pt)--cycle;
\node[anchor=north west,inner sep=0] at (17.6200pt,3.3425pt){\savebox{\marxupbox}{{\ensuremath{g}}}\immediate\write\boxesfile{73}\immediate\write\boxesfile{\number\wd\marxupbox}\immediate\write\boxesfile{\number\ht\marxupbox}\immediate\write\boxesfile{\number\dp\marxupbox}\box\marxupbox};
\path[-,line width=0.4000pt,line cap=butt,line join=miter,dash pattern=on 0.4000pt off 2.0000pt](13.6200pt,7.3425pt)--(26.3100pt,7.3425pt)--(26.3100pt,-7.3425pt)--(13.6200pt,-7.3425pt)--cycle;
\path[-,draw=black,line width=0.4000pt,line cap=butt,line join=miter,dash pattern=on 0.4000pt off 2.0000pt](13.6200pt,7.3425pt)--(26.3100pt,7.3425pt)--(26.3100pt,-7.3425pt)--(13.6200pt,-7.3425pt)--cycle;
\path[-,draw=black,line width=0.4000pt,line cap=butt,line join=miter,dash pattern=](30.3300pt,0.0000pt)--(26.3100pt,0.0000pt);
\path[-,draw=black,line width=0.4000pt,line cap=butt,line join=miter,dash pattern=](9.6000pt,0.0000pt)--(13.6200pt,0.0000pt);
\path[-,line width=0.4000pt,line cap=butt,line join=miter,dash pattern=](0.0000pt,11.0175pt)--(0.0000pt,11.0175pt)--(0.0000pt,11.0175pt)--(0.0000pt,11.0175pt)--cycle;
\path[-,draw=black,line width=0.4000pt,line cap=butt,line join=miter,dash pattern=](-3.0000pt,14.0175pt)--(3.0000pt,14.0175pt)--(3.0000pt,8.0175pt)--(-3.0000pt,8.0175pt)--cycle;
\path[-,fill=black,line width=0.4000pt,line cap=butt,line join=miter,dash pattern=](1.0000pt,11.0175pt)..controls(1.0000pt,11.5698pt)and(0.5523pt,12.0175pt)..(0.0000pt,12.0175pt)..controls(-0.5523pt,12.0175pt)and(-1.0000pt,11.5698pt)..(-1.0000pt,11.0175pt)..controls(-1.0000pt,10.4652pt)and(-0.5523pt,10.0175pt)..(0.0000pt,10.0175pt)..controls(0.5523pt,10.0175pt)and(1.0000pt,10.4652pt)..(1.0000pt,11.0175pt)--cycle;
\path[-,draw=black,line width=0.4000pt,line cap=butt,line join=miter,dash pattern=](0.0000pt,11.0175pt)..controls(4.0000pt,17.9457pt)and(1.6000pt,22.0350pt)..(9.6000pt,22.0350pt);
\path[-,draw=black,line width=0.4000pt,line cap=butt,line join=miter,dash pattern=](0.0000pt,11.0175pt)..controls(4.0000pt,4.0893pt)and(1.6000pt,0.0000pt)..(9.6000pt,0.0000pt);
\path[-,draw=black,line width=0.4000pt,line cap=butt,line join=miter,dash pattern=](13.7300pt,-24.5350pt)--(20.7300pt,-24.5350pt);
\path[-,draw=black,line width=0.4000pt,line cap=butt,line join=miter,dash pattern=](13.7300pt,-29.5350pt)--(20.7300pt,-29.5350pt);
\path[-,line width=0.4000pt,line cap=butt,line join=miter,dash pattern=](20.7300pt,-24.5350pt)--(20.7300pt,-24.5350pt)--(20.7300pt,-29.5350pt)--(20.7300pt,-29.5350pt)--cycle;
\path[-,draw=black,line width=0.4000pt,line cap=butt,line join=miter,dash pattern=](17.7300pt,-21.5350pt)--(23.7300pt,-21.5350pt)--(23.7300pt,-32.5350pt)--(17.7300pt,-32.5350pt)--cycle;
\path[-,fill=black,line width=0.4000pt,line cap=butt,line join=miter,dash pattern=](21.7300pt,-24.5350pt)..controls(21.7300pt,-23.9827pt)and(21.2823pt,-23.5350pt)..(20.7300pt,-23.5350pt)..controls(20.1777pt,-23.5350pt)and(19.7300pt,-23.9827pt)..(19.7300pt,-24.5350pt)..controls(19.7300pt,-25.0873pt)and(20.1777pt,-25.5350pt)..(20.7300pt,-25.5350pt)..controls(21.2823pt,-25.5350pt)and(21.7300pt,-25.0873pt)..(21.7300pt,-24.5350pt)--cycle;
\path[-,draw=black,line width=0.4000pt,line cap=butt,line join=miter,dash pattern=](20.7300pt,-29.5350pt)--(30.3300pt,-29.5350pt);
\path[-,line width=0.4000pt,line cap=butt,line join=miter,dash pattern=on 0.4000pt off 2.0000pt](4.0000pt,-22.3425pt)--(9.7300pt,-22.3425pt)--(9.7300pt,-31.7275pt)--(4.0000pt,-31.7275pt)--cycle;
\node[anchor=north west,inner sep=0] at (4.0000pt,-22.3425pt){\savebox{\marxupbox}{{\ensuremath{f}}}\immediate\write\boxesfile{74}\immediate\write\boxesfile{\number\wd\marxupbox}\immediate\write\boxesfile{\number\ht\marxupbox}\immediate\write\boxesfile{\number\dp\marxupbox}\box\marxupbox};
\path[-,line width=0.4000pt,line cap=butt,line join=miter,dash pattern=on 0.4000pt off 2.0000pt](0.0000pt,-18.3425pt)--(13.7300pt,-18.3425pt)--(13.7300pt,-35.7275pt)--(0.0000pt,-35.7275pt)--cycle;
\path[-,draw=black,line width=0.4000pt,line cap=butt,line join=miter,dash pattern=on 0.4000pt off 2.0000pt](0.0000pt,-18.3425pt)--(13.7300pt,-18.3425pt)--(13.7300pt,-35.7275pt)--(0.0000pt,-35.7275pt)--cycle;
\path[-,draw=black,line width=0.4000pt,line cap=butt,line join=miter,dash pattern=](13.7300pt,-24.5350pt)--(13.7300pt,-24.5350pt);
\path[-,draw=black,line width=0.4000pt,line cap=butt,line join=miter,dash pattern=](13.7300pt,-29.5350pt)--(13.7300pt,-29.5350pt);
\path[-,draw=black,line width=0.4000pt,line cap=butt,line join=miter,dash pattern=](0.0000pt,-27.0350pt)--(0.0000pt,-27.0350pt);
\path[-,line width=0.4000pt,line cap=butt,line join=miter,dash pattern=](-9.6000pt,-8.0087pt)--(-9.6000pt,-8.0087pt)--(-9.6000pt,-8.0087pt)--(-9.6000pt,-8.0087pt)--cycle;
\path[-,draw=black,line width=0.4000pt,line cap=butt,line join=miter,dash pattern=](-12.6000pt,-5.0087pt)--(-6.6000pt,-5.0087pt)--(-6.6000pt,-11.0087pt)--(-12.6000pt,-11.0087pt)--cycle;
\path[-,fill=black,line width=0.4000pt,line cap=butt,line join=miter,dash pattern=](-8.6000pt,-8.0087pt)..controls(-8.6000pt,-7.4565pt)and(-9.0477pt,-7.0087pt)..(-9.6000pt,-7.0087pt)..controls(-10.1523pt,-7.0087pt)and(-10.6000pt,-7.4565pt)..(-10.6000pt,-8.0087pt)..controls(-10.6000pt,-8.5610pt)and(-10.1523pt,-9.0087pt)..(-9.6000pt,-9.0087pt)..controls(-9.0477pt,-9.0087pt)and(-8.6000pt,-8.5610pt)..(-8.6000pt,-8.0087pt)--cycle;
\path[-,draw=black,line width=0.4000pt,line cap=butt,line join=miter,dash pattern=](-9.6000pt,-8.0087pt)..controls(-5.6000pt,-1.0805pt)and(-8.0000pt,11.0175pt)..(0.0000pt,11.0175pt);
\path[-,draw=black,line width=0.4000pt,line cap=butt,line join=miter,dash pattern=](-9.6000pt,-8.0087pt)..controls(-5.6000pt,-14.9369pt)and(-8.0000pt,-27.0350pt)..(0.0000pt,-27.0350pt);
\end{tikzpicture}} = 
{\begin{tikzpicture}\path[-,draw=black,line width=0.4000pt,line cap=butt,line join=miter,dash pattern=](-26.2000pt,-37.9812pt)--(-19.2000pt,-37.9812pt);
\path[-,draw=black,line width=0.4000pt,line cap=butt,line join=miter,dash pattern=](49.5300pt,0.0000pt)--(49.5300pt,0.0000pt);
\path[-,draw=black,line width=0.4000pt,line cap=butt,line join=miter,dash pattern=](49.5300pt,-5.0000pt)--(49.5300pt,-5.0000pt);
\path[-,draw=black,line width=0.4000pt,line cap=butt,line join=miter,dash pattern=](49.5300pt,-50.0775pt)--(49.5300pt,-50.0775pt);
\path[-,draw=black,line width=0.4000pt,line cap=butt,line join=miter,dash pattern=](20.7300pt,0.0000pt)--(20.7300pt,0.0000pt);
\path[-,draw=black,line width=0.4000pt,line cap=butt,line join=miter,dash pattern=](20.7300pt,-28.3850pt)--(20.7300pt,-28.3850pt);
\path[-,draw=black,line width=0.4000pt,line cap=butt,line join=miter,dash pattern=](20.7300pt,-61.7700pt)--(20.7300pt,-61.7700pt);
\path[-,draw=black,line width=0.4000pt,line cap=butt,line join=miter,dash pattern=](30.3300pt,0.0000pt)--(30.3300pt,0.0000pt);
\path[-,draw=black,line width=0.4000pt,line cap=butt,line join=miter,dash pattern=](30.3300pt,-33.3850pt)--(30.3300pt,-33.3850pt);
\path[-,draw=black,line width=0.4000pt,line cap=butt,line join=miter,dash pattern=](30.3300pt,-61.7700pt)--(30.3300pt,-61.7700pt);
\path[-,draw=black,line width=0.4000pt,line cap=butt,line join=miter,dash pattern=](39.9300pt,0.0000pt)--(39.9300pt,0.0000pt);
\path[-,draw=black,line width=0.4000pt,line cap=butt,line join=miter,dash pattern=](39.9300pt,-45.0775pt)--(39.9300pt,-45.0775pt);
\path[-,draw=black,line width=0.4000pt,line cap=butt,line join=miter,dash pattern=](39.9300pt,-50.0775pt)--(39.9300pt,-50.0775pt);
\path[-,draw=black,line width=0.4000pt,line cap=butt,line join=miter,dash pattern=](39.9300pt,0.0000pt)--(49.5300pt,0.0000pt);
\path[-,draw=black,line width=0.4000pt,line cap=butt,line join=miter,dash pattern=](39.9300pt,-45.0775pt)..controls(47.9300pt,-45.0775pt)and(41.5300pt,-5.0000pt)..(49.5300pt,-5.0000pt);
\path[-,draw=black,line width=0.4000pt,line cap=butt,line join=miter,dash pattern=](39.9300pt,-50.0775pt)--(49.5300pt,-50.0775pt);
\path[-,draw=black,line width=0.4000pt,line cap=butt,line join=miter,dash pattern=](30.3300pt,0.0000pt)--(39.9300pt,0.0000pt);
\path[-,draw=black,line width=0.4000pt,line cap=butt,line join=miter,dash pattern=](30.3300pt,-33.3850pt)..controls(38.3300pt,-33.3850pt)and(31.9300pt,-50.0775pt)..(39.9300pt,-50.0775pt);
\path[-,draw=black,line width=0.4000pt,line cap=butt,line join=miter,dash pattern=](30.3300pt,-61.7700pt)..controls(38.3300pt,-61.7700pt)and(31.9300pt,-45.0775pt)..(39.9300pt,-45.0775pt);
\path[-,draw=black,line width=0.4000pt,line cap=butt,line join=miter,dash pattern=](20.7300pt,0.0000pt)--(30.3300pt,0.0000pt);
\path[-,draw=black,line width=0.4000pt,line cap=butt,line join=miter,dash pattern=](20.7300pt,-28.3850pt)..controls(28.7300pt,-28.3850pt)and(22.3300pt,-33.3850pt)..(30.3300pt,-33.3850pt);
\path[-,draw=black,line width=0.4000pt,line cap=butt,line join=miter,dash pattern=](20.7300pt,-61.7700pt)--(30.3300pt,-61.7700pt);
\path[-,draw=black,line width=0.4000pt,line cap=butt,line join=miter,dash pattern=](-9.6000pt,-14.1925pt)--(-9.6000pt,-14.1925pt);
\path[-,draw=black,line width=0.4000pt,line cap=butt,line join=miter,dash pattern=](-9.6000pt,-61.7700pt)--(-9.6000pt,-61.7700pt);
\path[-,draw=black,line width=0.4000pt,line cap=butt,line join=miter,dash pattern=](0.0000pt,-2.5000pt)--(0.0000pt,-2.5000pt);
\path[-,draw=black,line width=0.4000pt,line cap=butt,line join=miter,dash pattern=](0.0000pt,-25.8850pt)--(0.0000pt,-25.8850pt);
\path[-,draw=black,line width=0.4000pt,line cap=butt,line join=miter,dash pattern=](13.7300pt,0.0000pt)--(20.7300pt,0.0000pt);
\path[-,draw=black,line width=0.4000pt,line cap=butt,line join=miter,dash pattern=](13.7300pt,-5.0000pt)--(20.7300pt,-5.0000pt);
\path[-,line width=0.4000pt,line cap=butt,line join=miter,dash pattern=](20.7300pt,0.0000pt)--(20.7300pt,0.0000pt)--(20.7300pt,-5.0000pt)--(20.7300pt,-5.0000pt)--cycle;
\path[-,draw=black,line width=0.4000pt,line cap=butt,line join=miter,dash pattern=](17.7300pt,3.0000pt)--(23.7300pt,3.0000pt)--(23.7300pt,-8.0000pt)--(17.7300pt,-8.0000pt)--cycle;
\path[-,fill=black,line width=0.4000pt,line cap=butt,line join=miter,dash pattern=](21.7300pt,-5.0000pt)..controls(21.7300pt,-4.4477pt)and(21.2823pt,-4.0000pt)..(20.7300pt,-4.0000pt)..controls(20.1777pt,-4.0000pt)and(19.7300pt,-4.4477pt)..(19.7300pt,-5.0000pt)..controls(19.7300pt,-5.5523pt)and(20.1777pt,-6.0000pt)..(20.7300pt,-6.0000pt)..controls(21.2823pt,-6.0000pt)and(21.7300pt,-5.5523pt)..(21.7300pt,-5.0000pt)--cycle;
\path[-,draw=black,line width=0.4000pt,line cap=butt,line join=miter,dash pattern=](20.7300pt,0.0000pt)--(20.7300pt,0.0000pt);
\path[-,line width=0.4000pt,line cap=butt,line join=miter,dash pattern=on 0.4000pt off 2.0000pt](4.0000pt,2.1925pt)--(9.7300pt,2.1925pt)--(9.7300pt,-7.1925pt)--(4.0000pt,-7.1925pt)--cycle;
\node[anchor=north west,inner sep=0] at (4.0000pt,2.1925pt){\savebox{\marxupbox}{{\ensuremath{f}}}\immediate\write\boxesfile{75}\immediate\write\boxesfile{\number\wd\marxupbox}\immediate\write\boxesfile{\number\ht\marxupbox}\immediate\write\boxesfile{\number\dp\marxupbox}\box\marxupbox};
\path[-,line width=0.4000pt,line cap=butt,line join=miter,dash pattern=on 0.4000pt off 2.0000pt](0.0000pt,6.1925pt)--(13.7300pt,6.1925pt)--(13.7300pt,-11.1925pt)--(0.0000pt,-11.1925pt)--cycle;
\path[-,draw=black,line width=0.4000pt,line cap=butt,line join=miter,dash pattern=on 0.4000pt off 2.0000pt](0.0000pt,6.1925pt)--(13.7300pt,6.1925pt)--(13.7300pt,-11.1925pt)--(0.0000pt,-11.1925pt)--cycle;
\path[-,draw=black,line width=0.4000pt,line cap=butt,line join=miter,dash pattern=](13.7300pt,0.0000pt)--(13.7300pt,0.0000pt);
\path[-,draw=black,line width=0.4000pt,line cap=butt,line join=miter,dash pattern=](13.7300pt,-5.0000pt)--(13.7300pt,-5.0000pt);
\path[-,draw=black,line width=0.4000pt,line cap=butt,line join=miter,dash pattern=](0.0000pt,-2.5000pt)--(0.0000pt,-2.5000pt);
\path[-,draw=black,line width=0.4000pt,line cap=butt,line join=miter,dash pattern=](13.7300pt,-23.3850pt)--(20.7300pt,-23.3850pt);
\path[-,draw=black,line width=0.4000pt,line cap=butt,line join=miter,dash pattern=](13.7300pt,-28.3850pt)--(20.7300pt,-28.3850pt);
\path[-,line width=0.4000pt,line cap=butt,line join=miter,dash pattern=](20.7300pt,-23.3850pt)--(20.7300pt,-23.3850pt)--(20.7300pt,-28.3850pt)--(20.7300pt,-28.3850pt)--cycle;
\path[-,draw=black,line width=0.4000pt,line cap=butt,line join=miter,dash pattern=](17.7300pt,-20.3850pt)--(23.7300pt,-20.3850pt)--(23.7300pt,-31.3850pt)--(17.7300pt,-31.3850pt)--cycle;
\path[-,fill=black,line width=0.4000pt,line cap=butt,line join=miter,dash pattern=](21.7300pt,-23.3850pt)..controls(21.7300pt,-22.8327pt)and(21.2823pt,-22.3850pt)..(20.7300pt,-22.3850pt)..controls(20.1777pt,-22.3850pt)and(19.7300pt,-22.8327pt)..(19.7300pt,-23.3850pt)..controls(19.7300pt,-23.9373pt)and(20.1777pt,-24.3850pt)..(20.7300pt,-24.3850pt)..controls(21.2823pt,-24.3850pt)and(21.7300pt,-23.9373pt)..(21.7300pt,-23.3850pt)--cycle;
\path[-,draw=black,line width=0.4000pt,line cap=butt,line join=miter,dash pattern=](20.7300pt,-28.3850pt)--(20.7300pt,-28.3850pt);
\path[-,line width=0.4000pt,line cap=butt,line join=miter,dash pattern=on 0.4000pt off 2.0000pt](4.0000pt,-21.1925pt)--(9.7300pt,-21.1925pt)--(9.7300pt,-30.5775pt)--(4.0000pt,-30.5775pt)--cycle;
\node[anchor=north west,inner sep=0] at (4.0000pt,-21.1925pt){\savebox{\marxupbox}{{\ensuremath{f}}}\immediate\write\boxesfile{76}\immediate\write\boxesfile{\number\wd\marxupbox}\immediate\write\boxesfile{\number\ht\marxupbox}\immediate\write\boxesfile{\number\dp\marxupbox}\box\marxupbox};
\path[-,line width=0.4000pt,line cap=butt,line join=miter,dash pattern=on 0.4000pt off 2.0000pt](0.0000pt,-17.1925pt)--(13.7300pt,-17.1925pt)--(13.7300pt,-34.5775pt)--(0.0000pt,-34.5775pt)--cycle;
\path[-,draw=black,line width=0.4000pt,line cap=butt,line join=miter,dash pattern=on 0.4000pt off 2.0000pt](0.0000pt,-17.1925pt)--(13.7300pt,-17.1925pt)--(13.7300pt,-34.5775pt)--(0.0000pt,-34.5775pt)--cycle;
\path[-,draw=black,line width=0.4000pt,line cap=butt,line join=miter,dash pattern=](13.7300pt,-23.3850pt)--(13.7300pt,-23.3850pt);
\path[-,draw=black,line width=0.4000pt,line cap=butt,line join=miter,dash pattern=](13.7300pt,-28.3850pt)--(13.7300pt,-28.3850pt);
\path[-,draw=black,line width=0.4000pt,line cap=butt,line join=miter,dash pattern=](0.0000pt,-25.8850pt)--(0.0000pt,-25.8850pt);
\path[-,line width=0.4000pt,line cap=butt,line join=miter,dash pattern=](-9.6000pt,-14.1925pt)--(-9.6000pt,-14.1925pt)--(-9.6000pt,-14.1925pt)--(-9.6000pt,-14.1925pt)--cycle;
\path[-,draw=black,line width=0.4000pt,line cap=butt,line join=miter,dash pattern=](-12.6000pt,-11.1925pt)--(-6.6000pt,-11.1925pt)--(-6.6000pt,-17.1925pt)--(-12.6000pt,-17.1925pt)--cycle;
\path[-,fill=black,line width=0.4000pt,line cap=butt,line join=miter,dash pattern=](-8.6000pt,-14.1925pt)..controls(-8.6000pt,-13.6402pt)and(-9.0477pt,-13.1925pt)..(-9.6000pt,-13.1925pt)..controls(-10.1523pt,-13.1925pt)and(-10.6000pt,-13.6402pt)..(-10.6000pt,-14.1925pt)..controls(-10.6000pt,-14.7448pt)and(-10.1523pt,-15.1925pt)..(-9.6000pt,-15.1925pt)..controls(-9.0477pt,-15.1925pt)and(-8.6000pt,-14.7448pt)..(-8.6000pt,-14.1925pt)--cycle;
\path[-,draw=black,line width=0.4000pt,line cap=butt,line join=miter,dash pattern=](-9.6000pt,-14.1925pt)..controls(-5.6000pt,-7.2643pt)and(-8.0000pt,-2.5000pt)..(0.0000pt,-2.5000pt);
\path[-,draw=black,line width=0.4000pt,line cap=butt,line join=miter,dash pattern=](-9.6000pt,-14.1925pt)..controls(-5.6000pt,-21.1207pt)and(-8.0000pt,-25.8850pt)..(0.0000pt,-25.8850pt);
\path[-,line width=0.4000pt,line cap=butt,line join=miter,dash pattern=on 0.4000pt off 2.0000pt](3.2200pt,-58.4275pt)--(7.9100pt,-58.4275pt)--(7.9100pt,-65.1124pt)--(3.2200pt,-65.1124pt)--cycle;
\node[anchor=north west,inner sep=0] at (3.2200pt,-58.4275pt){\savebox{\marxupbox}{{\ensuremath{g}}}\immediate\write\boxesfile{77}\immediate\write\boxesfile{\number\wd\marxupbox}\immediate\write\boxesfile{\number\ht\marxupbox}\immediate\write\boxesfile{\number\dp\marxupbox}\box\marxupbox};
\path[-,line width=0.4000pt,line cap=butt,line join=miter,dash pattern=on 0.4000pt off 2.0000pt](-0.7800pt,-54.4275pt)--(11.9100pt,-54.4275pt)--(11.9100pt,-69.1124pt)--(-0.7800pt,-69.1124pt)--cycle;
\path[-,draw=black,line width=0.4000pt,line cap=butt,line join=miter,dash pattern=on 0.4000pt off 2.0000pt](-0.7800pt,-54.4275pt)--(11.9100pt,-54.4275pt)--(11.9100pt,-69.1124pt)--(-0.7800pt,-69.1124pt)--cycle;
\path[-,draw=black,line width=0.4000pt,line cap=butt,line join=miter,dash pattern=](20.7300pt,-61.7700pt)--(11.9100pt,-61.7700pt);
\path[-,draw=black,line width=0.4000pt,line cap=butt,line join=miter,dash pattern=](-9.6000pt,-61.7700pt)--(-0.7800pt,-61.7700pt);
\path[-,line width=0.4000pt,line cap=butt,line join=miter,dash pattern=](-19.2000pt,-37.9812pt)--(-19.2000pt,-37.9812pt)--(-19.2000pt,-37.9812pt)--(-19.2000pt,-37.9812pt)--cycle;
\path[-,draw=black,line width=0.4000pt,line cap=butt,line join=miter,dash pattern=](-22.2000pt,-34.9812pt)--(-16.2000pt,-34.9812pt)--(-16.2000pt,-40.9812pt)--(-22.2000pt,-40.9812pt)--cycle;
\path[-,fill=black,line width=0.4000pt,line cap=butt,line join=miter,dash pattern=](-18.2000pt,-37.9812pt)..controls(-18.2000pt,-37.4289pt)and(-18.6477pt,-36.9812pt)..(-19.2000pt,-36.9812pt)..controls(-19.7523pt,-36.9812pt)and(-20.2000pt,-37.4289pt)..(-20.2000pt,-37.9812pt)..controls(-20.2000pt,-38.5335pt)and(-19.7523pt,-38.9812pt)..(-19.2000pt,-38.9812pt)..controls(-18.6477pt,-38.9812pt)and(-18.2000pt,-38.5335pt)..(-18.2000pt,-37.9812pt)--cycle;
\path[-,draw=black,line width=0.4000pt,line cap=butt,line join=miter,dash pattern=](-19.2000pt,-37.9812pt)..controls(-15.2000pt,-31.0530pt)and(-17.6000pt,-14.1925pt)..(-9.6000pt,-14.1925pt);
\path[-,draw=black,line width=0.4000pt,line cap=butt,line join=miter,dash pattern=](-19.2000pt,-37.9812pt)..controls(-15.2000pt,-44.9094pt)and(-17.6000pt,-61.7700pt)..(-9.6000pt,-61.7700pt);
\end{tikzpicture}} = 
{\begin{tikzpicture}\path[-,draw=black,line width=0.4000pt,line cap=butt,line join=miter,dash pattern=](-46.9300pt,-27.5000pt)--(-39.9300pt,-27.5000pt);
\path[-,draw=black,line width=0.4000pt,line cap=butt,line join=miter,dash pattern=](28.8000pt,0.0000pt)--(28.8000pt,0.0000pt);
\path[-,draw=black,line width=0.4000pt,line cap=butt,line join=miter,dash pattern=](28.8000pt,-5.0000pt)--(28.8000pt,-5.0000pt);
\path[-,draw=black,line width=0.4000pt,line cap=butt,line join=miter,dash pattern=](28.8000pt,-45.0000pt)--(28.8000pt,-45.0000pt);
\path[-,draw=black,line width=0.4000pt,line cap=butt,line join=miter,dash pattern=](0.0000pt,0.0000pt)--(0.0000pt,0.0000pt);
\path[-,draw=black,line width=0.4000pt,line cap=butt,line join=miter,dash pattern=](0.0000pt,-20.0000pt)--(0.0000pt,-20.0000pt);
\path[-,draw=black,line width=0.4000pt,line cap=butt,line join=miter,dash pattern=](0.0000pt,-45.0000pt)--(0.0000pt,-45.0000pt);
\path[-,draw=black,line width=0.4000pt,line cap=butt,line join=miter,dash pattern=](9.6000pt,0.0000pt)--(9.6000pt,0.0000pt);
\path[-,draw=black,line width=0.4000pt,line cap=butt,line join=miter,dash pattern=](9.6000pt,-25.0000pt)--(9.6000pt,-25.0000pt);
\path[-,draw=black,line width=0.4000pt,line cap=butt,line join=miter,dash pattern=](9.6000pt,-45.0000pt)--(9.6000pt,-45.0000pt);
\path[-,draw=black,line width=0.4000pt,line cap=butt,line join=miter,dash pattern=](19.2000pt,0.0000pt)--(19.2000pt,0.0000pt);
\path[-,draw=black,line width=0.4000pt,line cap=butt,line join=miter,dash pattern=](19.2000pt,-25.0000pt)--(19.2000pt,-25.0000pt);
\path[-,draw=black,line width=0.4000pt,line cap=butt,line join=miter,dash pattern=](19.2000pt,-45.0000pt)--(19.2000pt,-45.0000pt);
\path[-,draw=black,line width=0.4000pt,line cap=butt,line join=miter,dash pattern=](19.2000pt,0.0000pt)--(28.8000pt,0.0000pt);
\path[-,draw=black,line width=0.4000pt,line cap=butt,line join=miter,dash pattern=](19.2000pt,-25.0000pt)..controls(27.2000pt,-25.0000pt)and(20.8000pt,-5.0000pt)..(28.8000pt,-5.0000pt);
\path[-,draw=black,line width=0.4000pt,line cap=butt,line join=miter,dash pattern=](19.2000pt,-45.0000pt)--(28.8000pt,-45.0000pt);
\path[-,draw=black,line width=0.4000pt,line cap=butt,line join=miter,dash pattern=](9.6000pt,0.0000pt)--(19.2000pt,0.0000pt);
\path[-,draw=black,line width=0.4000pt,line cap=butt,line join=miter,dash pattern=](9.6000pt,-25.0000pt)..controls(17.6000pt,-25.0000pt)and(11.2000pt,-45.0000pt)..(19.2000pt,-45.0000pt);
\path[-,draw=black,line width=0.4000pt,line cap=butt,line join=miter,dash pattern=](9.6000pt,-45.0000pt)..controls(17.6000pt,-45.0000pt)and(11.2000pt,-25.0000pt)..(19.2000pt,-25.0000pt);
\path[-,draw=black,line width=0.4000pt,line cap=butt,line join=miter,dash pattern=](0.0000pt,0.0000pt)--(9.6000pt,0.0000pt);
\path[-,draw=black,line width=0.4000pt,line cap=butt,line join=miter,dash pattern=](0.0000pt,-20.0000pt)..controls(8.0000pt,-20.0000pt)and(1.6000pt,-25.0000pt)..(9.6000pt,-25.0000pt);
\path[-,draw=black,line width=0.4000pt,line cap=butt,line join=miter,dash pattern=](0.0000pt,-45.0000pt)--(9.6000pt,-45.0000pt);
\path[-,draw=black,line width=0.4000pt,line cap=butt,line join=miter,dash pattern=](-30.3300pt,-10.0000pt)--(-30.3300pt,-10.0000pt);
\path[-,draw=black,line width=0.4000pt,line cap=butt,line join=miter,dash pattern=](-30.3300pt,-45.0000pt)--(-30.3300pt,-45.0000pt);
\path[-,draw=black,line width=0.4000pt,line cap=butt,line join=miter,dash pattern=](-16.6000pt,-7.5000pt)--(-9.6000pt,-7.5000pt);
\path[-,draw=black,line width=0.4000pt,line cap=butt,line join=miter,dash pattern=](-16.6000pt,-12.5000pt)--(-9.6000pt,-12.5000pt);
\path[-,draw=black,line width=0.4000pt,line cap=butt,line join=miter,dash pattern=](0.0000pt,0.0000pt)--(0.0000pt,0.0000pt);
\path[-,draw=black,line width=0.4000pt,line cap=butt,line join=miter,dash pattern=](0.0000pt,-5.0000pt)--(0.0000pt,-5.0000pt);
\path[-,draw=black,line width=0.4000pt,line cap=butt,line join=miter,dash pattern=](0.0000pt,-15.0000pt)--(0.0000pt,-15.0000pt);
\path[-,draw=black,line width=0.4000pt,line cap=butt,line join=miter,dash pattern=](0.0000pt,-20.0000pt)--(0.0000pt,-20.0000pt);
\path[-,line width=0.4000pt,line cap=butt,line join=miter,dash pattern=](0.0000pt,0.0000pt)--(0.0000pt,0.0000pt)--(0.0000pt,-5.0000pt)--(0.0000pt,-5.0000pt)--cycle;
\path[-,draw=black,line width=0.4000pt,line cap=butt,line join=miter,dash pattern=](-3.0000pt,3.0000pt)--(3.0000pt,3.0000pt)--(3.0000pt,-8.0000pt)--(-3.0000pt,-8.0000pt)--cycle;
\path[-,fill=black,line width=0.4000pt,line cap=butt,line join=miter,dash pattern=](1.0000pt,-5.0000pt)..controls(1.0000pt,-4.4477pt)and(0.5523pt,-4.0000pt)..(0.0000pt,-4.0000pt)..controls(-0.5523pt,-4.0000pt)and(-1.0000pt,-4.4477pt)..(-1.0000pt,-5.0000pt)..controls(-1.0000pt,-5.5523pt)and(-0.5523pt,-6.0000pt)..(0.0000pt,-6.0000pt)..controls(0.5523pt,-6.0000pt)and(1.0000pt,-5.5523pt)..(1.0000pt,-5.0000pt)--cycle;
\path[-,draw=black,line width=0.4000pt,line cap=butt,line join=miter,dash pattern=](0.0000pt,0.0000pt)--(0.0000pt,0.0000pt);
\path[-,line width=0.4000pt,line cap=butt,line join=miter,dash pattern=](0.0000pt,-15.0000pt)--(0.0000pt,-15.0000pt)--(0.0000pt,-20.0000pt)--(0.0000pt,-20.0000pt)--cycle;
\path[-,draw=black,line width=0.4000pt,line cap=butt,line join=miter,dash pattern=](-3.0000pt,-12.0000pt)--(3.0000pt,-12.0000pt)--(3.0000pt,-23.0000pt)--(-3.0000pt,-23.0000pt)--cycle;
\path[-,fill=black,line width=0.4000pt,line cap=butt,line join=miter,dash pattern=](1.0000pt,-15.0000pt)..controls(1.0000pt,-14.4477pt)and(0.5523pt,-14.0000pt)..(0.0000pt,-14.0000pt)..controls(-0.5523pt,-14.0000pt)and(-1.0000pt,-14.4477pt)..(-1.0000pt,-15.0000pt)..controls(-1.0000pt,-15.5523pt)and(-0.5523pt,-16.0000pt)..(0.0000pt,-16.0000pt)..controls(0.5523pt,-16.0000pt)and(1.0000pt,-15.5523pt)..(1.0000pt,-15.0000pt)--cycle;
\path[-,draw=black,line width=0.4000pt,line cap=butt,line join=miter,dash pattern=](0.0000pt,-20.0000pt)--(0.0000pt,-20.0000pt);
\path[-,line width=0.4000pt,line cap=butt,line join=miter,dash pattern=](-9.6000pt,-7.5000pt)--(-9.6000pt,-7.5000pt)--(-9.6000pt,-12.5000pt)--(-9.6000pt,-12.5000pt)--cycle;
\path[-,draw=black,line width=0.4000pt,line cap=butt,line join=miter,dash pattern=](-12.6000pt,-4.5000pt)--(-6.6000pt,-4.5000pt)--(-6.6000pt,-15.5000pt)--(-12.6000pt,-15.5000pt)--cycle;
\path[-,fill=black,line width=0.4000pt,line cap=butt,line join=miter,dash pattern=](-8.6000pt,-7.5000pt)..controls(-8.6000pt,-6.9477pt)and(-9.0477pt,-6.5000pt)..(-9.6000pt,-6.5000pt)..controls(-10.1523pt,-6.5000pt)and(-10.6000pt,-6.9477pt)..(-10.6000pt,-7.5000pt)..controls(-10.6000pt,-8.0523pt)and(-10.1523pt,-8.5000pt)..(-9.6000pt,-8.5000pt)..controls(-9.0477pt,-8.5000pt)and(-8.6000pt,-8.0523pt)..(-8.6000pt,-7.5000pt)--cycle;
\path[-,fill=black,line width=0.4000pt,line cap=butt,line join=miter,dash pattern=](-8.6000pt,-12.5000pt)..controls(-8.6000pt,-11.9477pt)and(-9.0477pt,-11.5000pt)..(-9.6000pt,-11.5000pt)..controls(-10.1523pt,-11.5000pt)and(-10.6000pt,-11.9477pt)..(-10.6000pt,-12.5000pt)..controls(-10.6000pt,-13.0523pt)and(-10.1523pt,-13.5000pt)..(-9.6000pt,-13.5000pt)..controls(-9.0477pt,-13.5000pt)and(-8.6000pt,-13.0523pt)..(-8.6000pt,-12.5000pt)--cycle;
\path[-,draw=black,line width=0.4000pt,line cap=butt,line join=miter,dash pattern=](-9.6000pt,-7.5000pt)..controls(-5.6000pt,-0.5718pt)and(-8.0000pt,0.0000pt)..(0.0000pt,0.0000pt);
\path[-,draw=black,line width=0.4000pt,line cap=butt,line join=miter,dash pattern=](-9.6000pt,-12.5000pt)..controls(-5.6000pt,-5.5718pt)and(-8.0000pt,-5.0000pt)..(0.0000pt,-5.0000pt);
\path[-,draw=black,line width=0.4000pt,line cap=butt,line join=miter,dash pattern=](-9.6000pt,-7.5000pt)..controls(-5.6000pt,-14.4282pt)and(-8.0000pt,-15.0000pt)..(0.0000pt,-15.0000pt);
\path[-,draw=black,line width=0.4000pt,line cap=butt,line join=miter,dash pattern=](-9.6000pt,-12.5000pt)..controls(-5.6000pt,-19.4282pt)and(-8.0000pt,-20.0000pt)..(0.0000pt,-20.0000pt);
\path[-,line width=0.4000pt,line cap=butt,line join=miter,dash pattern=on 0.4000pt off 2.0000pt](-26.3300pt,-5.3075pt)--(-20.6000pt,-5.3075pt)--(-20.6000pt,-14.6925pt)--(-26.3300pt,-14.6925pt)--cycle;
\node[anchor=north west,inner sep=0] at (-26.3300pt,-5.3075pt){\savebox{\marxupbox}{{\ensuremath{f}}}\immediate\write\boxesfile{78}\immediate\write\boxesfile{\number\wd\marxupbox}\immediate\write\boxesfile{\number\ht\marxupbox}\immediate\write\boxesfile{\number\dp\marxupbox}\box\marxupbox};
\path[-,line width=0.4000pt,line cap=butt,line join=miter,dash pattern=on 0.4000pt off 2.0000pt](-30.3300pt,-1.3075pt)--(-16.6000pt,-1.3075pt)--(-16.6000pt,-18.6925pt)--(-30.3300pt,-18.6925pt)--cycle;
\path[-,draw=black,line width=0.4000pt,line cap=butt,line join=miter,dash pattern=on 0.4000pt off 2.0000pt](-30.3300pt,-1.3075pt)--(-16.6000pt,-1.3075pt)--(-16.6000pt,-18.6925pt)--(-30.3300pt,-18.6925pt)--cycle;
\path[-,draw=black,line width=0.4000pt,line cap=butt,line join=miter,dash pattern=](-16.6000pt,-7.5000pt)--(-16.6000pt,-7.5000pt);
\path[-,draw=black,line width=0.4000pt,line cap=butt,line join=miter,dash pattern=](-16.6000pt,-12.5000pt)--(-16.6000pt,-12.5000pt);
\path[-,draw=black,line width=0.4000pt,line cap=butt,line join=miter,dash pattern=](-30.3300pt,-10.0000pt)--(-30.3300pt,-10.0000pt);
\path[-,line width=0.4000pt,line cap=butt,line join=miter,dash pattern=on 0.4000pt off 2.0000pt](-17.5100pt,-41.6575pt)--(-12.8200pt,-41.6575pt)--(-12.8200pt,-48.3425pt)--(-17.5100pt,-48.3425pt)--cycle;
\node[anchor=north west,inner sep=0] at (-17.5100pt,-41.6575pt){\savebox{\marxupbox}{{\ensuremath{g}}}\immediate\write\boxesfile{79}\immediate\write\boxesfile{\number\wd\marxupbox}\immediate\write\boxesfile{\number\ht\marxupbox}\immediate\write\boxesfile{\number\dp\marxupbox}\box\marxupbox};
\path[-,line width=0.4000pt,line cap=butt,line join=miter,dash pattern=on 0.4000pt off 2.0000pt](-21.5100pt,-37.6575pt)--(-8.8200pt,-37.6575pt)--(-8.8200pt,-52.3425pt)--(-21.5100pt,-52.3425pt)--cycle;
\path[-,draw=black,line width=0.4000pt,line cap=butt,line join=miter,dash pattern=on 0.4000pt off 2.0000pt](-21.5100pt,-37.6575pt)--(-8.8200pt,-37.6575pt)--(-8.8200pt,-52.3425pt)--(-21.5100pt,-52.3425pt)--cycle;
\path[-,draw=black,line width=0.4000pt,line cap=butt,line join=miter,dash pattern=](0.0000pt,-45.0000pt)--(-8.8200pt,-45.0000pt);
\path[-,draw=black,line width=0.4000pt,line cap=butt,line join=miter,dash pattern=](-30.3300pt,-45.0000pt)--(-21.5100pt,-45.0000pt);
\path[-,line width=0.4000pt,line cap=butt,line join=miter,dash pattern=](-39.9300pt,-27.5000pt)--(-39.9300pt,-27.5000pt)--(-39.9300pt,-27.5000pt)--(-39.9300pt,-27.5000pt)--cycle;
\path[-,draw=black,line width=0.4000pt,line cap=butt,line join=miter,dash pattern=](-42.9300pt,-24.5000pt)--(-36.9300pt,-24.5000pt)--(-36.9300pt,-30.5000pt)--(-42.9300pt,-30.5000pt)--cycle;
\path[-,fill=black,line width=0.4000pt,line cap=butt,line join=miter,dash pattern=](-38.9300pt,-27.5000pt)..controls(-38.9300pt,-26.9477pt)and(-39.3777pt,-26.5000pt)..(-39.9300pt,-26.5000pt)..controls(-40.4823pt,-26.5000pt)and(-40.9300pt,-26.9477pt)..(-40.9300pt,-27.5000pt)..controls(-40.9300pt,-28.0523pt)and(-40.4823pt,-28.5000pt)..(-39.9300pt,-28.5000pt)..controls(-39.3777pt,-28.5000pt)and(-38.9300pt,-28.0523pt)..(-38.9300pt,-27.5000pt)--cycle;
\path[-,draw=black,line width=0.4000pt,line cap=butt,line join=miter,dash pattern=](-39.9300pt,-27.5000pt)..controls(-35.9300pt,-20.5718pt)and(-38.3300pt,-10.0000pt)..(-30.3300pt,-10.0000pt);
\path[-,draw=black,line width=0.4000pt,line cap=butt,line join=miter,dash pattern=](-39.9300pt,-27.5000pt)..controls(-35.9300pt,-34.4282pt)and(-38.3300pt,-45.0000pt)..(-30.3300pt,-45.0000pt);
\end{tikzpicture}} = 
{\begin{tikzpicture}\path[-,draw=black,line width=0.4000pt,line cap=butt,line join=miter,dash pattern=](-7.0000pt,11.0175pt)--(0.0000pt,11.0175pt);
\path[-,draw=black,line width=0.4000pt,line cap=butt,line join=miter,dash pattern=](52.1300pt,24.5350pt)--(52.1300pt,24.5350pt);
\path[-,draw=black,line width=0.4000pt,line cap=butt,line join=miter,dash pattern=](52.1300pt,14.7675pt)--(52.1300pt,14.7675pt);
\path[-,draw=black,line width=0.4000pt,line cap=butt,line join=miter,dash pattern=](52.1300pt,0.0000pt)--(52.1300pt,0.0000pt);
\path[-,draw=black,line width=0.4000pt,line cap=butt,line join=miter,dash pattern=](23.3300pt,24.5350pt)--(23.3300pt,24.5350pt);
\path[-,draw=black,line width=0.4000pt,line cap=butt,line join=miter,dash pattern=](23.3300pt,19.5350pt)--(23.3300pt,19.5350pt);
\path[-,draw=black,line width=0.4000pt,line cap=butt,line join=miter,dash pattern=](23.3300pt,0.0000pt)--(23.3300pt,0.0000pt);
\path[-,draw=black,line width=0.4000pt,line cap=butt,line join=miter,dash pattern=](32.9300pt,24.5350pt)--(32.9300pt,24.5350pt);
\path[-,draw=black,line width=0.4000pt,line cap=butt,line join=miter,dash pattern=](32.9300pt,5.0000pt)--(32.9300pt,5.0000pt);
\path[-,draw=black,line width=0.4000pt,line cap=butt,line join=miter,dash pattern=](32.9300pt,0.0000pt)--(32.9300pt,0.0000pt);
\path[-,draw=black,line width=0.4000pt,line cap=butt,line join=miter,dash pattern=](42.5300pt,24.5350pt)--(42.5300pt,24.5350pt);
\path[-,draw=black,line width=0.4000pt,line cap=butt,line join=miter,dash pattern=](42.5300pt,5.0000pt)--(42.5300pt,5.0000pt);
\path[-,draw=black,line width=0.4000pt,line cap=butt,line join=miter,dash pattern=](42.5300pt,0.0000pt)--(42.5300pt,0.0000pt);
\path[-,draw=black,line width=0.4000pt,line cap=butt,line join=miter,dash pattern=](42.5300pt,24.5350pt)--(52.1300pt,24.5350pt);
\path[-,draw=black,line width=0.4000pt,line cap=butt,line join=miter,dash pattern=](42.5300pt,5.0000pt)..controls(50.5300pt,5.0000pt)and(44.1300pt,14.7675pt)..(52.1300pt,14.7675pt);
\path[-,draw=black,line width=0.4000pt,line cap=butt,line join=miter,dash pattern=](42.5300pt,0.0000pt)--(52.1300pt,0.0000pt);
\path[-,draw=black,line width=0.4000pt,line cap=butt,line join=miter,dash pattern=](32.9300pt,24.5350pt)--(42.5300pt,24.5350pt);
\path[-,draw=black,line width=0.4000pt,line cap=butt,line join=miter,dash pattern=](32.9300pt,5.0000pt)..controls(40.9300pt,5.0000pt)and(34.5300pt,0.0000pt)..(42.5300pt,0.0000pt);
\path[-,draw=black,line width=0.4000pt,line cap=butt,line join=miter,dash pattern=](32.9300pt,0.0000pt)..controls(40.9300pt,0.0000pt)and(34.5300pt,5.0000pt)..(42.5300pt,5.0000pt);
\path[-,draw=black,line width=0.4000pt,line cap=butt,line join=miter,dash pattern=](23.3300pt,24.5350pt)--(32.9300pt,24.5350pt);
\path[-,draw=black,line width=0.4000pt,line cap=butt,line join=miter,dash pattern=](23.3300pt,19.5350pt)..controls(31.3300pt,19.5350pt)and(24.9300pt,5.0000pt)..(32.9300pt,5.0000pt);
\path[-,draw=black,line width=0.4000pt,line cap=butt,line join=miter,dash pattern=](23.3300pt,0.0000pt)--(32.9300pt,0.0000pt);
\path[-,draw=black,line width=0.4000pt,line cap=butt,line join=miter,dash pattern=](9.6000pt,22.0350pt)--(9.6000pt,22.0350pt);
\path[-,draw=black,line width=0.4000pt,line cap=butt,line join=miter,dash pattern=](9.6000pt,0.0000pt)--(9.6000pt,0.0000pt);
\path[-,line width=0.4000pt,line cap=butt,line join=miter,dash pattern=on 0.4000pt off 2.0000pt](13.6000pt,26.7275pt)--(19.3300pt,26.7275pt)--(19.3300pt,17.3425pt)--(13.6000pt,17.3425pt)--cycle;
\node[anchor=north west,inner sep=0] at (13.6000pt,26.7275pt){\savebox{\marxupbox}{{\ensuremath{f}}}\immediate\write\boxesfile{80}\immediate\write\boxesfile{\number\wd\marxupbox}\immediate\write\boxesfile{\number\ht\marxupbox}\immediate\write\boxesfile{\number\dp\marxupbox}\box\marxupbox};
\path[-,line width=0.4000pt,line cap=butt,line join=miter,dash pattern=on 0.4000pt off 2.0000pt](9.6000pt,30.7275pt)--(23.3300pt,30.7275pt)--(23.3300pt,13.3425pt)--(9.6000pt,13.3425pt)--cycle;
\path[-,draw=black,line width=0.4000pt,line cap=butt,line join=miter,dash pattern=on 0.4000pt off 2.0000pt](9.6000pt,30.7275pt)--(23.3300pt,30.7275pt)--(23.3300pt,13.3425pt)--(9.6000pt,13.3425pt)--cycle;
\path[-,draw=black,line width=0.4000pt,line cap=butt,line join=miter,dash pattern=](23.3300pt,24.5350pt)--(23.3300pt,24.5350pt);
\path[-,draw=black,line width=0.4000pt,line cap=butt,line join=miter,dash pattern=](23.3300pt,19.5350pt)--(23.3300pt,19.5350pt);
\path[-,draw=black,line width=0.4000pt,line cap=butt,line join=miter,dash pattern=](9.6000pt,22.0350pt)--(9.6000pt,22.0350pt);
\path[-,line width=0.4000pt,line cap=butt,line join=miter,dash pattern=on 0.4000pt off 2.0000pt](14.1200pt,3.3425pt)--(18.8100pt,3.3425pt)--(18.8100pt,-3.3425pt)--(14.1200pt,-3.3425pt)--cycle;
\node[anchor=north west,inner sep=0] at (14.1200pt,3.3425pt){\savebox{\marxupbox}{{\ensuremath{g}}}\immediate\write\boxesfile{81}\immediate\write\boxesfile{\number\wd\marxupbox}\immediate\write\boxesfile{\number\ht\marxupbox}\immediate\write\boxesfile{\number\dp\marxupbox}\box\marxupbox};
\path[-,line width=0.4000pt,line cap=butt,line join=miter,dash pattern=on 0.4000pt off 2.0000pt](10.1200pt,7.3425pt)--(22.8100pt,7.3425pt)--(22.8100pt,-7.3425pt)--(10.1200pt,-7.3425pt)--cycle;
\path[-,draw=black,line width=0.4000pt,line cap=butt,line join=miter,dash pattern=on 0.4000pt off 2.0000pt](10.1200pt,7.3425pt)--(22.8100pt,7.3425pt)--(22.8100pt,-7.3425pt)--(10.1200pt,-7.3425pt)--cycle;
\path[-,draw=black,line width=0.4000pt,line cap=butt,line join=miter,dash pattern=](23.3300pt,0.0000pt)--(22.8100pt,0.0000pt);
\path[-,draw=black,line width=0.4000pt,line cap=butt,line join=miter,dash pattern=](9.6000pt,0.0000pt)--(10.1200pt,0.0000pt);
\path[-,line width=0.4000pt,line cap=butt,line join=miter,dash pattern=](0.0000pt,11.0175pt)--(0.0000pt,11.0175pt)--(0.0000pt,11.0175pt)--(0.0000pt,11.0175pt)--cycle;
\path[-,draw=black,line width=0.4000pt,line cap=butt,line join=miter,dash pattern=](-3.0000pt,14.0175pt)--(3.0000pt,14.0175pt)--(3.0000pt,8.0175pt)--(-3.0000pt,8.0175pt)--cycle;
\path[-,fill=black,line width=0.4000pt,line cap=butt,line join=miter,dash pattern=](1.0000pt,11.0175pt)..controls(1.0000pt,11.5698pt)and(0.5523pt,12.0175pt)..(0.0000pt,12.0175pt)..controls(-0.5523pt,12.0175pt)and(-1.0000pt,11.5698pt)..(-1.0000pt,11.0175pt)..controls(-1.0000pt,10.4652pt)and(-0.5523pt,10.0175pt)..(0.0000pt,10.0175pt)..controls(0.5523pt,10.0175pt)and(1.0000pt,10.4652pt)..(1.0000pt,11.0175pt)--cycle;
\path[-,draw=black,line width=0.4000pt,line cap=butt,line join=miter,dash pattern=](0.0000pt,11.0175pt)..controls(4.0000pt,17.9457pt)and(1.6000pt,22.0350pt)..(9.6000pt,22.0350pt);
\path[-,draw=black,line width=0.4000pt,line cap=butt,line join=miter,dash pattern=](0.0000pt,11.0175pt)..controls(4.0000pt,4.0893pt)and(1.6000pt,0.0000pt)..(9.6000pt,0.0000pt);
\end{tikzpicture}}\caption{Undoing a split. Two copies of \ensuremath{\mathsf{f}} have been identified.  In the first step re-association is performed.
Then, \ensuremath{\mathsf{f}} is commuted with \ensuremath{\mathsf{δ}}. Finally duplication and
projections are simplified out.}\label{82}\end{figure}
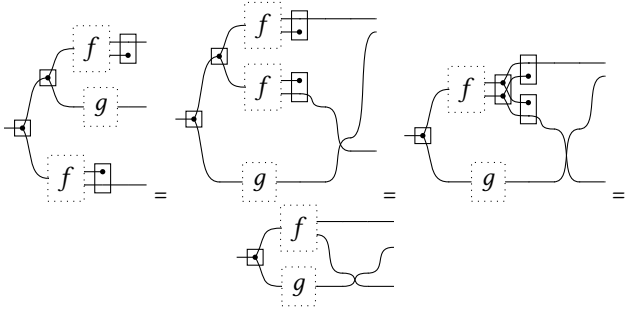 
The proof of \cref{69} gives a clear plan for how to
implement \ensuremath{\mathsf{reduce}}, namely reducing the form \ensuremath{\bigtriangleup} (f
id) by induction until we obtain a morphism in {\sc{}smc} form.

However, there are a couple of difficulties to overcome before we
actually have a usable algorithm.  First, the proof of
\cref{70} proceeds by case analysis on the form of the
input function (\ensuremath{\mathsf{f}\mskip 3.0mu}\ensuremath{\mathnormal{:}\mskip 3.0mu}\ensuremath{\mathsf{P}\mskip 3.0mu}\ensuremath{\mathsf{k}\mskip 3.0mu}\ensuremath{\mathsf{r}\mskip 3.0mu}\ensuremath{\mathsf{a}\mskip 3.0mu}\ensuremath{\mathnormal{⊸}\mskip 0.0mu}\ensuremath{\bigotimes_{i}\mskip 3.0mu}\ensuremath{\allowbreak{}\mathnormal{(}\mskip 0.0mu}\ensuremath{\mathsf{P}\mskip 3.0mu}\ensuremath{\mathsf{k}\mskip 3.0mu}\ensuremath{\mathsf{r}\mskip 3.0mu}\ensuremath{\mathsf{t}_{i}\mskip 0.0mu}\ensuremath{\mathnormal{)}\allowbreak{}}). But without
metaprogramming this form is inaccessible to programs in Haskell:
we only have access to the \ensuremath{\mathsf{FreeCartesian}} representation which is
\emph{produced} by \ensuremath{\mathsf{f}\mskip 3.0mu}\ensuremath{\mathsf{id}}.

Regardless, inspection of the proof of
\cref{70} reveals that the bulk of the work,
namely undoing \ensuremath{\mathsf{split}} operations, can be done by
finding two \ensuremath{\mathsf{FreeSMC}} morphisms of the form \ensuremath{\mathsf{π₁}\mskip 3.0mu}\ensuremath{\allowbreak{}\mathnormal{∘}\allowbreak{}\mskip 3.0mu}\ensuremath{\mathsf{h}} and \ensuremath{\mathsf{π₂}\mskip 3.0mu}\ensuremath{\allowbreak{}\mathnormal{∘}\allowbreak{}\mskip 3.0mu}\ensuremath{\mathsf{h}} in
the operands of \ensuremath{\bigtriangleup}, associate them to \ensuremath{\allowbreak{}\mathnormal{(}\mskip 0.0mu}\ensuremath{\mathsf{π₁}\mskip 3.0mu}\ensuremath{\allowbreak{}\mathnormal{∘}\allowbreak{}\mskip 3.0mu}\ensuremath{\mathsf{h}\mskip 0.0mu}\ensuremath{\mathnormal{)}\allowbreak{}\mskip 3.0mu}\ensuremath{\mathnormal{▵}\mskip 3.0mu}\ensuremath{\allowbreak{}\mathnormal{(}\mskip 0.0mu}\ensuremath{\mathsf{π₂}\mskip 3.0mu}\ensuremath{\allowbreak{}\mathnormal{∘}\allowbreak{}\mskip 3.0mu}\ensuremath{\mathsf{h}\mskip 0.0mu}\ensuremath{\mathnormal{)}\allowbreak{}} and reduce
them to \ensuremath{\mathsf{h}}. If we had access to the host language representation,
we'd know where these operands were.
But we don't: any permutation may be applied to the operands of
\ensuremath{\bigtriangleup}, and therefore an algorithm must start by re-associating them so
that \ensuremath{\mathsf{π₁}\mskip 3.0mu}\ensuremath{\allowbreak{}\mathnormal{∘}\allowbreak{}\mskip 3.0mu}\ensuremath{\mathsf{h}} and \ensuremath{\mathsf{π₂}\mskip 3.0mu}\ensuremath{\allowbreak{}\mathnormal{∘}\allowbreak{}\mskip 3.0mu}\ensuremath{\mathsf{h}} are connected to the same fork (▵).  This step
is illustrated in \cref{82}.  The process can then
continue until all splits have been undone.  A complete example involving several such
steps is depicted graphically in \begin{figure}
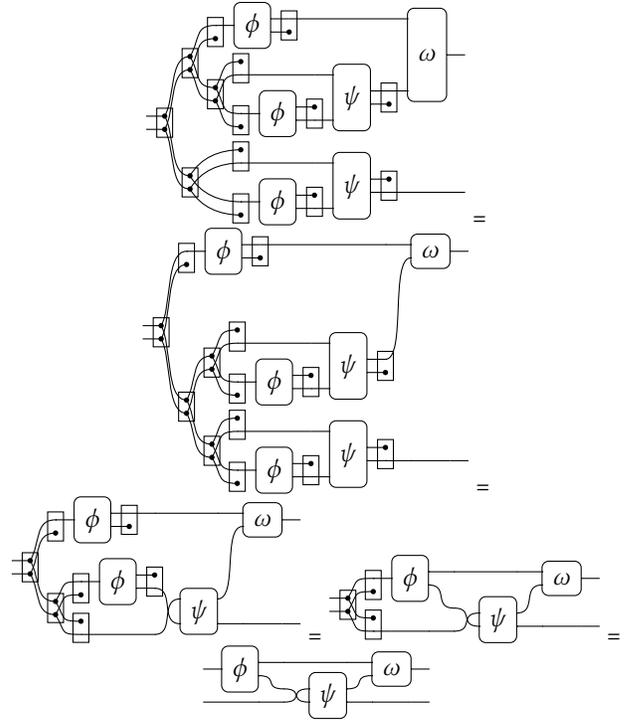

{
}\caption{Example of reduction
steps.}\label{105}\end{figure}\cref{105}.

We remark first that the above procedure is terminating, because
 every transformation reduces the size of the multary merge, as in the
 proof of \cref{70}. The same lemma also tells us
 that what remains after a reduction step is the computational prefix
 of the morphism, which is itself protolinear and thus subject to
 reduction by the same procedure.

Considering all possible re-associations of morphisms and testing for
equal prefixes is expensive. Therefore in our implementation we
maintain the arguments of \ensuremath{\bigtriangleup} as a sorted list of free
cartesian morphisms, \ensuremath{\mathsf{fs}}.
This ordering is defined lexicographically, considering the components
of a composition in computational order (right to left in textual
order).  Additionally, when comparing \ensuremath{\mathsf{f}\mskip 3.0mu}\ensuremath{\allowbreak{}\mathnormal{∘}\allowbreak{}\mskip 3.0mu}\ensuremath{\mathsf{g}} and \ensuremath{\mathsf{f'}\mskip 3.0mu}\ensuremath{\allowbreak{}\mathnormal{∘}\allowbreak{}\mskip 3.0mu}\ensuremath{\mathsf{g'}},
we ensure that neither \ensuremath{\mathsf{g}} nor \ensuremath{\mathsf{g'}} are compositions
themselves (otherwise we re-associate compositions).  This choice of
morphism ordering has two consequences. First, if the morphisms \ensuremath{π₁\mskip 3.0mu}\ensuremath{\allowbreak{}\mathnormal{∘}\allowbreak{}\mskip 3.0mu}\ensuremath{\mathsf{f}} and
\ensuremath{π₂\mskip 3.0mu}\ensuremath{\allowbreak{}\mathnormal{∘}\allowbreak{}\mskip 3.0mu}\ensuremath{\mathsf{f}} are in the sorted list of arguments \ensuremath{\mathsf{fs}}, they
must be adjacent to each other: so such a pair is easy to find.
Second, \ensuremath{\mathsf{f}} and \ensuremath{\mathsf{f'}} are compared only when \ensuremath{\mathsf{g}} and
\ensuremath{\mathsf{g'}} are equal, and this is important in what follows.

One final question remains:
how do we arrange to compare \ensuremath{\mathsf{g}} and
\ensuremath{\mathsf{g'}} if they are generators (say \ensuremath{\mathsf{g}\mskip 3.0mu}\ensuremath{\mathnormal{=}\mskip 3.0mu}\ensuremath{\mathsf{φ}} and \ensuremath{\mathsf{g'}\mskip 3.0mu}\ensuremath{\mathnormal{=}\mskip 3.0mu}\ensuremath{\mathsf{ψ}})? Do
we need to assume a decidable ordering on them? Perhaps surprisingly,
the answer is \emph{no}. Indeed, whenever we would need to compare two generators in the
reduction procedure, it turns out that they are necessarily equal.

This property can be explained by the conjunction of the following two
facts. 1. we compare morphisms only if they have the same source. That
is, when we compare \ensuremath{\mathsf{φ}\mskip 3.0mu}\ensuremath{\allowbreak{}\mathnormal{∘}\allowbreak{}\mskip 3.0mu}\ensuremath{\mathsf{f}} and \ensuremath{\mathsf{ψ}\mskip 3.0mu}\ensuremath{\allowbreak{}\mathnormal{∘}\allowbreak{}\mskip 3.0mu}\ensuremath{\mathsf{f'}}, we consider the generators φ
and ψ only if we already know that \ensuremath{\mathsf{f}\mskip 3.0mu}\ensuremath{\mathnormal{=}\mskip 3.0mu}\ensuremath{\mathsf{f'}} (thanks to using the lexicographical
ordering described above).  2. two generators which have the same
source are necessarily equal. This second property is grounded in
linearity: the same intermediate result can never be used more than
once. Consequently if a generator φ is fed an intermediate result
\ensuremath{\mathsf{x}}, this same \ensuremath{\mathsf{x}} can never be fed to a \emph{different} generator ψ.  (We can end up with two copies of generators in
the representation because \ensuremath{\mathsf{split}} makes such copies.)
\begin{mdframed}[linewidth=0pt,hidealllines,innerleftmargin=0pt,innerrightmargin=0pt,backgroundcolor=gray!15]Because we assume that we have two encoded generators φ and ψ with the
same source, the situation corresponds to them being embedded in a
single a morphism of the form

\begin{list}{}{\setlength\leftmargin{1.0em}}\item\relax
\ensuremath{\begin{parray}\column{B}{@{}>{}l<{}@{}}\column[0em]{1}{@{}>{}l<{}@{}}\column{E}{@{}>{}l<{}@{}}%
\>[1]{}{\mathsf{h}\mskip 3.0mu\allowbreak{}\mathnormal{∘}\allowbreak{}\mskip 3.0mu\allowbreak{}\mathnormal{(}\mskip 0.0mu\allowbreak{}\mathnormal{(}\mskip 0.0mu\mathsf{φ}\mskip 3.0mu\allowbreak{}\mathnormal{∘}\allowbreak{}\mskip 3.0mu\mathsf{f}\mskip 0.0mu\mathnormal{)}\allowbreak{}\mskip 3.0mu\mathnormal{▵}\mskip 3.0mu\allowbreak{}\mathnormal{(}\mskip 0.0mu\mathsf{ψ}\mskip 3.0mu\allowbreak{}\mathnormal{∘}\allowbreak{}\mskip 3.0mu\mathsf{f}\mskip 0.0mu\mathnormal{)}\allowbreak{}\mskip 3.0mu\mathnormal{▵}\mskip 3.0mu\mathsf{g}\mskip 0.0mu\mathnormal{)}\allowbreak{}}\<[E]{}\end{parray}}.\end{list} 
We start by proving the wanted result, but make a couple of additional
assumptions which we discharge later.

\begin{lemma}{}If \ensuremath{\mathsf{h}\mskip 3.0mu}\ensuremath{\allowbreak{}\mathnormal{∘}\allowbreak{}\mskip 3.0mu}\ensuremath{\allowbreak{}\mathnormal{(}\mskip 0.0mu}\ensuremath{\allowbreak{}\mathnormal{(}\mskip 0.0mu}\ensuremath{\mathsf{φ}\mskip 3.0mu}\ensuremath{\allowbreak{}\mathnormal{∘}\allowbreak{}\mskip 3.0mu}\ensuremath{\mathsf{f}\mskip 0.0mu}\ensuremath{\mathnormal{)}\allowbreak{}\mskip 3.0mu}\ensuremath{\mathnormal{▵}\mskip 3.0mu}\ensuremath{\allowbreak{}\mathnormal{(}\mskip 0.0mu}\ensuremath{\mathsf{ψ}\mskip 3.0mu}\ensuremath{\allowbreak{}\mathnormal{∘}\allowbreak{}\mskip 3.0mu}\ensuremath{\mathsf{f}\mskip 0.0mu}\ensuremath{\mathnormal{)}\allowbreak{}\mskip 3.0mu}\ensuremath{\mathnormal{▵}\mskip 3.0mu}\ensuremath{\mathsf{g}\mskip 0.0mu}\ensuremath{\mathnormal{)}\allowbreak{}} is
pseudolinear and \ensuremath{\mathsf{h}} discards neither the output of φ nor of ψ, then
\ensuremath{\mathsf{φ}\mskip 3.0mu}\ensuremath{\mathnormal{=}\mskip 3.0mu}\ensuremath{\mathsf{ψ}}.\label{106}\end{lemma}\begin{proof} 
We have the following equivalence:

\ensuremath{\begin{parray}\column{B}{@{}>{}l<{}@{}}\column[0em]{1}{@{}>{}l<{}@{}}\column{E}{@{}>{}l<{}@{}}%
\>[1]{}{\mathsf{h}\mskip 3.0mu\allowbreak{}\mathnormal{∘}\allowbreak{}\mskip 3.0mu\allowbreak{}\mathnormal{(}\mskip 0.0mu\allowbreak{}\mathnormal{(}\mskip 0.0mu\mathsf{φ}\mskip 3.0mu\allowbreak{}\mathnormal{∘}\allowbreak{}\mskip 3.0mu\mathsf{f}\mskip 0.0mu\mathnormal{)}\allowbreak{}\mskip 3.0mu\mathnormal{▵}\mskip 3.0mu\allowbreak{}\mathnormal{(}\mskip 0.0mu\mathsf{ψ}\mskip 3.0mu\allowbreak{}\mathnormal{∘}\allowbreak{}\mskip 3.0mu\mathsf{f}\mskip 0.0mu\mathnormal{)}\allowbreak{}\mskip 3.0mu\mathnormal{▵}\mskip 3.0mu\mathsf{g}\mskip 0.0mu\mathnormal{)}\allowbreak{}\mskip 3.0mu\mathnormal{=}\mskip 3.0mu\mathsf{h}\mskip 3.0mu\allowbreak{}\mathnormal{∘}\allowbreak{}\mskip 3.0mu\allowbreak{}\mathnormal{(}\mskip 0.0mu\allowbreak{}\mathnormal{(}\mskip 0.0mu\allowbreak{}\mathnormal{(}\mskip 0.0mu\mathsf{φ}\mskip 3.0mu\allowbreak{}\mathnormal{×}\allowbreak{}\mskip 3.0mu\mathsf{ψ}\mskip 0.0mu\mathnormal{)}\allowbreak{}\mskip 3.0mu\allowbreak{}\mathnormal{∘}\allowbreak{}\mskip 3.0muδ\mskip 3.0mu\allowbreak{}\mathnormal{∘}\allowbreak{}\mskip 3.0mu\mathsf{f}\mskip 0.0mu\mathnormal{)}\allowbreak{}\mskip 3.0mu\mathnormal{▵}\mskip 3.0mu\mathsf{g}\mskip 0.0mu\mathnormal{)}\allowbreak{}}\<[E]{}\end{parray}} 
So the morphism can be depicted as follows:

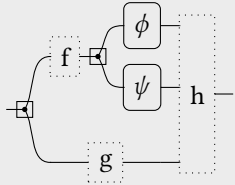
\begin{center}
{\begin{tikzpicture}\path[-,draw=black,line width=0.4000pt,line cap=butt,line join=miter,dash pattern=](-58.4650pt,-8.3462pt)--(-51.4650pt,-8.3462pt);
\path[-,draw=black,line width=0.4000pt,line cap=butt,line join=miter,dash pattern=](20.3800pt,-2.5000pt)--(27.3800pt,-2.5000pt);
\path[-,draw=black,line width=0.4000pt,line cap=butt,line join=miter,dash pattern=](-41.8650pt,11.6925pt)--(-41.8650pt,11.6925pt);
\path[-,draw=black,line width=0.4000pt,line cap=butt,line join=miter,dash pattern=](-41.8650pt,-28.3850pt)--(-41.8650pt,-28.3850pt);
\path[-,draw=black,line width=0.4000pt,line cap=butt,line join=miter,dash pattern=](0.0000pt,23.3850pt)--(7.0000pt,23.3850pt);
\path[-,draw=black,line width=0.4000pt,line cap=butt,line join=miter,dash pattern=](0.0000pt,0.0000pt)--(7.0000pt,0.0000pt);
\path[-,draw=black,line width=0.4000pt,line cap=butt,line join=miter,dash pattern=](0.0000pt,-28.3850pt)--(7.0000pt,-28.3850pt);
\path[-,line width=0.4000pt,line cap=butt,line join=miter,dash pattern=on 0.4000pt off 2.0000pt](11.0000pt,0.9475pt)--(16.3800pt,0.9475pt)--(16.3800pt,-5.9475pt)--(11.0000pt,-5.9475pt)--cycle;
\node[anchor=north west,inner sep=0] at (11.0000pt,0.9475pt){\savebox{\marxupbox}{{h}}\immediate\write\boxesfile{107}\immediate\write\boxesfile{\number\wd\marxupbox}\immediate\write\boxesfile{\number\ht\marxupbox}\immediate\write\boxesfile{\number\dp\marxupbox}\box\marxupbox};
\path[-,line width=0.4000pt,line cap=butt,line join=miter,dash pattern=on 0.4000pt off 2.0000pt](7.0000pt,4.9475pt)--(20.3800pt,4.9475pt)--(20.3800pt,-9.9475pt)--(7.0000pt,-9.9475pt)--cycle;
\path[-,draw=black,line width=0.4000pt,line cap=butt,line join=miter,dash pattern=on 0.4000pt off 2.0000pt](7.0000pt,27.3850pt)--(20.3800pt,27.3850pt)--(20.3800pt,-32.3850pt)--(7.0000pt,-32.3850pt)--cycle;
\path[-,draw=black,line width=0.4000pt,line cap=butt,line join=miter,dash pattern=](20.3800pt,-2.5000pt)--(20.3800pt,-2.5000pt);
\path[-,draw=black,line width=0.4000pt,line cap=butt,line join=miter,dash pattern=](7.0000pt,23.3850pt)--(7.0000pt,23.3850pt);
\path[-,draw=black,line width=0.4000pt,line cap=butt,line join=miter,dash pattern=](7.0000pt,0.0000pt)--(7.0000pt,0.0000pt);
\path[-,draw=black,line width=0.4000pt,line cap=butt,line join=miter,dash pattern=](7.0000pt,-28.3850pt)--(7.0000pt,-28.3850pt);
\path[-,draw=black,line width=0.4000pt,line cap=butt,line join=miter,dash pattern=](-30.7650pt,11.6925pt)--(-23.7650pt,11.6925pt);
\path[-,draw=black,line width=0.4000pt,line cap=butt,line join=miter,dash pattern=](-14.1650pt,23.3850pt)--(-14.1650pt,23.3850pt);
\path[-,draw=black,line width=0.4000pt,line cap=butt,line join=miter,dash pattern=](-14.1650pt,0.0000pt)--(-14.1650pt,0.0000pt);
\path[-,line width=0.4000pt,line cap=butt,line join=miter,dash pattern=](-10.0200pt,28.0775pt)--(-4.1450pt,28.0775pt)--(-4.1450pt,18.6925pt)--(-10.0200pt,18.6925pt)--cycle;
\node[anchor=north west,inner sep=0] at (-10.0200pt,28.0775pt){\savebox{\marxupbox}{{φ}}\immediate\write\boxesfile{108}\immediate\write\boxesfile{\number\wd\marxupbox}\immediate\write\boxesfile{\number\ht\marxupbox}\immediate\write\boxesfile{\number\dp\marxupbox}\box\marxupbox};
\path[-,line width=0.4000pt,line cap=butt,line join=miter,dash pattern=](-14.0200pt,32.0775pt)--(-0.1450pt,32.0775pt)--(-0.1450pt,14.6925pt)--(-14.0200pt,14.6925pt)--cycle;
\path[-,draw=black,line width=0.4000pt,line cap=butt,line join=miter,dash pattern=](-14.0200pt,29.5775pt)..controls(-14.0200pt,30.9582pt)and(-12.9007pt,32.0775pt)..(-11.5200pt,32.0775pt)--(-2.6450pt,32.0775pt)..controls(-1.2643pt,32.0775pt)and(-0.1450pt,30.9582pt)..(-0.1450pt,29.5775pt)--(-0.1450pt,17.1925pt)..controls(-0.1450pt,15.8118pt)and(-1.2643pt,14.6925pt)..(-2.6450pt,14.6925pt)--(-11.5200pt,14.6925pt)..controls(-12.9007pt,14.6925pt)and(-14.0200pt,15.8118pt)..(-14.0200pt,17.1925pt)--cycle;
\path[-,draw=black,line width=0.4000pt,line cap=butt,line join=miter,dash pattern=](0.0000pt,23.3850pt)--(-0.1450pt,23.3850pt);
\path[-,draw=black,line width=0.4000pt,line cap=butt,line join=miter,dash pattern=](-14.1650pt,23.3850pt)--(-14.0200pt,23.3850pt);
\path[-,line width=0.4000pt,line cap=butt,line join=miter,dash pattern=](-10.1650pt,4.6925pt)--(-4.0000pt,4.6925pt)--(-4.0000pt,-4.6925pt)--(-10.1650pt,-4.6925pt)--cycle;
\node[anchor=north west,inner sep=0] at (-10.1650pt,4.6925pt){\savebox{\marxupbox}{{ψ}}\immediate\write\boxesfile{109}\immediate\write\boxesfile{\number\wd\marxupbox}\immediate\write\boxesfile{\number\ht\marxupbox}\immediate\write\boxesfile{\number\dp\marxupbox}\box\marxupbox};
\path[-,line width=0.4000pt,line cap=butt,line join=miter,dash pattern=](-14.1650pt,8.6925pt)--(0.0000pt,8.6925pt)--(0.0000pt,-8.6925pt)--(-14.1650pt,-8.6925pt)--cycle;
\path[-,draw=black,line width=0.4000pt,line cap=butt,line join=miter,dash pattern=](-14.1650pt,6.1925pt)..controls(-14.1650pt,7.5732pt)and(-13.0457pt,8.6925pt)..(-11.6650pt,8.6925pt)--(-2.5000pt,8.6925pt)..controls(-1.1193pt,8.6925pt)and(0.0000pt,7.5732pt)..(0.0000pt,6.1925pt)--(0.0000pt,-6.1925pt)..controls(0.0000pt,-7.5732pt)and(-1.1193pt,-8.6925pt)..(-2.5000pt,-8.6925pt)--(-11.6650pt,-8.6925pt)..controls(-13.0457pt,-8.6925pt)and(-14.1650pt,-7.5732pt)..(-14.1650pt,-6.1925pt)--cycle;
\path[-,draw=black,line width=0.4000pt,line cap=butt,line join=miter,dash pattern=](0.0000pt,0.0000pt)--(0.0000pt,0.0000pt);
\path[-,draw=black,line width=0.4000pt,line cap=butt,line join=miter,dash pattern=](-14.1650pt,0.0000pt)--(-14.1650pt,0.0000pt);
\path[-,line width=0.4000pt,line cap=butt,line join=miter,dash pattern=](-23.7650pt,11.6925pt)--(-23.7650pt,11.6925pt)--(-23.7650pt,11.6925pt)--(-23.7650pt,11.6925pt)--cycle;
\path[-,draw=black,line width=0.4000pt,line cap=butt,line join=miter,dash pattern=](-26.7650pt,14.6925pt)--(-20.7650pt,14.6925pt)--(-20.7650pt,8.6925pt)--(-26.7650pt,8.6925pt)--cycle;
\path[-,fill=black,line width=0.4000pt,line cap=butt,line join=miter,dash pattern=](-22.7650pt,11.6925pt)..controls(-22.7650pt,12.2448pt)and(-23.2127pt,12.6925pt)..(-23.7650pt,12.6925pt)..controls(-24.3173pt,12.6925pt)and(-24.7650pt,12.2448pt)..(-24.7650pt,11.6925pt)..controls(-24.7650pt,11.1402pt)and(-24.3173pt,10.6925pt)..(-23.7650pt,10.6925pt)..controls(-23.2127pt,10.6925pt)and(-22.7650pt,11.1402pt)..(-22.7650pt,11.6925pt)--cycle;
\path[-,draw=black,line width=0.4000pt,line cap=butt,line join=miter,dash pattern=](-23.7650pt,11.6925pt)..controls(-19.7650pt,18.6207pt)and(-22.1650pt,23.3850pt)..(-14.1650pt,23.3850pt);
\path[-,draw=black,line width=0.4000pt,line cap=butt,line join=miter,dash pattern=](-23.7650pt,11.6925pt)..controls(-19.7650pt,4.7643pt)and(-22.1650pt,0.0000pt)..(-14.1650pt,0.0000pt);
\path[-,line width=0.4000pt,line cap=butt,line join=miter,dash pattern=on 0.4000pt off 2.0000pt](-37.8650pt,15.1400pt)--(-34.7650pt,15.1400pt)--(-34.7650pt,8.2450pt)--(-37.8650pt,8.2450pt)--cycle;
\node[anchor=north west,inner sep=0] at (-37.8650pt,15.1400pt){\savebox{\marxupbox}{{f}}\immediate\write\boxesfile{110}\immediate\write\boxesfile{\number\wd\marxupbox}\immediate\write\boxesfile{\number\ht\marxupbox}\immediate\write\boxesfile{\number\dp\marxupbox}\box\marxupbox};
\path[-,line width=0.4000pt,line cap=butt,line join=miter,dash pattern=on 0.4000pt off 2.0000pt](-41.8650pt,19.1400pt)--(-30.7650pt,19.1400pt)--(-30.7650pt,4.2450pt)--(-41.8650pt,4.2450pt)--cycle;
\path[-,draw=black,line width=0.4000pt,line cap=butt,line join=miter,dash pattern=on 0.4000pt off 2.0000pt](-41.8650pt,19.1400pt)--(-30.7650pt,19.1400pt)--(-30.7650pt,4.2450pt)--(-41.8650pt,4.2450pt)--cycle;
\path[-,draw=black,line width=0.4000pt,line cap=butt,line join=miter,dash pattern=](-30.7650pt,11.6925pt)--(-30.7650pt,11.6925pt);
\path[-,draw=black,line width=0.4000pt,line cap=butt,line join=miter,dash pattern=](-41.8650pt,11.6925pt)--(-41.8650pt,11.6925pt);
\path[-,line width=0.4000pt,line cap=butt,line join=miter,dash pattern=on 0.4000pt off 2.0000pt](-23.4325pt,-24.9575pt)--(-18.4325pt,-24.9575pt)--(-18.4325pt,-31.8125pt)--(-23.4325pt,-31.8125pt)--cycle;
\node[anchor=north west,inner sep=0] at (-23.4325pt,-24.9575pt){\savebox{\marxupbox}{{g}}\immediate\write\boxesfile{111}\immediate\write\boxesfile{\number\wd\marxupbox}\immediate\write\boxesfile{\number\ht\marxupbox}\immediate\write\boxesfile{\number\dp\marxupbox}\box\marxupbox};
\path[-,line width=0.4000pt,line cap=butt,line join=miter,dash pattern=on 0.4000pt off 2.0000pt](-27.4325pt,-20.9575pt)--(-14.4325pt,-20.9575pt)--(-14.4325pt,-35.8125pt)--(-27.4325pt,-35.8125pt)--cycle;
\path[-,draw=black,line width=0.4000pt,line cap=butt,line join=miter,dash pattern=on 0.4000pt off 2.0000pt](-27.4325pt,-20.9575pt)--(-14.4325pt,-20.9575pt)--(-14.4325pt,-35.8125pt)--(-27.4325pt,-35.8125pt)--cycle;
\path[-,draw=black,line width=0.4000pt,line cap=butt,line join=miter,dash pattern=](0.0000pt,-28.3850pt)--(-14.4325pt,-28.3850pt);
\path[-,draw=black,line width=0.4000pt,line cap=butt,line join=miter,dash pattern=](-41.8650pt,-28.3850pt)--(-27.4325pt,-28.3850pt);
\path[-,line width=0.4000pt,line cap=butt,line join=miter,dash pattern=](-51.4650pt,-8.3462pt)--(-51.4650pt,-8.3462pt)--(-51.4650pt,-8.3462pt)--(-51.4650pt,-8.3462pt)--cycle;
\path[-,draw=black,line width=0.4000pt,line cap=butt,line join=miter,dash pattern=](-54.4650pt,-5.3462pt)--(-48.4650pt,-5.3462pt)--(-48.4650pt,-11.3462pt)--(-54.4650pt,-11.3462pt)--cycle;
\path[-,fill=black,line width=0.4000pt,line cap=butt,line join=miter,dash pattern=](-50.4650pt,-8.3462pt)..controls(-50.4650pt,-7.7940pt)and(-50.9127pt,-7.3462pt)..(-51.4650pt,-7.3462pt)..controls(-52.0173pt,-7.3462pt)and(-52.4650pt,-7.7940pt)..(-52.4650pt,-8.3462pt)..controls(-52.4650pt,-8.8985pt)and(-52.0173pt,-9.3462pt)..(-51.4650pt,-9.3462pt)..controls(-50.9127pt,-9.3462pt)and(-50.4650pt,-8.8985pt)..(-50.4650pt,-8.3462pt)--cycle;
\path[-,draw=black,line width=0.4000pt,line cap=butt,line join=miter,dash pattern=](-51.4650pt,-8.3462pt)..controls(-47.4650pt,-1.4180pt)and(-49.8650pt,11.6925pt)..(-41.8650pt,11.6925pt);
\path[-,draw=black,line width=0.4000pt,line cap=butt,line join=miter,dash pattern=](-51.4650pt,-8.3462pt)..controls(-47.4650pt,-15.2744pt)and(-49.8650pt,-28.3850pt)..(-41.8650pt,-28.3850pt);
\end{tikzpicture}}\end{center} 
But, we also know that it is pseudo-linear, so it can be put in {\sc{}smc} form. In particular, this means that the \ensuremath{\mathsf{δ}} node connecting φ and ψ can
be eliminated. There are only three ways to reduce this node. We can either
1. assume φ=ψ, and then we can apply the rule \ensuremath{\allowbreak{}\mathnormal{(}\mskip 0.0mu}\ensuremath{\mathsf{φ}\mskip 3.0mu}\ensuremath{\allowbreak{}\mathnormal{×}\allowbreak{}\mskip 3.0mu}\ensuremath{\mathsf{ψ}\mskip 0.0mu}\ensuremath{\mathnormal{)}\allowbreak{}\mskip 3.0mu}\ensuremath{\allowbreak{}\mathnormal{∘}\allowbreak{}\mskip 3.0mu}\ensuremath{δ\mskip 3.0mu}\ensuremath{\mathnormal{=}\mskip 3.0mu}\ensuremath{δ\mskip 3.0mu}\ensuremath{\allowbreak{}\mathnormal{∘}\allowbreak{}\mskip 3.0mu}\ensuremath{\mathsf{φ}}, and let further reductions take place; 
2. prune away one of (or both) the branches; or
3. assume that there is another copy of φ or ψ in \ensuremath{\mathsf{g}} which cause \ensuremath{δ} commutation and elimination.

If we can rule out Case 2 and Case 3, then Case
1. must apply, and we have our result: φ=ψ.

Case 2. corresponds to one of the branches being equivalent to \ensuremath{\mathsf{ε}},
because some discard occurs inside \ensuremath{\mathsf{h}}. Let us assume without
loss of generality that the φ branch is the one equivalent to \ensuremath{\mathsf{ε}}. This
situation is depicted below:

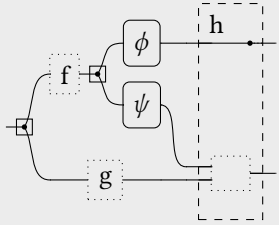
\begin{center}
{\begin{tikzpicture}\path[-,draw=black,line width=0.4000pt,line cap=butt,line join=miter,dash pattern=](-58.4650pt,-8.3462pt)--(-51.4650pt,-8.3462pt);
\path[-,draw=black,line width=0.4000pt,line cap=butt,line join=miter,dash pattern=on 0.4000pt off 1.0000pt](43.4681pt,23.3850pt)--(43.4681pt,23.3850pt);
\path[-,draw=black,line width=0.4000pt,line cap=butt,line join=miter,dash pattern=](43.4681pt,-25.8850pt)--(43.4681pt,-25.8850pt);
\path[-,draw=black,line width=0.4000pt,line cap=butt,line join=miter,dash pattern=](-41.8650pt,11.6925pt)--(-41.8650pt,11.6925pt);
\path[-,draw=black,line width=0.4000pt,line cap=butt,line join=miter,dash pattern=](-41.8650pt,-28.3850pt)--(-41.8650pt,-28.3850pt);
\path[-,draw=black,line width=0.4000pt,line cap=butt,line join=miter,dash pattern=](0.0000pt,23.3850pt)--(0.0000pt,23.3850pt);
\path[-,draw=black,line width=0.4000pt,line cap=butt,line join=miter,dash pattern=](0.0000pt,0.0000pt)--(0.0000pt,0.0000pt);
\path[-,draw=black,line width=0.4000pt,line cap=butt,line join=miter,dash pattern=](0.0000pt,-28.3850pt)--(0.0000pt,-28.3850pt);
\path[-,draw=black,line width=0.4000pt,line cap=butt,line join=miter,dash pattern=](9.6000pt,23.3850pt)--(9.6000pt,23.3850pt);
\path[-,draw=black,line width=0.4000pt,line cap=butt,line join=miter,dash pattern=](9.6000pt,-23.3850pt)--(9.6000pt,-23.3850pt);
\path[-,draw=black,line width=0.4000pt,line cap=butt,line join=miter,dash pattern=](9.6000pt,-28.3850pt)--(9.6000pt,-28.3850pt);
\path[-,line width=0.4000pt,line cap=butt,line join=miter,dash pattern=](18.4000pt,34.3850pt)--(23.7800pt,34.3850pt)--(23.7800pt,27.4900pt)--(18.4000pt,27.4900pt)--cycle;
\node[anchor=north west,inner sep=0] at (18.4000pt,34.3850pt){\savebox{\marxupbox}{{h}}\immediate\write\boxesfile{112}\immediate\write\boxesfile{\number\wd\marxupbox}\immediate\write\boxesfile{\number\ht\marxupbox}\immediate\write\boxesfile{\number\dp\marxupbox}\box\marxupbox};
\path[-,line width=0.4000pt,line cap=butt,line join=miter,dash pattern=](14.4000pt,38.3850pt)--(27.7800pt,38.3850pt)--(27.7800pt,23.4900pt)--(14.4000pt,23.4900pt)--cycle;
\path[-,draw=black,line width=0.4000pt,line cap=butt,line join=miter,dash pattern=on 3.0000pt off 3.0000pt](14.4000pt,38.3850pt)--(38.6681pt,38.3850pt)--(38.6681pt,-43.3850pt)--(14.4000pt,-43.3850pt)--cycle;
\path[-,draw=black,line width=0.4000pt,line cap=butt,line join=miter,dash pattern=](19.2000pt,23.3850pt)--(19.2000pt,23.3850pt);
\path[-,draw=black,line width=0.4000pt,line cap=butt,line join=miter,dash pattern=](19.2000pt,-23.3850pt)--(19.2000pt,-23.3850pt);
\path[-,draw=black,line width=0.4000pt,line cap=butt,line join=miter,dash pattern=](19.2000pt,-28.3850pt)--(19.2000pt,-28.3850pt);
\path[-,draw=black,line width=0.4000pt,line cap=butt,line join=miter,dash pattern=on 0.4000pt off 1.0000pt](33.8681pt,23.3850pt)--(33.8681pt,23.3850pt);
\path[-,draw=black,line width=0.4000pt,line cap=butt,line join=miter,dash pattern=](33.8681pt,-25.8850pt)--(33.8681pt,-25.8850pt);
\path[-,draw=black,line width=0.4000pt,line cap=butt,line join=miter,dash pattern=](9.6000pt,23.3850pt)--(19.2000pt,23.3850pt);
\path[-,draw=black,line width=0.4000pt,line cap=butt,line join=miter,dash pattern=](9.6000pt,-23.3850pt)--(19.2000pt,-23.3850pt);
\path[-,draw=black,line width=0.4000pt,line cap=butt,line join=miter,dash pattern=](9.6000pt,-28.3850pt)--(19.2000pt,-28.3850pt);
\path[-,draw=black,line width=0.4000pt,line cap=butt,line join=miter,dash pattern=](33.8681pt,23.3850pt)--(43.4681pt,23.3850pt);
\path[-,draw=black,line width=0.4000pt,line cap=butt,line join=miter,dash pattern=](33.8681pt,-25.8850pt)--(43.4681pt,-25.8850pt);
\path[-,draw=black,line width=0.4000pt,line cap=butt,line join=miter,dash pattern=on 0.4000pt off 1.0000pt](33.8681pt,23.3850pt)--(38.6681pt,23.3850pt);
\path[-,draw=black,line width=0.4000pt,line cap=butt,line join=miter,dash pattern=](33.8681pt,-25.8850pt)--(38.6681pt,-25.8850pt);
\path[-,draw=black,line width=0.4000pt,line cap=butt,line join=miter,dash pattern=](19.2000pt,23.3850pt)--(14.4000pt,23.3850pt);
\path[-,draw=black,line width=0.4000pt,line cap=butt,line join=miter,dash pattern=](19.2000pt,-23.3850pt)--(14.4000pt,-23.3850pt);
\path[-,draw=black,line width=0.4000pt,line cap=butt,line join=miter,dash pattern=](19.2000pt,-28.3850pt)--(14.4000pt,-28.3850pt);
\path[-,fill=black,line width=0.4000pt,line cap=butt,line join=miter,dash pattern=](34.8681pt,23.3850pt)..controls(34.8681pt,23.9373pt)and(34.4204pt,24.3850pt)..(33.8681pt,24.3850pt)..controls(33.3158pt,24.3850pt)and(32.8681pt,23.9373pt)..(32.8681pt,23.3850pt)..controls(32.8681pt,22.8327pt)and(33.3158pt,22.3850pt)..(33.8681pt,22.3850pt)..controls(34.4204pt,22.3850pt)and(34.8681pt,22.8327pt)..(34.8681pt,23.3850pt)--cycle;
\path[-,draw=black,line width=0.4000pt,line cap=butt,line join=miter,dash pattern=](19.2000pt,23.3850pt)--(33.8681pt,23.3850pt);
\path[-,line width=0.4000pt,line cap=butt,line join=miter,dash pattern=on 0.4000pt off 2.0000pt](23.2000pt,-25.8850pt)--(29.8681pt,-25.8850pt)--(29.8681pt,-25.8850pt)--(23.2000pt,-25.8850pt)--cycle;
\node[anchor=north west,inner sep=0] at (23.2000pt,-25.8850pt){\savebox{\marxupbox}{{    }}\immediate\write\boxesfile{113}\immediate\write\boxesfile{\number\wd\marxupbox}\immediate\write\boxesfile{\number\ht\marxupbox}\immediate\write\boxesfile{\number\dp\marxupbox}\box\marxupbox};
\path[-,line width=0.4000pt,line cap=butt,line join=miter,dash pattern=on 0.4000pt off 2.0000pt](19.2000pt,-21.8850pt)--(33.8681pt,-21.8850pt)--(33.8681pt,-29.8850pt)--(19.2000pt,-29.8850pt)--cycle;
\path[-,draw=black,line width=0.4000pt,line cap=butt,line join=miter,dash pattern=on 0.4000pt off 2.0000pt](19.2000pt,-19.3850pt)--(33.8681pt,-19.3850pt)--(33.8681pt,-32.3850pt)--(19.2000pt,-32.3850pt)--cycle;
\path[-,draw=black,line width=0.4000pt,line cap=butt,line join=miter,dash pattern=](33.8681pt,-25.8850pt)--(33.8681pt,-25.8850pt);
\path[-,draw=black,line width=0.4000pt,line cap=butt,line join=miter,dash pattern=](19.2000pt,-23.3850pt)--(19.2000pt,-23.3850pt);
\path[-,draw=black,line width=0.4000pt,line cap=butt,line join=miter,dash pattern=](19.2000pt,-28.3850pt)--(19.2000pt,-28.3850pt);
\path[-,draw=black,line width=0.4000pt,line cap=butt,line join=miter,dash pattern=](0.0000pt,23.3850pt)--(9.6000pt,23.3850pt);
\path[-,draw=black,line width=0.4000pt,line cap=butt,line join=miter,dash pattern=](0.0000pt,0.0000pt)..controls(8.0000pt,0.0000pt)and(1.6000pt,-23.3850pt)..(9.6000pt,-23.3850pt);
\path[-,draw=black,line width=0.4000pt,line cap=butt,line join=miter,dash pattern=](0.0000pt,-28.3850pt)--(9.6000pt,-28.3850pt);
\path[-,draw=black,line width=0.4000pt,line cap=butt,line join=miter,dash pattern=](-30.7650pt,11.6925pt)--(-23.7650pt,11.6925pt);
\path[-,draw=black,line width=0.4000pt,line cap=butt,line join=miter,dash pattern=](-14.1650pt,23.3850pt)--(-14.1650pt,23.3850pt);
\path[-,draw=black,line width=0.4000pt,line cap=butt,line join=miter,dash pattern=](-14.1650pt,0.0000pt)--(-14.1650pt,0.0000pt);
\path[-,line width=0.4000pt,line cap=butt,line join=miter,dash pattern=](-10.0200pt,28.0775pt)--(-4.1450pt,28.0775pt)--(-4.1450pt,18.6925pt)--(-10.0200pt,18.6925pt)--cycle;
\node[anchor=north west,inner sep=0] at (-10.0200pt,28.0775pt){\savebox{\marxupbox}{{φ}}\immediate\write\boxesfile{114}\immediate\write\boxesfile{\number\wd\marxupbox}\immediate\write\boxesfile{\number\ht\marxupbox}\immediate\write\boxesfile{\number\dp\marxupbox}\box\marxupbox};
\path[-,line width=0.4000pt,line cap=butt,line join=miter,dash pattern=](-14.0200pt,32.0775pt)--(-0.1450pt,32.0775pt)--(-0.1450pt,14.6925pt)--(-14.0200pt,14.6925pt)--cycle;
\path[-,draw=black,line width=0.4000pt,line cap=butt,line join=miter,dash pattern=](-14.0200pt,29.5775pt)..controls(-14.0200pt,30.9582pt)and(-12.9007pt,32.0775pt)..(-11.5200pt,32.0775pt)--(-2.6450pt,32.0775pt)..controls(-1.2643pt,32.0775pt)and(-0.1450pt,30.9582pt)..(-0.1450pt,29.5775pt)--(-0.1450pt,17.1925pt)..controls(-0.1450pt,15.8118pt)and(-1.2643pt,14.6925pt)..(-2.6450pt,14.6925pt)--(-11.5200pt,14.6925pt)..controls(-12.9007pt,14.6925pt)and(-14.0200pt,15.8118pt)..(-14.0200pt,17.1925pt)--cycle;
\path[-,draw=black,line width=0.4000pt,line cap=butt,line join=miter,dash pattern=](0.0000pt,23.3850pt)--(-0.1450pt,23.3850pt);
\path[-,draw=black,line width=0.4000pt,line cap=butt,line join=miter,dash pattern=](-14.1650pt,23.3850pt)--(-14.0200pt,23.3850pt);
\path[-,line width=0.4000pt,line cap=butt,line join=miter,dash pattern=](-10.1650pt,4.6925pt)--(-4.0000pt,4.6925pt)--(-4.0000pt,-4.6925pt)--(-10.1650pt,-4.6925pt)--cycle;
\node[anchor=north west,inner sep=0] at (-10.1650pt,4.6925pt){\savebox{\marxupbox}{{ψ}}\immediate\write\boxesfile{115}\immediate\write\boxesfile{\number\wd\marxupbox}\immediate\write\boxesfile{\number\ht\marxupbox}\immediate\write\boxesfile{\number\dp\marxupbox}\box\marxupbox};
\path[-,line width=0.4000pt,line cap=butt,line join=miter,dash pattern=](-14.1650pt,8.6925pt)--(0.0000pt,8.6925pt)--(0.0000pt,-8.6925pt)--(-14.1650pt,-8.6925pt)--cycle;
\path[-,draw=black,line width=0.4000pt,line cap=butt,line join=miter,dash pattern=](-14.1650pt,6.1925pt)..controls(-14.1650pt,7.5732pt)and(-13.0457pt,8.6925pt)..(-11.6650pt,8.6925pt)--(-2.5000pt,8.6925pt)..controls(-1.1193pt,8.6925pt)and(0.0000pt,7.5732pt)..(0.0000pt,6.1925pt)--(0.0000pt,-6.1925pt)..controls(0.0000pt,-7.5732pt)and(-1.1193pt,-8.6925pt)..(-2.5000pt,-8.6925pt)--(-11.6650pt,-8.6925pt)..controls(-13.0457pt,-8.6925pt)and(-14.1650pt,-7.5732pt)..(-14.1650pt,-6.1925pt)--cycle;
\path[-,draw=black,line width=0.4000pt,line cap=butt,line join=miter,dash pattern=](0.0000pt,0.0000pt)--(0.0000pt,0.0000pt);
\path[-,draw=black,line width=0.4000pt,line cap=butt,line join=miter,dash pattern=](-14.1650pt,0.0000pt)--(-14.1650pt,0.0000pt);
\path[-,line width=0.4000pt,line cap=butt,line join=miter,dash pattern=](-23.7650pt,11.6925pt)--(-23.7650pt,11.6925pt)--(-23.7650pt,11.6925pt)--(-23.7650pt,11.6925pt)--cycle;
\path[-,draw=black,line width=0.4000pt,line cap=butt,line join=miter,dash pattern=](-26.7650pt,14.6925pt)--(-20.7650pt,14.6925pt)--(-20.7650pt,8.6925pt)--(-26.7650pt,8.6925pt)--cycle;
\path[-,fill=black,line width=0.4000pt,line cap=butt,line join=miter,dash pattern=](-22.7650pt,11.6925pt)..controls(-22.7650pt,12.2448pt)and(-23.2127pt,12.6925pt)..(-23.7650pt,12.6925pt)..controls(-24.3173pt,12.6925pt)and(-24.7650pt,12.2448pt)..(-24.7650pt,11.6925pt)..controls(-24.7650pt,11.1402pt)and(-24.3173pt,10.6925pt)..(-23.7650pt,10.6925pt)..controls(-23.2127pt,10.6925pt)and(-22.7650pt,11.1402pt)..(-22.7650pt,11.6925pt)--cycle;
\path[-,draw=black,line width=0.4000pt,line cap=butt,line join=miter,dash pattern=](-23.7650pt,11.6925pt)..controls(-19.7650pt,18.6207pt)and(-22.1650pt,23.3850pt)..(-14.1650pt,23.3850pt);
\path[-,draw=black,line width=0.4000pt,line cap=butt,line join=miter,dash pattern=](-23.7650pt,11.6925pt)..controls(-19.7650pt,4.7643pt)and(-22.1650pt,0.0000pt)..(-14.1650pt,0.0000pt);
\path[-,line width=0.4000pt,line cap=butt,line join=miter,dash pattern=on 0.4000pt off 2.0000pt](-37.8650pt,15.1400pt)--(-34.7650pt,15.1400pt)--(-34.7650pt,8.2450pt)--(-37.8650pt,8.2450pt)--cycle;
\node[anchor=north west,inner sep=0] at (-37.8650pt,15.1400pt){\savebox{\marxupbox}{{f}}\immediate\write\boxesfile{116}\immediate\write\boxesfile{\number\wd\marxupbox}\immediate\write\boxesfile{\number\ht\marxupbox}\immediate\write\boxesfile{\number\dp\marxupbox}\box\marxupbox};
\path[-,line width=0.4000pt,line cap=butt,line join=miter,dash pattern=on 0.4000pt off 2.0000pt](-41.8650pt,19.1400pt)--(-30.7650pt,19.1400pt)--(-30.7650pt,4.2450pt)--(-41.8650pt,4.2450pt)--cycle;
\path[-,draw=black,line width=0.4000pt,line cap=butt,line join=miter,dash pattern=on 0.4000pt off 2.0000pt](-41.8650pt,19.1400pt)--(-30.7650pt,19.1400pt)--(-30.7650pt,4.2450pt)--(-41.8650pt,4.2450pt)--cycle;
\path[-,draw=black,line width=0.4000pt,line cap=butt,line join=miter,dash pattern=](-30.7650pt,11.6925pt)--(-30.7650pt,11.6925pt);
\path[-,draw=black,line width=0.4000pt,line cap=butt,line join=miter,dash pattern=](-41.8650pt,11.6925pt)--(-41.8650pt,11.6925pt);
\path[-,line width=0.4000pt,line cap=butt,line join=miter,dash pattern=on 0.4000pt off 2.0000pt](-23.4325pt,-24.9575pt)--(-18.4325pt,-24.9575pt)--(-18.4325pt,-31.8125pt)--(-23.4325pt,-31.8125pt)--cycle;
\node[anchor=north west,inner sep=0] at (-23.4325pt,-24.9575pt){\savebox{\marxupbox}{{g}}\immediate\write\boxesfile{117}\immediate\write\boxesfile{\number\wd\marxupbox}\immediate\write\boxesfile{\number\ht\marxupbox}\immediate\write\boxesfile{\number\dp\marxupbox}\box\marxupbox};
\path[-,line width=0.4000pt,line cap=butt,line join=miter,dash pattern=on 0.4000pt off 2.0000pt](-27.4325pt,-20.9575pt)--(-14.4325pt,-20.9575pt)--(-14.4325pt,-35.8125pt)--(-27.4325pt,-35.8125pt)--cycle;
\path[-,draw=black,line width=0.4000pt,line cap=butt,line join=miter,dash pattern=on 0.4000pt off 2.0000pt](-27.4325pt,-20.9575pt)--(-14.4325pt,-20.9575pt)--(-14.4325pt,-35.8125pt)--(-27.4325pt,-35.8125pt)--cycle;
\path[-,draw=black,line width=0.4000pt,line cap=butt,line join=miter,dash pattern=](0.0000pt,-28.3850pt)--(-14.4325pt,-28.3850pt);
\path[-,draw=black,line width=0.4000pt,line cap=butt,line join=miter,dash pattern=](-41.8650pt,-28.3850pt)--(-27.4325pt,-28.3850pt);
\path[-,line width=0.4000pt,line cap=butt,line join=miter,dash pattern=](-51.4650pt,-8.3462pt)--(-51.4650pt,-8.3462pt)--(-51.4650pt,-8.3462pt)--(-51.4650pt,-8.3462pt)--cycle;
\path[-,draw=black,line width=0.4000pt,line cap=butt,line join=miter,dash pattern=](-54.4650pt,-5.3462pt)--(-48.4650pt,-5.3462pt)--(-48.4650pt,-11.3462pt)--(-54.4650pt,-11.3462pt)--cycle;
\path[-,fill=black,line width=0.4000pt,line cap=butt,line join=miter,dash pattern=](-50.4650pt,-8.3462pt)..controls(-50.4650pt,-7.7940pt)and(-50.9127pt,-7.3462pt)..(-51.4650pt,-7.3462pt)..controls(-52.0173pt,-7.3462pt)and(-52.4650pt,-7.7940pt)..(-52.4650pt,-8.3462pt)..controls(-52.4650pt,-8.8985pt)and(-52.0173pt,-9.3462pt)..(-51.4650pt,-9.3462pt)..controls(-50.9127pt,-9.3462pt)and(-50.4650pt,-8.8985pt)..(-50.4650pt,-8.3462pt)--cycle;
\path[-,draw=black,line width=0.4000pt,line cap=butt,line join=miter,dash pattern=](-51.4650pt,-8.3462pt)..controls(-47.4650pt,-1.4180pt)and(-49.8650pt,11.6925pt)..(-41.8650pt,11.6925pt);
\path[-,draw=black,line width=0.4000pt,line cap=butt,line join=miter,dash pattern=](-51.4650pt,-8.3462pt)..controls(-47.4650pt,-15.2744pt)and(-49.8650pt,-28.3850pt)..(-41.8650pt,-28.3850pt);
\end{tikzpicture}}\end{center} 
Indeed, the only way that this branch can be pruned is when the output
of φ is discarded. However, by assumption, we have rejected this
situation.

Case 3. can only happen when \ensuremath{\mathsf{g}} is of the form
\ensuremath{\allowbreak{}\mathnormal{(}\mskip 0.0mu}\ensuremath{\mathsf{φ}\mskip 3.0mu}\ensuremath{\mathnormal{▵}\mskip 3.0mu}\ensuremath{\mathsf{i}\mskip 0.0mu}\ensuremath{\mathnormal{)}\allowbreak{}\mskip 3.0mu}\ensuremath{\allowbreak{}\mathnormal{∘}\allowbreak{}\mskip 3.0mu}\ensuremath{\mathsf{f}} or \ensuremath{\allowbreak{}\mathnormal{(}\mskip 0.0mu}\ensuremath{\mathsf{ψ}\mskip 3.0mu}\ensuremath{\mathnormal{▵}\mskip 3.0mu}\ensuremath{\mathsf{i}\mskip 0.0mu}\ensuremath{\mathnormal{)}\allowbreak{}\mskip 3.0mu}\ensuremath{\allowbreak{}\mathnormal{∘}\allowbreak{}\mskip 3.0mu}\ensuremath{\mathsf{f}}. Let us assume the latter
without loss of generality. The situation is then:

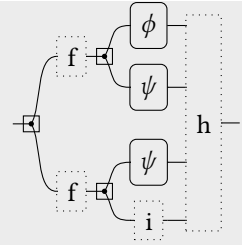
\begin{center}
{\begin{tikzpicture}\path[-,draw=black,line width=0.4000pt,line cap=butt,line join=miter,dash pattern=](-58.4650pt,14.5425pt)--(-51.4650pt,14.5425pt);
\path[-,draw=black,line width=0.4000pt,line cap=butt,line join=miter,dash pattern=](20.3800pt,14.8925pt)--(27.3800pt,14.8925pt);
\path[-,draw=black,line width=0.4000pt,line cap=butt,line join=miter,dash pattern=](-41.8650pt,40.0775pt)--(-41.8650pt,40.0775pt);
\path[-,draw=black,line width=0.4000pt,line cap=butt,line join=miter,dash pattern=](-41.8650pt,-10.9925pt)--(-41.8650pt,-10.9925pt);
\path[-,draw=black,line width=0.4000pt,line cap=butt,line join=miter,dash pattern=](0.0000pt,51.7700pt)--(7.0000pt,51.7700pt);
\path[-,draw=black,line width=0.4000pt,line cap=butt,line join=miter,dash pattern=](0.0000pt,28.3850pt)--(7.0000pt,28.3850pt);
\path[-,draw=black,line width=0.4000pt,line cap=butt,line join=miter,dash pattern=](0.0000pt,0.0000pt)--(7.0000pt,0.0000pt);
\path[-,draw=black,line width=0.4000pt,line cap=butt,line join=miter,dash pattern=](0.0000pt,-21.9850pt)--(7.0000pt,-21.9850pt);
\path[-,line width=0.4000pt,line cap=butt,line join=miter,dash pattern=on 0.4000pt off 2.0000pt](11.0000pt,18.3400pt)--(16.3800pt,18.3400pt)--(16.3800pt,11.4450pt)--(11.0000pt,11.4450pt)--cycle;
\node[anchor=north west,inner sep=0] at (11.0000pt,18.3400pt){\savebox{\marxupbox}{{h}}\immediate\write\boxesfile{118}\immediate\write\boxesfile{\number\wd\marxupbox}\immediate\write\boxesfile{\number\ht\marxupbox}\immediate\write\boxesfile{\number\dp\marxupbox}\box\marxupbox};
\path[-,line width=0.4000pt,line cap=butt,line join=miter,dash pattern=on 0.4000pt off 2.0000pt](7.0000pt,22.3400pt)--(20.3800pt,22.3400pt)--(20.3800pt,7.4450pt)--(7.0000pt,7.4450pt)--cycle;
\path[-,draw=black,line width=0.4000pt,line cap=butt,line join=miter,dash pattern=on 0.4000pt off 2.0000pt](7.0000pt,55.7700pt)--(20.3800pt,55.7700pt)--(20.3800pt,-25.9850pt)--(7.0000pt,-25.9850pt)--cycle;
\path[-,draw=black,line width=0.4000pt,line cap=butt,line join=miter,dash pattern=](20.3800pt,14.8925pt)--(20.3800pt,14.8925pt);
\path[-,draw=black,line width=0.4000pt,line cap=butt,line join=miter,dash pattern=](7.0000pt,51.7700pt)--(7.0000pt,51.7700pt);
\path[-,draw=black,line width=0.4000pt,line cap=butt,line join=miter,dash pattern=](7.0000pt,28.3850pt)--(7.0000pt,28.3850pt);
\path[-,draw=black,line width=0.4000pt,line cap=butt,line join=miter,dash pattern=](7.0000pt,0.0000pt)--(7.0000pt,0.0000pt);
\path[-,draw=black,line width=0.4000pt,line cap=butt,line join=miter,dash pattern=](7.0000pt,-21.9850pt)--(7.0000pt,-21.9850pt);
\path[-,draw=black,line width=0.4000pt,line cap=butt,line join=miter,dash pattern=](-30.7650pt,40.0775pt)--(-23.7650pt,40.0775pt);
\path[-,draw=black,line width=0.4000pt,line cap=butt,line join=miter,dash pattern=](-14.1650pt,51.7700pt)--(-14.1650pt,51.7700pt);
\path[-,draw=black,line width=0.4000pt,line cap=butt,line join=miter,dash pattern=](-14.1650pt,28.3850pt)--(-14.1650pt,28.3850pt);
\path[-,line width=0.4000pt,line cap=butt,line join=miter,dash pattern=](-10.0200pt,56.4624pt)--(-4.1450pt,56.4624pt)--(-4.1450pt,47.0775pt)--(-10.0200pt,47.0775pt)--cycle;
\node[anchor=north west,inner sep=0] at (-10.0200pt,56.4624pt){\savebox{\marxupbox}{{φ}}\immediate\write\boxesfile{119}\immediate\write\boxesfile{\number\wd\marxupbox}\immediate\write\boxesfile{\number\ht\marxupbox}\immediate\write\boxesfile{\number\dp\marxupbox}\box\marxupbox};
\path[-,line width=0.4000pt,line cap=butt,line join=miter,dash pattern=](-14.0200pt,60.4624pt)--(-0.1450pt,60.4624pt)--(-0.1450pt,43.0775pt)--(-14.0200pt,43.0775pt)--cycle;
\path[-,draw=black,line width=0.4000pt,line cap=butt,line join=miter,dash pattern=](-14.0200pt,57.9624pt)..controls(-14.0200pt,59.3432pt)and(-12.9007pt,60.4624pt)..(-11.5200pt,60.4624pt)--(-2.6450pt,60.4624pt)..controls(-1.2643pt,60.4624pt)and(-0.1450pt,59.3432pt)..(-0.1450pt,57.9624pt)--(-0.1450pt,45.5775pt)..controls(-0.1450pt,44.1968pt)and(-1.2643pt,43.0775pt)..(-2.6450pt,43.0775pt)--(-11.5200pt,43.0775pt)..controls(-12.9007pt,43.0775pt)and(-14.0200pt,44.1968pt)..(-14.0200pt,45.5775pt)--cycle;
\path[-,draw=black,line width=0.4000pt,line cap=butt,line join=miter,dash pattern=](0.0000pt,51.7700pt)--(-0.1450pt,51.7700pt);
\path[-,draw=black,line width=0.4000pt,line cap=butt,line join=miter,dash pattern=](-14.1650pt,51.7700pt)--(-14.0200pt,51.7700pt);
\path[-,line width=0.4000pt,line cap=butt,line join=miter,dash pattern=](-10.1650pt,33.0775pt)--(-4.0000pt,33.0775pt)--(-4.0000pt,23.6925pt)--(-10.1650pt,23.6925pt)--cycle;
\node[anchor=north west,inner sep=0] at (-10.1650pt,33.0775pt){\savebox{\marxupbox}{{ψ}}\immediate\write\boxesfile{120}\immediate\write\boxesfile{\number\wd\marxupbox}\immediate\write\boxesfile{\number\ht\marxupbox}\immediate\write\boxesfile{\number\dp\marxupbox}\box\marxupbox};
\path[-,line width=0.4000pt,line cap=butt,line join=miter,dash pattern=](-14.1650pt,37.0775pt)--(0.0000pt,37.0775pt)--(0.0000pt,19.6925pt)--(-14.1650pt,19.6925pt)--cycle;
\path[-,draw=black,line width=0.4000pt,line cap=butt,line join=miter,dash pattern=](-14.1650pt,34.5775pt)..controls(-14.1650pt,35.9582pt)and(-13.0457pt,37.0775pt)..(-11.6650pt,37.0775pt)--(-2.5000pt,37.0775pt)..controls(-1.1193pt,37.0775pt)and(0.0000pt,35.9582pt)..(0.0000pt,34.5775pt)--(0.0000pt,22.1925pt)..controls(0.0000pt,20.8118pt)and(-1.1193pt,19.6925pt)..(-2.5000pt,19.6925pt)--(-11.6650pt,19.6925pt)..controls(-13.0457pt,19.6925pt)and(-14.1650pt,20.8118pt)..(-14.1650pt,22.1925pt)--cycle;
\path[-,draw=black,line width=0.4000pt,line cap=butt,line join=miter,dash pattern=](0.0000pt,28.3850pt)--(0.0000pt,28.3850pt);
\path[-,draw=black,line width=0.4000pt,line cap=butt,line join=miter,dash pattern=](-14.1650pt,28.3850pt)--(-14.1650pt,28.3850pt);
\path[-,line width=0.4000pt,line cap=butt,line join=miter,dash pattern=](-23.7650pt,40.0775pt)--(-23.7650pt,40.0775pt)--(-23.7650pt,40.0775pt)--(-23.7650pt,40.0775pt)--cycle;
\path[-,draw=black,line width=0.4000pt,line cap=butt,line join=miter,dash pattern=](-26.7650pt,43.0775pt)--(-20.7650pt,43.0775pt)--(-20.7650pt,37.0775pt)--(-26.7650pt,37.0775pt)--cycle;
\path[-,fill=black,line width=0.4000pt,line cap=butt,line join=miter,dash pattern=](-22.7650pt,40.0775pt)..controls(-22.7650pt,40.6298pt)and(-23.2127pt,41.0775pt)..(-23.7650pt,41.0775pt)..controls(-24.3173pt,41.0775pt)and(-24.7650pt,40.6298pt)..(-24.7650pt,40.0775pt)..controls(-24.7650pt,39.5252pt)and(-24.3173pt,39.0775pt)..(-23.7650pt,39.0775pt)..controls(-23.2127pt,39.0775pt)and(-22.7650pt,39.5252pt)..(-22.7650pt,40.0775pt)--cycle;
\path[-,draw=black,line width=0.4000pt,line cap=butt,line join=miter,dash pattern=](-23.7650pt,40.0775pt)..controls(-19.7650pt,47.0057pt)and(-22.1650pt,51.7700pt)..(-14.1650pt,51.7700pt);
\path[-,draw=black,line width=0.4000pt,line cap=butt,line join=miter,dash pattern=](-23.7650pt,40.0775pt)..controls(-19.7650pt,33.1493pt)and(-22.1650pt,28.3850pt)..(-14.1650pt,28.3850pt);
\path[-,line width=0.4000pt,line cap=butt,line join=miter,dash pattern=on 0.4000pt off 2.0000pt](-37.8650pt,43.5250pt)--(-34.7650pt,43.5250pt)--(-34.7650pt,36.6300pt)--(-37.8650pt,36.6300pt)--cycle;
\node[anchor=north west,inner sep=0] at (-37.8650pt,43.5250pt){\savebox{\marxupbox}{{f}}\immediate\write\boxesfile{121}\immediate\write\boxesfile{\number\wd\marxupbox}\immediate\write\boxesfile{\number\ht\marxupbox}\immediate\write\boxesfile{\number\dp\marxupbox}\box\marxupbox};
\path[-,line width=0.4000pt,line cap=butt,line join=miter,dash pattern=on 0.4000pt off 2.0000pt](-41.8650pt,47.5250pt)--(-30.7650pt,47.5250pt)--(-30.7650pt,32.6300pt)--(-41.8650pt,32.6300pt)--cycle;
\path[-,draw=black,line width=0.4000pt,line cap=butt,line join=miter,dash pattern=on 0.4000pt off 2.0000pt](-41.8650pt,47.5250pt)--(-30.7650pt,47.5250pt)--(-30.7650pt,32.6300pt)--(-41.8650pt,32.6300pt)--cycle;
\path[-,draw=black,line width=0.4000pt,line cap=butt,line join=miter,dash pattern=](-30.7650pt,40.0775pt)--(-30.7650pt,40.0775pt);
\path[-,draw=black,line width=0.4000pt,line cap=butt,line join=miter,dash pattern=](-41.8650pt,40.0775pt)--(-41.8650pt,40.0775pt);
\path[-,draw=black,line width=0.4000pt,line cap=butt,line join=miter,dash pattern=](-30.7650pt,-10.9925pt)--(-23.7650pt,-10.9925pt);
\path[-,draw=black,line width=0.4000pt,line cap=butt,line join=miter,dash pattern=](-14.1650pt,0.0000pt)--(-14.1650pt,0.0000pt);
\path[-,draw=black,line width=0.4000pt,line cap=butt,line join=miter,dash pattern=](-14.1650pt,-21.9850pt)--(-14.1650pt,-21.9850pt);
\path[-,line width=0.4000pt,line cap=butt,line join=miter,dash pattern=](-10.1650pt,4.6925pt)--(-4.0000pt,4.6925pt)--(-4.0000pt,-4.6925pt)--(-10.1650pt,-4.6925pt)--cycle;
\node[anchor=north west,inner sep=0] at (-10.1650pt,4.6925pt){\savebox{\marxupbox}{{ψ}}\immediate\write\boxesfile{122}\immediate\write\boxesfile{\number\wd\marxupbox}\immediate\write\boxesfile{\number\ht\marxupbox}\immediate\write\boxesfile{\number\dp\marxupbox}\box\marxupbox};
\path[-,line width=0.4000pt,line cap=butt,line join=miter,dash pattern=](-14.1650pt,8.6925pt)--(0.0000pt,8.6925pt)--(0.0000pt,-8.6925pt)--(-14.1650pt,-8.6925pt)--cycle;
\path[-,draw=black,line width=0.4000pt,line cap=butt,line join=miter,dash pattern=](-14.1650pt,6.1925pt)..controls(-14.1650pt,7.5732pt)and(-13.0457pt,8.6925pt)..(-11.6650pt,8.6925pt)--(-2.5000pt,8.6925pt)..controls(-1.1193pt,8.6925pt)and(0.0000pt,7.5732pt)..(0.0000pt,6.1925pt)--(0.0000pt,-6.1925pt)..controls(0.0000pt,-7.5732pt)and(-1.1193pt,-8.6925pt)..(-2.5000pt,-8.6925pt)--(-11.6650pt,-8.6925pt)..controls(-13.0457pt,-8.6925pt)and(-14.1650pt,-7.5732pt)..(-14.1650pt,-6.1925pt)--cycle;
\path[-,draw=black,line width=0.4000pt,line cap=butt,line join=miter,dash pattern=](0.0000pt,0.0000pt)--(0.0000pt,0.0000pt);
\path[-,draw=black,line width=0.4000pt,line cap=butt,line join=miter,dash pattern=](-14.1650pt,0.0000pt)--(-14.1650pt,0.0000pt);
\path[-,line width=0.4000pt,line cap=butt,line join=miter,dash pattern=on 0.4000pt off 2.0000pt](-8.4375pt,-18.6925pt)--(-5.7275pt,-18.6925pt)--(-5.7275pt,-25.2775pt)--(-8.4375pt,-25.2775pt)--cycle;
\node[anchor=north west,inner sep=0] at (-8.4375pt,-18.6925pt){\savebox{\marxupbox}{{i}}\immediate\write\boxesfile{123}\immediate\write\boxesfile{\number\wd\marxupbox}\immediate\write\boxesfile{\number\ht\marxupbox}\immediate\write\boxesfile{\number\dp\marxupbox}\box\marxupbox};
\path[-,line width=0.4000pt,line cap=butt,line join=miter,dash pattern=on 0.4000pt off 2.0000pt](-12.4375pt,-14.6925pt)--(-1.7275pt,-14.6925pt)--(-1.7275pt,-29.2775pt)--(-12.4375pt,-29.2775pt)--cycle;
\path[-,draw=black,line width=0.4000pt,line cap=butt,line join=miter,dash pattern=on 0.4000pt off 2.0000pt](-12.4375pt,-14.6925pt)--(-1.7275pt,-14.6925pt)--(-1.7275pt,-29.2775pt)--(-12.4375pt,-29.2775pt)--cycle;
\path[-,draw=black,line width=0.4000pt,line cap=butt,line join=miter,dash pattern=](0.0000pt,-21.9850pt)--(-1.7275pt,-21.9850pt);
\path[-,draw=black,line width=0.4000pt,line cap=butt,line join=miter,dash pattern=](-14.1650pt,-21.9850pt)--(-12.4375pt,-21.9850pt);
\path[-,line width=0.4000pt,line cap=butt,line join=miter,dash pattern=](-23.7650pt,-10.9925pt)--(-23.7650pt,-10.9925pt)--(-23.7650pt,-10.9925pt)--(-23.7650pt,-10.9925pt)--cycle;
\path[-,draw=black,line width=0.4000pt,line cap=butt,line join=miter,dash pattern=](-26.7650pt,-7.9925pt)--(-20.7650pt,-7.9925pt)--(-20.7650pt,-13.9925pt)--(-26.7650pt,-13.9925pt)--cycle;
\path[-,fill=black,line width=0.4000pt,line cap=butt,line join=miter,dash pattern=](-22.7650pt,-10.9925pt)..controls(-22.7650pt,-10.4402pt)and(-23.2127pt,-9.9925pt)..(-23.7650pt,-9.9925pt)..controls(-24.3173pt,-9.9925pt)and(-24.7650pt,-10.4402pt)..(-24.7650pt,-10.9925pt)..controls(-24.7650pt,-11.5448pt)and(-24.3173pt,-11.9925pt)..(-23.7650pt,-11.9925pt)..controls(-23.2127pt,-11.9925pt)and(-22.7650pt,-11.5448pt)..(-22.7650pt,-10.9925pt)--cycle;
\path[-,draw=black,line width=0.4000pt,line cap=butt,line join=miter,dash pattern=](-23.7650pt,-10.9925pt)..controls(-19.7650pt,-4.0643pt)and(-22.1650pt,0.0000pt)..(-14.1650pt,0.0000pt);
\path[-,draw=black,line width=0.4000pt,line cap=butt,line join=miter,dash pattern=](-23.7650pt,-10.9925pt)..controls(-19.7650pt,-17.9207pt)and(-22.1650pt,-21.9850pt)..(-14.1650pt,-21.9850pt);
\path[-,line width=0.4000pt,line cap=butt,line join=miter,dash pattern=on 0.4000pt off 2.0000pt](-37.8650pt,-7.5450pt)--(-34.7650pt,-7.5450pt)--(-34.7650pt,-14.4400pt)--(-37.8650pt,-14.4400pt)--cycle;
\node[anchor=north west,inner sep=0] at (-37.8650pt,-7.5450pt){\savebox{\marxupbox}{{f}}\immediate\write\boxesfile{124}\immediate\write\boxesfile{\number\wd\marxupbox}\immediate\write\boxesfile{\number\ht\marxupbox}\immediate\write\boxesfile{\number\dp\marxupbox}\box\marxupbox};
\path[-,line width=0.4000pt,line cap=butt,line join=miter,dash pattern=on 0.4000pt off 2.0000pt](-41.8650pt,-3.5450pt)--(-30.7650pt,-3.5450pt)--(-30.7650pt,-18.4400pt)--(-41.8650pt,-18.4400pt)--cycle;
\path[-,draw=black,line width=0.4000pt,line cap=butt,line join=miter,dash pattern=on 0.4000pt off 2.0000pt](-41.8650pt,-3.5450pt)--(-30.7650pt,-3.5450pt)--(-30.7650pt,-18.4400pt)--(-41.8650pt,-18.4400pt)--cycle;
\path[-,draw=black,line width=0.4000pt,line cap=butt,line join=miter,dash pattern=](-30.7650pt,-10.9925pt)--(-30.7650pt,-10.9925pt);
\path[-,draw=black,line width=0.4000pt,line cap=butt,line join=miter,dash pattern=](-41.8650pt,-10.9925pt)--(-41.8650pt,-10.9925pt);
\path[-,line width=0.4000pt,line cap=butt,line join=miter,dash pattern=](-51.4650pt,14.5425pt)--(-51.4650pt,14.5425pt)--(-51.4650pt,14.5425pt)--(-51.4650pt,14.5425pt)--cycle;
\path[-,draw=black,line width=0.4000pt,line cap=butt,line join=miter,dash pattern=](-54.4650pt,17.5425pt)--(-48.4650pt,17.5425pt)--(-48.4650pt,11.5425pt)--(-54.4650pt,11.5425pt)--cycle;
\path[-,fill=black,line width=0.4000pt,line cap=butt,line join=miter,dash pattern=](-50.4650pt,14.5425pt)..controls(-50.4650pt,15.0948pt)and(-50.9127pt,15.5425pt)..(-51.4650pt,15.5425pt)..controls(-52.0173pt,15.5425pt)and(-52.4650pt,15.0948pt)..(-52.4650pt,14.5425pt)..controls(-52.4650pt,13.9902pt)and(-52.0173pt,13.5425pt)..(-51.4650pt,13.5425pt)..controls(-50.9127pt,13.5425pt)and(-50.4650pt,13.9902pt)..(-50.4650pt,14.5425pt)--cycle;
\path[-,draw=black,line width=0.4000pt,line cap=butt,line join=miter,dash pattern=](-51.4650pt,14.5425pt)..controls(-47.4650pt,21.4707pt)and(-49.8650pt,40.0775pt)..(-41.8650pt,40.0775pt);
\path[-,draw=black,line width=0.4000pt,line cap=butt,line join=miter,dash pattern=](-51.4650pt,14.5425pt)..controls(-47.4650pt,7.6143pt)and(-49.8650pt,-10.9925pt)..(-41.8650pt,-10.9925pt);
\end{tikzpicture}}\end{center} 
Which reduces to \ensuremath{\mathsf{h}\mskip 3.0mu}\ensuremath{\allowbreak{}\mathnormal{∘}\allowbreak{}\mskip 3.0mu}\ensuremath{\allowbreak{}\mathnormal{(}\mskip 0.0mu}\ensuremath{\allowbreak{}\mathnormal{(}\mskip 0.0mu}\ensuremath{\mathsf{φ}\mskip 3.0mu}\ensuremath{\mathnormal{▵}\mskip 3.0mu}\ensuremath{\allowbreak{}\mathnormal{(}\mskip 0.0mu}\ensuremath{δ\mskip 3.0mu}\ensuremath{\allowbreak{}\mathnormal{∘}\allowbreak{}\mskip 3.0mu}\ensuremath{\mathsf{ψ}\mskip 0.0mu}\ensuremath{\mathnormal{)}\allowbreak{}\mskip 0.0mu}\ensuremath{\mathnormal{)}\allowbreak{}\mskip 3.0mu}\ensuremath{\mathnormal{▵}\mskip 3.0mu}\ensuremath{\mathsf{i}\mskip 0.0mu}\ensuremath{\mathnormal{)}\allowbreak{}\mskip 3.0mu}\ensuremath{\allowbreak{}\mathnormal{∘}\allowbreak{}\mskip 3.0mu}\ensuremath{\mathsf{f}} 
But the only way to reduce the \ensuremath{δ} node is if one of its branches is connected to \ensuremath{\mathsf{ε}},
as depicted below:

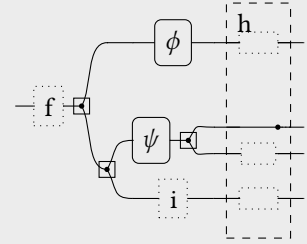
\begin{center}
{\begin{tikzpicture}\path[-,draw=black,line width=0.4000pt,line cap=butt,line join=miter,dash pattern=](-75.0650pt,7.9962pt)--(-68.0650pt,7.9962pt);
\path[-,draw=black,line width=0.4000pt,line cap=butt,line join=miter,dash pattern=](33.8681pt,31.9850pt)--(33.8681pt,31.9850pt);
\path[-,draw=black,line width=0.4000pt,line cap=butt,line join=miter,dash pattern=on 0.4000pt off 1.0000pt](33.8681pt,0.0000pt)--(33.8681pt,0.0000pt);
\path[-,draw=black,line width=0.4000pt,line cap=butt,line join=miter,dash pattern=](33.8681pt,-10.0000pt)--(33.8681pt,-10.0000pt);
\path[-,draw=black,line width=0.4000pt,line cap=butt,line join=miter,dash pattern=](33.8681pt,-26.9850pt)--(33.8681pt,-26.9850pt);
\path[-,draw=black,line width=0.4000pt,line cap=butt,line join=miter,dash pattern=](-56.9650pt,7.9962pt)--(-49.9650pt,7.9962pt);
\path[-,draw=black,line width=0.4000pt,line cap=butt,line join=miter,dash pattern=](-40.3650pt,31.9850pt)--(-40.3650pt,31.9850pt);
\path[-,draw=black,line width=0.4000pt,line cap=butt,line join=miter,dash pattern=](-40.3650pt,-15.9925pt)--(-40.3650pt,-15.9925pt);
\path[-,draw=black,line width=0.4000pt,line cap=butt,line join=miter,dash pattern=](-30.7650pt,31.9850pt)--(-30.7650pt,31.9850pt);
\path[-,draw=black,line width=0.4000pt,line cap=butt,line join=miter,dash pattern=](-30.7650pt,-5.0000pt)--(-30.7650pt,-5.0000pt);
\path[-,draw=black,line width=0.4000pt,line cap=butt,line join=miter,dash pattern=](-30.7650pt,-26.9850pt)--(-30.7650pt,-26.9850pt);
\path[-,draw=black,line width=0.4000pt,line cap=butt,line join=miter,dash pattern=](0.0000pt,31.9850pt)--(0.0000pt,31.9850pt);
\path[-,draw=black,line width=0.4000pt,line cap=butt,line join=miter,dash pattern=](0.0000pt,0.0000pt)--(0.0000pt,0.0000pt);
\path[-,draw=black,line width=0.4000pt,line cap=butt,line join=miter,dash pattern=](0.0000pt,-10.0000pt)--(0.0000pt,-10.0000pt);
\path[-,draw=black,line width=0.4000pt,line cap=butt,line join=miter,dash pattern=](0.0000pt,-26.9850pt)--(0.0000pt,-26.9850pt);
\path[-,line width=0.4000pt,line cap=butt,line join=miter,dash pattern=](8.8000pt,42.9850pt)--(14.1800pt,42.9850pt)--(14.1800pt,36.0900pt)--(8.8000pt,36.0900pt)--cycle;
\node[anchor=north west,inner sep=0] at (8.8000pt,42.9850pt){\savebox{\marxupbox}{{h}}\immediate\write\boxesfile{125}\immediate\write\boxesfile{\number\wd\marxupbox}\immediate\write\boxesfile{\number\ht\marxupbox}\immediate\write\boxesfile{\number\dp\marxupbox}\box\marxupbox};
\path[-,line width=0.4000pt,line cap=butt,line join=miter,dash pattern=](4.8000pt,46.9850pt)--(18.1800pt,46.9850pt)--(18.1800pt,32.0900pt)--(4.8000pt,32.0900pt)--cycle;
\path[-,draw=black,line width=0.4000pt,line cap=butt,line join=miter,dash pattern=on 3.0000pt off 3.0000pt](4.8000pt,46.9850pt)--(29.0681pt,46.9850pt)--(29.0681pt,-41.9850pt)--(4.8000pt,-41.9850pt)--cycle;
\path[-,draw=black,line width=0.4000pt,line cap=butt,line join=miter,dash pattern=](9.6000pt,31.9850pt)--(9.6000pt,31.9850pt);
\path[-,draw=black,line width=0.4000pt,line cap=butt,line join=miter,dash pattern=](9.6000pt,0.0000pt)--(9.6000pt,0.0000pt);
\path[-,draw=black,line width=0.4000pt,line cap=butt,line join=miter,dash pattern=](9.6000pt,-10.0000pt)--(9.6000pt,-10.0000pt);
\path[-,draw=black,line width=0.4000pt,line cap=butt,line join=miter,dash pattern=](9.6000pt,-26.9850pt)--(9.6000pt,-26.9850pt);
\path[-,draw=black,line width=0.4000pt,line cap=butt,line join=miter,dash pattern=](24.2681pt,31.9850pt)--(24.2681pt,31.9850pt);
\path[-,draw=black,line width=0.4000pt,line cap=butt,line join=miter,dash pattern=on 0.4000pt off 1.0000pt](24.2681pt,0.0000pt)--(24.2681pt,0.0000pt);
\path[-,draw=black,line width=0.4000pt,line cap=butt,line join=miter,dash pattern=](24.2681pt,-10.0000pt)--(24.2681pt,-10.0000pt);
\path[-,draw=black,line width=0.4000pt,line cap=butt,line join=miter,dash pattern=](24.2681pt,-26.9850pt)--(24.2681pt,-26.9850pt);
\path[-,draw=black,line width=0.4000pt,line cap=butt,line join=miter,dash pattern=](0.0000pt,31.9850pt)--(9.6000pt,31.9850pt);
\path[-,draw=black,line width=0.4000pt,line cap=butt,line join=miter,dash pattern=](0.0000pt,0.0000pt)--(9.6000pt,0.0000pt);
\path[-,draw=black,line width=0.4000pt,line cap=butt,line join=miter,dash pattern=](0.0000pt,-10.0000pt)--(9.6000pt,-10.0000pt);
\path[-,draw=black,line width=0.4000pt,line cap=butt,line join=miter,dash pattern=](0.0000pt,-26.9850pt)--(9.6000pt,-26.9850pt);
\path[-,draw=black,line width=0.4000pt,line cap=butt,line join=miter,dash pattern=](24.2681pt,31.9850pt)--(33.8681pt,31.9850pt);
\path[-,draw=black,line width=0.4000pt,line cap=butt,line join=miter,dash pattern=](24.2681pt,0.0000pt)--(33.8681pt,0.0000pt);
\path[-,draw=black,line width=0.4000pt,line cap=butt,line join=miter,dash pattern=](24.2681pt,-10.0000pt)--(33.8681pt,-10.0000pt);
\path[-,draw=black,line width=0.4000pt,line cap=butt,line join=miter,dash pattern=](24.2681pt,-26.9850pt)--(33.8681pt,-26.9850pt);
\path[-,draw=black,line width=0.4000pt,line cap=butt,line join=miter,dash pattern=](24.2681pt,31.9850pt)--(29.0681pt,31.9850pt);
\path[-,draw=black,line width=0.4000pt,line cap=butt,line join=miter,dash pattern=on 0.4000pt off 1.0000pt](24.2681pt,0.0000pt)--(29.0681pt,0.0000pt);
\path[-,draw=black,line width=0.4000pt,line cap=butt,line join=miter,dash pattern=](24.2681pt,-10.0000pt)--(29.0681pt,-10.0000pt);
\path[-,draw=black,line width=0.4000pt,line cap=butt,line join=miter,dash pattern=](24.2681pt,-26.9850pt)--(29.0681pt,-26.9850pt);
\path[-,draw=black,line width=0.4000pt,line cap=butt,line join=miter,dash pattern=](9.6000pt,31.9850pt)--(4.8000pt,31.9850pt);
\path[-,draw=black,line width=0.4000pt,line cap=butt,line join=miter,dash pattern=](9.6000pt,0.0000pt)--(4.8000pt,0.0000pt);
\path[-,draw=black,line width=0.4000pt,line cap=butt,line join=miter,dash pattern=](9.6000pt,-10.0000pt)--(4.8000pt,-10.0000pt);
\path[-,draw=black,line width=0.4000pt,line cap=butt,line join=miter,dash pattern=](9.6000pt,-26.9850pt)--(4.8000pt,-26.9850pt);
\path[-,line width=0.4000pt,line cap=butt,line join=miter,dash pattern=on 0.4000pt off 2.0000pt](13.6000pt,31.9850pt)--(20.2681pt,31.9850pt)--(20.2681pt,31.9850pt)--(13.6000pt,31.9850pt)--cycle;
\node[anchor=north west,inner sep=0] at (13.6000pt,31.9850pt){\savebox{\marxupbox}{{    }}\immediate\write\boxesfile{126}\immediate\write\boxesfile{\number\wd\marxupbox}\immediate\write\boxesfile{\number\ht\marxupbox}\immediate\write\boxesfile{\number\dp\marxupbox}\box\marxupbox};
\path[-,line width=0.4000pt,line cap=butt,line join=miter,dash pattern=on 0.4000pt off 2.0000pt](9.6000pt,35.9850pt)--(24.2681pt,35.9850pt)--(24.2681pt,27.9850pt)--(9.6000pt,27.9850pt)--cycle;
\path[-,draw=black,line width=0.4000pt,line cap=butt,line join=miter,dash pattern=on 0.4000pt off 2.0000pt](9.6000pt,35.9850pt)--(24.2681pt,35.9850pt)--(24.2681pt,27.9850pt)--(9.6000pt,27.9850pt)--cycle;
\path[-,draw=black,line width=0.4000pt,line cap=butt,line join=miter,dash pattern=](24.2681pt,31.9850pt)--(24.2681pt,31.9850pt);
\path[-,draw=black,line width=0.4000pt,line cap=butt,line join=miter,dash pattern=](9.6000pt,31.9850pt)--(9.6000pt,31.9850pt);
\path[-,fill=black,line width=0.4000pt,line cap=butt,line join=miter,dash pattern=](25.2681pt,0.0000pt)..controls(25.2681pt,0.5523pt)and(24.8204pt,1.0000pt)..(24.2681pt,1.0000pt)..controls(23.7158pt,1.0000pt)and(23.2681pt,0.5523pt)..(23.2681pt,0.0000pt)..controls(23.2681pt,-0.5523pt)and(23.7158pt,-1.0000pt)..(24.2681pt,-1.0000pt)..controls(24.8204pt,-1.0000pt)and(25.2681pt,-0.5523pt)..(25.2681pt,0.0000pt)--cycle;
\path[-,draw=black,line width=0.4000pt,line cap=butt,line join=miter,dash pattern=](9.6000pt,0.0000pt)--(24.2681pt,0.0000pt);
\path[-,line width=0.4000pt,line cap=butt,line join=miter,dash pattern=on 0.4000pt off 2.0000pt](14.4335pt,-10.0000pt)--(19.4346pt,-10.0000pt)--(19.4346pt,-10.0000pt)--(14.4335pt,-10.0000pt)--cycle;
\node[anchor=north west,inner sep=0] at (14.4335pt,-10.0000pt){\savebox{\marxupbox}{{   }}\immediate\write\boxesfile{127}\immediate\write\boxesfile{\number\wd\marxupbox}\immediate\write\boxesfile{\number\ht\marxupbox}\immediate\write\boxesfile{\number\dp\marxupbox}\box\marxupbox};
\path[-,line width=0.4000pt,line cap=butt,line join=miter,dash pattern=on 0.4000pt off 2.0000pt](10.4335pt,-6.0000pt)--(23.4346pt,-6.0000pt)--(23.4346pt,-14.0000pt)--(10.4335pt,-14.0000pt)--cycle;
\path[-,draw=black,line width=0.4000pt,line cap=butt,line join=miter,dash pattern=on 0.4000pt off 2.0000pt](10.4335pt,-6.0000pt)--(23.4346pt,-6.0000pt)--(23.4346pt,-14.0000pt)--(10.4335pt,-14.0000pt)--cycle;
\path[-,draw=black,line width=0.4000pt,line cap=butt,line join=miter,dash pattern=](24.2681pt,-10.0000pt)--(23.4346pt,-10.0000pt);
\path[-,draw=black,line width=0.4000pt,line cap=butt,line join=miter,dash pattern=](9.6000pt,-10.0000pt)--(10.4335pt,-10.0000pt);
\path[-,line width=0.4000pt,line cap=butt,line join=miter,dash pattern=on 0.4000pt off 2.0000pt](13.6000pt,-26.9850pt)--(20.2681pt,-26.9850pt)--(20.2681pt,-26.9850pt)--(13.6000pt,-26.9850pt)--cycle;
\node[anchor=north west,inner sep=0] at (13.6000pt,-26.9850pt){\savebox{\marxupbox}{{    }}\immediate\write\boxesfile{128}\immediate\write\boxesfile{\number\wd\marxupbox}\immediate\write\boxesfile{\number\ht\marxupbox}\immediate\write\boxesfile{\number\dp\marxupbox}\box\marxupbox};
\path[-,line width=0.4000pt,line cap=butt,line join=miter,dash pattern=on 0.4000pt off 2.0000pt](9.6000pt,-22.9850pt)--(24.2681pt,-22.9850pt)--(24.2681pt,-30.9850pt)--(9.6000pt,-30.9850pt)--cycle;
\path[-,draw=black,line width=0.4000pt,line cap=butt,line join=miter,dash pattern=on 0.4000pt off 2.0000pt](9.6000pt,-22.9850pt)--(24.2681pt,-22.9850pt)--(24.2681pt,-30.9850pt)--(9.6000pt,-30.9850pt)--cycle;
\path[-,draw=black,line width=0.4000pt,line cap=butt,line join=miter,dash pattern=](24.2681pt,-26.9850pt)--(24.2681pt,-26.9850pt);
\path[-,draw=black,line width=0.4000pt,line cap=butt,line join=miter,dash pattern=](9.6000pt,-26.9850pt)--(9.6000pt,-26.9850pt);
\path[-,line width=0.4000pt,line cap=butt,line join=miter,dash pattern=](-18.3200pt,36.6775pt)--(-12.4450pt,36.6775pt)--(-12.4450pt,27.2925pt)--(-18.3200pt,27.2925pt)--cycle;
\node[anchor=north west,inner sep=0] at (-18.3200pt,36.6775pt){\savebox{\marxupbox}{{φ}}\immediate\write\boxesfile{129}\immediate\write\boxesfile{\number\wd\marxupbox}\immediate\write\boxesfile{\number\ht\marxupbox}\immediate\write\boxesfile{\number\dp\marxupbox}\box\marxupbox};
\path[-,line width=0.4000pt,line cap=butt,line join=miter,dash pattern=](-22.3200pt,40.6775pt)--(-8.4450pt,40.6775pt)--(-8.4450pt,23.2925pt)--(-22.3200pt,23.2925pt)--cycle;
\path[-,draw=black,line width=0.4000pt,line cap=butt,line join=miter,dash pattern=](-22.3200pt,38.1775pt)..controls(-22.3200pt,39.5582pt)and(-21.2007pt,40.6775pt)..(-19.8200pt,40.6775pt)--(-10.9450pt,40.6775pt)..controls(-9.5643pt,40.6775pt)and(-8.4450pt,39.5582pt)..(-8.4450pt,38.1775pt)--(-8.4450pt,25.7925pt)..controls(-8.4450pt,24.4118pt)and(-9.5643pt,23.2925pt)..(-10.9450pt,23.2925pt)--(-19.8200pt,23.2925pt)..controls(-21.2007pt,23.2925pt)and(-22.3200pt,24.4118pt)..(-22.3200pt,25.7925pt)--cycle;
\path[-,draw=black,line width=0.4000pt,line cap=butt,line join=miter,dash pattern=](0.0000pt,31.9850pt)--(-8.4450pt,31.9850pt);
\path[-,draw=black,line width=0.4000pt,line cap=butt,line join=miter,dash pattern=](-30.7650pt,31.9850pt)--(-22.3200pt,31.9850pt);
\path[-,draw=black,line width=0.4000pt,line cap=butt,line join=miter,dash pattern=](-16.6000pt,-5.0000pt)--(-9.6000pt,-5.0000pt);
\path[-,line width=0.4000pt,line cap=butt,line join=miter,dash pattern=](-9.6000pt,-5.0000pt)--(-9.6000pt,-5.0000pt)--(-9.6000pt,-5.0000pt)--(-9.6000pt,-5.0000pt)--cycle;
\path[-,draw=black,line width=0.4000pt,line cap=butt,line join=miter,dash pattern=](-12.6000pt,-2.0000pt)--(-6.6000pt,-2.0000pt)--(-6.6000pt,-8.0000pt)--(-12.6000pt,-8.0000pt)--cycle;
\path[-,fill=black,line width=0.4000pt,line cap=butt,line join=miter,dash pattern=](-8.6000pt,-5.0000pt)..controls(-8.6000pt,-4.4477pt)and(-9.0477pt,-4.0000pt)..(-9.6000pt,-4.0000pt)..controls(-10.1523pt,-4.0000pt)and(-10.6000pt,-4.4477pt)..(-10.6000pt,-5.0000pt)..controls(-10.6000pt,-5.5523pt)and(-10.1523pt,-6.0000pt)..(-9.6000pt,-6.0000pt)..controls(-9.0477pt,-6.0000pt)and(-8.6000pt,-5.5523pt)..(-8.6000pt,-5.0000pt)--cycle;
\path[-,draw=black,line width=0.4000pt,line cap=butt,line join=miter,dash pattern=](-9.6000pt,-5.0000pt)..controls(-5.6000pt,1.9282pt)and(-8.0000pt,0.0000pt)..(0.0000pt,0.0000pt);
\path[-,draw=black,line width=0.4000pt,line cap=butt,line join=miter,dash pattern=](-9.6000pt,-5.0000pt)..controls(-5.6000pt,-11.9282pt)and(-8.0000pt,-10.0000pt)..(0.0000pt,-10.0000pt);
\path[-,line width=0.4000pt,line cap=butt,line join=miter,dash pattern=](-26.7650pt,-0.3075pt)--(-20.6000pt,-0.3075pt)--(-20.6000pt,-9.6925pt)--(-26.7650pt,-9.6925pt)--cycle;
\node[anchor=north west,inner sep=0] at (-26.7650pt,-0.3075pt){\savebox{\marxupbox}{{ψ}}\immediate\write\boxesfile{130}\immediate\write\boxesfile{\number\wd\marxupbox}\immediate\write\boxesfile{\number\ht\marxupbox}\immediate\write\boxesfile{\number\dp\marxupbox}\box\marxupbox};
\path[-,line width=0.4000pt,line cap=butt,line join=miter,dash pattern=](-30.7650pt,3.6925pt)--(-16.6000pt,3.6925pt)--(-16.6000pt,-13.6925pt)--(-30.7650pt,-13.6925pt)--cycle;
\path[-,draw=black,line width=0.4000pt,line cap=butt,line join=miter,dash pattern=](-30.7650pt,1.1925pt)..controls(-30.7650pt,2.5732pt)and(-29.6457pt,3.6925pt)..(-28.2650pt,3.6925pt)--(-19.1000pt,3.6925pt)..controls(-17.7193pt,3.6925pt)and(-16.6000pt,2.5732pt)..(-16.6000pt,1.1925pt)--(-16.6000pt,-11.1925pt)..controls(-16.6000pt,-12.5732pt)and(-17.7193pt,-13.6925pt)..(-19.1000pt,-13.6925pt)--(-28.2650pt,-13.6925pt)..controls(-29.6457pt,-13.6925pt)and(-30.7650pt,-12.5732pt)..(-30.7650pt,-11.1925pt)--cycle;
\path[-,draw=black,line width=0.4000pt,line cap=butt,line join=miter,dash pattern=](-16.6000pt,-5.0000pt)--(-16.6000pt,-5.0000pt);
\path[-,draw=black,line width=0.4000pt,line cap=butt,line join=miter,dash pattern=](-30.7650pt,-5.0000pt)--(-30.7650pt,-5.0000pt);
\path[-,line width=0.4000pt,line cap=butt,line join=miter,dash pattern=on 0.4000pt off 2.0000pt](-16.7375pt,-23.6925pt)--(-14.0275pt,-23.6925pt)--(-14.0275pt,-30.2775pt)--(-16.7375pt,-30.2775pt)--cycle;
\node[anchor=north west,inner sep=0] at (-16.7375pt,-23.6925pt){\savebox{\marxupbox}{{i}}\immediate\write\boxesfile{131}\immediate\write\boxesfile{\number\wd\marxupbox}\immediate\write\boxesfile{\number\ht\marxupbox}\immediate\write\boxesfile{\number\dp\marxupbox}\box\marxupbox};
\path[-,line width=0.4000pt,line cap=butt,line join=miter,dash pattern=on 0.4000pt off 2.0000pt](-20.7375pt,-19.6925pt)--(-10.0275pt,-19.6925pt)--(-10.0275pt,-34.2775pt)--(-20.7375pt,-34.2775pt)--cycle;
\path[-,draw=black,line width=0.4000pt,line cap=butt,line join=miter,dash pattern=on 0.4000pt off 2.0000pt](-20.7375pt,-19.6925pt)--(-10.0275pt,-19.6925pt)--(-10.0275pt,-34.2775pt)--(-20.7375pt,-34.2775pt)--cycle;
\path[-,draw=black,line width=0.4000pt,line cap=butt,line join=miter,dash pattern=](0.0000pt,-26.9850pt)--(-10.0275pt,-26.9850pt);
\path[-,draw=black,line width=0.4000pt,line cap=butt,line join=miter,dash pattern=](-30.7650pt,-26.9850pt)--(-20.7375pt,-26.9850pt);
\path[-,draw=black,line width=0.4000pt,line cap=butt,line join=miter,dash pattern=](-40.3650pt,31.9850pt)--(-30.7650pt,31.9850pt);
\path[-,line width=0.4000pt,line cap=butt,line join=miter,dash pattern=](-40.3650pt,-15.9925pt)--(-40.3650pt,-15.9925pt)--(-40.3650pt,-15.9925pt)--(-40.3650pt,-15.9925pt)--cycle;
\path[-,draw=black,line width=0.4000pt,line cap=butt,line join=miter,dash pattern=](-43.3650pt,-12.9925pt)--(-37.3650pt,-12.9925pt)--(-37.3650pt,-18.9925pt)--(-43.3650pt,-18.9925pt)--cycle;
\path[-,fill=black,line width=0.4000pt,line cap=butt,line join=miter,dash pattern=](-39.3650pt,-15.9925pt)..controls(-39.3650pt,-15.4402pt)and(-39.8127pt,-14.9925pt)..(-40.3650pt,-14.9925pt)..controls(-40.9173pt,-14.9925pt)and(-41.3650pt,-15.4402pt)..(-41.3650pt,-15.9925pt)..controls(-41.3650pt,-16.5448pt)and(-40.9173pt,-16.9925pt)..(-40.3650pt,-16.9925pt)..controls(-39.8127pt,-16.9925pt)and(-39.3650pt,-16.5448pt)..(-39.3650pt,-15.9925pt)--cycle;
\path[-,draw=black,line width=0.4000pt,line cap=butt,line join=miter,dash pattern=](-40.3650pt,-15.9925pt)..controls(-36.3650pt,-9.0643pt)and(-38.7650pt,-5.0000pt)..(-30.7650pt,-5.0000pt);
\path[-,draw=black,line width=0.4000pt,line cap=butt,line join=miter,dash pattern=](-40.3650pt,-15.9925pt)..controls(-36.3650pt,-22.9207pt)and(-38.7650pt,-26.9850pt)..(-30.7650pt,-26.9850pt);
\path[-,line width=0.4000pt,line cap=butt,line join=miter,dash pattern=](-49.9650pt,7.9962pt)--(-49.9650pt,7.9962pt)--(-49.9650pt,7.9962pt)--(-49.9650pt,7.9962pt)--cycle;
\path[-,draw=black,line width=0.4000pt,line cap=butt,line join=miter,dash pattern=](-52.9650pt,10.9962pt)--(-46.9650pt,10.9962pt)--(-46.9650pt,4.9962pt)--(-52.9650pt,4.9962pt)--cycle;
\path[-,fill=black,line width=0.4000pt,line cap=butt,line join=miter,dash pattern=](-48.9650pt,7.9962pt)..controls(-48.9650pt,8.5485pt)and(-49.4127pt,8.9962pt)..(-49.9650pt,8.9962pt)..controls(-50.5173pt,8.9962pt)and(-50.9650pt,8.5485pt)..(-50.9650pt,7.9962pt)..controls(-50.9650pt,7.4440pt)and(-50.5173pt,6.9962pt)..(-49.9650pt,6.9962pt)..controls(-49.4127pt,6.9962pt)and(-48.9650pt,7.4440pt)..(-48.9650pt,7.9962pt)--cycle;
\path[-,draw=black,line width=0.4000pt,line cap=butt,line join=miter,dash pattern=](-49.9650pt,7.9962pt)..controls(-45.9650pt,14.9244pt)and(-48.3650pt,31.9850pt)..(-40.3650pt,31.9850pt);
\path[-,draw=black,line width=0.4000pt,line cap=butt,line join=miter,dash pattern=](-49.9650pt,7.9962pt)..controls(-45.9650pt,1.0680pt)and(-48.3650pt,-15.9925pt)..(-40.3650pt,-15.9925pt);
\path[-,line width=0.4000pt,line cap=butt,line join=miter,dash pattern=on 0.4000pt off 2.0000pt](-64.0650pt,11.4437pt)--(-60.9650pt,11.4437pt)--(-60.9650pt,4.5488pt)--(-64.0650pt,4.5488pt)--cycle;
\node[anchor=north west,inner sep=0] at (-64.0650pt,11.4437pt){\savebox{\marxupbox}{{f}}\immediate\write\boxesfile{132}\immediate\write\boxesfile{\number\wd\marxupbox}\immediate\write\boxesfile{\number\ht\marxupbox}\immediate\write\boxesfile{\number\dp\marxupbox}\box\marxupbox};
\path[-,line width=0.4000pt,line cap=butt,line join=miter,dash pattern=on 0.4000pt off 2.0000pt](-68.0650pt,15.4437pt)--(-56.9650pt,15.4437pt)--(-56.9650pt,0.5488pt)--(-68.0650pt,0.5488pt)--cycle;
\path[-,draw=black,line width=0.4000pt,line cap=butt,line join=miter,dash pattern=on 0.4000pt off 2.0000pt](-68.0650pt,15.4437pt)--(-56.9650pt,15.4437pt)--(-56.9650pt,0.5488pt)--(-68.0650pt,0.5488pt)--cycle;
\path[-,draw=black,line width=0.4000pt,line cap=butt,line join=miter,dash pattern=](-56.9650pt,7.9962pt)--(-56.9650pt,7.9962pt);
\path[-,draw=black,line width=0.4000pt,line cap=butt,line join=miter,dash pattern=](-68.0650pt,7.9962pt)--(-68.0650pt,7.9962pt);
\end{tikzpicture}}\end{center} 
But this could happen only if one of the ψ was discarded to begin
with, which is ruled out by assumption. So we can again reject this
case.

\end{proof} 

The next step in the argument is to prove that, if any generator is in a decoded
morphism, its output never (fully) discarded. That is, the situation
depicted in the following diagram cannot occur:

\begin{center}
{\begin{tikzpicture}\path[-,draw=black,line width=0.4000pt,line cap=butt,line join=miter,dash pattern=](-28.6681pt,5.8462pt)--(-21.6681pt,5.8462pt);
\path[-,draw=black,line width=0.4000pt,line cap=butt,line join=miter,dash pattern=on 0.4000pt off 1.0000pt](42.5431pt,11.6925pt)--(49.5431pt,11.6925pt);
\path[-,draw=black,line width=0.4000pt,line cap=butt,line join=miter,dash pattern=](42.5431pt,0.0000pt)--(49.5431pt,0.0000pt);
\path[-,draw=black,line width=0.4000pt,line cap=butt,line join=miter,dash pattern=](-7.0000pt,11.6925pt)--(0.0000pt,11.6925pt);
\path[-,draw=black,line width=0.4000pt,line cap=butt,line join=miter,dash pattern=](-7.0000pt,0.0000pt)--(0.0000pt,0.0000pt);
\path[-,draw=black,line width=0.4000pt,line cap=butt,line join=miter,dash pattern=](13.8750pt,11.6925pt)--(20.8750pt,11.6925pt);
\path[-,draw=black,line width=0.4000pt,line cap=butt,line join=miter,dash pattern=](35.5431pt,11.6925pt)--(42.5431pt,11.6925pt);
\path[-,fill=black,line width=0.4000pt,line cap=butt,line join=miter,dash pattern=](43.5431pt,11.6925pt)..controls(43.5431pt,12.2448pt)and(43.0953pt,12.6925pt)..(42.5431pt,12.6925pt)..controls(41.9908pt,12.6925pt)and(41.5431pt,12.2448pt)..(41.5431pt,11.6925pt)..controls(41.5431pt,11.1402pt)and(41.9908pt,10.6925pt)..(42.5431pt,10.6925pt)..controls(43.0953pt,10.6925pt)and(43.5431pt,11.1402pt)..(43.5431pt,11.6925pt)--cycle;
\path[-,draw=black,line width=0.4000pt,line cap=butt,line join=miter,dash pattern=](42.5431pt,11.6925pt)--(42.5431pt,11.6925pt);
\path[-,line width=0.4000pt,line cap=butt,line join=miter,dash pattern=on 0.4000pt off 2.0000pt](24.8750pt,11.6925pt)--(31.5431pt,11.6925pt)--(31.5431pt,11.6925pt)--(24.8750pt,11.6925pt)--cycle;
\node[anchor=north west,inner sep=0] at (24.8750pt,11.6925pt){\savebox{\marxupbox}{{    }}\immediate\write\boxesfile{133}\immediate\write\boxesfile{\number\wd\marxupbox}\immediate\write\boxesfile{\number\ht\marxupbox}\immediate\write\boxesfile{\number\dp\marxupbox}\box\marxupbox};
\path[-,line width=0.4000pt,line cap=butt,line join=miter,dash pattern=on 0.4000pt off 2.0000pt](20.8750pt,15.6925pt)--(35.5431pt,15.6925pt)--(35.5431pt,7.6925pt)--(20.8750pt,7.6925pt)--cycle;
\path[-,draw=black,line width=0.4000pt,line cap=butt,line join=miter,dash pattern=on 0.4000pt off 2.0000pt](20.8750pt,15.6925pt)--(35.5431pt,15.6925pt)--(35.5431pt,7.6925pt)--(20.8750pt,7.6925pt)--cycle;
\path[-,draw=black,line width=0.4000pt,line cap=butt,line join=miter,dash pattern=](35.5431pt,11.6925pt)--(35.5431pt,11.6925pt);
\path[-,draw=black,line width=0.4000pt,line cap=butt,line join=miter,dash pattern=](20.8750pt,11.6925pt)--(20.8750pt,11.6925pt);
\path[-,line width=0.4000pt,line cap=butt,line join=miter,dash pattern=](4.0000pt,16.3850pt)--(9.8750pt,16.3850pt)--(9.8750pt,7.0000pt)--(4.0000pt,7.0000pt)--cycle;
\node[anchor=north west,inner sep=0] at (4.0000pt,16.3850pt){\savebox{\marxupbox}{{φ}}\immediate\write\boxesfile{134}\immediate\write\boxesfile{\number\wd\marxupbox}\immediate\write\boxesfile{\number\ht\marxupbox}\immediate\write\boxesfile{\number\dp\marxupbox}\box\marxupbox};
\path[-,line width=0.4000pt,line cap=butt,line join=miter,dash pattern=](0.0000pt,20.3850pt)--(13.8750pt,20.3850pt)--(13.8750pt,3.0000pt)--(0.0000pt,3.0000pt)--cycle;
\path[-,draw=black,line width=0.4000pt,line cap=butt,line join=miter,dash pattern=](0.0000pt,17.8850pt)..controls(0.0000pt,19.2657pt)and(1.1193pt,20.3850pt)..(2.5000pt,20.3850pt)--(11.3750pt,20.3850pt)..controls(12.7557pt,20.3850pt)and(13.8750pt,19.2657pt)..(13.8750pt,17.8850pt)--(13.8750pt,5.5000pt)..controls(13.8750pt,4.1193pt)and(12.7557pt,3.0000pt)..(11.3750pt,3.0000pt)--(2.5000pt,3.0000pt)..controls(1.1193pt,3.0000pt)and(0.0000pt,4.1193pt)..(0.0000pt,5.5000pt)--cycle;
\path[-,draw=black,line width=0.4000pt,line cap=butt,line join=miter,dash pattern=](13.8750pt,11.6925pt)--(13.8750pt,11.6925pt);
\path[-,draw=black,line width=0.4000pt,line cap=butt,line join=miter,dash pattern=](0.0000pt,11.6925pt)--(0.0000pt,11.6925pt);
\path[-,draw=black,line width=0.4000pt,line cap=butt,line join=miter,dash pattern=](0.0000pt,0.0000pt)--(42.5431pt,0.0000pt);
\path[-,line width=0.4000pt,line cap=butt,line join=miter,dash pattern=on 0.4000pt off 2.0000pt](-17.6681pt,5.8462pt)--(-11.0000pt,5.8462pt)--(-11.0000pt,5.8462pt)--(-17.6681pt,5.8462pt)--cycle;
\node[anchor=north west,inner sep=0] at (-17.6681pt,5.8462pt){\savebox{\marxupbox}{{    }}\immediate\write\boxesfile{135}\immediate\write\boxesfile{\number\wd\marxupbox}\immediate\write\boxesfile{\number\ht\marxupbox}\immediate\write\boxesfile{\number\dp\marxupbox}\box\marxupbox};
\path[-,line width=0.4000pt,line cap=butt,line join=miter,dash pattern=on 0.4000pt off 2.0000pt](-21.6681pt,9.8462pt)--(-7.0000pt,9.8462pt)--(-7.0000pt,1.8462pt)--(-21.6681pt,1.8462pt)--cycle;
\path[-,draw=black,line width=0.4000pt,line cap=butt,line join=miter,dash pattern=on 0.4000pt off 2.0000pt](-21.6681pt,15.6925pt)--(-7.0000pt,15.6925pt)--(-7.0000pt,-4.0000pt)--(-21.6681pt,-4.0000pt)--cycle;
\path[-,draw=black,line width=0.4000pt,line cap=butt,line join=miter,dash pattern=](-7.0000pt,11.6925pt)--(-7.0000pt,11.6925pt);
\path[-,draw=black,line width=0.4000pt,line cap=butt,line join=miter,dash pattern=](-7.0000pt,0.0000pt)--(-7.0000pt,0.0000pt);
\path[-,draw=black,line width=0.4000pt,line cap=butt,line join=miter,dash pattern=](-21.6681pt,5.8462pt)--(-21.6681pt,5.8462pt);
\end{tikzpicture}}\end{center} 
Formally:

\begin{lemma}{}For every \ensuremath{\mathsf{g}\mskip 0.0mu}\ensuremath{\mathnormal{,}\mskip 3.0mu}\ensuremath{\mathsf{h}\mskip 0.0mu}\ensuremath{\mathnormal{,}\mskip 3.0mu}\ensuremath{\mathsf{i}} if \ensuremath{\bigtriangleup\mskip 3.0mu}\ensuremath{\allowbreak{}\mathnormal{(}\mskip 0.0mu}\ensuremath{\mathsf{f}\mskip 3.0mu}\ensuremath{\mathsf{id}\mskip 0.0mu}\ensuremath{\mathnormal{)}\allowbreak{}\mskip 3.0mu}\ensuremath{\mathnormal{=}\mskip 3.0mu}\ensuremath{\mathsf{h}\mskip 3.0mu}\ensuremath{\allowbreak{}\mathnormal{∘}\allowbreak{}\mskip 3.0mu}\ensuremath{\allowbreak{}\mathnormal{(}\mskip 0.0mu}\ensuremath{\allowbreak{}\mathnormal{(}\mskip 0.0mu}\ensuremath{\mathsf{i}\mskip 3.0mu}\ensuremath{\allowbreak{}\mathnormal{∘}\allowbreak{}\mskip 3.0mu}\ensuremath{\mathsf{φ}\mskip 0.0mu}\ensuremath{\mathnormal{)}\allowbreak{}\mskip 3.0mu}\ensuremath{\allowbreak{}\mathnormal{×}\allowbreak{}\mskip 3.0mu}\ensuremath{\mathsf{id}\mskip 0.0mu}\ensuremath{\mathnormal{)}\allowbreak{}\mskip 3.0mu}\ensuremath{\allowbreak{}\mathnormal{∘}\allowbreak{}\mskip 3.0mu}\ensuremath{\mathsf{g}}, and \ensuremath{\mathsf{h}} is protolinear, then \ensuremath{\mathsf{i}} cannot be equivalent to \ensuremath{\mathsf{ε}}.\label{136}\end{lemma}\begin{proof} The proof is an immediate consequence of the possible ways to
construct \ensuremath{\mathsf{f}}--- namely, the combinators of our interface. The
only way to discard fully a value is via \ensuremath{\mathsf{ε}}, which is itself
available only via \ensuremath{\mathsf{unit}}. However, 1. \ensuremath{\mathsf{unit}} applies \ensuremath{\mathsf{ε}} to
its input directly and 2. the only construction which places something
before another morphism is \ensuremath{\mathsf{merge}}, which places another \ensuremath{\mathsf{δ}} before the whole construction. Consequently \ensuremath{\mathsf{ε}} can only be connected
directly to the input of \emph{mergeA (f id)}, never after a generator
φ.\end{proof} 
To get the desired result, it suffices to put all the pieces together.

\begin{theorem}{} For every function \ensuremath{\mathsf{f}\mskip 3.0mu}\ensuremath{\mathnormal{:}\mskip 3.0mu}\ensuremath{\mathsf{P}\mskip 3.0mu}\ensuremath{\mathsf{k}\mskip 3.0mu}\ensuremath{\mathsf{r}\mskip 3.0mu}\ensuremath{\mathsf{a}\mskip 3.0mu}\ensuremath{\mathnormal{⊸}\mskip 3.0mu}\ensuremath{\mathsf{P}\mskip 3.0mu}\ensuremath{\mathsf{k}\mskip 3.0mu}\ensuremath{\mathsf{r}\mskip 3.0mu}\ensuremath{\mathsf{b}}, if \ensuremath{\bigtriangleup\mskip 3.0mu}\ensuremath{\allowbreak{}\mathnormal{(}\mskip 0.0mu}\ensuremath{\mathsf{f}\mskip 3.0mu}\ensuremath{\mathsf{id}\mskip 0.0mu}\ensuremath{\mathnormal{)}\allowbreak{}\mskip 3.0mu}\ensuremath{\mathnormal{=}\mskip 3.0mu}\ensuremath{\mathsf{h}\mskip 3.0mu}\ensuremath{\allowbreak{}\mathnormal{∘}\allowbreak{}\mskip 3.0mu}\ensuremath{\allowbreak{}\mathnormal{(}\mskip 0.0mu}\ensuremath{\allowbreak{}\mathnormal{(}\mskip 0.0mu}\ensuremath{\mathsf{φ}\mskip 3.0mu}\ensuremath{\allowbreak{}\mathnormal{∘}\allowbreak{}\mskip 3.0mu}\ensuremath{\mathsf{f}\mskip 0.0mu}\ensuremath{\mathnormal{)}\allowbreak{}\mskip 3.0mu}\ensuremath{\mathnormal{▵}\mskip 3.0mu}\ensuremath{\allowbreak{}\mathnormal{(}\mskip 0.0mu}\ensuremath{\mathsf{ψ}\mskip 3.0mu}\ensuremath{\allowbreak{}\mathnormal{∘}\allowbreak{}\mskip 3.0mu}\ensuremath{\mathsf{f}\mskip 0.0mu}\ensuremath{\mathnormal{)}\allowbreak{}\mskip 3.0mu}\ensuremath{\mathnormal{▵}\mskip 3.0mu}\ensuremath{\mathsf{g}\mskip 0.0mu}\ensuremath{\mathnormal{)}\allowbreak{}}, then \ensuremath{\mathsf{φ}\mskip 3.0mu}\ensuremath{\mathnormal{=}\mskip 3.0mu}\ensuremath{\mathsf{ψ}}.\label{137}\end{theorem}\begin{proof}We apply \cref{106}. The pseudolinearity condition is given by \cref{69}, and non-discardability by
\cref{136}.\end{proof} \end{mdframed} 
\subsection{Haskell implementation of \ensuremath{\mathsf{reduce}}}\label{138}  In this section we present the main components of the Haskell
implementation of the \ensuremath{\mathsf{reduce}} function. We start by showing the underlying
data structure which is manipulated by \ensuremath{\mathsf{reduce}}. This data structure is a list
of morphisms of type \ensuremath{\mathsf{FreeCartesian}\mskip 3.0mu}\ensuremath{\mathsf{k}\mskip 3.0mu}\ensuremath{\mathsf{a}\mskip 3.0mu}\ensuremath{\mathsf{xᵢ}}, for varying \ensuremath{\mathsf{xᵢ}}. (This list corresponds to the arguments of \ensuremath{\bigtriangleup}.)  Because we have to keep track in the
type that all these morphisms share the same source object, we need to
use a {\sc{}gadt} to store them instead of a plain Haskell list:

\begin{list}{}{\setlength\leftmargin{1.0em}}\item\relax
\ensuremath{\begin{parray}\column{B}{@{}>{}l<{}@{}}\column[0em]{1}{@{}>{}l<{}@{}}\column[1em]{2}{@{}>{}l<{}@{}}\column{3}{@{}>{}l<{}@{}}\column{4}{@{}>{}l<{}@{}}\column{5}{@{}>{}l<{}@{}}\column{6}{@{}>{}l<{}@{}}\column{7}{@{}>{}l<{}@{}}\column{8}{@{}>{}l<{}@{}}\column{9}{@{}>{}l<{}@{}}\column{E}{@{}>{}l<{}@{}}%
\>[1]{}{\mathbf{data}\mskip 3.0mu\mathsf{Merge}\mskip 3.0mu\mathsf{k}\mskip 3.0mu}\>[4]{}{\mathsf{a}\mskip 3.0mu\mathsf{xs}\mskip 3.0mu\mathbf{where}}\<[E]{}\\
\>[2]{}{\allowbreak{}\mathnormal{(}\mskip 0.0mu\allowbreak{}\mathnormal{:\kern -3pt +}\allowbreak{}\mskip 0.0mu\mathnormal{)}\allowbreak{}\mskip 3.0mu}\>[3]{}{\mathnormal{::}\mskip 3.0mu}\>[7]{}{\mathsf{FreeCartesian}\mskip 3.0mu\mathsf{k}\mskip 3.0mu}\>[8]{}{\mathsf{a}\mskip 3.0mu\mathsf{x}\mskip 3.0mu\mathnormal{\rightarrow }\mskip 3.0mu\mathsf{Merge}\mskip 3.0mu\mathsf{k}\mskip 3.0mu}\>[9]{}{\mathsf{a}\mskip 3.0mu\mathsf{xs}}\<[E]{}\\
\>[3]{}{\mathnormal{\rightarrow }\mskip 3.0mu\mathsf{Merge}\mskip 3.0mu\mathsf{k}\mskip 3.0mu}\>[6]{}{\mathsf{a}\mskip 3.0mu\allowbreak{}\mathnormal{(}\mskip 0.0mu\mathsf{Cons}\mskip 3.0mu\mathsf{x}\mskip 3.0mu\mathsf{xs}\mskip 0.0mu\mathnormal{)}\allowbreak{}}\<[E]{}\\
\>[2]{}{\mathsf{Nil}\mskip 3.0mu}\>[3]{}{\mathnormal{::}\mskip 3.0mu\mathsf{Merge}\mskip 3.0mu\mathsf{k}\mskip 3.0mu}\>[5]{}{\mathsf{a}\mskip 3.0mu\mathsf{Null}}\<[E]{}\end{parray}}\end{list} 
The output of \ensuremath{\mathsf{reduce}} is a morphism whose target object is the
monoidal product of the above \ensuremath{\mathsf{xᵢ}}. So we need to encode the
product of a list of types, as a type family:

\begin{list}{}{\setlength\leftmargin{1.0em}}\item\relax
\ensuremath{\begin{parray}\column{B}{@{}>{}l<{}@{}}\column[0em]{1}{@{}>{}l<{}@{}}\column[1em]{2}{@{}>{}l<{}@{}}\column{3}{@{}>{}l<{}@{}}\column{E}{@{}>{}l<{}@{}}%
\>[1]{}{\mathbf{type}\mskip 3.0mu\mathsf{family}\mskip 3.0mu\mathsf{Prod}\mskip 3.0mu\allowbreak{}\mathnormal{(}\mskip 0.0mu\mathsf{xs}\mskip 3.0mu\mathnormal{::}\mskip 3.0mu\allowbreak{}\mathnormal{[}\mskip 0.0mu\mathsf{Type}\mskip 0.0mu\mathnormal{]}\allowbreak{}\mskip 0.0mu\mathnormal{)}\allowbreak{}\mskip 3.0mu}\>[3]{}{\mathbf{where}}\<[E]{}\\
\>[2]{}{\mathsf{Prod}\mskip 3.0mu\mathsf{Null}\mskip 3.0mu\mathnormal{=}\mskip 3.0mu\allowbreak{}\mathnormal{(}\mskip 0.0mu\mathnormal{)}\allowbreak{}}\<[E]{}\\
\>[2]{}{\mathsf{Prod}\mskip 3.0mu\allowbreak{}\mathnormal{(}\mskip 0.0mu\mathsf{Cons}\mskip 3.0mu\mathsf{x}\mskip 3.0mu\mathsf{ys}\mskip 0.0mu\mathnormal{)}\allowbreak{}\mskip 3.0mu\mathnormal{=}\mskip 3.0mu\mathsf{x}\mskip 3.0mu\mathnormal{⊗}\mskip 3.0mu\mathsf{Prod}\mskip 3.0mu\mathsf{ys}}\<[E]{}\end{parray}}\end{list} 
As explained above, lists of morphisms will be sorted according to the
lexicographical order on \ensuremath{\mathsf{FreeCartesian}\mskip 3.0mu}\ensuremath{\mathsf{k}\mskip 3.0mu}\ensuremath{\mathsf{a}}. To keep lists
in sorted order, we will need a function to merge them while preserving the order.
Even though this kind of function is entirely standard, our version must
only return the resulting merged list, but it must also keep track of
the permutations and re-associations which it applies.
 Indeed, this permutation is necessary for the purpose of the
algorithm, because overall the meaning of the morphism must remain the
same: the
composition of the sorted merge operation and the permutations is the
identity. Because we do not know, from types only, the ordering of the
resulting list and hence its type, we must quantify existentially
over it. Because Haskell does not support native existential types, we
use a {\sc{}cps} encoding to define the append function:

\begin{list}{}{\setlength\leftmargin{1.0em}}\item\relax
\ensuremath{\begin{parray}\column{B}{@{}>{}l<{}@{}}\column[0em]{1}{@{}>{}l<{}@{}}\column[1em]{2}{@{}>{}l<{}@{}}\column[2em]{3}{@{}>{}l<{}@{}}\column{4}{@{}>{}l<{}@{}}\column{5}{@{}>{}l<{}@{}}\column{6}{@{}>{}l<{}@{}}\column[6em]{7}{@{}>{}l<{}@{}}\column{8}{@{}>{}l<{}@{}}\column{9}{@{}>{}l<{}@{}}\column{10}{@{}>{}l<{}@{}}\column{11}{@{}>{}l<{}@{}}\column{12}{@{}>{}l<{}@{}}\column{13}{@{}>{}l<{}@{}}\column{14}{@{}>{}l<{}@{}}\column{15}{@{}>{}l<{}@{}}\column{16}{@{}>{}l<{}@{}}\column{17}{@{}>{}l<{}@{}}\column{18}{@{}>{}l<{}@{}}\column{19}{@{}>{}l<{}@{}}\column{20}{@{}>{}l<{}@{}}\column{21}{@{}>{}l<{}@{}}\column{E}{@{}>{}l<{}@{}}%
\>[1]{}{\mathsf{appendSorted}\mskip 3.0mu}\>[6]{}{\mathnormal{::}\mskip 3.0mu}\>[19]{}{\mathsf{Merge}\mskip 3.0mu\mathsf{cat}\mskip 3.0mu}\>[20]{}{\mathsf{a}\mskip 3.0mu\mathsf{xs}\mskip 3.0mu\mathnormal{\rightarrow }\mskip 3.0mu\mathsf{Merge}\mskip 3.0mu\mathsf{cat}\mskip 3.0mu}\>[21]{}{\mathsf{a}\mskip 3.0mu\mathsf{ys}\mskip 3.0mu\mathnormal{\rightarrow }}\<[E]{}\\
\>[7]{}{\allowbreak{}\mathnormal{(}\mskip 0.0mu}\>[8]{}{∀\mskip 3.0mu\mathsf{zs}\mskip 1.0mu.\mskip 3.0mu}\>[14]{}{\mathsf{FreeSMC}\mskip 3.0mu\mathsf{cat}\mskip 3.0mu}\>[17]{}{\allowbreak{}\mathnormal{(}\mskip 0.0mu\mathsf{Prod}\mskip 3.0mu\mathsf{zs}\mskip 0.0mu\mathnormal{)}\allowbreak{}}\<[E]{}\\
\>[17]{}{\allowbreak{}\mathnormal{(}\mskip 0.0mu\mathsf{Prod}\mskip 3.0mu\mathsf{xs}\mskip 3.0mu\mathnormal{⊗}\mskip 3.0mu\mathsf{Prod}\mskip 3.0mu\mathsf{ys}\mskip 0.0mu\mathnormal{)}\allowbreak{}\mskip 3.0mu}\>[18]{}{\mathnormal{\rightarrow }}\<[E]{}\\
\>[8]{}{\mathsf{Merge}\mskip 3.0mu\mathsf{cat}\mskip 3.0mu}\>[12]{}{\mathsf{a}\mskip 3.0mu\mathsf{zs}\mskip 3.0mu\mathnormal{\rightarrow }\mskip 3.0mu\mathsf{k}\mskip 0.0mu\mathnormal{)}\allowbreak{}\mskip 3.0mu\mathnormal{\rightarrow }\mskip 3.0mu\mathsf{k}}\<[E]{}\\
\>[1]{}{\mathsf{appendSorted}\mskip 3.0mu}\>[6]{}{\mathsf{Nil}\mskip 3.0mu}\>[9]{}{\mathsf{ys}\mskip 3.0mu}\>[10]{}{\mathsf{k}\mskip 3.0mu\mathnormal{=}\mskip 3.0mu\mathsf{k}\mskip 3.0mu\allowbreak{}\mathnormal{(}\mskip 0.0muσ\mskip 3.0mu\allowbreak{}\mathnormal{∘}\allowbreak{}\mskip 3.0mu}\>[13]{}{ρ\mskip 0.0mu\mathnormal{)}\allowbreak{}\mskip 3.0mu}\>[16]{}{\mathsf{ys}}\<[E]{}\\
\>[1]{}{\mathsf{appendSorted}\mskip 3.0mu}\>[6]{}{\mathsf{xs}\mskip 3.0mu}\>[9]{}{\mathsf{Nil}\mskip 3.0mu}\>[10]{}{\mathsf{k}\mskip 3.0mu\mathnormal{=}\mskip 3.0mu\mathsf{k}\mskip 3.0mu}\>[13]{}{ρ\mskip 3.0mu}\>[16]{}{\mathsf{xs}}\<[E]{}\\
\>[1]{}{\mathsf{appendSorted}\mskip 3.0mu}\>[6]{}{\allowbreak{}\mathnormal{(}\mskip 0.0mu\mathsf{x}\mskip 3.0mu\allowbreak{}\mathnormal{:\kern -3pt +}\allowbreak{}\mskip 3.0mu\mathsf{xs}\mskip 0.0mu\mathnormal{)}\allowbreak{}\mskip 3.0mu}\>[9]{}{\allowbreak{}\mathnormal{(}\mskip 0.0mu\mathsf{y}\mskip 3.0mu\allowbreak{}\mathnormal{:\kern -3pt +}\allowbreak{}\mskip 3.0mu\mathsf{ys}\mskip 0.0mu\mathnormal{)}\allowbreak{}\mskip 3.0mu}\>[10]{}{\mathsf{k}\mskip 3.0mu\mathnormal{=}}\<[E]{}\\
\>[2]{}{\mathbf{case}\mskip 3.0mu\mathsf{compareMorphisms}\mskip 3.0mu\mathsf{x}\mskip 3.0mu\mathsf{y}\mskip 3.0mu\mathbf{of}}\<[E]{}\\
\>[3]{}{\mathsf{GT}\mskip 3.0mu}\>[4]{}{\mathnormal{\rightarrow }\mskip 3.0mu}\>[5]{}{\mathsf{appendSorted}\mskip 3.0mu\allowbreak{}\mathnormal{(}\mskip 0.0mu\mathsf{x}\mskip 3.0mu\allowbreak{}\mathnormal{:\kern -3pt +}\allowbreak{}\mskip 3.0mu\mathsf{xs}\mskip 0.0mu\mathnormal{)}\allowbreak{}\mskip 3.0mu\mathsf{ys}\mskip 3.0mu\mathnormal{\$}\mskip 3.0muλ\mskip 3.0mu\mathsf{a}\mskip 3.0mu\mathsf{zs}\mskip 3.0mu\mathnormal{\rightarrow }}\<[E]{}\\
\>[5]{}{\mathsf{k}\mskip 3.0mu\allowbreak{}\mathnormal{(}\mskip 0.0muα\mskip 3.0mu\allowbreak{}\mathnormal{∘}\allowbreak{}\mskip 3.0mu\allowbreak{}\mathnormal{(}\mskip 0.0muσ\mskip 3.0mu\allowbreak{}\mathnormal{×}\allowbreak{}\mskip 3.0mu\mathsf{id}\mskip 0.0mu\mathnormal{)}\allowbreak{}\mskip 3.0mu\allowbreak{}\mathnormal{∘}\allowbreak{}\mskip 3.0mu}\>[11]{}{\bar{α}\mskip 3.0mu\allowbreak{}\mathnormal{∘}\allowbreak{}\mskip 3.0mu\allowbreak{}\mathnormal{(}\mskip 0.0mu\mathsf{id}\mskip 3.0mu\allowbreak{}\mathnormal{×}\allowbreak{}\mskip 3.0mu\mathsf{a}\mskip 0.0mu\mathnormal{)}\allowbreak{}\mskip 0.0mu\mathnormal{)}\allowbreak{}\mskip 3.0mu}\>[15]{}{\allowbreak{}\mathnormal{(}\mskip 0.0mu\mathsf{y}\mskip 3.0mu\allowbreak{}\mathnormal{:\kern -3pt +}\allowbreak{}\mskip 3.0mu\mathsf{zs}\mskip 0.0mu\mathnormal{)}\allowbreak{}}\<[E]{}\\
\>[3]{}{\mathnormal{\_}\mskip 3.0mu}\>[4]{}{\mathnormal{\rightarrow }\mskip 3.0mu}\>[5]{}{\mathsf{appendSorted}\mskip 3.0mu\mathsf{xs}\mskip 3.0mu\allowbreak{}\mathnormal{(}\mskip 0.0mu\mathsf{y}\mskip 3.0mu\allowbreak{}\mathnormal{:\kern -3pt +}\allowbreak{}\mskip 3.0mu\mathsf{ys}\mskip 0.0mu\mathnormal{)}\allowbreak{}\mskip 3.0mu\mathnormal{\$}\mskip 3.0muλ\mskip 3.0mu\mathsf{a}\mskip 3.0mu\mathsf{zs}\mskip 3.0mu\mathnormal{\rightarrow }}\<[E]{}\\
\>[5]{}{\mathsf{k}\mskip 3.0mu\allowbreak{}\mathnormal{(}\mskip 0.0mu}\>[11]{}{\bar{α}\mskip 3.0mu\allowbreak{}\mathnormal{∘}\allowbreak{}\mskip 3.0mu\allowbreak{}\mathnormal{(}\mskip 0.0mu\mathsf{id}\mskip 3.0mu\allowbreak{}\mathnormal{×}\allowbreak{}\mskip 3.0mu\mathsf{a}\mskip 0.0mu\mathnormal{)}\allowbreak{}\mskip 0.0mu\mathnormal{)}\allowbreak{}\mskip 3.0mu}\>[15]{}{\allowbreak{}\mathnormal{(}\mskip 0.0mu\mathsf{x}\mskip 3.0mu\allowbreak{}\mathnormal{:\kern -3pt +}\allowbreak{}\mskip 3.0mu\mathsf{zs}\mskip 0.0mu\mathnormal{)}\allowbreak{}}\<[E]{}\end{parray}}\end{list} 
Like \ensuremath{\mathsf{appendSorted}}, the rest of the functions must record permutations
which they might apply, and thus are written in the same style, with
existentials encoded in {\sc{}cps}. In fact, when the input sorted list of
morphisms is empty, the accumulated permutation contains the result
morphism in \ensuremath{\mathsf{FreeSMC}} form.

The purpose of the next function is to expose forks (▵) as a sorted
list of morphisms to merge (of type \ensuremath{\mathsf{Merge}}). Additionally, it
shifts embedded morphisms (and \ensuremath{ε}, which when merged is a
no-op) to the accumulated result.

\begin{list}{}{\setlength\leftmargin{1.0em}}\item\relax
\ensuremath{\begin{parray}\column{B}{@{}>{}l<{}@{}}\column[0em]{1}{@{}>{}l<{}@{}}\column{2}{@{}>{}l<{}@{}}\column[2em]{3}{@{}>{}l<{}@{}}\column{4}{@{}>{}l<{}@{}}\column{5}{@{}>{}l<{}@{}}\column{6}{@{}>{}l<{}@{}}\column{7}{@{}>{}l<{}@{}}\column{8}{@{}>{}l<{}@{}}\column{9}{@{}>{}l<{}@{}}\column{10}{@{}>{}l<{}@{}}\column{11}{@{}>{}l<{}@{}}\column{E}{@{}>{}l<{}@{}}%
\>[1]{}{\mathsf{expose}\mskip 3.0mu}\>[2]{}{\mathnormal{::}\mskip 3.0mu}\>[10]{}{\mathsf{Cat}\mskip 3.0mu\mathsf{cat}\mskip 3.0mu}\>[11]{}{\mathsf{a}\mskip 3.0mu\mathsf{b}\mskip 3.0mu\mathnormal{\rightarrow }}\<[E]{}\\
\>[3]{}{\allowbreak{}\mathnormal{(}\mskip 0.0mu}\>[4]{}{∀\mskip 3.0mu\mathsf{x}\mskip 1.0mu.\mskip 3.0mu}\>[8]{}{\mathsf{FreeSMC}\mskip 3.0mu\mathsf{cat}\mskip 3.0mu}\>[9]{}{\allowbreak{}\mathnormal{(}\mskip 0.0mu\mathsf{Prod}\mskip 3.0mu\mathsf{x}\mskip 0.0mu\mathnormal{)}\allowbreak{}\mskip 3.0mu\mathsf{b}\mskip 3.0mu\mathnormal{\rightarrow }}\<[E]{}\\
\>[4]{}{\mathsf{Merge}\mskip 3.0mu\mathsf{cat}\mskip 3.0mu}\>[7]{}{\mathsf{a}\mskip 3.0mu\mathsf{x}\mskip 3.0mu\mathnormal{\rightarrow }\mskip 3.0mu\mathsf{k}\mskip 0.0mu\mathnormal{)}\allowbreak{}\mskip 3.0mu\mathnormal{\rightarrow }\mskip 3.0mu\mathsf{k}}\<[E]{}\\
\>[1]{}{\mathsf{expose}\mskip 3.0mu\allowbreak{}\mathnormal{(}\mskip 0.0mu\mathsf{f}_{1}\mskip 3.0mu\mathnormal{:\kern -3pt ▵ \kern -3pt:}\mskip 3.0mu\mathsf{f}_{2}\mskip 0.0mu\mathnormal{)}\allowbreak{}\mskip 3.0mu\mathsf{k}\mskip 3.0mu}\>[5]{}{\mathnormal{=}\mskip 3.0mu}\>[6]{}{\mathsf{expose}\mskip 3.0mu\mathsf{f}_{1}\mskip 3.0mu\mathnormal{\$}\mskip 3.0muλ\mskip 3.0mu\mathsf{g}_{1}\mskip 3.0mu\mathsf{fs}_{1}\mskip 3.0mu\mathnormal{\rightarrow }}\<[E]{}\\
\>[6]{}{\mathsf{expose}\mskip 3.0mu\mathsf{f}_{2}\mskip 3.0mu\mathnormal{\$}\mskip 3.0muλ\mskip 3.0mu\mathsf{g}_{2}\mskip 3.0mu\mathsf{fs}_{2}\mskip 3.0mu\mathnormal{\rightarrow }}\<[E]{}\\
\>[6]{}{\mathsf{appendSorted}\mskip 3.0mu\mathsf{fs}_{1}\mskip 3.0mu\mathsf{fs}_{2}\mskip 3.0mu\mathnormal{\$}\mskip 3.0muλ\mskip 3.0mu\mathsf{g}\mskip 3.0mu\mathsf{fs}\mskip 3.0mu\mathnormal{\rightarrow }}\<[E]{}\\
\>[6]{}{\mathsf{k}\mskip 3.0mu\allowbreak{}\mathnormal{(}\mskip 0.0mu\allowbreak{}\mathnormal{(}\mskip 0.0mu\mathsf{g}_{1}\mskip 3.0mu\allowbreak{}\mathnormal{×}\allowbreak{}\mskip 3.0mu\mathsf{g}_{2}\mskip 0.0mu\mathnormal{)}\allowbreak{}\mskip 3.0mu\allowbreak{}\mathnormal{∘}\allowbreak{}\mskip 3.0mu\mathsf{g}\mskip 0.0mu\mathnormal{)}\allowbreak{}\mskip 3.0mu\mathsf{fs}}\<[E]{}\\
\>[1]{}{\mathsf{expose}\mskip 3.0mu\allowbreak{}\mathnormal{(}\mskip 0.0mu\mathsf{Embed}\mskip 3.0mu\mathsf{ϕ}\mskip 3.0mu\mathsf{:\ensuremath{<}:}\mskip 3.0mu\mathsf{f}\mskip 0.0mu\mathnormal{)}\allowbreak{}\mskip 3.0mu\mathsf{k}\mskip 3.0mu}\>[5]{}{\mathnormal{=}\mskip 3.0mu}\>[6]{}{\mathsf{expose}\mskip 3.0mu\mathsf{f}\mskip 3.0mu\mathnormal{\$}\mskip 3.0muλ\mskip 3.0mu\mathsf{g}\mskip 3.0mu\mathsf{fs}\mskip 3.0mu\mathnormal{\rightarrow }}\<[E]{}\\
\>[6]{}{\mathsf{k}\mskip 3.0mu\allowbreak{}\mathnormal{(}\mskip 0.0mu\mathsf{FreeSMC.Embed}\mskip 3.0mu\mathsf{ϕ}\mskip 3.0mu\allowbreak{}\mathnormal{∘}\allowbreak{}\mskip 3.0mu\mathsf{g}\mskip 0.0mu\mathnormal{)}\allowbreak{}\mskip 3.0mu\mathsf{fs}}\<[E]{}\\
\>[1]{}{\mathsf{expose}\mskip 3.0mu\allowbreak{}\mathnormal{(}\mskip 0.0mu\mathsf{E}\mskip 3.0mu\mathsf{:\ensuremath{<}:}\mskip 3.0mu\mathnormal{\_}\mskip 0.0mu\mathnormal{)}\allowbreak{}\mskip 3.0mu\mathsf{k}\mskip 3.0mu}\>[5]{}{\mathnormal{=}\mskip 3.0mu\mathsf{k}\mskip 3.0mu\mathsf{id}\mskip 3.0mu\mathsf{Nil}}\<[E]{}\\
\>[1]{}{\mathsf{expose}\mskip 3.0mu\mathsf{x}\mskip 3.0mu\mathsf{k}\mskip 3.0mu}\>[5]{}{\mathnormal{=}\mskip 3.0mu\mathsf{k}\mskip 3.0mu\bar{ρ}\mskip 3.0mu\allowbreak{}\mathnormal{(}\mskip 0.0mu\mathsf{x}\mskip 3.0mu\allowbreak{}\mathnormal{:\kern -3pt +}\allowbreak{}\mskip 3.0mu\mathsf{Nil}\mskip 0.0mu\mathnormal{)}\allowbreak{}}\<[E]{}\end{parray}}\end{list} 
To finish we show the code to undo a \ensuremath{\mathsf{split}}. Even though it is
somewhat obscured by the necessary accumulation of result morphisms,
its purpose is simple: searching the sorted list for a pair \ensuremath{\mathsf{π₁}\mskip 3.0mu}\ensuremath{\allowbreak{}\mathnormal{∘}\allowbreak{}\mskip 3.0mu}\ensuremath{\mathsf{f}} and \ensuremath{\mathsf{π₂}\mskip 3.0mu}\ensuremath{\allowbreak{}\mathnormal{∘}\allowbreak{}\mskip 3.0mu}\ensuremath{\mathsf{f}} and apply the appropriate reduction.

\begin{list}{}{\setlength\leftmargin{1.0em}}\item\relax
\ensuremath{\begin{parray}\column{B}{@{}>{}l<{}@{}}\column[0em]{1}{@{}>{}l<{}@{}}\column[1em]{2}{@{}>{}l<{}@{}}\column{3}{@{}>{}l<{}@{}}\column[3em]{4}{@{}>{}l<{}@{}}\column[4em]{5}{@{}>{}l<{}@{}}\column{6}{@{}>{}l<{}@{}}\column{7}{@{}>{}l<{}@{}}\column{8}{@{}>{}l<{}@{}}\column{9}{@{}>{}l<{}@{}}\column{10}{@{}>{}l<{}@{}}\column{11}{@{}>{}l<{}@{}}\column{12}{@{}>{}l<{}@{}}\column{13}{@{}>{}l<{}@{}}\column{14}{@{}>{}l<{}@{}}\column{15}{@{}>{}l<{}@{}}\column{E}{@{}>{}l<{}@{}}%
\>[1]{}{\mathsf{reduceStep}\mskip 3.0mu}\>[3]{}{\mathnormal{::}\mskip 3.0mu}\>[14]{}{\mathsf{Merge}\mskip 3.0mu\mathsf{cat}\mskip 3.0mu}\>[15]{}{\mathsf{a}\mskip 3.0mu\mathsf{xs}\mskip 3.0mu\mathnormal{\rightarrow }}\<[E]{}\\
\>[4]{}{\allowbreak{}\mathnormal{(}\mskip 0.0mu}\>[6]{}{∀\mskip 3.0mu\mathsf{zs}\mskip 1.0mu.\mskip 3.0mu}\>[11]{}{\mathsf{FreeSMC}\mskip 3.0mu\mathsf{cat}\mskip 3.0mu}\>[12]{}{\allowbreak{}\mathnormal{(}\mskip 0.0mu\mathsf{Prod}\mskip 3.0mu\mathsf{zs}\mskip 0.0mu\mathnormal{)}\allowbreak{}\mskip 3.0mu\allowbreak{}\mathnormal{(}\mskip 0.0mu\mathsf{Prod}\mskip 3.0mu\mathsf{xs}\mskip 0.0mu\mathnormal{)}\allowbreak{}\mskip 3.0mu}\>[13]{}{\mathnormal{\rightarrow }}\<[E]{}\\
\>[5]{}{\mathsf{Merge}\mskip 3.0mu\mathsf{cat}\mskip 3.0mu}\>[10]{}{\mathsf{a}\mskip 3.0mu\mathsf{zs}\mskip 3.0mu\mathnormal{\rightarrow }\mskip 3.0mu\mathsf{k}\mskip 0.0mu\mathnormal{)}\allowbreak{}\mskip 3.0mu\mathnormal{\rightarrow }\mskip 3.0mu\mathsf{k}}\<[E]{}\\
\>[1]{}{\mathsf{reduceStep}\mskip 3.0mu\allowbreak{}\mathnormal{(}\mskip 0.0mu\allowbreak{}\mathnormal{(}\mskip 0.0mu\mathsf{P}_{1}\mskip 3.0mu\mathsf{:\ensuremath{<}:}\mskip 3.0mu\mathsf{f₁}\mskip 0.0mu\mathnormal{)}\allowbreak{}\mskip 3.0mu\allowbreak{}\mathnormal{:\kern -3pt +}\allowbreak{}\mskip 3.0mu\allowbreak{}\mathnormal{(}\mskip 0.0mu\mathsf{P}_{2}\mskip 3.0mu\mathsf{:\ensuremath{<}:}\mskip 3.0mu\mathsf{f₂}\mskip 0.0mu\mathnormal{)}\allowbreak{}\mskip 3.0mu\allowbreak{}\mathnormal{:\kern -3pt +}\allowbreak{}\mskip 3.0mu\mathsf{rest}\mskip 0.0mu\mathnormal{)}\allowbreak{}\mskip 3.0mu\mathsf{k}}\<[E]{}\\
\>[2]{}{\mathnormal{|}\mskip 3.0mu\mathsf{EQ}\mskip 3.0mu\mathnormal{\leftarrow }\mskip 3.0mu\mathsf{compareMorphisms}\mskip 3.0mu\mathsf{f₁}\mskip 3.0mu\mathsf{f₂}\mskip 3.0mu\mathnormal{=}}\<[E]{}\\
\>[2]{}{\mathsf{expose}\mskip 3.0mu\mathsf{f₁}\mskip 3.0mu}\>[7]{}{\mathnormal{\$}\mskip 3.0muλ\mskip 3.0mu\mathsf{g}\mskip 3.0mu\mathsf{f'}\mskip 3.0mu\mathnormal{\rightarrow }}\<[E]{}\\
\>[2]{}{\mathsf{appendSorted}\mskip 3.0mu\mathsf{f'}\mskip 3.0mu\mathsf{rest}\mskip 3.0mu}\>[7]{}{\mathnormal{\$}\mskip 3.0muλ\mskip 3.0mu\mathsf{g'}\mskip 3.0mu\mathsf{rest'}\mskip 3.0mu\mathnormal{\rightarrow }}\<[E]{}\\
\>[2]{}{\mathsf{k}\mskip 3.0mu\allowbreak{}\mathnormal{(}\mskip 0.0muα\mskip 3.0mu\allowbreak{}\mathnormal{∘}\allowbreak{}\mskip 3.0mu\allowbreak{}\mathnormal{(}\mskip 0.0mu\mathsf{g}\mskip 3.0mu\allowbreak{}\mathnormal{×}\allowbreak{}\mskip 3.0mu\mathsf{id}\mskip 0.0mu\mathnormal{)}\allowbreak{}\mskip 3.0mu\allowbreak{}\mathnormal{∘}\allowbreak{}\mskip 3.0mu\mathsf{g'}\mskip 0.0mu\mathnormal{)}\allowbreak{}\mskip 3.0mu\mathsf{rest'}}\<[E]{}\\
\>[1]{}{\mathsf{reduceStep}\mskip 3.0mu\allowbreak{}\mathnormal{(}\mskip 0.0mu\mathsf{f}\mskip 3.0mu\allowbreak{}\mathnormal{:\kern -3pt +}\allowbreak{}\mskip 3.0mu\mathsf{rest}\mskip 0.0mu\mathnormal{)}\allowbreak{}\mskip 3.0mu\mathsf{k}\mskip 3.0mu\mathnormal{=}}\<[E]{}\\
\>[2]{}{\mathsf{reduceStep}\mskip 3.0mu\mathsf{rest}\mskip 3.0mu}\>[9]{}{\mathnormal{\$}\mskip 3.0muλ\mskip 3.0mu\mathsf{g}\mskip 3.0mu\mathsf{rest'}\mskip 3.0mu\mathnormal{\rightarrow }}\<[E]{}\\
\>[2]{}{\mathsf{appendSorted}\mskip 3.0mu\allowbreak{}\mathnormal{(}\mskip 0.0mu\mathsf{f}\mskip 3.0mu\allowbreak{}\mathnormal{:\kern -3pt +}\allowbreak{}\mskip 3.0mu\mathsf{Nil}\mskip 0.0mu\mathnormal{)}\allowbreak{}\mskip 3.0mu\mathsf{rest'}\mskip 3.0mu}\>[8]{}{\mathnormal{\$}\mskip 3.0muλ\mskip 3.0mu\mathsf{g'}\mskip 3.0mu\mathsf{rest''}\mskip 3.0mu\mathnormal{\rightarrow }}\<[E]{}\\
\>[2]{}{\mathsf{k}\mskip 3.0mu\allowbreak{}\mathnormal{(}\mskip 0.0mu\allowbreak{}\mathnormal{(}\mskip 0.0mu\bar{ρ}\mskip 3.0mu\allowbreak{}\mathnormal{×}\allowbreak{}\mskip 3.0mu\mathsf{g}\mskip 0.0mu\mathnormal{)}\allowbreak{}\mskip 3.0mu\allowbreak{}\mathnormal{∘}\allowbreak{}\mskip 3.0mu\mathsf{g'}\mskip 0.0mu\mathnormal{)}\allowbreak{}\mskip 3.0mu\mathsf{rest''}}\<[E]{}\end{parray}}\end{list}

\section{Discussion and Related work}\label{139} 
\subsection{Dynamically checking for linearity}\label{140} 
Could we implement a variant of our {\sc{}api} and implementation which performs linearity
checks at runtime, rather than relying on Haskell to perform
them? This sounds reasonable: after all we already construct a representation
of decoded morphisms, and we can run a protolinearity (\cref{68}) check on it.
This could be done, but only if generators are equipped with a
decidable equality. Indeed, consider the morphism \ensuremath{\allowbreak{}\mathnormal{(}\mskip 0.0mu}\ensuremath{π₁\mskip 3.0mu}\ensuremath{\allowbreak{}\mathnormal{∘}\allowbreak{}\mskip 3.0mu}\ensuremath{\mathsf{φ}\mskip 0.0mu}\ensuremath{\mathnormal{)}\allowbreak{}\mskip 3.0mu}\ensuremath{\mathnormal{▵}\mskip 3.0mu}\ensuremath{\allowbreak{}\mathnormal{(}\mskip 0.0mu}\ensuremath{π₂\mskip 3.0mu}\ensuremath{\allowbreak{}\mathnormal{∘}\allowbreak{}\mskip 3.0mu}\ensuremath{\mathsf{ψ}\mskip 0.0mu}\ensuremath{\mathnormal{)}\allowbreak{}}. It is protolinear if φ=ψ, but not otherwise. The
tradeoff is simple to express: one either needs static linearity checks or a
dynamic equality check on generators (but not both). However, if one
would choose dynamic equality checks, it may be more sensible to
evaluate to cartesian categories instead, as discussed in \cref{141}.

\subsection{Evaluating to Cartesian Categories}\label{141} 
Our technique can be adapted to cartesian structures (instead of
monoidal symmetric ones).  To do so one shall 1. retain the encoding
of ports as morphisms from an abstract object \ensuremath{\mathsf{r}}: \ensuremath{\mathsf{P}\mskip 3.0mu}\ensuremath{\mathsf{k}\mskip 3.0mu}\ensuremath{\mathsf{r}\mskip 3.0mu}\ensuremath{\mathsf{a}\mskip 3.0mu}\ensuremath{\mathnormal{=}\mskip 3.0mu}\ensuremath{\mathsf{FreeCartesian}\mskip 3.0mu}\ensuremath{\mathsf{k}\mskip 3.0mu}\ensuremath{\mathsf{r}\mskip 3.0mu}\ensuremath{\mathsf{a}}, 2. relax the requirement to work with
linear functions, and 3. drop the projection from free cartesian to
free monoidal structures in the implementation of \ensuremath{\mathsf{decode}}. We must however underline that such a
technique places \ensuremath{\mathsf{δ}} at the earliest points in the morphisms, thus
generators are duplicated every time their output is
split. This behaviour follows cartesian laws to the letter: indeed
they stipulate that such duplication has no effect: \ensuremath{\allowbreak{}\mathnormal{(}\mskip 0.0mu}\ensuremath{\mathsf{f}\mskip 3.0mu}\ensuremath{\allowbreak{}\mathnormal{×}\allowbreak{}\mskip 3.0mu}\ensuremath{\mathsf{f}\mskip 0.0mu}\ensuremath{\mathnormal{)}\allowbreak{}\mskip 3.0mu}\ensuremath{\allowbreak{}\mathnormal{∘}\allowbreak{}\mskip 3.0mu}\ensuremath{δ\mskip 3.0mu}\ensuremath{\mathnormal{=}\mskip 3.0mu}\ensuremath{δ\mskip 3.0mu}\ensuremath{\allowbreak{}\mathnormal{∘}\allowbreak{}\mskip 3.0mu}\ensuremath{\mathsf{f}}. However, categories which make both sides of the
above equation equivalent in all respects are rare. For example the
presence of effects tend to break the property. For example,
a Kleisli category is cartesian only if the embedded effects are
commutative and idempotent.
In particular if one takes runtime costs into account,
the equivalence vanishes. Worse, in the presence of other
optimisations, one cannot tell \textit{a priori} which side of the
equation has the lowest cost: it may be beneficial to have a single
instance of \ensuremath{\mathsf{f}} so that work is not duplicated, but it may just
as well be more beneficial to have two instances, so that for example
they can fuse with whatever follows in the computation. Indeed if the
output of \ensuremath{\mathsf{f}} is large, following it with \ensuremath{\mathsf{δ}} may require storing
(parts of) it, whereas each copy of \ensuremath{\mathsf{f}} may be followed by a
function which only require \ensuremath{\mathsf{f}} to be ran lazily, not requiring any storage. In
general, programmers must decide for themselves if it is best to place
\ensuremath{\mathsf{f}} before or after \ensuremath{\mathsf{δ}}. Thus the {\sc{}smc} approach, which we follow, is to ask the
programmer to place \ensuremath{\mathsf{δ}} \emph{explicitly}, using \ensuremath{\mathsf{copy}} from \cref{33}.\footnote{Indeed, duplications which we consider here are coming from
user code, and they are disjoint from those that we insert in the
intermediate free cartesian representations discussed in
\cref{60}. In fact, there is no interaction between
the two.} We regard this approach to be the most appropriate
when there is a significant difference between the left-
and the right-hand side of the above equation. Another possibility
would be to use a decidable check over generators (as discussed in
\cref{140}) and enforce that all applications of a generator
to the same input are realised as a single generator in the
resulting morphism.  For example, one can use identity in the
source code (of the host language) as generator equality. Then, each
occurrence in the source is mapped one-to-one with its occurrences in
the representation as a (cartesian) morphism.  This sort of
source-code identity is available when one has access to the
representation of the source code (see \cref{143}), or by using
any approach to observable sharing (see \cref{142}).

\subsection{Observable Sharing}\label{142} 
One way to recover representations from embedded {\sc{}dsl}s is to leverage
observable sharing techniques. \citet{gill_type-safe_2009} provides a
review of the possible approaches, but in short, one uses unique names
equipped with testable equality for what we call here generators.
Explicit unique names can be provided directly by the programmer, or
generated using a state monad.  Alternatively, testable equality can be
implemented by pointer equality. The version of
\citet{claessen_observable_1999} is native, but it breaks referential
transparency. The version of \citet{peyton_jones_stretching_1999} preserves referential transparency, but resides in the catch-all IO
monad.

Turning the problem on its head, we can see our approach as a
principled solution to observable sharing. Essentially, forcing the
programmer to be explicit about duplication means that no implicit
sharing needs to be recovered, and therefore {\sc{}edsl} backends (such as those
presented in \cref{36}) need not deal with it.

\subsection{Compiling to Categories}\label{143} 
\citet{elliott_compiling_2017} advertises a compiler plugin which
translates a source code representation to a categorical
representation. This plugin is close in purpose to what we propose here. The first
obvious difference is that our solution is entirely programmed within
Haskell, while Elliott's acts at the level of the
compiler. This makes our approach much less tied to a particular
implementation, and we even expect it to be portable to other
languages with linear types. In return, it demands paying an extra
cost at runtime.

There are more fundamental differences however: because the
input of the plugin is Haskell source code, it is forced to target cartesian
\ensuremath{\mathsf{closed}} categories, even though most of Elliott's applications
naturally reside at the simple cartesian level (keeping in mind the
\textit{caveat} discussed in \cref{141}). This forces one to
provide cartesian closed instances for all applications, or add a
translation layer from cartesian closed to just cartesian categories. Our approach
avoids any of those complications, but
\citet{valliappan_exponential_2019} provide a detailed study of the alternative.

In fact, the present work has much synergy with Elliott's: all his examples
are supported by our technique, out of the box, and we recommend
consulting them for a broader view of the applications of categorical
approaches. Accordingly, in \cref{36} we have focused on the stones left unturned
by Elliott. In
particular the
applications to quantum gates is out of reach when one targets
cartesian closed categories.

\subsection{A practical type theory for {\sc{}smc}s}\label{144} 
\citet{shulman_type_theory_for_smc_2019} proposes a type theory for
{\sc{}smc}s. The motivation is different than ours: Shulman wants to provide
set-like reasoning on {\sc{}smc} morphisms to mathematicians, whereas we are
providing a notation to describe {\sc{}smc} morphisms in a programming
language. Shulman's is a dedicated language while ours is embedded in
Haskell. The means are rather different too: Shulman's
type theory doesn't require giving a meaning to ports (Shulman calls
ports ``terms'') instead the semantic is given globally over an entire
judgement. We give a local semantics by giving a meaning to ports. This
difference follows from our implementing the port interface within
Haskell, rather than using metaprogramming, as described in
\cref{143}.

Nevertheless, the end product, as far as the user is concerned, is
pretty similar: one writes expressions on ports that one then needs to
combine together to form a legal expression. This convergence suggests
that there may be value in a further investigation of the mathematical structure
of ports.

\subsection{Quantification over \ensuremath{\mathsf{r}} }\label{145} 
Another minor possible improvement in the {\sc{}api} (\cref{33}) would be
to remove the variable \ensuremath{\mathsf{r}} in the type \ensuremath{\mathsf{P}\mskip 3.0mu}\ensuremath{\mathsf{k}\mskip 3.0mu}\ensuremath{\mathsf{r}\mskip 3.0mu}\ensuremath{\mathsf{a}}. Such
locally quantified variables are used to ensure that two instances of
an {\sc{}edsl} are not mixed together.
For instance, \citet{launchbury_lazy_1994} use them to capture the identity of state threads.
However, in our case, this role is already fulfilled by the
use of linear types. Indeed, if an initial value type \ensuremath{\mathsf{P}\mskip 3.0mu}\ensuremath{\mathsf{k}\mskip 3.0mu}\ensuremath{\mathsf{r}\mskip 3.0mu}\ensuremath{\mathsf{a}} is introduced
by an instance of \ensuremath{\mathsf{decode}}, linearity checks already prevent it from occuring free
in another instance of \ensuremath{\mathsf{decode}}. The same property holds for values derived from it (using \ensuremath{\mathsf{split}}, \ensuremath{\mathsf{merge}}, etc.).
We leave a proof of this fact to future work.
Besides, even though the implementation does not strictly need the variable \ensuremath{\mathsf{r}}, our proof of its correctness does,
therefore we have not explored this route further.

\section{Conclusion}\label{146} 
When defining an embedded domain specific language, there is often a
tension between making the syntax ({\sc{}api}) convenient for the user, and
making the implementation simple. In particular, how to compose
objects is an important choice in the design space. Using explicit
names for intermediate computations is often most convenient for the
user, but can be hard to support by the implementation. In this paper
we have shown a way to bridge the gap between the convenience of lambda
notation on the user-facing side with the convenience of categorical
combinators on the implementation side. The price to pay is
linearity: the user must make duplication and discarding of values explicit.

Indeed, our technique is grounded in the equivalence between symmetric monoidal
categories and linear functions.
While this equivalence is well known, we have pushed the state of
the art by showing that linear functions can \emph{compute their own
representation}  in a symmetric monoidal category.

Our technique has several positive aspects: it is usable in practice
in a wide range of contexts; it is comprised of a small interface and
reasonably short implementation; and it does not depend on any
special-purpose compiler modification, nor on metaprogramming. As such,
in the context of Haskell, it has the potential to displace the arrow
notation as a standard means to represent computations whose static
structure is accessible.

In general, we think that this paper provides suitable means to work
with commutative effects in functional languages.  Commutative effects
are numerous (environment, supply of unique names, random number
generation, etc.), and proper support for them has been recognised as a
challenge for a long time, for example by
\citet[challenge 2, slide 38]{peyton_jones_wearing_2003}.  This paper provides evidence that
{\sc{}smc}s constitute the right abstraction for commutative effects, and that
linear types are key to providing a convenient notation for them.

\begin{acks}
 We warmly thank James Haydon and Georgios Karachalias as well as anonymous reviewers for feedback on previous versions of this paper.
Jean-Philippe Bernardy is supported by grant \grantnum{clasp}{2014-39} from the
\grantsponsor{clasp}{Swedish Research Council}{https://www.vr.se}, which funds the Centre for Linguistic Theory and
Studies in Probability (CLASP) in the Department of Philosophy, Linguistics,
and Theory of Science at the University of Gothenburg.
\end{acks}

\bibliographystyle{ACM-Reference-Format}
\bibliography{PaperTools/bibtex/jp,local}
\end{document}